\documentclass[11pt, a4paper]{scrartcl}
\usepackage[affil-it]{authblk}

\usepackage[T1]{fontenc}
\usepackage[latin9]{inputenc}
\usepackage[title]{appendix}

\usepackage[dvipsnames]{xcolor}

\usepackage[shortlabels]{enumitem}	

\usepackage[algoruled,linesnumbered]{algorithm2e} 

\makeatletter
\let\original@algocf@latexcaption\algocf@latexcaption
\long\def\algocf@latexcaption#1[#2]{%
  \@ifundefined{NR@gettitle}{%
    \def\@currentlabelname{#2}%
  }{%
    \NR@gettitle{#2}%
  }%
  \original@algocf@latexcaption{#1}[{#2}]%
}
\makeatother

\makeatletter
\newcommand{\mylabel}[2]{#2\def\@currentlabel{#2}\label{#1}}
\makeatother

\usepackage{lscape}
\usepackage{etex} 
\usepackage{amsmath}            
\usepackage{amssymb}            
\usepackage{amsthm}             
\usepackage{pst-all}
\usepackage{pstricks-add}
\usepackage{pst-3dplot}
\usepackage{float}
\usepackage{floatflt}
\usepackage{wrapfig}
\usepackage{booktabs,framed}
\newtheorem{theorem}{Theorem}
\newtheorem{lemma}{Lemma}

\newtheorem{conjecture}{Conjecture}
\newtheorem{proposition}{Proposition}

\theoremstyle{definition}
\newtheorem{definition}{Definition}
\usepackage{geometry}
\usepackage[caption=true]{subfig}  
\usepackage[colorlinks=true, linkcolor=blue, citecolor=OliveGreen, urlcolor=darkgray]{hyperref}
\newtheorem*{theorem*}{Theorem}
\renewcommand{\epsilon}{\varepsilon}

\DeclareMathOperator{\ord}{ord}

\newcommand{\Bogen}[4]{\ncarc[arcangle=#4]{-}{#1}{#2}\ncput{\colorbox{almostwhite}{$#3$}}}
\newcommand{\Bogendashed}[4]{\ncarc[arcangle=#4, linestyle=dashed]{-}{#1}{#2}\ncput{\colorbox{almostwhite}{$#3$}}}
\newcommand{\Kante}[3]{\ncline{-}{#1}{#2}\ncput{\colorbox{almostwhite}{$#3$}}}
\newcommand{\Kantedashed}[3]{\ncline[linestyle=dashed]{-}{#1}{#2}\ncput{\colorbox{almostwhite}{$#3$}}}
\newcommand{\Knoten}[3]{\cnode*(#1,#2){2.2pt}{#3}}
\newcommand{\Knotenthick}[3]{\cnode*(#1,#2){3.7pt}{#3}}
\newgray{almostwhite}{0.9}
\newcommand{\Bogendef}[4]{\ncarc[arcangle=#4]{-}{#1}{#2}}

\newcommand{\ptl}{PL}
\newcommand{\lp}{LP($F$)}
\begin{document}

\title{A Characterization of Undirected Graphs Admitting
Optimal Cost Shares}

\author{Tobias Harks}
\author{Anja Huber}
\author{Manuel Surek}
\affil{\small Universit\"at Augsburg, Institut f\"ur Mathematik, 86135 Augsburg\\
\href{mailto:tobias.harks@math.uni-augsburg.de}{\{\texttt{tobias.harks,anja.huber,manuel.surek\}@math.uni-augsburg.de}}}


\date{}

\maketitle

\begin{abstract}
In a seminal paper, Chen, Roughgarden and Valiant~\cite{ChenRV10} studied cost sharing protocols for network design with the objective to implement a low-cost Steiner forest as a Nash equilibrium of an induced cost-sharing game.
One of the most intriguing open problems to date is to understand the power of
budget-balanced and separable cost sharing protocols in order
to induce low-cost Steiner forests. 
In this work, we focus on \emph{undirected} networks and analyze topological properties of the underlying graph so that an \emph{optimal Steiner
forest} can be implemented as a
Nash equilibrium (by some separable cost sharing protocol)
\emph{independent} of the edge costs.
We term a graph  \emph{efficient} if the above stated property holds. As our main result, we
 give a complete characterization of efficient undirected graphs for two-player network design games: an undirected graph is efficient if and only if it does not contain (at least) one out of \emph{few forbidden subgraphs}. 
 Our characterization implies that several graph classes are efficient: generalized series-parallel graphs,
 fan and wheel graphs and graphs with small cycles. \nocite{Koch2004} \nocite{GamrathFischerGallyetal.2016}

\end{abstract}

\section{Introduction}

In the \emph{Steiner forest}
problem, there is a network $(G,c)$ with an undirected graph $G=(V,E)$ and nonnegative edge costs $c(e), e \in E$. Furthermore, there are $n \geq 1$ pairs $(s_1,t_1), \ldots, (s_n,t_n)$ of vertices in $G$ and each such pair $(s_i,t_i)$  needs the vertices $s_i$ and $t_i$ to be connected by (at least one) path. 
Thus, a feasible solution for the Steiner forest problem is a subset $F \subseteq E$ so that each pair $(s_i,t_i)$ is connected in the subgraph induced by $F$. Since edge costs are nonnegative, there are no cycles in any  optimal solution, thus, one can restrict the search to Steiner forests. An optimal Steiner forest $F$ is a Steiner forest with minimum cost, that is $c(F):=\sum_{e \in F}{c(e)}$ is minimal under all possible Steiner forests $F$.

\subsection{Network Cost Sharing Games}
In this article, we consider a game-theoretic variant
of the Steiner forest problem (introduced in Chen, Roughgarden and Valiant~\cite{ChenRV10}) assuming that a system manager can \emph{design} a protocol
that determines how the edge costs of the forest are shared among its users. Formally, the $n$ pairs $(s_i,t_i)$ correspond to players $N:=\{1,\dots,n\}$ that each want to establish an $(s_i,t_i)$ connection
with minimum cost. Thus, a strategy profile is
a tuple $P=(P_1,\dots,P_n)$, where every $P_i$ is an 
$s_i$-$t_i$ path. Given $P$, the cost of player $i$ using edge $e$ is
 $\xi_{i,e}(P)\geq 0$ and the $\xi_{i,e}(P)$-values are determined by a \emph{cost-sharing protocol} $\Xi$.
The total cost that player $i\in N$ needs to pay under $P$ is defined as
\[ \xi_i(P):=\sum_{e\in P_i}\xi_{i,e}(P),\]
and a pure Nash equilibrium of the strategic game induced by $\Xi$
is a strategy profile $P$ from which no player can unilaterally deviate, say to another path $P'_i$, and strictly pay less. Chen et al.~\cite{ChenRV10} axiomatized 
cost sharing protocols by the following
three fundamental properties (see also~\cite{FalkeHarks13,Christodoulou16}):
\begin{enumerate}[itemsep=-0em]
\item\label{enum:bb} \emph{Budget-balance}: The cost $c(e)\geq 0$ of each edge $e$ is
  exactly covered by the collected cost shares of the players using
  the edge, that is, $\sum_{i\in S_e(P)}\xi_{i,e}(P)=c(e)$ for all $e\in E$,
  where $S_e(P):=\{i \in N: e \in P_i\}$.
\item\label{enum:stab} \emph{Stability}: There is at least one pure strategy Nash equilibrium in each game induced by the cost sharing protocol.
\item\label{enum:sep} \emph{Separability}: The cost shares on an edge only depend on the set of players using the edge, that is,
$S_e(P)=S_e(P')\Rightarrow \xi_{i,e}(P)=\xi_{i,e}(P')$ for all $P,P'$ and $e\in E$. 
\end{enumerate}
Budget-balance (Condition~\ref{enum:bb}.) is straightforward, Stability (Condition~\ref{enum:stab}.) requires the existence of at least one
Nash equilibrium in pure strategies (abbreviated PNE).
This requirement is important for applications in which mixed or correlated strategies have no
meaningful physical interpretation (see also the
discussion in Osborne and Rubinstein~\cite[\S~3.2]{Osborne94}).
Separability (Condition~\ref{enum:sep}.)
allows for a distributed implementation of the cost sharing protocol as each edge needs only 
to know its own player set. A cost sharing protocol is called \emph{separable}, if it
satisfies~\ref{enum:bb}.-\ref{enum:sep}.


One important example for a separable cost sharing protocol is the Shapley cost sharing protocol (see~\cite{Anshelevich08,KolliasR15}). For the case of two players, the corresponding PoS is known to be $4/3$, see Figure~\ref{fig1shapleya} for an example. The solid lines build the unique optimal Steiner forest OPT with cost $3+2\epsilon$, but OPT is no PNE since Player 1 has to pay $2+2\epsilon>2+\epsilon$. On the other hand, each player taking her direct $s_i-t_i$-edge is the unique PNE with cost $4+\epsilon$. 

Can we improve the PoS for this example by using a different separable cost sharing protocol? 
Note that for the case of two players, a separable cost sharing protocol is uniquely determined by one value per edge, namely the amount Player 1 has to pay if both players use this edge. In Figure~\ref{fig1shapleyb} we display a cost sharing protocol $\Xi$ for which OPT is a PNE. The edges are labelled by their costs followed by the value described above which determines $\Xi$. 

\begin{figure}[H] \centering \psset{unit=1.3cm}
\subfloat[OPT is no PNE for Shapley]{
 \begin{pspicture}(-0.5,-0.5)(4,1.3) 
\Knoten{0}{0}{s}\nput{270}{s}{$s_1=s_2$}
\Knoten{2.5}{0}{t2}\nput{-45}{t2}{$t_2$}
\Knoten{3.5}{1}{t1}\nput{0}{t1}{$t_1$}
\Kante{s}{t2}{2}
\Kante{t2}{t1}{1+2\epsilon}
\Bogendashed{s}{t1}{2+\epsilon}{30}
\end{pspicture} \label{fig1shapleya}
}
\subfloat[OPT is a PNE for $\Xi$]{
 \begin{pspicture}(-0.5,-0.5)(4,1.3) 
\Knoten{0}{0}{s}\nput{270}{s}{$s_1=s_2$}
\Knoten{2.5}{0}{t2}\nput{-45}{t2}{$t_2$}
\Knoten{3.5}{1}{t1}\nput{0}{t1}{$t_1$}
\Kante{s}{t2}{2|0}
\Kante{t2}{t1}{1+2\epsilon|0}
\Bogendashed{s}{t1}{2+\epsilon|2+\epsilon}{30}
\end{pspicture} \label{fig1shapleyb}
}
\caption{Examples for separable cost sharing protocols} \label{fig1}
\end{figure}
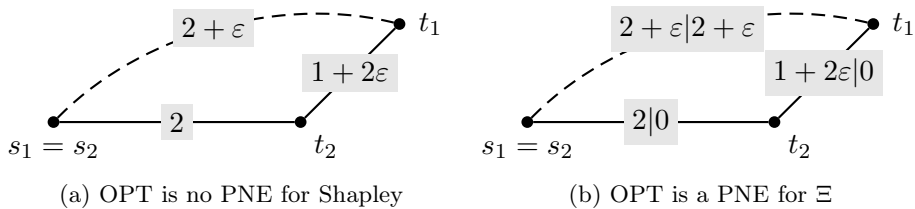



\subsection{Our Results} 
We study \emph{efficient} graphs $G=(V,E)$ having the property that there is an optimal Steiner forest that can be implemented as a pure Nash equilibrium by some separable cost sharing protocol (we speak of an \emph{enforceable} Steiner forest).
The above definition does not specify a priori the cost structure of the graph since any graph
can be made efficient by assigning infinite or very high cost on some edges, thus, deleting ``problematic'' edges and effectively making the combinatorial structure of the graph irrelevant.
An equivalent formulation of the research question we study is the following: what is the largest class of undirected graphs for which the worst-case ratio of the cost of the best Nash equilibrium and that
of an optimal Steiner forest (PoS)  is $1$?
An even stronger condition is the following: $G$ is said to be \emph{strongly} efficient, if \emph{every}
optimal Steiner forest can be enforced as a pure Nash equilibrium.

Our main result gives a complete characterization of efficient and strongly efficient graphs for two-player games:
\begin{theorem*}[Main Result (Informal)]
\[ G \text{ is efficient }\Leftrightarrow G \text{ is strongly efficient }\Leftrightarrow G \text{ does not contain
certain subgraphs (see Figure~\ref{allbc})}.\]
\end{theorem*}

 Some of the forbidden subgraphs are reminiscent to an instance for
 directed network design showing a lower bound of $5/4$ for the PoS, see Chen et al.~\cite{ChenRV10}.
Our characterization implies that
several well-known graph classes are efficient, while for others, we immediately get counterexamples, see Table~\ref{tabclasses} for a (non-exhaustive) overview. The proof of the efficiency of the listed graph classes and the explanation why the other classes contain non-efficient graphs can be found in \autoref{secclasses}.

\begin{table}[H] 
\begin{center}
\begin{tabular}{c c@{\hspace{1cm}} c}
\toprule
efficient classes & & classes containing non-efficient graphs \\ \midrule
generalized series-parallel graphs & & bipartite graphs  \\ \addlinespace[0.25em]
wheel and fan graphs & & chordal graphs \\ \addlinespace[0.25em]
graphs with longest cycle $\leq 6$ & & planar graphs \\
\bottomrule
\end{tabular}
\caption{Efficiency of graph classes}
\label{tabclasses}
\end{center} \end{table}
\vspace{-2.5em}

\subsection{Used Proof Techniques and Significance}
Showing that graphs which contain a forbidden subgraph are not (strongly) efficient is straightforward: 
It suffices to give costs for each forbidden subgraph so that the PoS is greater than $1$.
Here, we can effectively delete edges which are not part of a forbidden subgraph by assigning high costs to them. The property \emph{not strongly efficient} is derived by proving that the optimal Steiner
forest of the used instance is unique.

The reverse direction, that is, showing that every graph which does not contain
a forbidden subgraph (called \emph{bad configuration}) is efficient, is much more involved. 
As a first step we derive a novel LP-characterization of enforceable Steiner forests.
An optimal LP solution (for a given Steiner forest) corresponds to cost shares that are budget balanced
and induce a separable cost sharing protocol so that the Steiner forest
becomes a pure Nash equilibrium. The proof proceeds now by contraposition:
assume we are given a graph without a forbidden subgraph and assume (by contradiction) 
that there is an optimal Steiner forest that is not enforceable. We solve the corresponding LP
for the Steiner forest and since the Steiner forest is not enforceable, there
exists an inequality which is not tight and corresponds to an edge that is not completely paid by the players.
We use this unpaid edge to derive the existence of an alternative strategy (path)
for some player with costs equal to a fraction of the currently paid cost shares (this
alternative strategy corresponds to a tight inequality of the LP).
These alternative paths are now iteratively generated until we can either argue
that there exists a cheaper Steiner forest compared to the initial optimal Steiner forest (contradiction),
or, there is a bad configuration (contradiction).
Along this main approach, however, several additional
ideas are required: the location of the unpaid edge leads to different subcases for which we need to use special optimal LP-solutions in order to 
derive the proper alternative strategies. 

We believe that our approach is a promising step towards
better understanding the power 
of separable cost sharing protocols in general. For the PoS-question in directed or undirected graphs, there
has been no progress since the initial conference version of Chen et al.~\cite{ChenRV10} roughly 10 years ago.
Our characterization and the proof exactly prescribes substructures of a worst-case instance (namely
a bad configuration must exist whose subpaths have costs  corresponding to tight
inequalities of an LP solution). We are confident that our proof technique gives a blue-print for both, characterizing
efficient graphs for 
the general $n$-Player case, and for resolving the PoS-question. 

\subsection{Related Work}
 For the PoA of uniform cost sharing protocols\footnote{
Uniform protocols require that the cost shares on an edge only depend on the edge cost and the set of players, but not on the network itself.}, Chen et al.~\cite{ChenRV10} proved (tight) bounds of 
$2$  for undirected single-sink networks and
$\Theta(\text{polylog}(n))$ for undirected multi-commodity networks. For
directed single-sink networks the achievable PoA is $n$. For the PoS, it is shown
 that single-sink instances (directed or undirected) admit an optimal Steiner forest as PNE (that is the PoS is $1$). For multi-commodity directed network, the achievable PoS lies
 in $[3/2,\log(n)]$ and since the initial work of Chen et al.~\cite{ChenRV10},
 no improvement has been made on this question. For undirected networks, 
the only known upper bounds are derived by analyzing the Shapley cost sharing protocol
and they are of order $O(\log(n))$, see Anshelevich et al.~\cite{Anshelevich08}. Several works improved lower and upper bounds for the PoS of Shapley cost
sharing in undirected networks (cf.~\cite{BiloCFM13,BiloB11,ChristodoulouCLPS09,DisserFKM15,FiatKLOS06,Mamageishvili2014}) but up to day it is open whether  the PoS is of order $\log(n)$ or even in $O(1)$. 
For several special cases,
the price of stability is shown to
be significantly lower (cf.~\cite{FiatKLOS06,LeeL13,LI2009}). Recently Bil\'{o} et al.~\cite{BiloFM13} could show that the PoS for broadcast games is $O(1)$.
For the design of separable cost sharing protocols in undirected networks,
we are not aware of \emph{any} known lower bounds regarding the PoS.

	 Chen and Roughgarden~\cite{Chen09} and Kollias and Roughgarden~\cite{KolliasR15}  focused on network design with weighted players (where Kollias and Roughgarden analyzed this variant as a special case of weighted congestion games) and derived tight bounds on the PoA for the Shapley cost sharing protocol. 
	Gkatzelis, Kollias and Roughgarden~\cite{GKR16}
  further showed that the Shapley cost sharing protocol is optimal among all uniform protocols
  for polynomial and convex cost functions.\footnote{The certificate for optimality uses a characterization
of uniform protocols by Gopalakrishnan et al.~\cite{GopalakrishnanMW14}.}
For further works analyzing the Shapley protocol or arbitrary cost sharing, see \cite{GairingKK15,harkspeis2014resource,Hoefer11,KlimmS15,RoughgardenS16,Gairing2017}.
 Harks and von Falkenhausen~\cite{HarksF14,FalkeHarks13} studied the design of separable cost sharing protocols in a model, where players want to buy a basis  of a matroid. 
They derived tight bounds for the achievable PoS and PoA of order $\log(n)$
and $n$, respectively. Christodoulou and Sgouritsa~\cite{Christodoulou16} considered
multicast cost sharing games under the assumption
that input parameters (such as the set of terminals
and their location in the graph) are not known or only
known probabilistically. Among other results they show constant PoA
bounds for outer planar graphs even without knowing the parameters. 
On the other hand, they derive strong lower bounds on the PoA of order $\log(n)$
even if the graph metric is known in advance.

Cost sharing approaches for facility location problems and network design problems were analyzed in \cite{KonemannLSZ08,PalT03}. In these works, however, the collected
cost shares need not be budget balanced per edge, thus, leading to a structurally different 
setting.

There exist several characterizations of efficient graph topologies, albeit for
the simpler setting of average cost sharing (or Shapley cost sharing).
Epstein et al.~\cite{Epstein09topo} investigated efficient graph topologies for
Shapley cost sharing and showed that for symmetric $s$-$t$ network congestion games, only extension parallel
graphs (a subclass of series-parallel graphs) are efficient. For asymmetric (multi-commodity)
games, only trees or nodes with parallel edges are efficient. These works are closely related to Milchtaichs~\cite{Milchtaich06GEB} work on the Braess paradox (see also \cite{CenciarelliGS16,ChenD16,ChenDH16}).

 \section{An LP-Characterization of Enforceability}
Let $\mathcal P_i$ be the set of simple $(s_i,t_i)$-paths of $G$. Furthermore, let $F$ be a fixed Steiner forest and $P=(P_1, \ldots, P_n)$ with $P_i \in \mathcal P_i$ the uniquely defined $(s_i,t_i)$-paths in $F$. In addition, let $S_e(P)$ be the set of players which use edge $e$ in their $(s_i,t_i)$-path $P_i$, i.e. $S_e(P):=\{i \in N: e \in P_i\}$. 
An important technical tool for obtaining
characterizations of \emph{efficient graphs} relies on a novel characterization of Steiner forests $F$ that can actually appear as a pure Nash equilibrium for a given graph $G=(V,E)$ and given costs $c$.
We define this property formally.

\begin{definition}\label{def:enforceable} Let $(G,c)$ be an undirected network and
$N$ be a set of players with given connectivity constraints.
A Steiner forest  $F\subseteq E$ is called \emph{enforceable},
if there is a separable cost sharing protocol so that $P=(P_1,\dots,P_n)$ (where every $P_i$ is the unique path in $F$) is a pure Nash equilibrium of the induced game.
\end{definition}
We give a characterization of enforceability of $F$ based on the following linear program \lp{}: 
\begin{align*}
\max &\sum_{i \in N, e \in P_i}\xi_{i,e} \\[5pt]
\text{s.t.: } & \sum_{i \in S_e(P)} \xi_{i,e} \leq c(e) & \forall e \in F \tag{1}\\[5pt]
& \sum_{e \in P_i \setminus P_i'} \xi_{i,e} \leq \sum_{e \in P_i' \setminus P_i} c(e)  & \forall P_i' \in \mathcal P_i \ \forall i \in N \tag{2} \label{Nash1}\\[5pt]
&\xi_{i,e} \geq 0 &\forall e \in P_i \ \forall i \in N \tag{3}\\
\end{align*}

\vspace{-1em}
\begin{theorem}
The Steiner forest $F$ with corresponding strategy profile $P=(P_1, \ldots, P_n)$ is enforceable if and only if there is an optimal solution $(\xi_{i,e})_{i \in N, e \in P_i}$ for \lp{} with 
\[
\sum_{i \in S_e(P)} \xi_{i,e} = c(e) \ \forall e \in F. \tag{BB} \label{Nash2}
\]
\end{theorem}

\begin{proof}
We first assume that there is an optimal solution $(\xi_{i,e})_{i \in N, e \in P_i}$ for \lp{} with \eqref{Nash2}.

Consider the following cost sharing protocol, which assigns for each $i \in N$ and $e \in E$ and each strategy profile $P'=(P_1', \ldots, P_n')$ the following cost shares: 
\[
\xi_{i,e}(P') = \begin{cases} \xi_{i,e}, & \text{if } i \in S_e(P')=S_e(P),\\[3pt]
c(e), & \text{if } i \in (S_e(P')\setminus S_e(P)) \text{ and } i=\min (S_e(P')\setminus S_e(P)), \\[3pt]
c(e), & \text{if } i \in S_e(P') \subset S_e(P) \text{ and } i=\min S_e(P'), \\[3pt]
0, & \text{else.}
\end{cases}
\] 
It is clear that this protocol fulfills separability and budget-balance. We now show that $P$ is a pure Nash equilibrium in the game induced by this protocol. For each player $i$ and each $(s_i,t_i)$-path $P_i' \in \mathcal P_i$, we have 
\[
 \sum_{e \in P_i}{\xi_{i,e}(P)}= \sum_{e \in P_i \setminus P_i'} \xi_{i,e} + \sum_{e \in P_i \cap P_i'}{\xi_{i,e}} \leq \sum_{e \in P_i' \setminus P_i} c(e) + \sum_{e \in P_i \cap P_i'}{\xi_{i,e}}=\sum_{e \in P_i'}{\xi_{i,e}(P_i', P_{-i})}.
\]
Using standard notation in game theory, $P_{-i}$ denotes the strategy profile $(P_1, \ldots, P_{i-1}, P_{i+1}, \ldots P_n)$ consisting of the strategies of the other players.
The inequality follows from condition \eqref{Nash1} of the LP-formulation and the last equality from the definition of the protocol.

\vspace{0.5em}
For the other direction, we assume that we have a separable protocol $\Xi$, so that $P$ is a pure Nash equilibrium in the induced game. We show that the corresponding cost shares $(\xi_{i,e}(P))_{i \in N, e \in P_i}$ are an optimal solution for \lp{} with the desired property \eqref{Nash2}. 

It is clear that we only have to verify condition \eqref{Nash1}. Since $P$ is a pure Nash equilibrium and $\Xi$ is a separable protocol, for each player $i \in N$ and each path $P_i' \in \mathcal P_i$  we get 
\begin{align*} 
\sum_{e \in P_i}{\xi_{i,e}(P)} &\leq \sum_{e \in P_i'}{\xi_{i,e}(P_i', P_{-i})} \\[7pt]
\Leftrightarrow \sum_{e \in P_i \setminus P_i'} \xi_{i,e}(P) + \sum_{e \in P_i \cap P_i'}{\xi_{i,e}(P)} &\leq \sum_{e \in P_i' \setminus P_i} \xi_{i,e}(P_i', P_{-i}) + \sum_{e \in P_i \cap P_i'}{\xi_{i,e}(P_i', P_{-i})} \\[7pt]
\Leftrightarrow \sum_{e \in P_i \setminus P_i'} \xi_{i,e}(P) &\leq \sum_{e \in P_i' \setminus P_i} \xi_{i,e}(P_i', P_{-i}).
\end{align*}
Since $\xi_{i,e}(P_i', P_{-i}) \leq c(e)$ for each edge $e \in P_i'$, condition \eqref{Nash1} follows.
\end{proof}
\section{A Characterization of Efficient Graphs for Two Player Games} \label{section_theorem}

We now consider the case of two players, $N=\{1,2\}$, and first show that an optimal Steiner forest is not necessarily enforceable. To see this, consider the network displayed in Figure \ref{fig1bc1}. 
The solid lines build the unique optimal Steiner forest OPT with cost 22 (which can be easily verified by considering all 19 possible Steiner forests). But the sum of cost shares that one can collect by any separable cost sharing protocol is obviously bounded by 9 for Player 1 and $6+6=12$ for Player 2, thus the objective value for LP(OPT) is bounded by $21<22$ and therefore OPT is not enforceable. By optimizing the costs for this graph, we get a lower bound of $\frac{15}{14}$ for the PoS, see~\autoref{theopos}.

\begin{figure}[H] \centering \psset{unit=1.9cm}
 \begin{pspicture}(0.5,-1.2)(6.6,1.1) 
\Knoten{0.8}{0.8}{v1}\nput{180}{v1}{$s_1$}
\Knoten{0.8}{-0.8}{v2}\nput{180}{v2}{$s_2$}
\Knoten{1}{-0.6}{n2}
\Knoten{1.6}{0}{v3}
\Knoten{2.125}{0}{n5}
\Knoten{3.25}{0}{v4}
\Knoten{4.375}{0}{v6}
\Knoten{5.5}{0}{v7}
\Knoten{6.1}{0.6}{n3}
\Knoten{6.1}{-0.6}{n4}
\Knoten{6.3}{0.8}{v10}\nput{0}{v10}{$t_1$}
\Knoten{6.3}{-0.8}{v11}\nput{0}{v11}{$t_2$}
%
%
%
%
\ncline{v1}{v3}
\ncline{v2}{n2}
\Kante{n2}{v3}{5|-|5}
\ncline{v3}{n5}
\Kante{n5}{v4}{3|2|1}
\Kante{v4}{v6}{2|2|0}
\Kante{v6}{v7}{3|2|1}
\Kante{v7}{n3}{4|3|-}
\ncline{v10}{n3}
\Kante{v7}{n4}{5|-|5}
\ncline{v11}{n4}
\Bogendashed{n5}{n3}{9}{25}
\Bogendashed{n2}{v6}{6}{-25}
\Bogendashed{v4}{n4}{6}{-25}
\end{pspicture}
\caption{OPT is not enforceable (edges of OPT with cost $>0$ are labelled with their cost, followed by an optimal solution $\xi_{1,e}|\xi_{2,e}$ for LP(OPT))} \label{fig1bc1}
\end{figure}
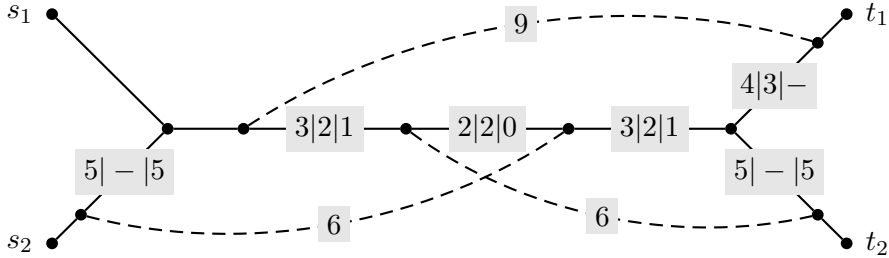


\vspace{0.5em}
As we will show in the rest of the paper, the configuration displayed in Figure~\ref{fig1bc1} is one of few cases in which an optimal Steiner forest is not enforceable.
Before we can state this as a theorem, we need a few definitions.

\begin{definition}[(Strongly) Efficient Graph] \ \
\begin{enumerate}[itemsep=-0em]
\vspace{-0.1em}
		\item We call $(G, (s_1,t_1), (s_2, t_2))$ \textit{efficient}, if, for every cost function $c$, there is an optimal Steiner forest of $(G, (s_1,t_1), (s_2, t_2),c)$ which is enforceable (that means the PoS is 1).
		\item We call $(G, (s_1,t_1), (s_2, t_2))$ \textit{strongly efficient}, if, for every cost function $c$, every optimal Steiner forest of $(G, (s_1,t_1), (s_2, t_2),c)$ is enforceable.
\end{enumerate}
\end{definition}

%
%

\clearpage
\begin{definition}\label{def:BC}
We call a subgraph $H$ of $(G,(s_1,t_1),(s_2,t_2))$ a \textit{Bad Configuration} (BC), if $H$ is one of the graphs in the set $\mathcal{BC}$, where
\[
\mathcal{BC} = \{
	\text{BC1a, BC1b, BC2a, BC2b, BC2c, BC2d, BC3, BC4a, BC4b}\}.
\]

The graphs of $\mathcal{BC}$ are displayed in Figure~\ref{allbc} (see \autoref{formaldef:BC} for the exact definition), where one should note the following: 

\vspace{-0.4em}
\begin{itemize}[itemsep=-0.0ex]
	\item $u,v$ are the terminal nodes of one player and $w,x$ the terminal nodes of the other player;
	\item lines represent simple paths and paths are node-disjoint (except for endnodes);
	\item solid paths have to consist of at least one edge, whereas dashed paths can consist of only one node.
\end{itemize}
\end{definition}

\begin{figure}[H] \centering \psset{unit=1cm}
\captionsetup[subfigure]{labelformat=empty}
\psset{dash= 2pt 2pt}
\subfloat[BC1a]{
 \begin{pspicture}(-0.3,-1)(7.1,1) 
\Knoten{0.2}{0.8}{v1}\nput{180}{v1}{$u$}
\Knoten{0.2}{-0.8}{v2}\nput{180}{v2}{$w$}
\Knoten{0.6}{-0.4}{n2}
\Knoten{1}{0}{v3}
\Knoten{2.125}{0}{n5}
\Knoten{3.25}{0}{v4}
\Knoten{4.375}{0}{v6}
\Knoten{5.5}{0}{v7}
\Knoten{5.9}{0.4}{n3}
\Knoten{5.9}{-0.4}{n4}
\Knoten{6.3}{0.8}{v10}\nput{0}{v10}{$v$}
\Knoten{6.3}{-0.8}{v11}\nput{0}{v11}{$x$}

\ncline[linestyle=dashed]{-}{v1}{v3}
\ncline{-}{n2}{v3}
\ncline[linestyle=dashed]{-}{v2}{n2}
\ncline[linestyle=dashed]{-}{v3}{n5}
\ncline{-}{v4}{n5}
\ncline{-}{v4}{v6}
\ncline{-}{v6}{v7}
\ncline{-}{v7}{n3}
\ncline[linestyle=dashed]{-}{n3}{v10}
\ncline{-}{v7}{n4}
\ncline[linestyle=dashed]{-}{n4}{v11}
\Bogendef{n5}{n3}{q_1}{35}
\Bogendef{n2}{v6}{q_2}{-35}
\Bogendef{v4}{n4}{q_3}{-35}
\end{pspicture}\label{BC1a}
}
\hspace{0.2cm}
\subfloat[BC1b]{
 \begin{pspicture}(0.5,-1)(7.9,1) 
\Knoten{1.3125}{0.8}{v1}\nput{180}{v1}{$u$}
\Knoten{1.7125}{0.4}{n1}
\Knoten{1.3125}{-0.8}{v2}\nput{180}{v2}{$w$}
\Knoten{1.7125}{-0.4}{n2}
\Knoten{2.125}{0}{v3}
\Knoten{3.25}{0}{v4}
\Knoten{4.375}{0}{v6}
\Knoten{5.5}{0}{v7}
\Knoten{5.9}{0.4}{n3}
\Knoten{5.9}{-0.4}{n4}
\Knoten{6.3}{0.8}{v10}\nput{0}{v10}{$v$}
\Knoten{6.3}{-0.8}{v11}\nput{0}{v11}{$x$}

\ncline[linestyle=dashed]{-}{v1}{n1}
\ncline{-}{n1}{v3}
\ncline{-}{n2}{v3}
\ncline[linestyle=dashed]{-}{v2}{n2}
\ncline{-}{v3}{v4}
\ncline{-}{v4}{v6}
\ncline{-}{v6}{v7}
\ncline{-}{v7}{n3}
\ncline[linestyle=dashed]{-}{n3}{v10}
\ncline{-}{v7}{n4}
\ncline[linestyle=dashed]{-}{n4}{v11}
\Bogendef{n1}{n3}{q_1}{25}
\Bogendef{n2}{v6}{q_2}{-35}
\Bogendef{v4}{n4}{q_3}{-35}
\end{pspicture}\label{BC1b}
}

\subfloat[BC2a]{
\begin{pspicture}(-0.8,-1)(6.6,1.2) 
\Knoten{0.2}{0.8}{v1}\nput{180}{v1}{$u$}
\Knoten{0.2}{-0.8}{v2}\nput{180}{v2}{$w$}
\Knoten{0.6}{-0.4}{n2}
\Knoten{1}{0}{v3}
\Knoten{2.125}{0}{n5}
\Knoten{3.25}{0}{v4}
\Knoten{4.375}{0}{v6}
\Knoten{5.5}{0}{v7}
\Knoten{5.9}{0.4}{n3}
\Knoten{5.9}{-0.4}{n4}
\Knoten{6.3}{0.8}{v10}\nput{0}{v10}{$v$}
\Knoten{6.3}{-0.8}{v11}\nput{0}{v11}{$x$}

\Knoten{3.25}{0.65}{n6}
\Knoten{4.375}{0.65}{n7}

\ncline[linestyle=dashed]{n6}{n7}
\ncline{n7}{n3}

\ncline{n5}{n6}
\ncline{n7}{v6}

\ncline[linestyle=dashed]{-}{v1}{v3}
\ncline{-}{n2}{v3}
\ncline[linestyle=dashed]{-}{v2}{n2}
\ncline[linestyle=dashed]{-}{v3}{n5}
\ncline{-}{v4}{n5}
\ncline{-}{v4}{v6}
\ncline{-}{v6}{v7}
\ncline{-}{v7}{n3}
\ncline[linestyle=dashed]{-}{n3}{v10}
\ncline{-}{v7}{n4}
\ncline[linestyle=dashed]{-}{n4}{v11}

\pscurve(0.6,-0.4)(-0.4,0.8)(3.25,0.65)
\Bogendef{v4}{n4}{q_3}{-35}
\end{pspicture}\label{BC2a}
}
\hspace{0.2cm}
\subfloat[BC2b]{
\begin{pspicture}(-0.65,-1)(6.75,1.2) 
\Knoten{0.2}{0.8}{v1}\nput{180}{v1}{$u$}
\Knoten{0.2}{-0.8}{v2}\nput{180}{v2}{$w$}
\Knoten{0.6}{-0.4}{n2}
\Knoten{1}{0}{v3}
\Knoten{2.125}{0}{n5}
\Knoten{3.25}{0}{v4}
\Knoten{4.375}{0}{v6}
\Knoten{5.5}{0}{v7}
\Knoten{5.9}{0.4}{n3}
\Knoten{5.9}{-0.4}{n4}
\Knoten{6.3}{0.8}{v10}\nput{0}{v10}{$v$}
\Knoten{6.3}{-0.8}{v11}\nput{0}{v11}{$x$}

\Knoten{3.25}{0.65}{n6}
\Knoten{4.375}{0.65}{n7}

\ncline{n5}{n6}
\ncline[linestyle=dashed]{n6}{n7}
\ncline{n7}{n3}
\ncline{n6}{v6}

\ncline[linestyle=dashed]{-}{v1}{v3}
\ncline{-}{n2}{v3}
\ncline[linestyle=dashed]{-}{v2}{n2}
\ncline[linestyle=dashed]{-}{v3}{n5}
\ncline{-}{v4}{n5}
\ncline{-}{v4}{v6}
\ncline{-}{v6}{v7}
\ncline{-}{v7}{n3}
\ncline[linestyle=dashed]{-}{n3}{v10}
\ncline{-}{v7}{n4}
\ncline[linestyle=dashed]{-}{n4}{v11}

\pscurve(0.6,-0.4)(-0.4,0.8)(4.375,0.65)
\Bogendef{v4}{n4}{q_3}{-35}
\end{pspicture} 
\label{BC2b}
}
\end{figure}

\begin{figure}[H] \centering \psset{unit=1cm}
\captionsetup[subfigure]{labelformat=empty}
\psset{dash= 2pt 2pt}
\subfloat[BC2c (BC2d is the variant which arises from BC2c in the same way than BC2b from BC2a)]{
\begin{pspicture}(0.3,-1.2)(7.7,1.3) 
\Knoten{1.3125}{0.8}{v1}\nput{180}{v1}{$u$}
\Knoten{1.7125}{0.4}{n1}
\Knoten{1.3125}{-0.8}{v2}\nput{180}{v2}{$w$}
\Knoten{1.7125}{-0.4}{n2}
\Knoten{2.125}{0}{v3}
\Knoten{3.25}{0}{v4}
\Knoten{4.375}{0}{v6}
\Knoten{5.5}{0}{v7}
\Knoten{5.9}{0.4}{n3}
\Knoten{5.9}{-0.4}{n4}
\Knoten{6.3}{0.8}{v10}\nput{0}{v10}{$v$}
\Knoten{6.3}{-0.8}{v11}\nput{0}{v11}{$x$}

\Knoten{3.25}{0.65}{n6}
\Knoten{4.375}{0.65}{n7}

\ncline[linestyle=dashed]{n6}{n7}
\ncline{n7}{n3}

\ncline{n1}{n6}
\ncline{n7}{v6}

\ncline[linestyle=dashed]{-}{v1}{n1}
\ncline{-}{n1}{v3}
\ncline{-}{n2}{v3}
\ncline[linestyle=dashed]{-}{v2}{n2}
\ncline{-}{v3}{v4}
\ncline{-}{v4}{v6}
\ncline{-}{v6}{v7}
\ncline{-}{v7}{n3}
\ncline[linestyle=dashed]{-}{n3}{v10}
\ncline{-}{v7}{n4}
\ncline[linestyle=dashed]{-}{n4}{v11}

\pscurve(1.7125,-0.4)(0.6,0.9)(3.25,0.65)
\Bogendef{v4}{n4}{q_3}{-35}
\end{pspicture}\label{BC3a}
}
\hspace{0.5cm}
\subfloat[BC3]{
\begin{pspicture}(0.5,-1.2)(6.5,1) 
\Knoten{1.2}{0.8}{v1}\nput{180}{v1}{$u$}
\Knoten{1}{-1}{v2}\nput{180}{v2}{$w$}
\Knoten{1.3}{-0.7}{n2}
\Knoten{1.6}{-0.4}{n7}
\Knoten{2}{0}{v3}
\Knoten{3.25}{0}{v4}
\Knoten{4.375}{0}{v6}
\Knoten{5.5}{0}{v7}
\Knoten{6}{0.5}{n3}
\Knoten{6}{-0.5}{n4}
\Knoten{6.5}{1}{v10}\nput{0}{v10}{$v$}
\Knoten{6.5}{-1}{v11}\nput{0}{v11}{$x$}

\ncline[linestyle=dashed]{-}{v1}{v3}
\ncline{-}{n2}{v3}
\ncline[linestyle=dashed]{-}{v2}{n2}
\ncline{-}{v3}{v4}
\ncline{-}{v4}{v6}
\ncline{-}{v6}{v7}
\ncline{-}{v7}{n3}
\ncline[linestyle=dashed]{-}{n3}{v10}
\ncline{-}{v7}{n4}
\ncline[linestyle=dashed]{-}{n4}{v11}
\ncline{v3}{n7}
\pscurve(1.6,-0.4)(0.4,0.7)(6,0.5)
\ncarc[arcangle=-35]{n2}{v6}
\ncarc[arcangle=-35]{v4}{n4}
\end{pspicture}
}

\subfloat[BC4a]{
\begin{pspicture}(0,-1.2)(6.5,1.7) 
\Knoten{1.2}{0.8}{v1}\nput{180}{v1}{$u$}
\Knoten{1}{-1}{v2}\nput{180}{v2}{$w$}
\Knoten{1.25}{-0.75}{n9}
\Knoten{1.5}{-0.5}{n2}
\Knoten{2}{0}{v3}
\Knoten{3.25}{0}{v4}
\Knoten{4.375}{0}{v6}
\Knoten{5.5}{0}{v7}
\Knoten{6}{0.5}{n3}
\Knoten{6}{-0.5}{n4}
\Knoten{6.5}{1}{v10}\nput{0}{v10}{$v$}
\Knoten{6.5}{-1}{v11}\nput{0}{v11}{$x$}

\Knoten{3.25}{0.65}{n6}
\Knoten{4.375}{0.65}{n7}

\ncline[linestyle=dashed]{n6}{n7}
\ncline{n7}{n3}

\ncline{n7}{v6}

\ncline{n2}{v3}
\ncline[linestyle=dashed]{-}{v1}{v3}
\ncline[linestyle=dashed]{-}{v2}{n9}
\ncline{-}{n9}{n2}
\ncline{-}{v3}{v4}
\ncline{-}{v4}{v6}
\ncline{-}{v6}{v7}
\ncline{-}{v7}{n3}
\ncline[linestyle=dashed]{-}{n3}{v10}
\ncline{-}{v7}{n4}
\ncline[linestyle=dashed]{-}{n4}{v11}

\pscurve(1.5,-0.5)(0.5,0.9)(3.25,0.65)
\pscurve(1.25,-0.75)(0,1.3)(3.25,0.65)

\ncarc[arcangle=-35]{v4}{n4}
\end{pspicture}
}
\hspace{1cm}
\subfloat[BC4b]{
\begin{pspicture}(0,-1.2)(6.5,1.7) 
\Knoten{1.2}{0.8}{v1}\nput{180}{v1}{$u$}
\Knoten{1}{-1}{v2}\nput{180}{v2}{$w$}
\Knoten{1.25}{-0.75}{n9}
\Knoten{1.5}{-0.5}{n2}
\Knoten{2}{0}{v3}
\Knoten{3.25}{0}{v4}
\Knoten{4.375}{0}{v6}
\Knoten{5.5}{0}{v7}
\Knoten{6}{0.5}{n3}
\Knoten{6}{-0.5}{n4}
\Knoten{6.5}{1}{v10}\nput{0}{v10}{$v$}
\Knoten{6.5}{-1}{v11}\nput{0}{v11}{$x$}

\Knoten{3.25}{0.65}{n6}
\Knoten{4.375}{0.65}{n7}

\ncline[linestyle=dashed]{n6}{n7}
\ncline{n7}{n3}

\ncline{n6}{v6}

\ncline{n2}{v3}
\ncline[linestyle=dashed]{-}{v1}{v3}
\ncline[linestyle=dashed]{-}{v2}{n9}
\ncline{-}{n9}{n2}
\ncline{-}{v3}{v4}
\ncline{-}{v4}{v6}
\ncline{-}{v6}{v7}
\ncline{-}{v7}{n3}
\ncline[linestyle=dashed]{-}{n3}{v10}
\ncline{-}{v7}{n4}
\ncline[linestyle=dashed]{-}{n4}{v11}

\pscurve(1.5,-0.5)(0.5,0.9)(3.25,0.65)
\pscurve(1.25,-0.75)(0,1.3)(4.375,0.65)

\ncarc[arcangle=-35]{v4}{n4}
\end{pspicture}
}
\caption{Bad Configurations}
\label{allbc}
\end{figure}
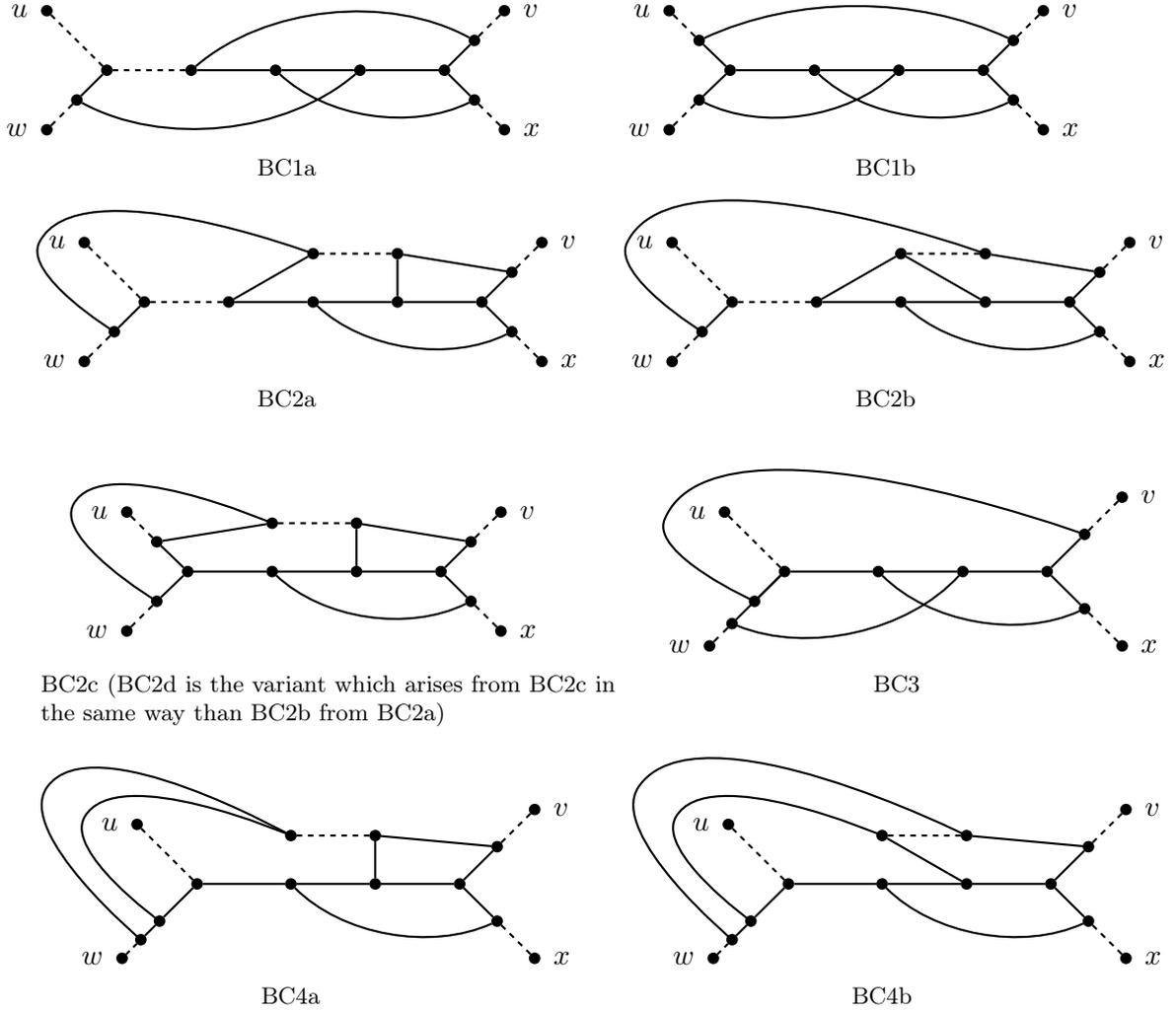

\begin{theorem}[Main Theorem]\label{theoremrobust}
The following three statements are equivalent:
\begin{enumerate}[label=(\arabic*)]
  \item\label{maintheo1} $(G, (s_1,t_1), (s_2, t_2))$ does not contain a subgraph which is a Bad Configuration.
	\item\label{maintheo2} $(G, (s_1,t_1), (s_2, t_2))$ is efficient.
	\item\label{maintheo3} $(G, (s_1,t_1), (s_2, t_2))$ is strongly efficient.
\end{enumerate} 
\end{theorem}

%

It is clear that \ref{maintheo3} implies \ref{maintheo2};  a sketch of the proof of $\ref{maintheo2} \Rightarrow \ref{maintheo1}$ and $\ref{maintheo1} \Rightarrow \ref{maintheo3}$ can be found in the next section.


%
%

\section{Sketch of the Proof of \autoref{theoremrobust}}
\subsection{\ref{maintheo2} implies \ref{maintheo1}}
We assume that $(G, (s_1,t_1), (s_2,t_2))$ contains a Bad Configuration. Then we define a cost function~$c$ so that the optimal Steiner forest is unique and not enforceable, showing the claim. 
	To this end we choose a subgraph of $G$ that is a BC and set $c(e)=\infty$ if the edge $e$ is not contained in this subgraph. We now have to distinguish between the different types of BCs.
	For BC1a, the costs displayed in Figure~\ref{fig1bc1} carry over (if a path consists of more than one edge, choose the costs of the corresponding edges arbitrarily so that they sum up to the displayed cost on the path; all paths with nonzero costs contain at least one edge because of the definition of the corresponding type of BC).	Costs for the other types of BCs can be found in \autoref{subsecproof}.
	

\subsection{\ref{maintheo1} implies \ref{maintheo3}}\label{subsec:proof_sketch}
Consider an arbitrary optimal Steiner forest $F$ (w.r.t. an arbitrary cost function $c$) and an optimal solution $(\xi_{i,e})_{i \in N, e \in P_i}$ of the corresponding \lp{}. Assume that condition \eqref{Nash2} is not satisfied, i.e. there is an edge that is not paid completely.

%
%
%

Note that $P_1 \cap P_2$ has to contain at least one edge, since otherwise $F$ is enforceable. Furthermore, $P_1 \cap P_2$ has to be a simple path, since $F$ contains no cycles. We refer to the edges of $P_1 \cap P_2$ as the \emph{commonly used edges} or \emph{the middle part} (cf. Figure~\ref{lmr}).
Note that we can w.l.o.g. assume that $s_1$ and $s_2$ are in the left part, otherwise we can just swap source and sink since the graph is undirected. Figure~\ref{lmr} also illustrates the complete ordering on the edges of $F$ that we use throughout the proof (the numbers indicate in which order we consider the subpaths; the arrows indicate increasing order within the subpaths).


\begin{figure}[H] \centering \psset{unit=1.2cm}
\begin{pspicture}(-1,-2.3)(6,1.5) 
\Knoten{-1.1}{1}{v1}\nput{180}{v1}{$s_1$}
\Knoten{-1.1}{-1}{v2}\nput{180}{v2}{$s_2$}
\Knoten{0.4}{0}{v3}
\Knoten{4}{0}{v4}
\Knoten{5.5}{1}{v5}\nput{0}{v5}{$t_1$}
\Knoten{5.5}{-1}{v6}\nput{0}{v6}{$t_2$}

\ncline{-}{v1}{v3}
\ncline{-}{v2}{v3}
\ncline{-}{v3}{v4}
\ncline{v4}{v5}
\ncline{-}{v4}{v6}

\psbrace[rot=90, nodesepA=-0.7cm, nodesepB=0.4cm, braceWidthInner=0.5em, braceWidthOuter=0.5em](-1.1,-1.5)(0.35,-1.5){left part}
\psbrace[rot=90, nodesepA=-1cm, nodesepB=0.4cm, braceWidthInner=0.5em, braceWidthOuter=0.5em](0.45,-1.5)(3.95,-1.5){middle part}
\psbrace[rot=90, nodesepA=-0.7cm, nodesepB=0.4cm, braceWidthInner=0.5em, braceWidthOuter=0.5em](4.05,-1.5)(5.5,-1.5){right part}

\psline[arrows=->, linestyle=dashed, arrowsize=5pt](-0.5,1.1)(0.5,0.44)\rput{-33}(0.2,1.1){(1)}
\psline[arrows=->, linestyle=dashed, arrowsize=5pt](-0.5,-1.16)(0.5,-0.44)\rput{33}(0.2,-1.1){(2)}
\psline[arrows=->, linestyle=dashed, arrowsize=5pt](1.2,0.44)(3.2,0.44)\rput{0}(2.2,0.8){(3)}
\psline[arrows=->, linestyle=dashed, arrowsize=5pt](3.9,0.5)(4.9,1.16)\rput{33}(4.1,1.1){(4)}
\psline[arrows=->, linestyle=dashed, arrowsize=5pt](3.9,-0.5)(4.9,-1.16)\rput{-33}(4.1,-1.1){(5)}
\end{pspicture}
\caption{Left, middle and right part of $F$; complete ordering on edges of $F$} \label{lmr}
\end{figure}
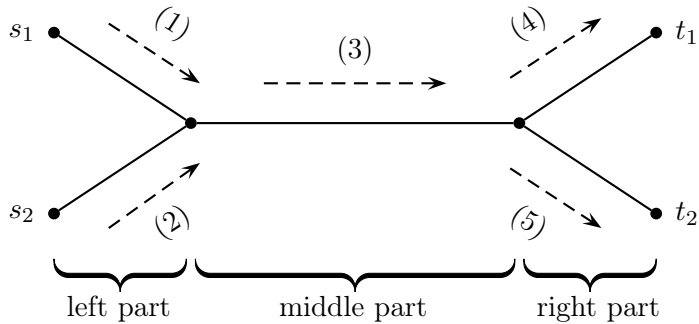

\begin{definition}\label{def:ptl}
We call an optimal solution $(\xi_{i,e})_{i \in N, e \in P_i}$ for \lp{} \emph{pushed to the left (\ptl)}, if the following changes of the cost shares (which we denote a push operation) do not yield a feasible solution for \lp{} (for every choice of $i,e,f$ and $\epsilon>0$): 
\vspace{-0.25em}
\begin{center}
	Increase the cost share $\xi_{i,e}$ of Player $i$ on edge $e$ by $\epsilon$ and simultaneously\\
	decrease $\xi_{i,f}$ by $\epsilon$, where $f$ is an edge with higher order than $e$.
\end{center} 
\vspace{-0.25em}
\end{definition}
To obtain \ptl-cost shares $(\xi_{i,e})_{i \in N, e \in P_i}$, we can use Algorithm~\nameref{push} (see \autoref{subsec:proof_2}).
%
%
Let $e$ be the first edge (with respect to the order) which is not completely paid according to the computed \ptl-cost shares. We distinguish between the cases that $e$ is in the left part of $F$ (Case $L$), the middle part of $F$ (Case $M$) or the right part of $F$ (Case $R$).
In each of these cases, we get a contradiction (see \autoref{lemma1}, \autoref{lemma2} and \autoref{lemma3} for complete proofs). 

We now describe some of the main ideas for the Cases $L$, $M$ and $R$. 
If $e \in P_i$, Player $i$ needs to have a \emph{tight alternative} $q$ for $e$ (corresponding to a tight inequality in~\eqref{Nash1}), i.e., $q$ is a simple path which closes a unique cycle $C(q)$ with $P_i$ containing $e$, and the cost of $q$ equals the sum of cost shares Player $i$ pays on the edges of $P_i \cap C(q)$:
If there is no such tight alternative, increasing the cost share $\xi_{i,e}<c(e)$ without changing the other cost shares would yield a feasible solution for \lp{} with higher objective function value. 
We denote the edges of $P_i \cap C(q)$ as the edges which are \emph{substituted by} $q$. 

\textbf{Case $L$:} Assume that $e \in P_1$ holds ($e \in P_2$ follows analogously). 
Under all tight alternatives for Player $1$ which substitute $e$, let $q_1$ be \emph{smallest}, that is, 
$q_1$ minimizes the maximum occurring order of an edge in $C(q_1)\cap P_1$.
Let $f$ be the ``last'' edge which is substituted by $q_1$, i.e., where the maximum is attained.
The situation for the case that $f$ is in the middle part is illustrated in Figure~\ref{sketchcaseLa}; the other cases ($f$ is in the right or left part) can be treated very similarly. 

\begin{figure}[H] \centering \psset{unit=1cm}
\subfloat[$f$ is in the middle part]{
 \begin{pspicture}(-0.6,-1)(5.2,1.6) 
\Knoten{-0.2}{-0.8}{v7}\nput{180}{v7}{$s_2$}
\Knoten{-0.4}{1.4}{v2}\nput{180}{v2}{$s_1$}
\Knoten{-0.1}{1.1}{v3}
\Knoten{0.2}{0.8}{v4}
\Knoten{0.7}{0.3}{v5}
\Knoten{1}{0}{v6}
\Knoten{2}{0}{v12}
\Knoten{3}{0}{v13}
\Knoten{4}{0}{v14}
\Knoten{5}{1}{v15}\nput{0}{v15}{$t_1$}
\Knoten{4.8}{-0.8}{v16}\nput{0}{v16}{$t_2$}

\ncline{v2}{v3}
\Kante{v13}{v12}{f}
\Kante{v4}{v5}{e}
\Bogendashed{v3}{v13}{q_1}{35}
\ncline{v3}{v4}
\ncline{v5}{v6}
\ncline{v6}{v7}
\ncline{v6}{v12}
\ncline{v13}{v14}
\ncline{v14}{v15}
\ncline{v14}{v16}
\end{pspicture} \label{sketchcaseLa}
}
\hspace{1cm}
\subfloat[$F^*$]{
 \begin{pspicture}(-0.6,-1)(5.2,1.6) 
\Knoten{-0.2}{-0.8}{v7}\nput{180}{v7}{$s_2$}
\Knoten{-0.4}{1.4}{v2}\nput{180}{v2}{$s_1$}
\Knoten{-0.1}{1.1}{v3}
\Knoten{1}{0}{v6}
\Knoten{2}{0}{v12}
\Knoten{3}{0}{v13}
\Knoten{4}{0}{v14}
\Knoten{5}{1}{v15}\nput{0}{v15}{$t_1$}
\Knoten{4.8}{-0.8}{v16}\nput{0}{v16}{$t_2$}

\Kante{v13}{v12}{f}
\ncline{v2}{v3}
\Bogen{v3}{v13}{q_1}{35}
\ncline[linestyle=dashed]{v3}{v6}
\ncline{v6}{v7}
\ncline{v6}{v12}
\ncline{v13}{v14}
\ncline{v14}{v15}
\ncline{v14}{v16}
\end{pspicture} \label{sketchcaseLb}
}
\caption{Case $L$}
\label{sketchcaseL}
\end{figure}

\vspace{-1em}
We get that Player $1$ pays the edges of $C(q_1)\cap P_1$ before $e$ (w.r.t. the ordering) completely since those edges are not contained in $P_2$ and $e$ is the first edge that is not paid completely. The same reasons yield $\xi_{1,e}<c(e)$. Furthermore, Player $1$ pays nothing on the edges of $C(q_1)\cap P_1$ after $e$. To see this, assume that there exists such an edge $h$ with $\xi_{1, h}>0$. But then a push operation with $e$ and $h$ (and a suitably small $\epsilon$) yields a feasible solution for \lp{}, because $q_1$ is a smallest alternative for $e$, contradiction.
Let $F^*$ be the Steiner forest which arises from $F$ by adding $q_1$ and deleting the edges of $C(q_1) \cap P_1$ which are in the left part (cf. Figure~\ref{sketchcaseLb}). 
Since the cost of $q_1$ equals the sum of cost shares of the deleted edges and this sum is strictly smaller than the costs of these edges, $c(F^*)<c(F)$. The full proof for Case $L$ can be found in the proof of~\autoref{lemma1},~\autoref{subsubsec:caseL}.

\textbf{Case $M$:} Now both players need to have tight alternatives $q_1$ and $q_2$ for $e$. It is clear that we can construct a cheaper Steiner forest if there are tight alternatives $q_1$ and $q_2$ for $e$ which substitute the same edges of the middle part, or a tight alternative for $e$ which substitutes only edges of the middle part. 
We then distinguish between the two cases that all tight alternatives for $e$ of one player substitute edges of the right part, or both players have a tight alternative for $e$ which substitutes edges of the left part. 
Since the first case can be treated similarly to Case $L$, we describe how to proceed in the second case. Let $q_1$ (for Player 1) and $q_2$ (for Player 2) be smallest tight alternatives for $e$ which substitute edges of the left part. Consider the case that $q_1$ substitutes less edges of the middle part than $q_2$, and $q_2$ does not substitute edges of the right part, see Figure~\ref{sketchcaseM} (the other cases follow similarly). Adding $q_1$ and $q_2$, and deleting the dashed edges (cf. Figure~\ref{sketchcaseM}) yields a  Steiner forest $F^*$ with smaller cost than $F$ (note that Player $2$ pays nothing on the edges after $e$ which are substituted by $q_2$), and thus we get a contradiction. 
The full proof for Case $M$ can be found in the proof of~\autoref{lemma2},~\autoref{subsubsec:caseM}.

\begin{figure}[H] \centering \psset{unit=1.3cm}
\subfloat{
 \begin{pspicture}(-0.1,-1)(5.3,1) 
\Knoten{0.2}{0.8}{v1}\nput{180}{v1}{$s_1$}
\Knoten{0.2}{-0.8}{v2}\nput{180}{v2}{$s_2$}
\Knoten{0.6}{0.4}{v11}
\Knoten{0.6}{-0.4}{v12}
\Knoten{1}{0}{v3}
\Knoten{1.6}{0}{v4}
\Knoten{2.3}{0}{v5}
\Knoten{2.9}{0}{v6}
\Knoten{3.6}{0}{v7}
\Knoten{4.2}{0}{v8}
\Knoten{5}{0.8}{v13}\nput{0}{v13}{$t_1$}
\Knoten{5}{-0.8}{v14}\nput{0}{v14}{$t_2$}

\ncline{v1}{v11}
\ncline{-}{v11}{v3}
\ncline{v2}{v12}
\ncline{-}{v12}{v3}
\ncline{-}{v3}{v4}
\Kante{v4}{v5}{e}
\ncline{-}{v5}{v6}
\ncline{v6}{v7}
\ncline{v7}{v8}
\ncline{-}{v8}{v13}
\ncline{-}{v8}{v14}
\Bogendashed{v11}{v6}{q_1}{35}
\Bogendashed{v12}{v7}{q_2}{-35}
\end{pspicture}\label{sketchcaseMa}
}
\hspace{0.7cm}
\subfloat{
\begin{pspicture}(-0.1,-1)(5.3,1) 
\Knoten{0.2}{0.8}{v1}\nput{180}{v1}{$s_1$}
\Knoten{0.2}{-0.8}{v2}\nput{180}{v2}{$s_2$}
\Knoten{0.6}{0.4}{v11}
\Knoten{0.6}{-0.4}{v12}
\Knoten{1}{0}{v3}
\Knoten{2.9}{0}{v6}
\Knoten{3.6}{0}{v7}
\Knoten{4.2}{0}{v8}
\Knoten{5}{0.8}{v13}\nput{0}{v13}{$t_1$}
\Knoten{5}{-0.8}{v14}\nput{0}{v14}{$t_2$}

\ncline{v1}{v11}
\ncline[linestyle=dashed]{v11}{v3}
\ncline{v2}{v12}
\ncline[linestyle=dashed]{-}{v12}{v3}
\ncline[linestyle=dashed]{-}{v3}{v6}
\ncline{v6}{v7}
\ncline{-}{v7}{v8}
\ncline{-}{v8}{v13}
\ncline{-}{v8}{v14}
\Bogen{v11}{v6}{q_1}{35}
\Bogen{v12}{v7}{q_2}{-35}
\end{pspicture}\label{sketchcaseMb}
}
\caption{Case $M$}\label{sketchcaseM}
\end{figure}

\newpage
\textbf{Case $R$:}
Note that in the Cases $L$ and $M$ we did not need any arguments according to BCs. This already indicates that Case $R$ is more complicated. 
We mention here only a few of the proof ideas (a full proof can be found in \autoref{lemma3}, \autoref{subsubsec:caseR}). 

Again, we consider a smallest alternative $q_1$ of Player $1$ for $e$ (assuming that $e \in P_1$). It is easy to see that $q_1$ has to substitute some edges of the middle part, since otherwise there is a cheaper Steiner forest. 
The same argument shows that there has to be an edge $f$ in the middle part (substituted by $q_1$) which Player $2$ does not pay completely. Now we want to argue that Player~$2$ needs to have a tight alternative which substitutes $f$. Note that for an arbitrary \ptl-solution of \lp{} 
we cannot guarantee that whenever a player does not pay an edge in the middle part completely, this player has a tight alternative for this edge.
However, we can achieve this property for one fixed Player~$i$ by maximizing the sum of cost shares of Player~$i$ among all optimal solutions for \lp{} for which $e$ is the first edge which is not completely paid. Let us assume that this property holds for Player~$2$ and consider a tight alternative of Player~$2$ which substitutes $f$. If this alternative substitutes only edges of the middle part, or the same edges of the middle part as $q_1$, one can construct a cheaper Steiner forest. The remaining subcases can be organized as follows: If there is no tight alternative of Player~$2$ which substitutes $f$ and edges of the right part, let~$q_2$ be any tight alternative for~$f$ (which then substitutes edges of the left part; Subcase~$R.3$). 
Otherwise, let $q_2$ be a tight alternative of Player~$2$ which substitutes~$f$ and edges of the right part maximizing the minimum occuring order of an edge in $C(q_2)\cap P_2$. Then~$q_2$ can either substitute less (Subcase~$R.1$) or more (Subcase~$R.2$) edges of the middle part than~$q_1$, see Figure~\ref{sketchcaseR} for the subcases that~$q_1$ (or~$q_2$ in~$R.2$) does not substitute edges of the left part. 

\begin{figure}[H] \centering \psset{unit=1.1cm}
\subfloat[Subcase $R.1$]{
\begin{pspicture}(0.8,-1.3)(7.1,1.3) 
\Knoten{0.96}{1.04}{v-1}\nput{180}{v-1}{$s_1$}
\Knoten{0.96}{-1.04}{v0}\nput{180}{v0}{$s_2$}
\Knoten{2}{0}{v-3}
\Knoten{2.8}{0}{v4}
\Knoten{3.6}{0}{v5}
\Knoten{4.4}{0}{v6}
\Knoten{5.2}{0}{v7}
\Knoten{6}{0}{v8}
\Knoten{7.04}{1.04}{v1}\nput{0}{v1}{$t_1$}
\Knoten{7.04}{-1.04}{v2}\nput{0}{v2}{$t_2$}
\Knoten{6.87}{0.87}{v11}
\Knoten{6.52}{-0.52}{v12}
\Knoten{6.64}{0.64}{n1}
\Knoten{6.24}{0.24}{n2}

\ncline{-}{v-1}{v-3}
\ncline{-}{v0}{v-3}
\ncline{-}{v-3}{v4}
\ncline{-}{v4}{v5}
\ncline{-}{v5}{v6}
\Kante{v7}{v6}{f}
\ncline{-}{v7}{v8}
\ncline{-}{v8}{v2}
\ncline{v8}{n2}
\Kante{n2}{n1}{e}
\ncline{n1}{v1}
\Bogendashed{v11}{v4}{q_1}{-30}
\Bogendashed{v5}{v12}{q_2}{-30}
\end{pspicture}\label{sketchcaseRa}
}
\hspace{1cm}
\subfloat[Subcase $R.2$]{
\begin{pspicture}(0.8,-1.3)(7.1,1.3) 

\Knoten{0.96}{1.04}{v-1}\nput{180}{v-1}{$s_1$}
\Knoten{0.96}{-1.04}{v0}\nput{180}{v0}{$s_2$}
\Knoten{2}{0}{v-3}
\Knoten{2.8}{0}{v4}
\Knoten{3.6}{0}{v5}
\Knoten{4.4}{0}{v6}
\Knoten{5.2}{0}{v7}
\Knoten{6}{0}{v8}
\Knoten{7.04}{1.04}{v1}\nput{0}{v1}{$t_1$}
\Knoten{7.04}{-1.04}{v2}\nput{0}{v2}{$t_2$}
\Knoten{6.87}{0.87}{v11}
\Knoten{6.64}{0.64}{n1}
\Knoten{6.24}{0.24}{n2}
\Knoten{6.52}{-0.52}{v12}

\ncline{-}{v-1}{v-3}
\ncline{-}{v0}{v-3}
\ncline{-}{v-3}{v4}
\ncline{-}{v4}{v5}
\ncline{-}{v5}{v6}
\Kante{v7}{v6}{f}
\ncline{-}{v7}{v8}
\ncline{-}{v8}{v2}
\ncline{-}{v8}{n2}
\Kante{n2}{n1}{e}
\ncline{-}{n1}{v1}
\Bogendashed{v11}{v5}{q_1}{-30}
\Bogendashed{v4}{v12}{q_2}{-30}
\end{pspicture}\label{sketchcaseRb}
}
\caption{Subcases $R.1$ and $R.2$} \label{sketchcaseR}
\end{figure}
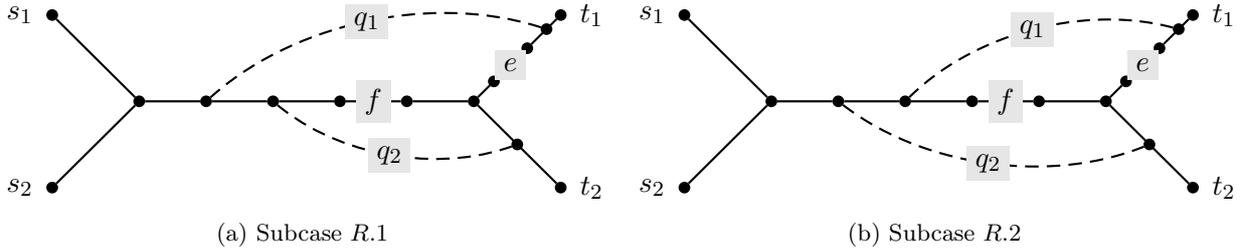

We now describe how to proceed with the situation illustrated in Figure~\ref{sketchcaseRa}.
One can construct a cheaper Steiner forest if Player $2$ completely pays the edges of the commonly used part which are substituted by $q_1$, but not by $q_2$. Thus, there has to be such an edge $h$ which Player $2$ does not pay completely, together with a tight alternative $q_2'$ for $h$. Similarly as for $q_2$, we have to distinguish between several subcases depending on the properties of $q_2'$, see Figure~\ref{sketchcaseR2} for two possible subcases. 

\begin{figure}[H] \centering \psset{unit=1.1cm}
\subfloat[Existence of tight alternatives for Player 1?]{
\begin{pspicture}(0.8,-1.4)(7.1,1.3) 
\Knoten{0.96}{1.04}{v-1}\nput{180}{v-1}{$s_1$}
\Knoten{0.96}{-1.04}{v0}\nput{180}{v0}{$s_2$}
\Knoten{1.48}{-0.52}{v-2}
\Knoten{2}{0}{v-3}
\Knoten{2.8}{0}{v4}
\Knoten{3.6}{0}{v5}
\Knoten{4.4}{0}{v6}
\Knoten{5.2}{0}{v7}
\Knoten{6}{0}{v8}
\Knoten{7.04}{1.04}{v1}\nput{0}{v1}{$t_1$}
\Knoten{7.04}{-1.04}{v2}\nput{0}{v2}{$t_2$}
\Knoten{6.87}{0.87}{v11}
\Knoten{6.52}{-0.52}{v12}
\Knoten{6.64}{0.64}{n1}
\Knoten{6.24}{0.24}{n2}

\ncline{-}{v-1}{v-3}
\ncline{-}{v0}{v-3}
\ncline{-}{v-3}{v4}
\Kante{v4}{v5}{h}
\ncline{-}{v5}{v6}
\Kante{v6}{v7}{f}
\ncline{-}{v7}{v8}
\ncline{-}{v8}{v2}
\ncline{v8}{n2}
\Kante{n2}{n1}{e}
\ncline{n1}{v1}
\Bogendashed{v11}{v4}{q_1}{-30}
\Bogendashed{v5}{v12}{q_2}{-30}
\Bogendashed{v-2}{v2}{q_2'}{-10}
\end{pspicture}\label{sketchcaseR2a}
}
\hspace{1cm}
\subfloat[Bad Configuration?]{
\begin{pspicture}(0.8,-1.4)(7.1,1.3) 
\Knoten{0.96}{1.04}{v-1}\nput{180}{v-1}{$s_1$}
\Knoten{0.96}{-1.04}{v0}\nput{180}{v0}{$s_2$}
\Knoten{1.48}{-0.52}{v-2}
\Knoten{2}{0}{v-3}
\Knoten{2.8}{0}{v4}
\Knoten{3.6}{0}{v5}
\Knoten{4.4}{0}{v6}
\Knoten{5.2}{0}{v7}
\Knoten{6}{0}{v8}
\Knoten{7.04}{1.04}{v1}\nput{0}{v1}{$t_1$}
\Knoten{7.04}{-1.04}{v2}\nput{0}{v2}{$t_2$}
\Knoten{6.87}{0.87}{v11}
\Knoten{6.52}{-0.52}{v12}
\Knoten{6.64}{0.64}{n1}
\Knoten{6.24}{0.24}{n2}

\ncline{-}{v-1}{v-3}
\ncline{-}{v0}{v-3}
\ncline{-}{v-3}{v4}
\Kante{v4}{v5}{h}
\ncline{-}{v5}{v6}
\Kante{v6}{v7}{f}
\ncline{-}{v7}{v8}
\ncline{-}{v8}{v2}
\ncline{v8}{n2}
\Kante{n2}{n1}{e}
\ncline{n1}{v1}
\Bogendashed{v11}{v4}{q_1}{-30}
\Bogendashed{v5}{v12}{q_2}{-35}
\Bogendashed{v-2}{v6}{q_2'}{-35}
\end{pspicture}\label{sketchcaseR2b}
}
\caption{Different Subcases of $R.1$} \label{sketchcaseR2}
\end{figure}
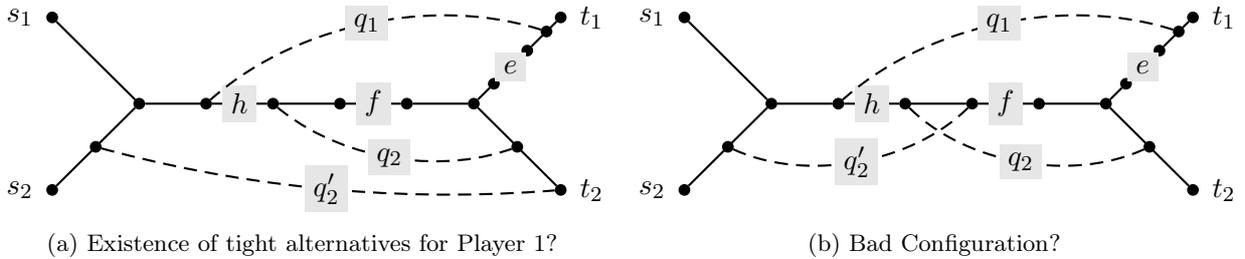

First consider the subcase illustrated in Figure~\ref{sketchcaseR2a}. If Player $1$ completely pays the edges of the middle part which are not substituted by $q_1$, using $q_1$ and $q_2'$ yields a cheaper Steiner forest, so we can assume that this does not hold. Now we would like to argue that there has to be a tight alternative for Player $1$ substituting such an edge; but as mentioned above, this is not immediately clear. To ensure this, we introduced another additional property for the given cost shares (for more details, see \autoref{def:properties} and \autoref{lemma:properties} in \autoref{subsubsec:caseR}). 
We now consider the subcase of Figure~\ref{sketchcaseR2b}, which turned out to be the most challenging problem in the proof.

Note that the subgraph illustrated Figure~\ref{sketchcaseR2b} is a BC1a only if the paths $q_1,q_2$ and $q_2'$ are pairwise node-disjoint and furthermore internal node-disjoint with $P_1 \cup P_2$. Depending on these properties, we grouped all possible different situations in twelve ``types'' (Figure~\ref{sketchcaseR3} illustrates two of them). 


\begin{figure}[H] \centering \psset{unit=1.1cm}
\subfloat{
\begin{pspicture}(0.8,-1.3)(7.1,1.3) 
\Knoten{0.96}{1.04}{v-1}\nput{180}{v-1}{$s_1$}
\Knoten{0.96}{-1.04}{v0}\nput{180}{v0}{$s_2$}
\Knoten{1.48}{-0.52}{v-2}
\Knoten{2}{0}{v-3}
\Knoten{2.8}{0}{v4}
\Knoten{3.6}{0}{v5}
\Knoten{4.4}{0}{v6}
\Knoten{5.2}{0}{v7}
\Knoten{6}{0}{v8}
\Knoten{7.04}{1.04}{v1}\nput{0}{v1}{$t_1$}
\Knoten{7.04}{-1.04}{v2}\nput{0}{v2}{$t_2$}
\Knoten{6.87}{0.87}{v11}
\Knoten{6.52}{-0.52}{v12}
\Knoten{6.64}{0.64}{n1}
\Knoten{6.24}{0.24}{n2}

\Knoten{4}{-0.52}{n3}

\ncline{-}{v-1}{v-3}
\ncline{-}{v0}{v-3}
\ncline{-}{v-3}{v4}
\Kante{v4}{v5}{h}
\ncline{-}{v5}{v6}
\Kante{v6}{v7}{f}
\ncline{-}{v7}{v8}
\ncline{-}{v8}{v2}
\ncline{v8}{n2}
\Kante{n2}{n1}{e}
\ncline{n1}{v1}
\ncarc[linestyle=dashed, arcangle=10]{n3}{v5}
\ncarc[linestyle=dashed, arcangle=-10]{n3}{v6}
\Bogendashed{v11}{v4}{q_1}{-30}
\Bogendashed{n3}{v12}{q_2}{-25}
\Bogendashed{v-2}{n3}{q_2'}{-25}
\end{pspicture}\label{sketchcaseR3a}
}
\hspace{1cm}
\subfloat{
 \begin{pspicture}(0.8,-1.3)(7.1,1.3) 
\Knoten{0.96}{1.04}{v-1}\nput{180}{v-1}{$s_1$}
\Knoten{0.96}{-1.04}{v0}\nput{180}{v0}{$s_2$}
\Knoten{1.48}{-0.52}{v-2}
\Knoten{1.48}{0.52}{n3}
\Knoten{2}{0}{v-3}
\Knoten{2.8}{0}{v4}
\Knoten{3.6}{0}{v5}
\Knoten{4.4}{0}{v6}
\Knoten{5.2}{0}{v7}
\Knoten{6}{0}{v8}
\Knoten{7.04}{1.04}{v1}\nput{0}{v1}{$t_1$}
\Knoten{7.04}{-1.04}{v2}\nput{0}{v2}{$t_2$}
\Knoten{6.87}{0.87}{v11}
\Knoten{6.52}{-0.52}{v12}
\Knoten{6.64}{0.64}{n1}
\Knoten{6.24}{0.24}{n2}

\ncline{-}{v-1}{v-3}
\ncline{-}{v0}{v-3}
\ncline{-}{v-3}{v4}
\Kante{v4}{v5}{h}
\ncline{-}{v5}{v6}
\Kante{v6}{v7}{f}
\ncline{-}{v7}{v8}
\ncline{-}{v8}{v2}
\ncline{v8}{n2}
\Kante{n2}{n1}{e}
\ncline{n1}{v1}
\Bogendashed{n3}{v6}{q_2'}{35}
\Bogendashed{v11}{v4}{q_1}{-30}
\Bogendashed{v5}{v12}{q_2}{-35}
\Bogendashed{v-2}{n3}{q_2'}{25}
\end{pspicture}\label{sketchcaseR3ab}
}
\caption{Two types of No Bad Configurations} \label{sketchcaseR3}
\end{figure}
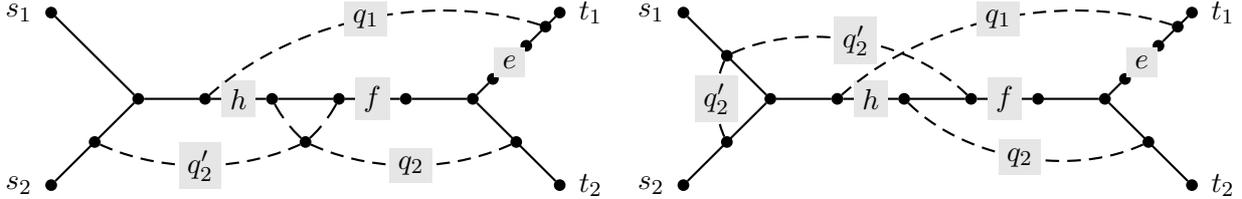

In total, 16 subgraphs similar to the one illustrated in Figure~\ref{sketchcaseR2b} occur, for which we have to investigate all twelve types, leading to $16\cdot 12=192$ subcases (see~\autoref{subsubsec:caseR} and \ref{subsubsec:pbcs}). 

\newcommand{\Knotensmall}[3]{\cnode*(#1,#2){0.1pt}{#3}}

\section{Analyzing Different Graph Classes} \label{secclasses}
In this section we show that various well-studied classes of graphs only contain strongly efficient graphs. Furthermore some negative results are given, i.e. we present graph classes which do contain non-efficient graphs.
\begin{theorem}
The following graph classes only contain strongly efficient graphs. This holds even without specification of the terminal nodes $s_1,t_1,s_2$ and $t_2$:
\begin{enumerate}[label=(\arabic*),itemsep=-0em,leftmargin=*]
\item \label{effsepa} generalized series-parallel graphs,
\item \label{effcycle} graphs with the property that all cycles have length $\leq 6$,
\item \label{effwheel} wheel and fan graphs.
\end{enumerate}
\end{theorem}

\begin{proof} \
\begin{enumerate}[itemsep=-0em,leftmargin=0.7cm]
\item[\ref{effsepa}] 
We start with generalized series-parallel graphs. These graphs are known to be $K_4$-minor-free, i.e. $K_4$ can not be obtained by a sequence of vertex deletions, edge deletions and / or edge contractions. We will now show that every Bad Configuration contains $K_4$ as a minor, and therefore any generalized series-parallel graph can not contain a Bad Configuration, hence every generalized series-parallel graph is strongly efficient. 

The Bad Configuration BC1a is shown in Figure~\ref{SP1a}. By contracting suitable edges we can w.l.o.g. assume that all displayed lines are in fact edges. Deleting the vertices $u,v,w$ and $x$ yields the graph displayed in Figure~\ref{SP1b}. By contracting the thick edges, one gets the graph shown in Figure~\ref{SP1c} which is a $K_4$.

For the Bad Configuration BC1b, we just need one additional edge contraction. 
It is easy to check (and intuitively clear) that all other types of Bad Configurations must contain $K_4$ as a minor (we leave this to the reader). 

This shows that generalized series-parallel graphs only contain strongly efficient graphs.
\begin{figure}[H] \centering \psset{unit=0.95cm}\psset{dash= 2pt 2pt}
\subfloat[BC1a]{
 \begin{pspicture}(-0.5,-1.2)(7,1.2) 
\Knoten{0}{1}{v1}\nput{180}{v1}{$u$}
\Knoten{0}{-1}{v2}\nput{180}{v2}{$w$}
\Knoten{0.5}{-0.5}{n2}
\Knoten{1}{0}{v3}
\Knoten{2.125}{0}{n5}
\Knoten{3.25}{0}{v4}
\Knoten{4.375}{0}{v6}
\Knoten{5.5}{0}{v7}
\Knoten{6}{0.5}{n3}
\Knoten{6}{-0.5}{n4}
\Knoten{6.5}{1}{v10}\nput{0}{v10}{$v$}
\Knoten{6.5}{-1}{v11}\nput{0}{v11}{$x$}

\ncline[linestyle=dashed]{-}{v1}{v3}
\ncline{-}{n2}{v3}
\ncline[linestyle=dashed]{-}{v2}{n2}
\ncline[linestyle=dashed]{-}{v3}{n5}
\ncline{-}{v4}{n5}
\ncline{-}{v4}{v6}
\ncline{-}{v6}{v7}
\ncline{-}{v7}{n3}
\ncline[linestyle=dashed]{-}{n3}{v10}
\ncline{-}{v7}{n4}
\ncline[linestyle=dashed]{-}{n4}{v11}
\ncarc[arcangle=-35]{v4}{n4}
\ncarc[arcangle=-35]{n2}{v6}
\ncarc[arcangle=35]{n5}{n3}
\end{pspicture}
\label{SP1a}
}
\subfloat[After deleting $u,v,w,x$]{
 \begin{pspicture}(-0.5,-1.2)(7,1.2) 
\Knoten{0.5}{-0.5}{n2}
\Knoten{1}{0}{v3}
\Knoten{2.125}{0}{n5}
\Knoten{3.25}{0}{v4}
\Knoten{4.375}{0}{v6}
\Knoten{5.5}{0}{v7}
\Knoten{6}{0.5}{n3}
\Knoten{6}{-0.5}{n4}

\ncline[linewidth=2.5pt]{-}{n2}{v3}
\ncline[linewidth=2.5pt,linestyle=dashed]{-}{v3}{n5}
\ncline{-}{v4}{n5}
\ncline{-}{v4}{v6}
\ncline{-}{v6}{v7}
\ncline[linewidth=2.5pt]{-}{v7}{n3}
\ncline[linewidth=2.5pt]{-}{v7}{n4}
\ncarc[arcangle=-35]{v4}{n4}
\ncarc[arcangle=-35]{n2}{v6}
\ncarc[arcangle=35]{n5}{n3}
\end{pspicture}\label{SP1b}
} \\
\subfloat[After contracting the thick edges]{
 \begin{pspicture}(0.5,-1.2)(7.5,1.5) 
\Knoten{1}{0}{n5}
\Knoten{3}{0}{v4}
\Knoten{5}{0}{v6}
\Knoten{7}{0}{v7}

\ncline{-}{v4}{n5}
\ncline{-}{v4}{v6}
\ncline{-}{v6}{v7}
\ncarc[arcangle=-35]{v4}{v7}
\ncarc[arcangle=-35]{n5}{v6}
\ncarc[arcangle=35]{n5}{v7}
\end{pspicture}\label{SP1c}
}
\caption{$K_4$-minor of BC1a}
\label{SC1}
\end{figure}
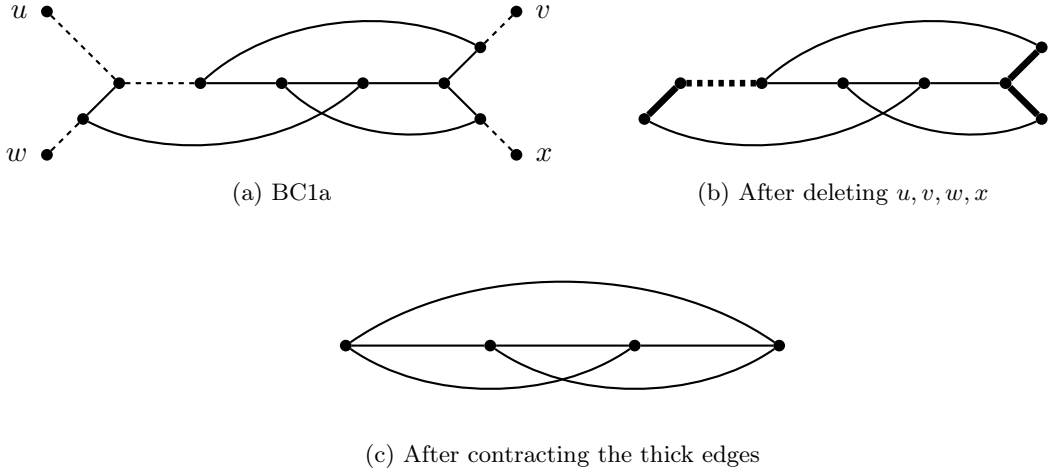

\item[\ref{effcycle}]
As all Bad Configurations contain at least one cycle with length greater or equal to seven, graphs with the property that all cycles have length at most six are strongly efficient. 

\item[\ref{effwheel}]
Taking the cycle $C_n$ and adding an additional vertex $z$, which is adjacent to all vertices of $C_n$, yields the Wheel $W_n$. We will now show that for every $n$, the Wheel $W_n$ can not contain a Bad Configuration. Assume (on the contrary) that there is a BC1a (the other Bad Configurations follow similarly). Figure~\ref{wheelbc} displays BC1a together with some notations needed in the following.

If the paths $a,b,c,d$ are all part of $C_n$, i.e. $z$ is not contained in one of these paths, $q_2$ and $q_3$ both have to contain $z$. Since $q_2$ and $q_3$ have to be node-disjoint in a BC1a, this is not possible and $z$ has to be contained in one of the paths $a,b,c,d$.
More exactly, $z$ is either the endnode $y$ of $d$ or the node which is adjacent to $y$ in $d$. Otherwise the paths $f$ and $g$ can not exist, since they both have to contain at least one edge (and can not contain $z$). Figure~\ref{wheel} displays the subcase that $z$ is adjacent to $y$ in $d$.
%
But both remaining subcases also yield contradictions: On the one hand $q_1$ has to be node-disjoint with $d$, but on the other hand it also has to contain $z$ (by taking into account that $q_1$ has to end in $b$, this holds even if $v$ is an endnode of $q_1$).
Therefore, $W_n$ can not contain a BC1a.
\begin{figure}[H] \centering \psset{unit=1.1cm}
\subfloat[BC1a]{
\begin{pspicture}(-0.5,-2.2)(7,1.2) \psset{dash= 2pt 2pt}
\Knoten{0}{1}{v1}\nput{180}{v1}{$u$}
\Knoten{0}{-1}{v2}\nput{180}{v2}{$w$}
\Knoten{0.5}{-0.5}{n2}
\Knoten{1}{0}{v3}
\Knoten{2.125}{0}{n5}
\Knoten{3.25}{0}{v4}
\Knoten{4.375}{0}{v6}
\Knoten{5.5}{0}{v7}\nput{110}{v7}{$y$}
\Knoten{6.25}{0.75}{n3}
\Knoten{6.25}{-0.75}{n4}
\Knoten{6.5}{1}{v10}\nput{0}{v10}{$v$}
\Knoten{6.5}{-1}{v11}\nput{0}{v11}{$x$}

\ncline[linestyle=dashed]{-}{v1}{v3}
\ncline{-}{n2}{v3}
\ncline[linestyle=dashed]{-}{v2}{n2}
\ncline[linestyle=dashed]{-}{v3}{n5}\ncput{\colorbox{almostwhite}{$a$}}
\ncline{-}{v4}{n5}\ncput{\colorbox{almostwhite}{$b$}}
\ncline{-}{v4}{v6}\ncput{\colorbox{almostwhite}{$c$}}
\ncline{-}{v6}{v7}\ncput{\colorbox{almostwhite}{$d$}}
\ncline{-}{v7}{n3}\ncput{\colorbox{almostwhite}{$f$}}
\ncline[linestyle=dashed]{-}{n3}{v10}
\ncline{-}{v7}{n4}\ncput{\colorbox{almostwhite}{$g$}}
\ncline[linestyle=dashed]{-}{n4}{v11}
\Bogen{n5}{n3}{q_1}{35}
\Bogen{n2}{v6}{q_2}{-35}
\Bogen{v4}{n4}{q_3}{-35}
\end{pspicture}\label{wheelbc}
}
\subfloat[Subcase $z$ is adjacent to $y$]{
 \begin{pspicture}(1,0.3)(7,4.5) 
\psrotate(4.5,2.3){-90}{
\Knoten{4}{4}{y}\nput[rot=90]{90}{y}{$y$}
\Knoten{5.5}{3}{v1}
\Knoten{6.2}{1.5}{v2}\nput[rot=90]{80}{v2}{$x$}
\Knoten{2.5}{3}{v3}
\Knoten{1.8}{1.5}{v4}\nput[rot=90]{100}{v4}{$v$}
\Knotensmall{6.2}{0.7}{v5}
\Knotensmall{1.8}{0.7}{v6}
\Knoten{4}{1}{z}\nput[rot=90]{235}{z}{$z$}
\Knotensmall{3}{0.5}{v7}
\Knotensmall{5}{0.5}{v8}
\Knotensmall{4}{0.2}{v9}

\ncline[linestyle=dashed,linewidth=3pt,dash=1.2pt 1.2pt]{-}{v1}{v2}
\ncline[linestyle=dashed,linewidth=3pt,dash=1.2pt 1.2pt]{-}{y}{v1}
\ncline[linestyle=dashed,linewidth=3pt,dash=1.2pt 1.2pt]{-}{y}{v3}
\ncline[linestyle=dashed,linewidth=3pt,dash=1.2pt 1.2pt]{-}{v3}{v4}
\ncline{z}{v1}
\ncline{-}{z}{v2}
\ncline{-}{z}{v3}
\ncline{-}{z}{v4}
\ncline[linewidth=2.5pt]{-}{z}{y}

\ncline[linestyle=dashed]{-}{v2}{v5}
\ncline[linestyle=dashed]{-}{v4}{v6}
\ncline[linestyle=dashed]{-}{z}{v7}
\ncline[linewidth=2.5pt, linestyle=dashed]{z}{v9}
\ncline[linestyle=dashed]{-}{z}{v8}
}
\end{pspicture}\label{wheel}
}
\caption{Situation for $W_n$}
\end{figure}
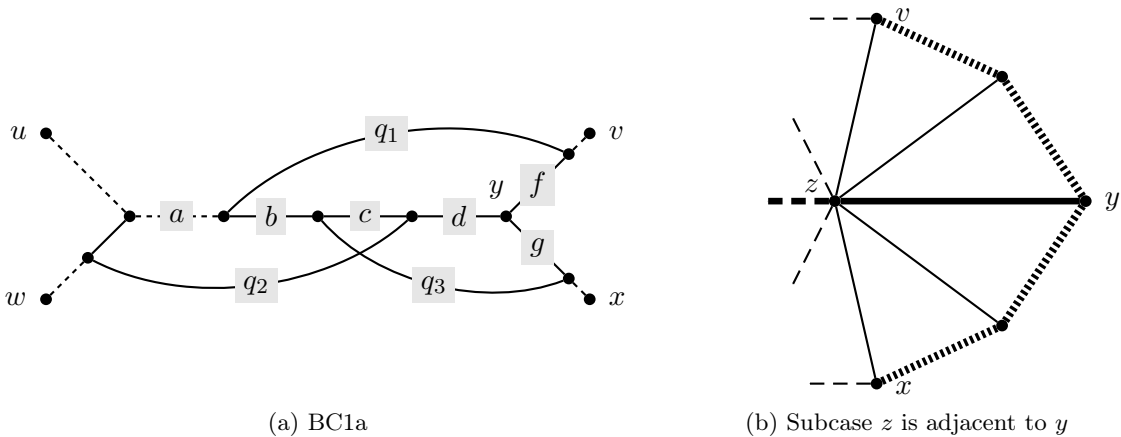

A $n$-fan is a graph consisting of a path of length $n$ and an additional vertex which is adjacent to all vertices of the path. As these graphs are subgraphs of wheels, fan graphs are also strongly efficient.
\end{enumerate} 
\vspace{-2em}
\end{proof}

Note that other $K_4$-minor-free graphs are Cactus graphs, unicyclic graphs and trees. Therefore these classes are also strongly efficient.

As all Bad Configurations contain at least 9 edges, all graphs with at most 8 edges are strongly efficient.
Furthermore, all graphs with at most 6 vertices are strongly efficient (as all Bad Configurations contain at least 7 vertices).

\vspace{1em}
Now we give some negative results. 
\begin{proposition}
The following graph classes contain non-efficient graphs:
\begin{enumerate}[label=(\arabic*),itemsep=-0em,leftmargin=*]
	\item \label{neffbip} bipartite graphs,
	\item \label{neffchord} chordal graphs, 
	\item \label{neffplan} planar graphs.
\end{enumerate}
\end{proposition}

\begin{proof}
We show that there exists a graph which contains a BC1a in each of the above mentioned graph classes. 
\begin{enumerate}[itemsep=-0em,leftmargin=0.7cm]
	\item[\ref{neffbip}] Figure~\ref{bipartit} shows a bipartite graph (large and small nodes constitute the bipartition) which is a BC1a.
	\item[\ref{neffchord}] BC1a is obviously contained in a complete graph (with a suitable number of vertices) and complete graphs are chordal. 
	\item[\ref{neffplan}] Figure~\ref{planar} shows a planar graph which is a BC1a.
\end{enumerate}

\begin{figure}[H] \centering \psset{unit=0.75cm}
\subfloat[A bipartite graph which is a BC1a]{
\begin{pspicture}(-1.3,-1.45)(7.9,1.15) 
\Knotenthick{-1.1}{1}{v1}\nput{180}{v1}{$s_1$}
\Knotenthick{-1.1}{-1}{v2}\nput{180}{v2}{$s_2$}
\Knoten{0.4}{0}{v3}
\Knotenthick{1.4}{0}{v4}
\Knoten{2.4}{0}{v5}
\Knotenthick{3.4}{0}{v6}
\Knoten{4.4}{0}{v7}
\Knotenthick{5.4}{0}{v8}
\Knoten{6.9}{1}{v9}\nput{0}{v9}{$t_1$}
\Knoten{6.9}{-1}{v10}\nput{0}{v10}{$t_2$}

\ncline{-}{v1}{v3}
\ncline{-}{v2}{v3}
\ncline{-}{v3}{v8}
\ncline{v8}{v9}
\ncline{-}{v8}{v10}

\Bogendef{v2}{v7}{q_2}{-25}
\Bogendef{v4}{v9}{q_1}{30}
\Bogendef{v6}{v10}{q_3}{-35}
\end{pspicture} \label{bipartit}
}
\subfloat[A planar graph which is a BC1a]{
\begin{pspicture}(-2.1,-1.45)(7.1,1.15) 
\Knoten{-1.1}{1}{v1}\nput{180}{v1}{$s_1$}
\Knoten{-1.1}{-1}{v2}\nput{180}{v2}{$s_2$}
\Knoten{0.4}{0}{v3}
\Knoten{1.4}{0}{v4}
\Knoten{3.4}{0}{v6}
\Knoten{4.4}{0}{v7}
\Knoten{5.4}{0}{v8}
\Knoten{6.9}{1}{v9}\nput{0}{v9}{$t_1$}
\Knoten{4.7}{0.6}{v10}\nput{0}{v10}{$t_2$}

\ncline{-}{v1}{v3}
\ncline{-}{v2}{v3}
\ncline{-}{v3}{v8}
\ncline{v8}{v9}
\ncline{-}{v8}{v10}

\Bogendef{v2}{v7}{q_2}{-25}
\Bogendef{v4}{v9}{q_1}{30}
\Bogendef{v6}{v10}{q_3}{25}
\end{pspicture}
\label{planar}
}
\caption{Graphs containing a BC1a}
\end{figure}
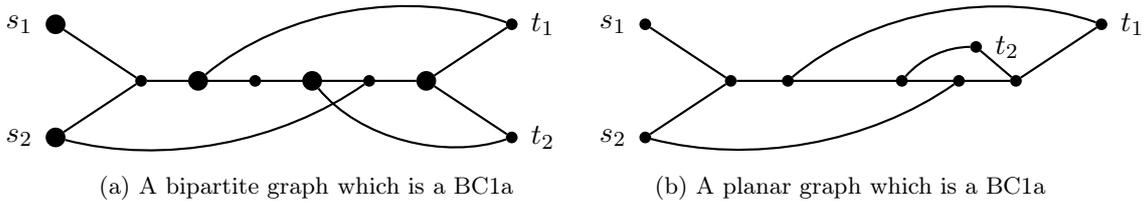
\vspace{-2em}
\end{proof}

We want to note that even if all edges of a graph have the same cost, an optimal Steiner forest may not be enforceable: To see that, observe that the costs displayed in Figure~\ref{fig1bc1} are all integral. Contracting the edges with cost 0, replacing all remaining edges by paths with length equal to the edge costs and assigning cost 1 to all edges yields a graph with the desired property. 


%
%
%
%
%

\section{A Lower Bound for the PoS for Two Player Games}
In this section we derive a lower bound for the price of stability. 

\begin{proposition}\label{theopos}
The PoS for two-player undirected network design games is at least $\frac{15}{14}$. 
\end{proposition}

\begin{proof}
Let us consider the instance of a Bad Configuration shown in Figure~\ref{posgraph}, where the edges are labelled with their costs and $x \in \mathbb{R}^+$. 

By considering all possible Steiner forests, one can easily show that the unique optimal Steiner forest consists of all solid edges and has cost $14x+8$, whereas all other possible Steiner forests have cost at least $15x+8$. 
\begin{figure}[H] \centering \psset{unit=1.5cm}
 \begin{pspicture}(-0.5,-1.5)(7,1.5) 
\Knoten{0}{-1}{v2}\nput{180}{v2}{$s_2$}
\Knoten{1}{0}{v3}\nput{180}{v3}{$s_1$}
\Knoten{2.5}{0}{v4}
\Knoten{4}{0}{v6}
\Knoten{5.5}{0}{v7}
\Knoten{6.5}{1}{v10}\nput{0}{v10}{$t_1$}
\Knoten{6.5}{-1}{v11}\nput{0}{v11}{$t_2$}

\Kante{v2}{v3}{3x+2}
\Kante{v3}{v4}{2x+1}
\Kante{v4}{v6}{x+1}
\Kante{v6}{v7}{2x+1}
\Kante{v7}{v10}{3x+1}
\Kante{v7}{v11}{3x+2}
\Bogendashed{v4}{v11}{4x+2}{-35}
\Bogendashed{v2}{v6}{4x+2}{-35}
\Bogendashed{v3}{v10}{6x+3}{35}
\end{pspicture}
\caption{A lower bound for the PoS} \label{posgraph}
\end{figure}
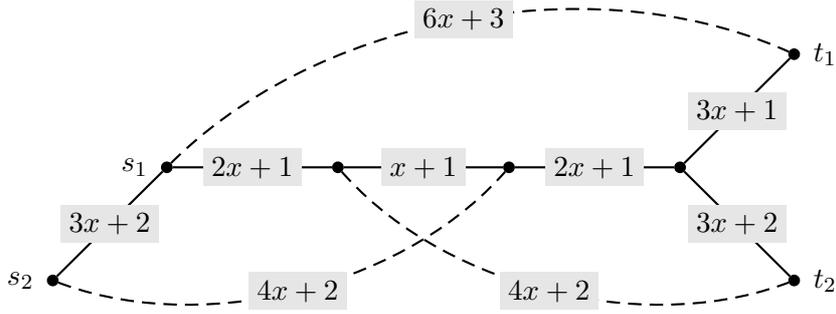
Furthermore, the optimal Steiner forest is not enforceable since the sum of collectable cost shares is obviously bounded above by $14x+7$ ($6x+3$ for Player 1 and $8x+4$ for Player 2). 
But there is a Steiner forest with cost $15x+8$ which is enforceable: Player 1 takes the edge $s_1-t_1$ with cost $6x+3$, Player 2 takes the two dashed edges with cost $4x+2$ and the solid edge with cost $x+1$. This yields an enforceable Steiner forest since both players take (disjoint) shortest paths.
Therefore, the PoS for two-player undirected network design games is at least $\frac{15x+8}{14x+8}$. As $x$ tends to infinity, we achieve a lower bound for the PoS of $\frac{15}{14}$. 
\end{proof}
\section{Summary and Open Problems} \label{secsummary}
We derived a complete characterization of efficient
graphs for two-player network design games
 showing
 that a graph is efficient iff certain forbidden subgraphs
 are not present. Our work leads to several interesting research questions
 as outlined below.
 \begin{itemize}
 \item What is the computational complexity of recognizing
 a Bad Configuration?
 \item  How does a characterization look like for three or more players?
 \end{itemize}
Our characterization prescribes substructures of worst-case
instances regarding the long-standing PoS question
for separable protocols. 
We conjecture the following bounds:
\begin{conjecture}
The PoS for two-player undirected network design games is  $15/14$ (\autoref{theopos}).
\end{conjecture}
\begin{conjecture}
The PoS for undirected network design games with $n$ players 
is $<2$.
\end{conjecture}
\begin{conjecture}
The PoS for directed network design games with $n$ players is $2$.
\end{conjecture}
\newpage
\section*{Appendix}
\appendix

\section{Proof of our Main Theorem}\label{sec:main_proof}

\subsection{\ref{maintheo2} implies \ref{maintheo1}} \label{subsecproof}
We show this by contraposition, thus we assume that $(G, (s_1,t_1), (s_2,t_2))$ contains a Bad Configuration. Then we define a cost function $c$ so that the optimal Steiner forest is unique and not enforceable, showing the claim. 
	
	To this end we choose a subgraph that is a BC and set $c(e)=\infty$ if the edge $e$ is not contained in this subgraph. The costs of the edges in the BC depend on the type of the BC, thus we now have to distinguish between the different types of BCs.
	In the following we have displayed costs for each type, where one should note:
	\begin{itemize}[itemsep=-0em,leftmargin=*]
		\item we have always chosen $(u,v,w,x)=(s_1,t_1,s_2,t_2)$ since the other cases are analogous;
		\item paths with cost $>0$ are labelled with their cost; 
		\item if a path consists of more than one edge, one can choose the costs of the corresponding edges arbitrarily so that they sum up to the displayed cost on the path;
		\item all paths with nonzero costs contain at least one edge because of the definition of the corresponding type of BC; 
		\item the solid lines always build the (unique) optimal Steiner Forest OPT (with cost $c(OPT)$);
		\item we do not display costs for all subcases (not displayed subcases can be treated analogously).
	\end{itemize}

	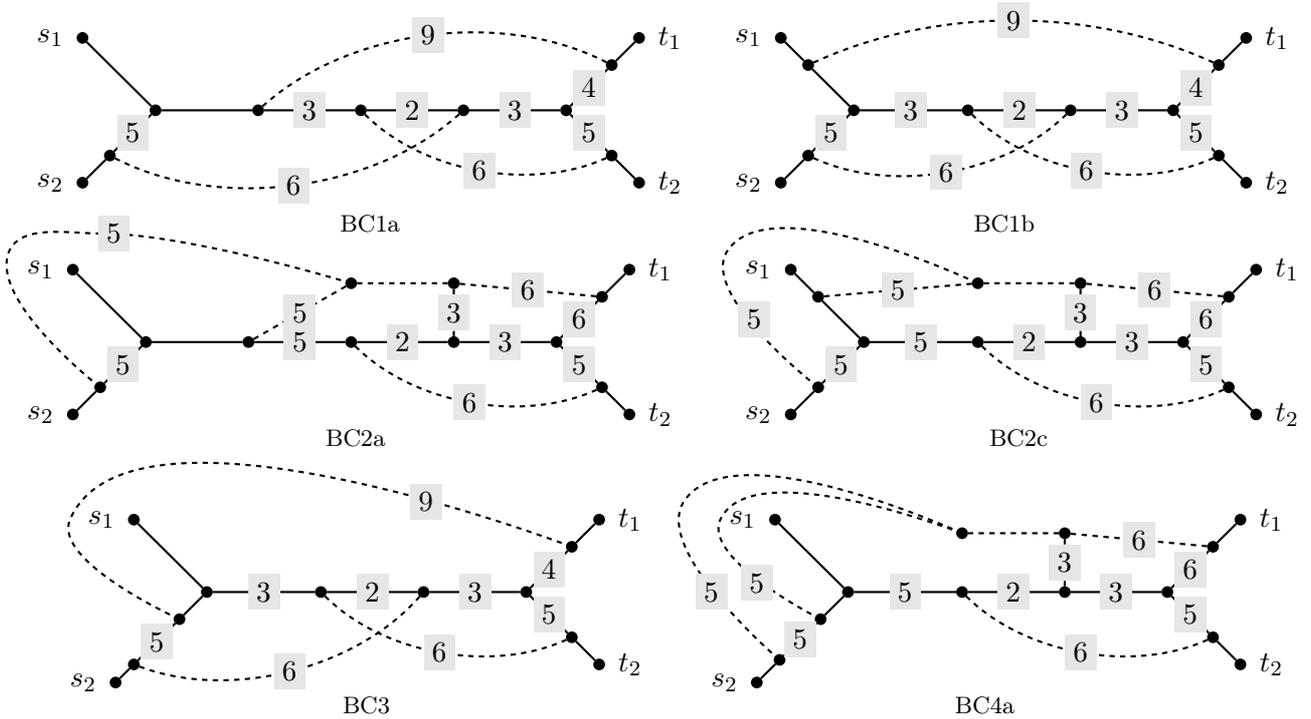
\begin{figure}[H] \centering \psset{unit=1.2cm}\captionsetup[subfigure]{labelformat=empty}
	\psset{dash= 2pt 2pt}
	\subfloat[BC1a]{
 \begin{pspicture}(0.2,-1)(6.5,1) 
\Knoten{0.2}{0.8}{v1}\nput{180}{v1}{$s_1$}
\Knoten{0.2}{-0.8}{v2}\nput{180}{v2}{$s_2$}
\Knoten{0.5}{-0.5}{n2}
\Knoten{1}{0}{v3}
\Knoten{2.125}{0}{n5}
\Knoten{3.25}{0}{v4}
\Knoten{4.375}{0}{v6}
\Knoten{5.5}{0}{v7}
\Knoten{6}{0.5}{n3}
\Knoten{6}{-0.5}{n4}
\Knoten{6.3}{0.8}{v10}\nput{0}{v10}{$t_1$}
\Knoten{6.3}{-0.8}{v11}\nput{0}{v11}{$t_2$}

\ncline{v1}{v3}
\ncline{v2}{n2}
\Kante{n2}{v3}{5}
\ncline{v3}{n5}
\Kante{n5}{v4}{3}
\Kante{v4}{v6}{2}
\Kante{v6}{v7}{3}
\Kante{v7}{n3}{4}
\ncline{v10}{n3}
\Kante{v7}{n4}{5}
\ncline{v11}{n4}
\Bogendashed{n5}{n3}{9}{35}
\Bogendashed{n2}{v6}{6}{-35}
\Bogendashed{v4}{n4}{6}{-35}
\end{pspicture} 
}
\hspace{1cm}
\subfloat[BC1b]{
 \begin{pspicture}(1,-1)(6.3,1) 
\Knoten{1.2}{0.8}{v1}\nput{180}{v1}{$s_1$}
\Knoten{1.5}{0.5}{n1}
\Knoten{1.2}{-0.8}{v2}\nput{180}{v2}{$s_2$}
\Knoten{1.5}{-0.5}{n2}
\Knoten{2}{0}{v3}
\Knoten{3.25}{0}{v4}
\Knoten{4.375}{0}{v6}
\Knoten{5.5}{0}{v7}
\Knoten{6}{0.5}{n3}
\Knoten{6}{-0.5}{n4}
\Knoten{6.3}{0.8}{v10}\nput{0}{v10}{$t_1$}
\Knoten{6.3}{-0.8}{v11}\nput{0}{v11}{$t_2$}

\ncline{v1}{n1}
\ncline{n1}{v3}
\ncline{v2}{n2}
\Kante{n2}{v3}{5}
\Kante{v3}{v4}{3}
\Kante{v4}{v6}{2}
\Kante{v6}{v7}{3}
\Kante{v7}{n3}{4}
\ncline{v10}{n3}
\Kante{v7}{n4}{5}
\ncline{v11}{n4}
\Bogendashed{n1}{n3}{9}{25}
\Bogendashed{n2}{v6}{6}{-35}
\Bogendashed{v4}{n4}{6}{-35}
\end{pspicture} 
}

	\subfloat[BC2a]{
\begin{pspicture}(0.2,-0.8)(6.4,0.8)
\Knoten{0.2}{0.8}{v1}\nput{180}{v1}{$s_1$}
\Knoten{0.2}{-0.8}{v2}\nput{180}{v2}{$s_2$}
\Knoten{0.5}{-0.5}{n2}
\Knoten{1}{0}{v3}
\Knoten{2.125}{0}{n5}
\Knoten{3.25}{0}{v4}
\Knoten{4.375}{0}{v6}
\Knoten{5.5}{0}{v7}
\Knoten{6}{0.5}{n3}
\Knoten{6}{-0.5}{n4}
\Knoten{6.3}{0.8}{v10}\nput{0}{v10}{$t_1$}
\Knoten{6.3}{-0.8}{v11}\nput{0}{v11}{$t_2$}
\Knoten{3.25}{0.65}{n6}
\Knoten{4.375}{0.65}{n7}

\Kantedashed{n5}{n6}{5}
\ncline[linestyle=dashed]{n6}{n7}
\Kantedashed{n7}{n3}{6}
\Kantedashed{n7}{v6}{3}
\ncline{v1}{v3}
\Kante{n2}{v3}{5}
\ncline{v2}{n2}
\ncline{v3}{n5}
\Kante{n5}{v4}{5}
\Kante{v4}{v6}{2}
\Kante{v6}{v7}{3}
\Kante{v7}{n3}{6}
\ncline{n3}{v10}
\Kante{v7}{n4}{5}
\ncline{n4}{v11}

\pscurve[linestyle=dashed](0.5,-0.5)(-0.5,1)(0.5,1.2)(3.25,0.65)\rput{0}(0.65,1.2){\colorbox{almostwhite}{$5$}}
\Bogendashed{v4}{n4}{6}{-35}

\end{pspicture}
}
\hspace{1.5cm}
\subfloat[BC2c]{
\begin{pspicture}(1.1,-0.8)(6.3,0.8) 

\Knoten{1.2}{0.8}{v1}\nput{180}{v1}{$s_1$}
\Knoten{1.5}{0.5}{n1}
\Knoten{1.2}{-0.8}{v2}\nput{180}{v2}{$s_2$}
\Knoten{1.5}{-0.5}{n2}
\Knoten{2}{0}{v3}
\Knoten{3.25}{0}{v4}
\Knoten{4.375}{0}{v6}
\Knoten{5.5}{0}{v7}
\Knoten{6}{0.5}{n3}
\Knoten{6}{-0.5}{n4}
\Knoten{6.3}{0.8}{v10}\nput{0}{v10}{$t_1$}
\Knoten{6.3}{-0.8}{v11}\nput{0}{v11}{$t_2$}
\Knoten{3.25}{0.65}{n6}
\Knoten{4.375}{0.65}{n7}

\Kantedashed{n1}{n6}{5}
\ncline[linestyle=dashed]{n6}{n7}
\Kantedashed{n7}{n3}{6}
\Kantedashed{n7}{v6}{3}
\ncline{v1}{n1}
\ncline{n1}{v3}
\Kante{n2}{v3}{5}
\ncline{v2}{n2}
\Kante{v3}{v4}{5}
\Kante{v4}{v6}{2}
\Kante{v6}{v7}{3}
\Kante{v7}{n3}{6}
\ncline{n3}{v10}
\Kante{v7}{n4}{5}
\ncline{n4}{v11}

\pscurve[linestyle=dashed](1.5,-0.5)(0.5,1)(3.25,0.65)\rput{0}(0.8,0.25){\colorbox{almostwhite}{$5$}}

\Bogendashed{v4}{n4}{6}{-35}
\end{pspicture}
}

\subfloat[BC3]{
	\begin{pspicture}(0.5,-1)(7,1.2) 
\Knoten{1.2}{0.8}{v1}\nput{180}{v1}{$s_1$}
\Knoten{1}{-1}{v2}\nput{180}{v2}{$s_2$}
\Knoten{1.2}{-0.8}{n6}
\Knoten{1.7}{-0.3}{n2}
\Knoten{2}{0}{v3}
\Knoten{3.25}{0}{v4}
\Knoten{4.375}{0}{v6}
\Knoten{5.5}{0}{v7}
\Knoten{6}{0.5}{n3}
\Knoten{6}{-0.5}{n4}
\Knoten{6.3}{0.8}{v10}\nput{0}{v10}{$t_1$}
\Knoten{6.3}{-0.8}{v11}\nput{0}{v11}{$t_2$}

\ncline{v1}{v3}
\ncline{v2}{n6}
\Kante{n6}{n2}{5}
\ncline{n2}{v3}
\Kante{v3}{v4}{3}
\Kante{v4}{v6}{2}
\Kante{v6}{v7}{3}
\Kante{v7}{n3}{4}
\ncline{n3}{v10}
\Kante{v7}{n4}{5}
\ncline{n4}{v11}
\pscurve[linestyle=dashed](6,0.5)(0.5,0.8)(1.7,-0.3)\rput(4.4,1){\colorbox{almostwhite}{9}}
\Bogendashed{n6}{v6}{6}{-35}
\Bogendashed{v4}{n4}{6}{-35}
\end{pspicture}
}
\hspace{0.25cm}
\subfloat[BC4a]{
\begin{pspicture}(0.5,-1)(6.5,1.2) 
\Knoten{1.2}{0.8}{v1}\nput{180}{v1}{$s_1$}
\Knoten{1}{-1}{v2}\nput{180}{v2}{$s_2$}
\Knoten{1.7}{-0.3}{n2}
\Knoten{2}{0}{v3}
\Knoten{3.25}{0}{v4}
\Knoten{4.375}{0}{v6}
\Knoten{5.5}{0}{v7}
\Knoten{6}{0.5}{n3}
\Knoten{6}{-0.5}{n4}
\Knoten{6.3}{0.8}{v10}\nput{0}{v10}{$t_1$}
\Knoten{6.3}{-0.8}{v11}\nput{0}{v11}{$t_2$}
\Knoten{1.25}{-0.75}{n9}
\Knoten{3.25}{0.65}{n6}
\Knoten{4.375}{0.65}{n7}

\ncline[linestyle=dashed]{n6}{n7}
\ncline[linestyle=dashed]{n7}{n3}\ncput{\colorbox{almostwhite}{$6$}}
\ncline[linestyle=dashed]{n7}{v6}\ncput{\colorbox{almostwhite}{$3$}}
\ncline{v1}{v3}
\ncline{n2}{v3}
\ncline{v2}{n9}
\Kante{n9}{n2}{5}
\Kante{v3}{v4}{5}
\Kante{v4}{v6}{2}
\Kante{v6}{v7}{3}
\Kante{v7}{n3}{6}
\ncline{n3}{v10}
\Kante{v7}{n4}{5}
\ncline{n4}{v11}
\pscurve[linestyle=dashed](1.7,-0.3)(0.6,0.8)(3.25,0.65)\rput{0}(1,0.1){\colorbox{almostwhite}{$5$}}
\pscurve[linestyle=dashed](1.25,-0.75)(0.2,1)(3.25,0.65)\rput{0}(0.5,0){\colorbox{almostwhite}{$5$}}
\Bogendashed{v4}{n4}{6}{-35}
\end{pspicture}
}
\caption{For BC1a, b and BC3: $c(OPT)=22$; sum of cost shares $\leq 9+6+6=21$; \\
for BC2a, c and BC4a: $c(OPT)=26$; sum of cost shares $\leq 11+8+6=25$.}
\label{costsBC}
\end{figure}
\subsection{\ref{maintheo1} implies \ref{maintheo3}} \label{subsec:proof_2}
Consider an arbitrary optimal Steiner forest $F$ (w.r.t. an arbitrary cost function $c$) and an optimal solution $(\xi_{i,e})_{i \in N, e \in P_i}$ of the corresponding \lp{}. If condition (\ref{Nash2}) is satisfied, the statement is true. So we assume that there is an edge that is not completely paid. 
Note that $P_1 \cap P_2$ has to contain at least one edge, otherwise the statement is clear. Furthermore, since $F$ contains no cycles, $P_1 \cap P_2$ has to be a simple path (in the following, when we speak of a path, we always mean a simple path). We refer to the edges of $P_1 \cap P_2$ as the \emph{commonly used edges} or the \emph{commonly used part}.
To obtain \ptl-cost shares (see~\autoref{def:ptl} in~\autoref{subsec:proof_sketch}), we apply Algorithm~\nameref{push} (see \autoref{subsubsec:push}; this procedure terminates and yields the desired property).
Let $e$ be the first edge (with respect to the order displayed in Figure~\ref{order}) which is not completely paid according to the computed \ptl-cost shares $(\xi_{i,e})_{i \in N, e \in P_i}$.

\begin{figure}[H] \centering \psset{unit=1cm}
 \begin{pspicture}(-3.5,-3)(9.5,3) 
\Knoten{-3}{3}{v1}\nput{180}{v1}{$s_1$}
\Knoten{-3}{-3}{v2}\nput{180}{v2}{$s_2$}
\Knoten{-2}{2}{v3}
\Knoten{-1}{1}{v4}
\Knoten{-2}{-2}{v5}
\Knoten{-1}{-1}{v6}
\Knoten{0}{0}{v7}
\Knoten{2.5}{0}{v8}
\Knoten{3.5}{0}{v9}
\Knoten{6}{0}{v10}
\Knoten{7}{1}{v11}
\Knoten{7}{-1}{v12}
\Knoten{8}{2}{v13}
\Knoten{8}{-2}{v14}
\Knoten{9}{3}{v15}\nput{0}{v15}{$t_1$}
\Knoten{9}{-3}{v16}\nput{0}{v16}{$t_2$}

\Kante{v1}{v3}{1}
\ncline[linestyle=dotted]{v3}{v4}
\Kante{v4}{v7}{\ell_1}
\Kante{v2}{v5}{\ell_1+1}
\ncline[linestyle=dotted]{v5}{v6}
\Kante{v7}{v6}{\ell_1+\ell_2}
\Kante{v7}{v8}{\ell_1+\ell_2+1}
\ncline[linestyle=dotted]{v8}{v9}
\Kante{v10}{v9}{\ell_1+\ell_2+m}
\Kante{v10}{v11}{\ell_1+\ell_2+m+1}
\ncline[linestyle=dotted]{v11}{v13}
\Kante{v13}{v15}{\ell_1+\ell_2+m+r_1}
\Kante{v10}{v12}{\ell_1+\ell_2+m+r_1+1}
\ncline[linestyle=dotted]{v12}{v14}
\Kante{v14}{v16}{\ell_1+\ell_2+m+r_1+r_2}
\end{pspicture}
\caption{Ordering of the edges of $F$}
\label{order}
\end{figure}
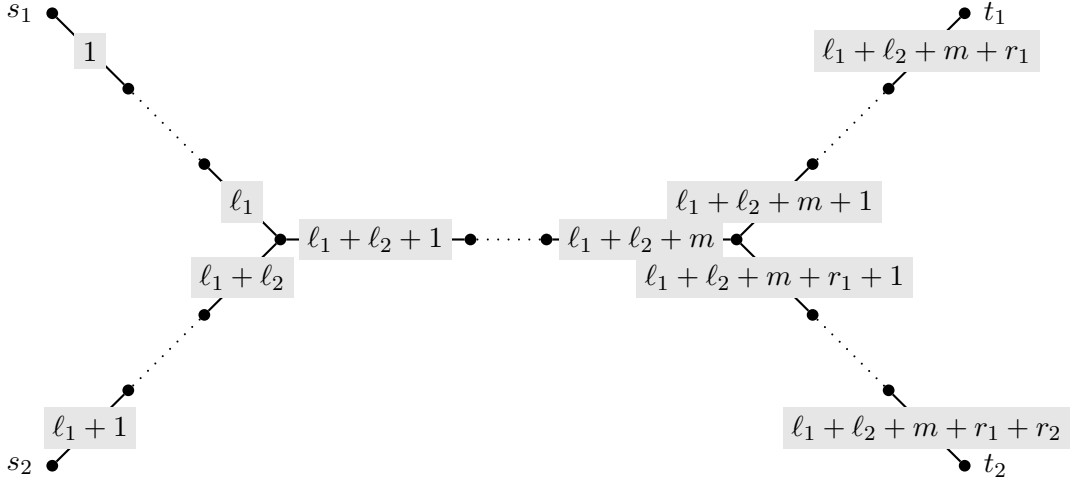

We distinguish between the following three cases:
\begin{enumerate}[itemsep=-0em]
	\item[]Case $L$: \ $\ord(e) \leq \ell_1+\ell_2$,
	\item[]Case $M$: \ $\ell_1+\ell_2+1 \leq \ord(e) \leq \ell_1+\ell_2+m$,
	\item[]Case $R$: \ $\ell_1+\ell_2+m+1 \leq \ord(e)$.
\end{enumerate}


The Cases $M$ and $R$ have further subcases, mostly depending on properties of certain paths which we construct during the cases. Figure~\ref{tree} illustrates the case distinction for the most important subcases, where the nodes are labelled with the subcases (note that the Subcases $R.1.1$, $R.1.2$, $R.1.2.1$ and $R.2$ have further subsubcases which are illustrated in Figure~\ref{tree2}). 
In each of these subcases, we get a contradiction (see \autoref{lemma1}, \autoref{lemma2} and \autoref{lemma3}), therefore each edge in $F$ has to be completely paid and thus $F$ is enforceable.



We now introduce some notation we will use throughout the paper. 
\begin{itemize}[itemsep=-0em,leftmargin=*]
	\item For $\alpha \in \{1,\ldots,\ell_1+\ell_2+m+r_1+r_2\}$, we denote the edge of order $\alpha$ with $e_{\alpha}$. In all following figures, we label the edges (except $e$) with their order. 
	\item For $I \subseteq \{1,\ldots,\ell_1+\ell_2+m+r_1+r_2\}$ we call $e_{\alpha}$ the \emph{smallest (largest)} edge with respect to $I$, if $\alpha$ is the smallest (largest) number in $I$.
	\item We call $P_i' \in \mathcal P_i$ a \emph{(tight) alternative} for Player $i$ and edge $e_{\alpha}$, if $e_{\alpha} \notin P_i'$ and the corresponding restriction (\ref{Nash1}) of \lp{} is tight. Adding the edges of $P_i'$ to $P_i$ yields a unique cycle $C$ with $e_{\alpha} \in C$. Let $e_{\beta}$ ($e_{\gamma}$) be the smallest (largest) edge in $C\cap P_i$ and $q := C \setminus P_i$. For our proof it is sufficient to consider the subpath $q$ of $P_i'$ and therefore we refer to $q$ as a \emph{(tight) alternative} for $e_{\alpha}$, defined by $\beta$ and $\gamma$ ($\beta \leq \gamma$; omitting $P_i'$).
	\item If $q$ is a (tight) alternative for Player $i$, defined by $\beta$ and $\gamma$, we will denote the edges $\{e_{\beta}, \ldots, e_{\gamma}\}\cap P_i$ as the \emph{edges which are substituted by $q$}. We call the nodes of these edges (except for the two endnodes of $q$) the \emph{nodes which are substituted by $q$}. 
	\item We call a (tight) alternative $q$, defined by $\beta$ and $\gamma$, a \emph{left alternative } if $\beta < \ell_1+\ell_2$, and a \emph{right alternative} if $\gamma > \ell_1+\ell_2+m$. 
	\item If $q$ (defined by $\beta$ and $\gamma$) and $p$ (defined by $\beta'$ and $\gamma'$) are two left alternatives of a player, we say that $q$ is \emph{smaller (larger)} than $p$ if $\gamma < \gamma'$ ($\gamma > \gamma'$). If $q$ and $p$ are right alternatives, $q$ is \emph{smaller (larger)} than $p$ if $\beta > \beta'$ ($\beta < \beta')$.
\end{itemize} 

Note that an alternative can simultaneously be a right and a left alternative. Furthermore, if $p$ is a smaller alternative than $q$, this means that $p$ substitutes less commonly used edges than $q$.

For illustration, consider the graph in Figure~\ref{erklaerung} where the solid lines form an optimal Steiner forest and the dashed lines are all existing tight alternatives. 

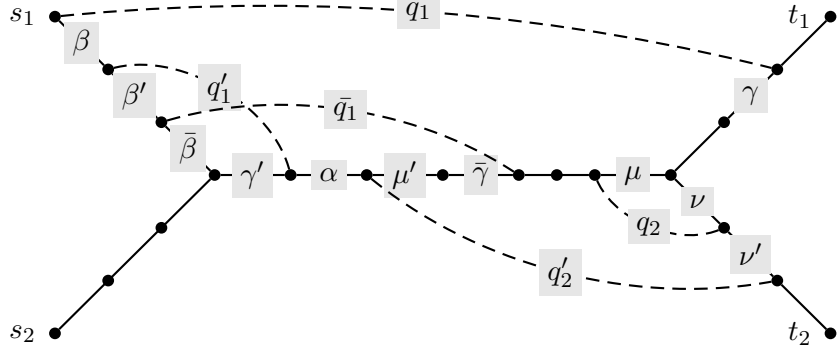
\begin{figure}[H]\centering \psset{unit=1cm}
 \begin{pspicture}(-2.5,-2)(8.2,2.3) 
\Knoten{-2.1}{2.1}{s_1}\nput{180}{s_1}{$s_1$}
\Knoten{-1.4}{1.4}{v1}
\Knoten{-0.7}{0.7}{v2}
\Knoten{-2.1}{-2.1}{s_2}\nput{180}{s_2}{$s_2$}
\Knoten{-1.4}{-1.4}{v3}
\Knoten{-0.7}{-0.7}{v4}
\Knoten{0}{0}{v5}
\Knoten{1}{0}{v6}
\Knoten{2}{0}{v7}
\Knoten{3}{0}{v8}
\Knoten{4}{0}{v9}
\Knoten{4.5}{0}{v10}
\Knoten{5}{0}{v11}
\Knoten{6}{0}{v12}
\Knoten{6.7}{0.7}{v13}
\Knoten{7.4}{1.4}{v14}
\Knoten{8.1}{2.1}{t_1}\nput{180}{t_1}{$t_1$}
\Knoten{6.7}{-0.7}{v15}
\Knoten{7.4}{-1.4}{v16}
\Knoten{8.1}{-2.1}{t_2}\nput{180}{t_2}{$t_2$}

\Kante{s_1}{v1}{\beta}
\Kante{v1}{v2}{\beta '}
\Kante{v2}{v5}{\bar{\beta}}
\ncline{s_2}{v5}
\Kante{v5}{v6}{\gamma '}
\Kante{v6}{v7}{\alpha}
\Kante{v7}{v8}{\mu '}
\Kante{v8}{v9}{\bar{\gamma}}
\ncline{v9}{v11}
\Kante{v11}{v12}{\mu}
\ncline{v12}{v13}
\Kante{v13}{v14}{\gamma}
\ncline{v14}{t_1}
\Kante{v12}{v15}{\nu}
\Kante{v15}{v16}{\nu '}
\ncline{v16}{t_2}

\Bogendashed{s_1}{v14}{q_1}{10}
\Bogendashed{v1}{v6}{q_1'}{45}
\Bogendashed{v2}{v9}{\bar{q_1}}{25}
\Bogendashed{v7}{v16}{q_2'}{-25}
\Bogendashed{v11}{v15}{q_2}{-45}
\end{pspicture}
\caption{Used notation}
\label{erklaerung}
\end{figure}
The alternatives $q_1$ (defined by $\beta$ and $\gamma$), $q_1'$ (defined by $\beta'$ and $\gamma'$) and $\bar{q_1}$ (defined by $\bar{\beta}$ and $\bar{\gamma}$) are all left alternatives of Player 1, where $q_1$ (which is also a right alternative) is the largest and $q_1'$ is the smallest left alternative. We will sometimes need the smallest (or largest) alternative that substitutes a certain edge. For example, $\bar{q_1}$ is the smallest left alternative that substitutes $e_{\alpha}$. 
 For Player 2, $q_2$ (defined by $\mu$ and $\nu$) and $q_2'$ (defined by $\mu'$ and $\nu'$) are right alternatives where $q_2'$ is larger than $q_2$.

\begin{figure}[H]
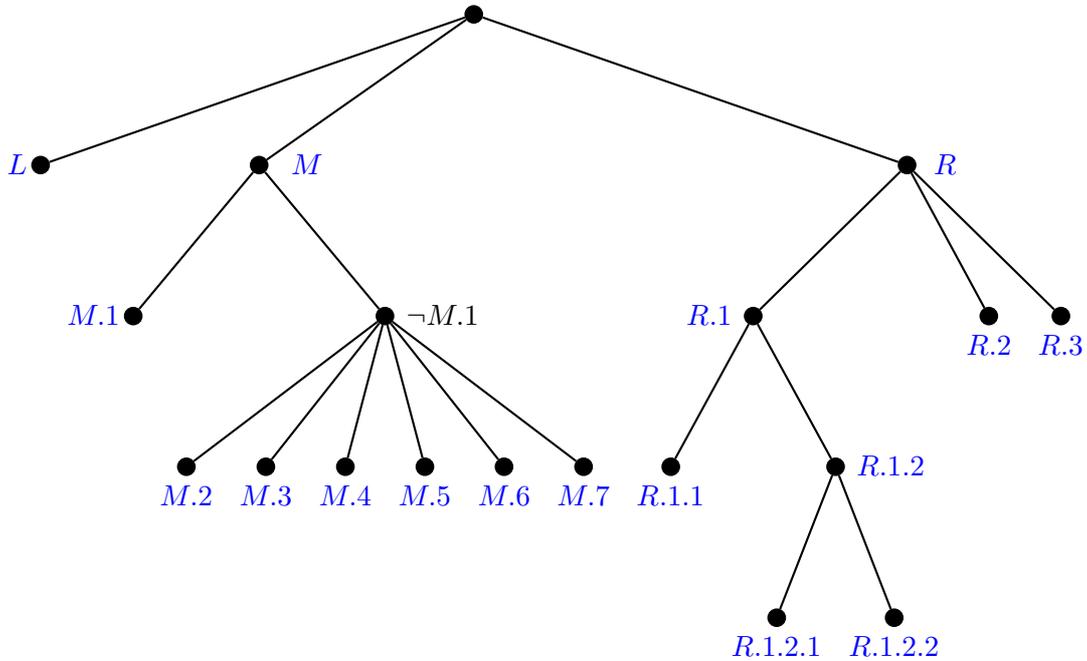
\centering\psset{treemode=D, dotsize=7pt, treefit=loose}
\pstree[treesep=10pt]{\Tdot}{
	\Tdot~[tnpos=l]{\hyperlink{l}{$L$}}
	\pstree{\Tdot~[tnpos=r, tnsep=12pt]{\hyperlink{m}{$M$}}}{
		\Tdot~[tnpos=l]{\hyperlink{m1}{$M.1$}}
		\pstree{\Tdot~[tnpos=r, tnsep=8pt]{$\neg M.1$}}{
			\Tdot~[tnpos=b]{\hyperlink{m2}{$M.2$}}
			\Tdot~[tnpos=b]{\hyperlink{m3}{$M.3$}}
			\Tdot~[tnpos=b]{\hyperlink{m4}{$M.4$}}		
			\Tdot~[tnpos=b]{\hyperlink{m5}{$M.5$}}
			\Tdot~[tnpos=b]{\hyperlink{m6}{$M.6$}}
			\Tdot~[tnpos=b]{\hyperlink{m7}{$M.7$}}
		}
	}
	\pstree{\Tdot~[tnpos=r, tnsep=10pt]{\hyperlink{r}{$R$}}}{
		\pstree{\Tdot~[tnpos=l, tnsep=8pt]{\hyperlink{r1}{$R.1$}}}{
			\Tdot~[tnpos=b]{\hyperlink{r11}{$R.1.1$}}
			\pstree{\Tdot~[tnpos=r, tnsep=8pt]{\hyperlink{r12}{$R.1.2$}}}{
				\Tdot~[tnpos=b]{\hyperlink{r121}{$R.1.2.1$}}
				\Tdot~[tnpos=b]{\hyperlink{r122}{$R.1.2.2$}}
			}
		}
		\Tdot~[tnpos=b]{\hyperlink{r2}{$R.2$}}
		\Tdot~[tnpos=b]{\hyperlink{r3}{$R.3$}}
	}
}
\caption{Tree for case distinction in the proof of~\autoref{theoremrobust} (further explanation below)} \label{tree} 
\end{figure}
We now explain the case distinction displayed in Figure~\ref{tree}.\label{treesubcases} In all three cases ($L$, $M$ and $R$) we construct tight alternatives. Different properties of these alternatives yield different subcases. 

For Case \hyperlink{m}{$M$}, we distinguish between the case that there is a tight alternative which substitutes $e$ and is neither a right nor a left alternative (\hyperlink{m1}{$M.1$}), or not. The Subcases $\hyperlink{m2}{$M.2$}$, \hyperlink{m3}{$M.3$}, \hyperlink{m4}{$M.4$} cover the case that there are two tight alternatives (one for each Player) which substitute $e$ and  the same commonly used edges. For the case that this is not true, we distinguish between further properties of two tight alternatives (\hyperlink{m5}{$M.5$}, \hyperlink{m6}{$M.6$}, \hyperlink{m7}{$M.7$}). 

For Case $R$, we consider two tight alternatives $q_1$ (for Player 1) and $q_2$ (for Player 2) (with certain properties). Either $q_2$ is a left alternative (\hyperlink{r3}{$R.3$}), or not. If $q_2$ is a right alternative, we distinguish between the case that $q_2$ substitutes more (\hyperlink{r2}{$R.2$}) or less (\hyperlink{r1}{$R.1$}) commonly used edges than $q_1$. In $R.1$, we consider a third alternative $q_2'$ which can be a left (\hyperlink{r11}{$R.1.1$}) or a right (\hyperlink{r12}{$R.1.2$}) alternative. The subcase $R.1.2$ has two further subcases (\hyperlink{r121}{$R.1.2.1$} and \hyperlink{r122}{$R.1.2.2$}) which we explain in more detail in \autoref{subsubsec:caseR}, where even further subcases of Case $R.1$ and $R.2$ are discussed.

\subsubsection{Case L}\label{subsubsec:caseL}
\begin{proposition} \label{lemma1}
In Case \hypertarget{l}{$L$}, that means $\ord(e) \leq \ell_1+\ell_2$, we get a contradiction.
\end{proposition}

\begin{proof}
We describe the case $\ord(e) \leq \ell_1$ in detail, the case $\ell_1+1 \leq \ord(e) \leq \ell_1 + \ell_2$ follows analogously. Since $\xi_{1,e}<c(e)$ (and $e \notin P_2$), we get that Player 1 has a (tight) left alternative that substitutes $e$:
Suppose there would be no such tight alternative. In this case, we could increase the cost share $\xi_{1,e}<c(e)$ of Player 1 on edge $e$ without changing the cost shares on the other edges and would still get a feasible solution for \lp{}. Furthermore, this solution has a higher objective function value, a contradiction to the optimality of the given cost shares for \lp{}.

We consider a smallest left alternative $q_1$ for $e$ (defined by $\beta$ and $\gamma$).
The situation is illustrated in Figure \ref{case1}\footnote{Note that $\gamma \geq \ell_1+\ell_2+m+1$ or $\gamma \leq \ell_1$ doesn't change the argumentation; in the latter case, we only have to choose a slightly different $F^*$.}, where the solid lines represent the optimal Steiner forest $F$.

\begin{figure}[H] \centering \psset{unit=0.8cm}
\subfloat[Case $L$]{
 \begin{pspicture}(-3.5,-1.5)(6.1,1.5) 
\Knoten{-3}{1.5}{v1}\nput{180}{v1}{$s_1$}
\Knoten{-2.25}{1.2}{v2}
\Knoten{-1.5}{0.9}{v3}
\Knoten{-0.75}{0.6}{v4}
\Knoten{0}{0.3}{v5}
\Knoten{0.75}{0}{v6}
\Knoten{-3}{-1.5}{v7}\nput{180}{v7}{$s_2$}
\Knoten{2}{0}{v12}
\Knoten{3}{0}{v13}
\Knoten{4}{0}{v14}
\Knoten{5.25}{1.5}{v15}\nput{0}{v15}{$t_1$}
\Knoten{5.25}{-1.5}{v16}\nput{0}{v16}{$t_2$}

\Kante{v2}{v3}{\beta}
\Kante{v13}{v12}{\gamma}
\Kante{v4}{v5}{e}
\Bogendashed{v2}{v13}{q_1}{40}
\ncline{v1}{v2}
\ncline{v3}{v4}
\ncline{v5}{v6}
\ncline{v6}{v7}
\ncline{v6}{v12}
\ncline{v13}{v14}
\ncline{v14}{v15}
\ncline{v14}{v16}

\end{pspicture}
\label{case1}
} 
\subfloat[$F^*$ in Case $L$]{
 \begin{pspicture}(-3.9,-1.5)(5.75,1.5) 
\Knoten{-3}{1.5}{v1}\nput{180}{v1}{$s_1$}
\Knoten{-2.25}{1.2}{v2}
\Knoten{-1.5}{0.9}{v3}
\Knoten{-0.75}{0.6}{v4}
\Knoten{0}{0.3}{v5}
\Knoten{0.75}{0}{v6}
\Knoten{-3}{-1.5}{v7}\nput{180}{v7}{$s_2$}
\Knoten{2}{0}{v12}
\Knoten{3}{0}{v13}
\Knoten{4}{0}{v14}
\Knoten{5.25}{1.5}{v15}\nput{0}{v15}{$t_1$}
\Knoten{5.25}{-1.5}{v16}\nput{0}{v16}{$t_2$}

\Kantedashed{v2}{v3}{\beta}
\Kante{v13}{v12}{\gamma}
\Kantedashed{v4}{v5}{e}
\Bogen{v2}{v13}{q_1}{40}
\ncline{v1}{v2}
\ncline[linestyle=dashed]{v3}{v4}
\ncline[linestyle=dashed]{v5}{v6}
\ncline{v6}{v7}
\ncline{v6}{v12}
\ncline{v13}{v14}
\ncline{v14}{v15}
\ncline{v14}{v16}

\end{pspicture}
\label{F*case1}
}
\caption{Situation in Case $L$}
\end{figure}
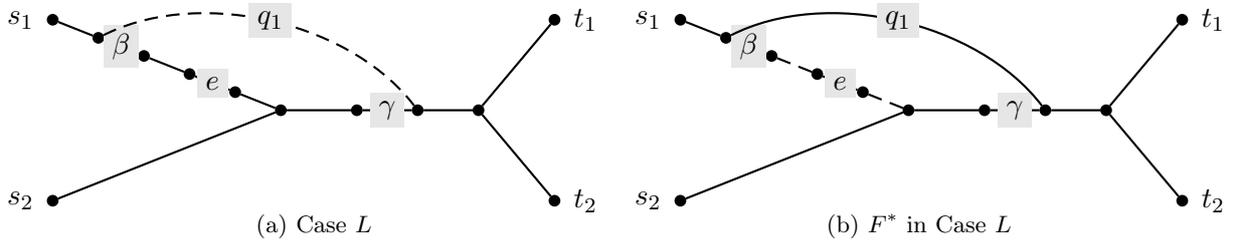


Regarding the cost shares of Player 1, we get
\[
\xi_{1,e_{i}} \ \ \begin{cases} =c(e_{i}), & \beta \leq i < \ord(e), \\
< c(e), & i = \ord(e), \\
=0, & \ord(e) < i \leq \gamma.
\end{cases}
\]
The first two cases are clear due to the fact that those edges are not contained in $P_2$ and the choice of $e$ as the first edge that is not paid completely. 
The third case holds since $q_1$ is a smallest left alternative for $e$ and the cost shares are pushed to the left.

Now we use this to construct a different Steiner forest $F^*$ with $c(F^*)<c(F)$, which is a contradiction.
Let $F^*=F \setminus \{e_{\beta}, \ldots, e_{\ell_1}\} \cup q_1$, see Figure \ref{F*case1}, where the solid lines represent $F^*$. 
Note that it is possible that $F^*$ (as defined above) is no Steiner forest, because $q_1$ can contain nodes or edges of $P_2 \setminus P_1$. But then it is clear that $F^*$ contains a Steiner forest with cost at most $c(F^*)$. In the following, we will not mention that in detail again: If we speak of a cheaper Steiner forest or a cheaper solution, we include the possibility that we have to delete edges of the considered subgraph to get a Steiner forest with lower costs than $F$. 

Considering the difference in costs, we get (for simplification, we omit $i \in P_1$ in the summation indices of the following sums):
\begin{align*}
c(F^*)-c(F) &=c(q_1)-\sum_{i=\beta }^{\ell_1}{c(e_{i})} = \sum_{i=\beta}^{\gamma}{\xi_{1,e_{i}}}-\sum_{i=\beta}^{\ell_1}{c(e_{i})} = \\ &= \sum_{i=\beta}^{\ord(e)-1}{\underbrace{\xi_{1,e_{i}}}_{=c(e_i)}}+\underbrace{\xi_{1,e}}_{<c(e)}+\sum_{i=\ord(e)+1}^{\gamma}{\underbrace{\xi_{1,e_{i}}}_{=0}}-\sum_{i=\beta}^{\ell_1}{c(e_{i})}  < \\
&< \sum_{i=\beta}^{\ord(e)}{c(e_{i})}-\sum_{i=\beta}^{\ell_1}{c(e_{i})} \leq \sum_{i=\beta}^{\ell_1}{c(e_{i})}-\sum_{i=\beta}^{\ell_1}{c(e_{i})} =0.
\end{align*}

The second equality follows from the fact that $q_1$ is a tight alternative. The strict inequality is due to our above observations of $\xi_{1, e_i}$ and the last inequality follows from $\ord(e)\leq \ell_1$ and that costs are nonnegative. 
\end{proof}

\subsubsection{Case M}\label{subsubsec:caseM}
\begin{proposition} \label{lemma2}
In Case \hypertarget{m}{$M$}, that means $\ell_1+\ell_2+1 \leq \ord(e) \leq \ell_1+\ell_2+m$, we get a contradiction.
\end{proposition}

\begin{proof}
Since the cost shares are optimal and $\xi_{1,e}+\xi_{2,e}<c(e)$, 
both players need to have a tight alternative for $e$.

In Subcase \hypertarget{m1}{$M.1$}, Player 1 has a tight alternative $q_1$ for $e$, defined by $\beta$ and $\gamma$ (or Player 2 has a tight alternative $q_2$ for $e$, defined by $\mu$ and $\nu$), which is neither a left nor a right alternative. Figure \ref{case2} illustrates this Subcase for Player 1.
%
We get a contradiction since one can construct a cheaper Steiner forest $F^*$:
For Player 1, consider $F^*:= F \setminus \{e_{\beta},\ldots,e_{\gamma}\} \cup q_1$.
Using that $q_1$ is a tight alternative, $\xi_{1,e_i} \leq c(e_i)$ for each edge $e_i$ and 
$\xi_{1,e}<c(e)$, we get
\begin{align*}
c(F^*)-c(F) = c(q_1)-\sum_{i=\beta }^{\gamma}{c(e_{i})} = \sum_{i=\beta }^{\gamma}{\xi_{1,e_{i}}}-\sum_{i=\beta}^{\gamma}{c(e_{i})} < 0.
\end{align*}

\vspace{-1em}
\begin{figure}[H] \centering \psset{unit=0.85cm}
\subfloat[Subcase $M.1$]{
%
 \begin{pspicture}(-1,-1.5)(8.5,1.5) 
\Knoten{-0.5}{1.5}{v-1}\nput{180}{v-1}{$s_1$}
\Knoten{-0.5}{-1.5}{v0}\nput{180}{v0}{$s_2$}
\Knoten{1}{0}{v3}
\Knoten{1.7}{0}{v4}
\Knoten{2.4}{0}{v5}
\Knoten{3.1}{0}{v6}
\Knoten{3.8}{0}{v7}
\Knoten{4.5}{0}{v8}
\Knoten{5.2}{0}{v11}
\Knoten{6}{0}{v12}
\Knoten{7.5}{1.5}{v9}\nput{0}{v9}{$t_1$}
\Knoten{7.5}{-1.5}{v10}\nput{0}{v10}{$t_2$}

\ncline{-}{v-1}{v3}
\ncline{-}{v0}{v3}
\ncline{-}{v2}{v12}
\ncline{-}{v3}{v4}
\Kante{v4}{v5}{\beta}
\ncline{-}{v5}{v6}
\Kante{v7}{v6}{e}
\ncline{-}{v7}{v8}
\Kante{v8}{v11}{\gamma}
\ncline{-}{v11}{v12}
\ncline{-}{v12}{v9}
\ncline{-}{v12}{v10}
\Bogendashed{v4}{v11}{q_1}{60}
\end{pspicture}
\label{case2}
}
\subfloat[Subcase $M.2$]{
 \begin{pspicture}(-1,-1.5)(8,1.5) 
\Knoten{-0.5}{1.5}{v-1}\nput{180}{v-1}{$s_1$}
\Knoten{-0.5}{-1.5}{v0}\nput{180}{v0}{$s_2$}
\Knoten{0}{1}{v1}
\Knoten{0}{-1}{v2}
\Knoten{0.5}{0.5}{v11}
\Knoten{0.5}{-0.5}{v12}
\Knoten{1}{0}{v3}
\Knoten{1.8}{0}{v4}
\Knoten{2.6}{0}{v5}
\Knoten{3.4}{0}{v6}
\Knoten{5}{0}{v7}
\Knoten{6}{0}{v8}
\Knoten{7.5}{1.5}{v9}\nput{0}{v9}{$t_1$}
\Knoten{7.5}{-1.5}{v10}\nput{0}{v10}{$t_2$}

\ncline{-}{v-1}{v1}
\Kante{v1}{v11}{\beta}
\ncline{-}{v0}{v2}
\ncline{-}{v11}{v3}
\Kante{v2}{v12}{\mu}
\ncline{-}{v12}{v3}
\ncline{-}{v3}{v4}
\Kante{v4}{v5}{e}
\ncline{-}{v5}{v6}
\Kante{v7}{v6}{\gamma=\nu}
\ncline{-}{v7}{v8}
\ncline{-}{v8}{v9}
\ncline{-}{v8}{v10}
\Bogendashed{v1}{v7}{q_1}{35}
\Bogendashed{v2}{v7}{q_2}{-30}
\end{pspicture}\label{case3a}
}
\end{figure}

\begin{figure}[H] \centering \psset{unit=0.85cm}
\subfloat[Subcase $M.3$]{
\begin{pspicture}(-1,-1.6)(8.5,1.5) 
\Knoten{-0.5}{1.5}{v-1}\nput{180}{v-1}{$s_1$}
\Knoten{-0.5}{-1.5}{v0}\nput{180}{v0}{$s_2$}
\Knoten{7}{1}{v1}
\Knoten{7}{-1}{v2}
\Knoten{6.5}{0.5}{v11}
\Knoten{6.5}{-0.5}{v12}
\Knoten{1}{0}{v3}
\Knoten{2}{0}{v4}
\Knoten{3.6}{0}{v5}
\Knoten{4.4}{0}{v6}
\Knoten{5.2}{0}{v7}
\Knoten{6}{0}{v8}
\Knoten{7.5}{1.5}{v9}\nput{0}{v9}{$t_1$}
\Knoten{7.5}{-1.5}{v10}\nput{0}{v10}{$t_2$}

\ncline{-}{v-1}{v3}
\ncline{-}{v0}{v3}
\ncline{-}{v3}{v4}
\Kante{v4}{v5}{\beta=\mu}
\ncline{-}{v5}{v6}
\Kante{v7}{v6}{e}
\ncline{-}{v7}{v8}
\ncline{-}{v8}{v11}
\ncline{-}{v8}{v12}
\Kante{v11}{v1}{\gamma}
\Kante{v12}{v2}{\nu}
\ncline{-}{v2}{v10}
\ncline{-}{v1}{v9}
\Bogendashed{v1}{v4}{q_1}{-30}
\Bogendashed{v2}{v4}{q_2}{30}
\end{pspicture}\label{case3b}
}
\subfloat[Subcase $M.4$]{
 \begin{pspicture}(-1,-1.6)(8,1.5) 
\Knoten{-0.5}{1.5}{v-1}\nput{180}{v-1}{$s_1$}
\Knoten{-0.5}{-1.5}{v0}\nput{180}{v0}{$s_2$}
\Knoten{0}{1}{v13}
\Knoten{0.5}{0.5}{v14}
\Knoten{0}{-1}{v15}
\Knoten{0.5}{-0.5}{v16}
\Knoten{7}{1}{v1}
\Knoten{7}{-1}{v2}
\Knoten{6.5}{0.5}{v11}
\Knoten{6.5}{-0.5}{v12}
\Knoten{1}{0}{v3}
\Knoten{2.67}{0}{v6}
\Knoten{4.33}{0}{v7}
\Knoten{6}{0}{v8}
\Knoten{7.5}{1.5}{v9}\nput{0}{v9}{$t_1$}
\Knoten{7.5}{-1.5}{v10}\nput{0}{v10}{$t_2$}

\ncline{-}{v-1}{v13}
\ncline{-}{v0}{v15}
\Kante{v13}{v14}{\beta}
\Kante{v15}{v16}{\mu}
\ncline{-}{v14}{v3}
\ncline{-}{v16}{v3}
\ncline{-}{v3}{v6}
\Kante{v7}{v6}{e}
\ncline{-}{v7}{v8}
\ncline{-}{v8}{v11}
\ncline{-}{v8}{v12}
\Kante{v11}{v1}{\gamma}
\Kante{v12}{v2}{\nu}
\ncline{-}{v2}{v10}
\ncline{-}{v1}{v9}
\Bogendashed{v1}{v13}{q_1}{-15}
\Bogendashed{v2}{v15}{q_2}{15}
\end{pspicture}\label{case3c}
}
\caption{Some subcases in Case $M$}
\end{figure}
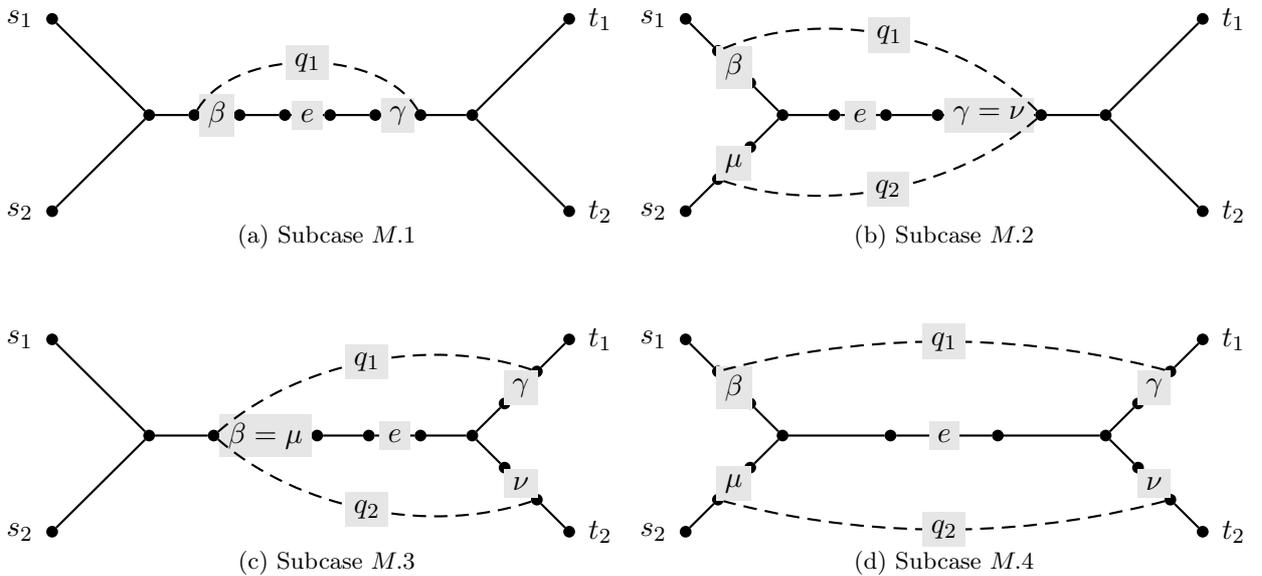
%

Obviously, we can construct a similar $F^*$ if Player 2 has a tight alternative which substitutes only commonly used edges. 
Therefore we can assume that all tight alternatives which substitute $e$ are either left or right alternatives (or both). 


Another subcase is that Player 1 has a tight alternative $q_1$ for $e$, defined by $\beta$ and $\gamma$, and Player 2 has a tight alternative $q_2$ for $e$, defined by $\mu$ and $\nu$, so that both alternatives substitute the same edges of $P_1 \cap P_2$. Referring to the assumption above this leads to three different subcases, illustrated in Figures \ref{case3a}, \ref{case3b} and \ref{case3c}:
\begin{enumerate}[label=Subcase $M$.\arabic*:, leftmargin=2.4cm] \setcounter{enumi}{1}
	\item If $q_1$ is a left alternative which does not substitute all commonly used edges, $q_2$ also has to be a left alternative and $\gamma=\nu$ has to hold;
	\item If $q_1$ is a right alternative which does not substitute all commonly used edges, $q_2$ also has to be a right alternative and $\beta=\mu$ has to hold; 
	\item  If $q_1$ is an alternative which substitutes all commonly used edges, $q_2$ also has to be an alternative which substitutes all commonly used edges. 
\end{enumerate} 

In all three subcases, we can again construct a cheaper Steiner forest $F^*$:
For Subcase \hypertarget{m2}{$M.2$}, consider
\[
F^*:= F \setminus \left(\{e_{\beta}, \ldots , e_{\ell_1}\} \cup \{e_{\mu}, \ldots , e_{\ell_1+\ell_2}\} \cup \{e_{\ell_1+\ell_2+1}, \ldots, e, \ldots, e_{\gamma}\} \right) \cup \left(q_1 \cup q_2 \right),
\]
illustrated in Figure \ref{case3aF*}.

\begin{figure}[H] \centering \psset{unit=0.95cm}
 \begin{pspicture}(-1,-1.5)(8,1.5) 
\Knoten{-0.5}{1.5}{v-1}\nput{180}{v-1}{$s_1$}
\Knoten{-0.5}{-1.5}{v0}\nput{180}{v0}{$s_2$}
\Knoten{0}{1}{v1}
\Knoten{0}{-1}{v2}
\Knoten{0.5}{0.5}{v11}
\Knoten{0.5}{-0.5}{v12}
\Knoten{1}{0}{v3}
\Knoten{1.8}{0}{v4}
\Knoten{2.6}{0}{v5}
\Knoten{3.4}{0}{v6}
\Knoten{5}{0}{v7}
\Knoten{6}{0}{v8}
\Knoten{7.5}{1.5}{v9}\nput{0}{v9}{$t_1$}
\Knoten{7.5}{-1.5}{v10}\nput{0}{v10}{$t_2$}

\ncline{-}{v-1}{v1}
\Kantedashed{v1}{v11}{\beta}
\ncline{-}{v0}{v2}
\ncline[linestyle=dashed]{-}{v11}{v3}
\Kantedashed{v2}{v12}{\mu}
\ncline[linestyle=dashed]{-}{v12}{v3}
\ncline[linestyle=dashed]{-}{v3}{v4}
\Kantedashed{v4}{v5}{e}
\ncline[linestyle=dashed]{-}{v5}{v6}
\Kantedashed{v7}{v6}{\gamma=\nu}
\ncline{-}{v7}{v8}
\ncline{-}{v8}{v9}
\ncline{-}{v8}{v10}
\Bogen{v1}{v7}{q_1}{35}
\Bogen{v2}{v7}{q_2}{-35}
\end{pspicture}
\caption{$F^*$ for Subcase $M.2$}
\label{case3aF*}
\end{figure}
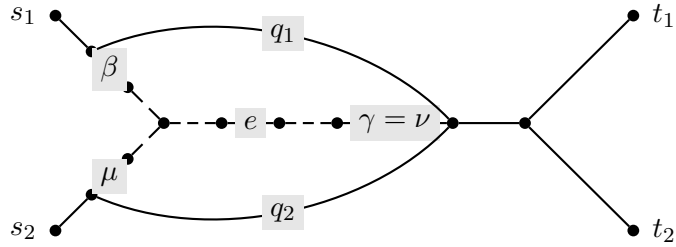
For the difference in costs, we get 
\begin{align*}
&c(F^*)-c(F) = c(q_1)+c(q_2)-\sum_{i=\beta}^{\ell_1}{c(e_{i})} - \sum_{i=\mu}^{\ell_1+\ell_2}{c(e_{i})} -\sum_{i=\ell_1+\ell_2+1}^{\gamma}{c(e_{i})} = \\
&= \sum_{i=\beta}^{\ell_1}{\underbrace{\xi_{1,e_{i}}}_{\leq c(e_i)}} + \sum_{i=\mu}^{\ell_1+\ell_2}{\underbrace{\xi_{2,e_{i}}}_{\leq c(e_i)}} +\sum_{i=\ell_1+\ell_2+1}^{\gamma}{\underbrace{(\xi_{1,e_{i}}+\xi_{2,e_{i}})}_{\leq c(e_i)}} -\sum_{i=\beta}^{\ell_1}{c(e_{i})} - \sum_{i=\mu}^{\ell_1+\ell_2}{c(e_{i})} -\sum_{i=\ell_1+\ell_2+1}^{\gamma}{c(e_{i})} <0.
\end{align*}
The second equality follows from the fact that $q_1$ and $q_2$ are tight alternatives and
the strict inequality is due to $\xi_{1,e}+\xi_{2,e}<c(e)$.
The Subcases \hypertarget{m3}{$M.3$} and \hypertarget{m4}{$M.4$} are very similar, only $F^*$ differs slightly; the proof is therefore left to the reader.

We can now assume that there are no tight alternatives for the players in which they substitute the same commonly used edges. 
The following subcases, for which we will need the properties of the \ptl-cost shares, cover the remaining situation of Case~$M$:
%
%
\begin{enumerate}[leftmargin=2.4cm,itemsep=-0em]
	\item[Subcase M.5:] All tight alternatives of Player 1 for $e$ are right alternatives,
	\item[Subcase M.6:] All tight alternatives of Player 2 for $e$ are right alternatives,
	\item[Subcase M.7:] Player 1 and Player 2 both have a tight left alternative for $e$. 
\end{enumerate}

For Subcase \hypertarget{m5}{$M.5$}, let $q_1$, defined by $\beta$ and $\gamma$,  be a tight right alternative of Player 1 for $e$ which has the smallest possible value of $\gamma$ among all such alternatives. The situation is illustrated in Figure \ref{case2bwas}. 

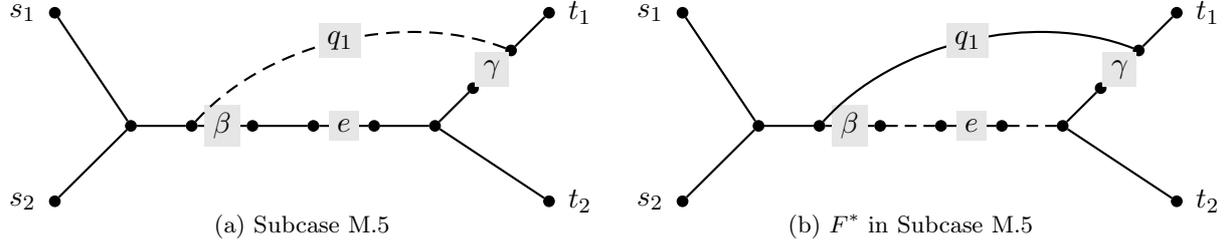
\begin{figure}[H] \centering \psset{unit=1cm}
\subfloat[Subcase M.5]{
\begin{pspicture}(0.5,-1)(8,1.5) 
\Knoten{1}{1.5}{v-1}\nput{180}{v-1}{$s_1$}
\Knoten{1}{-1}{v0}\nput{180}{v0}{$s_2$}
\Knoten{2}{0}{v-3}
\Knoten{2.8}{0}{v4}
\Knoten{3.6}{0}{v5}
\Knoten{4.4}{0}{v6}
\Knoten{5.2}{0}{v7}
\Knoten{6}{0}{v8}
\Knoten{7.5}{1.5}{v9}\nput{0}{v9}{$t_1$}
\Knoten{7.5}{-1}{v10}\nput{0}{v10}{$t_2$}
\Knoten{7}{1}{v1}
\Knoten{6.5}{0.5}{v11}

\ncline{-}{v-1}{v-3}
\ncline{-}{v0}{v-3}
\ncline{-}{v-3}{v4}
\Kante{v4}{v5}{\beta}
\ncline{-}{v5}{v6}
\Kante{v7}{v6}{e}
\ncline{-}{v7}{v8}
\ncline{-}{v8}{v11}
\ncline{-}{v8}{v10}
\Kante{v11}{v1}{\gamma}
\ncline{-}{v1}{v9}
\Bogendashed{v1}{v4}{q_1}{-35}
\end{pspicture}
\label{case2bwas}
}
\subfloat[$F^*$ in Subcase M.5]{
\begin{pspicture}(0,-1)(8,1.5) 
\Knoten{1}{1.5}{v-1}\nput{180}{v-1}{$s_1$}
\Knoten{1}{-1}{v0}\nput{180}{v0}{$s_2$}
\Knoten{2}{0}{v-3}
\Knoten{2.8}{0}{v4}
\Knoten{3.6}{0}{v5}
\Knoten{4.4}{0}{v6}
\Knoten{5.2}{0}{v7}
\Knoten{6}{0}{v8}
\Knoten{7.5}{1.5}{v9}\nput{0}{v9}{$t_1$}
\Knoten{7.5}{-1}{v10}\nput{0}{v10}{$t_2$}
\Knoten{7}{1}{v1}
\Knoten{6.5}{0.5}{v11}

\ncline{-}{v-1}{v-3}
\ncline{-}{v0}{v-3}
\ncline{-}{v-3}{v4}
\Kantedashed{v4}{v5}{\beta}
\ncline[linestyle=dashed]{-}{v5}{v6}
\Kantedashed{v7}{v6}{e}
\ncline[linestyle=dashed]{-}{v7}{v8}
\ncline{-}{v8}{v11}
\ncline{-}{v8}{v10}
\Kante{v11}{v1}{\gamma}
\ncline{-}{v1}{v9}
\Bogen{v1}{v4}{q_1}{-35}
\end{pspicture}
\label{case2bF*}
}
\caption{Situation in Subcase M.5}
\end{figure}

Since $q_1$ is a tight alternative, $c(q_1)$ equals the sum of cost shares of Player 1 for the edges which are substituted by $q_1$.
Furthermore, we get for the \ptl-cost~shares (by our choice of $q_1$)
\[
\xi_{1,e_i} \begin{cases}
\leq c(e_i), & \beta \leq i \leq \ord(e)-1, \\
< c(e_i), & i=\ord(e), \\
=0, & \ord(e)+1 \leq i \leq \gamma.
\end{cases}
\]
Altogether, since costs are nonnegative and $\ord(e)\leq \ell_1+\ell_2+m$,
\[
c(q_1)<\sum_{i=\beta}^{\ell_1+\ell_2+m}{c(e_i)}
\]
holds and thus $F^*:=F \setminus \{e_{\beta}, \ldots, e, \ldots, e_{\ell_1+\ell_2+m}\} \cup q_1$ (see Figure \ref{case2bF*}) is a cheaper Steiner forest.
It is clear that a similar $F^*$ with lower cost than $F$ can be constructed using a suitable tight alternative for Player 2 in Subcase \hypertarget{m6}{$M.6$}. 

Concluding, we consider Subcase \hypertarget{m7}{$M.7$}. Let $q_1$ (defined by $\beta$ and $\gamma$) be a smallest left alternative of Player 1 for $e$ and $q_2$ (defined by $\mu$ and $\nu$) a smallest left alternative of Player 2 for $e$. We now describe the subcase $\gamma < \nu$ in detail (see Figure \ref{case4a} for $\nu \leq \ell_1+\ell_2+m$), the subcase $\gamma>\nu$ follows analogously. 

\begin{figure}[H] \centering \psset{unit=0.9cm}
\subfloat[Subcase M.7, $\gamma<\nu$]{
 \begin{pspicture}(-1,-1.5)(7.8,1.5) 
\Knoten{-0.5}{1.5}{v-1}\nput{180}{v-1}{$s_1$}
\Knoten{-0.5}{-1.5}{v0}\nput{180}{v0}{$s_2$}
\Knoten{0}{1}{v1}
\Knoten{0}{-1}{v2}
\Knoten{0.5}{0.5}{v11}
\Knoten{0.5}{-0.5}{v12}
\Knoten{1}{0}{v3}
\Knoten{1.8}{0}{v4}
\Knoten{2.6}{0}{v5}
\Knoten{3.4}{0}{v6}
\Knoten{4.2}{0}{v7}
\Knoten{4.7}{0}{v8}
\Knoten{5.5}{0}{v9}
\Knoten{6.0}{0}{v10}
\Knoten{7}{1.5}{v13}\nput{0}{v13}{$t_1$}
\Knoten{7}{-1.5}{v14}\nput{0}{v14}{$t_2$}

\ncline{-}{v-1}{v1}
\Kante{v1}{v11}{\beta}
\ncline{-}{v0}{v2}
\ncline{-}{v11}{v3}
\Kante{v2}{v12}{\mu}
\ncline{-}{v12}{v3}
\ncline{-}{v3}{v4}
\Kante{v4}{v5}{e}
\ncline{-}{v5}{v6}
\Kante{v7}{v6}{\gamma}
\Kante{v8}{v9}{\nu}
\ncline{-}{v7}{v8}
\ncline{-}{v9}{v10}
\ncline{-}{v10}{v13}
\ncline{-}{v10}{v14}
\Bogendashed{v1}{v7}{q_1}{35}
\Bogendashed{v2}{v9}{q_2}{-30}
\end{pspicture}
\label{case4a}
}
\subfloat[$F^*$ in Subcase M.7, $\gamma<\nu$]{
 \begin{pspicture}(-1,-1.5)(7.5,1.5) 
\Knoten{-0.5}{1.5}{v-1}\nput{180}{v-1}{$s_1$}
\Knoten{-0.5}{-1.5}{v0}\nput{180}{v0}{$s_2$}
\Knoten{0}{1}{v1}
\Knoten{0}{-1}{v2}
\Knoten{0.5}{0.5}{v11}
\Knoten{0.5}{-0.5}{v12}
\Knoten{1}{0}{v3}
\Knoten{1.8}{0}{v4}
\Knoten{2.6}{0}{v5}
\Knoten{3.4}{0}{v6}
\Knoten{4.2}{0}{v7}
\Knoten{4.7}{0}{v8}
\Knoten{5.5}{0}{v9}
\Knoten{6.0}{0}{v10}
\Knoten{7}{1.5}{v13}\nput{0}{v13}{$t_1$}
\Knoten{7}{-1.5}{v14}\nput{0}{v14}{$t_2$}

\ncline{-}{v-1}{v1}
\Kantedashed{v1}{v11}{\beta}
\ncline{-}{v0}{v2}
\ncline[linestyle=dashed]{-}{v11}{v3}
\Kantedashed{v2}{v12}{\mu}
\ncline[linestyle=dashed]{-}{v12}{v3}
\ncline[linestyle=dashed]{-}{v3}{v4}
\Kantedashed{v4}{v5}{e}
\ncline[linestyle=dashed]{-}{v5}{v6}
\Kantedashed{v7}{v6}{\gamma}
\Kante{v8}{v9}{\nu}
\ncline{-}{v7}{v8}
\ncline{-}{v9}{v10}
\ncline{-}{v10}{v13}
\ncline{-}{v10}{v14}
\Bogen{v1}{v7}{q_1}{35}
\Bogen{v2}{v9}{q_2}{-30}
\end{pspicture}
\label{F*case4a}
}
\caption{Situation in Subcase M.7, $\gamma < \nu$}
\end{figure}
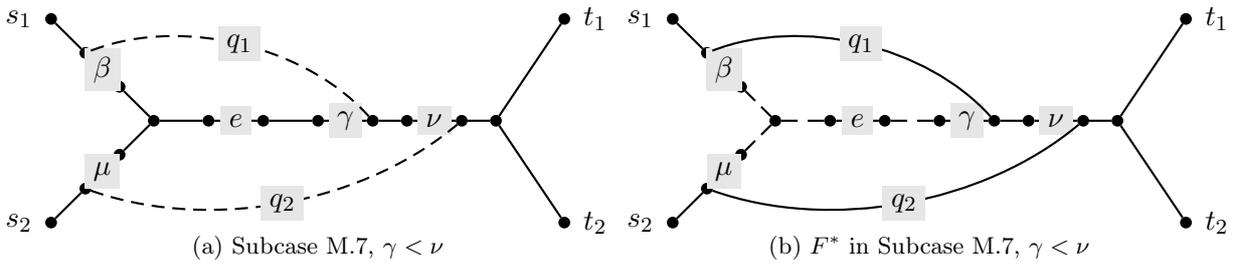

Since $q_1$ and $q_2$ are tight alternatives, we get
\[
c(q_1)=\sum_{i=\beta}^{\ell_1}{\xi_{1,e_{i}}}+\sum_{i=\ell_1+\ell_2+1}^{\gamma}{\xi_{1,e_{i}}} \text{ and } c(q_2)=\sum_{i=\mu}^{\ell_1+\ell_2}{\xi_{2,e_{i}}}+\sum_{i=\ell_1+\ell_2+1}^{\nu}{\xi_{2,e_{i}}}.
\]

Furthermore we get for the \ptl-cost~shares (by our choice of $q_1$, $q_2$ and $e$)
\[
\xi_{1,e_{i}}= \ \ \begin{cases} c(e_{i}), & \beta \leq i \leq \ell_1, \\
0, & \ord(e) < i \leq \gamma,
\end{cases}
\hspace{1cm}
\xi_{2,e_{i}}= \ \ \begin{cases} c(e_{i}), & \mu \leq i \leq \ell_1+\ell_2, \\
0, & \ord(e) < i \leq \nu,
\end{cases}
\]
and
\[
\xi_{1,e_{i}}+\xi_{2,e_{i}} \ \ \begin{cases} =c(e_{i}), & \ell_1+\ell_2+1 \leq i <\ord(e), \\
<c(e), &  i= \ord(e).
\end{cases}
\]
This yields
\[
c(q_1)+c(q_2)=\sum_{i=\beta}^{\ell_1}{c(e_i)}+\sum_{i=\mu}^{\ell_1+\ell_2}{c(e_i)}+\sum_{i=\ell_1+\ell_2+1}^{\ord(e)-1}{c(e_{i})}+\underbrace{\xi_{1,e}+\xi_{2,e}}_{<c(e)}
\]
and therefore 
\[
F^*:=F \setminus \left(\{e_{\beta}, \ldots , e_{\ell_1}\} \cup \{e_{\mu}, \ldots , e_{\ell_1+\ell_2}\} \cup \{e_{\ell_1+\ell_2+1}, \ldots, e, \ldots, e_{\gamma}\} \right) \cup (q_1 \cup q_2),
\]
see Figure \ref{F*case4a}, is a cheaper Steiner forest.

\end{proof}
\subsubsection{Case R}\label{subsubsec:caseR}

Now we consider Case $R$, that is, $\ell_1+ \ell_2+ m + 1 \leq \ord(e)$.
We can assume w.l.o.g. $e \in P_1$  since the other case ($e \in P_2$) follows analogously. 
To proof this last case, we need special properties of the cost shares. We introduce the following operation (which we denote as \texttt{CHANGE($j,i$)}): 
\begin{center}
	Increase $\xi_{2,e_i}$ and $\xi_{1,e_j}$, and simultaneously decrease $\xi_{2,e_j}$ and $\xi_{1,e_i}$, until either
\begin{itemize}[itemsep=-0em]
	\item an alternative of Player 1 for $e_j$ gets tight, or
	\item an alternative of Player 2 for $e_i$ gets tight, or
	\item $\xi_{2,e_i}=c(e_i)$, or 
	\item $\xi_{2,e_j}=0$.
\end{itemize}  
\end{center}
If this changes the cost shares we call \texttt{CHANGE($j,i$)} \emph{feasible} (and otherwise \emph{infeasible}).
The reason why we introduce the described operation is to get cost shares with the following properties:

\begin{definition}\label{def:properties} Let $F$ be an optimal Steiner forest which is not enforceable, $(\xi_{i,f})_{i \in \{1,2\}, f \in P_i}$ an optimal solution for \lp{} 
and $e$ (with $\ell_1+\ell_2+m+1 \leq \ord(e)$) the first edge which is not completely paid. We call this assignment of cost shares \emph{maximized for Player 2} if the following two properties hold:
\begin{enumerate}
\item[\mylabel{2m}{(2M)}] The sum of cost shares of Player 2 is maximal among all optimal assignments in which $e$ is the first edge that is not completely paid.
	\item[\mylabel{nc}{(NC)}] For every pair $i,j \in \{\ell_1+\ell_2+1, \ldots, \ell_1+\ell_2+m\}$ with $j<i$, \texttt{CHANGE($j,i$)} is not feasible.
\end{enumerate}
\end{definition}
Note that property~\ref{2m} implies that for every commonly used edge $e_i$ that Player 2 does not pay completely
she has a tight alternative for $e_i$,
otherwise we could simultaneously increase $\xi_{2,e_i}$ and decrease $\xi_{1,e_i}$ and get a solution in which the sum of cost shares of Player 2 is larger than before.

We now want to give a short intuition why it seems reasonable to consider cost shares which are maximized for Player 2. In an optimal assignment of cost shares, changing the cost shares in a feasible way according to \lp{} can not result in a higher objective function value. In our case this means that changing the cost shares can not yield that Player 1 can increase her cost share on $e$ (while the sum of the other cost shares remains the same). 
Property~\ref{2m} is therefore clear: If Player 2 could pay more, this could possibly yield that there is no tight alternative left for $e$ and then Player 1 can increase her cost share on $e$.
%
Property~\ref{nc} is linked with the fact that the edge $e$ is not contained in the set of the commonly used edges, in particular that the order of $e$ is at least $\ell_1+\ell_2+m+1$. As we will see in the proof of \autoref{lemma3}, there has to be a right alternative for Player 1 which substitutes $e$. 
If there is a pair $i,j \in \{\ell_1+\ell_2+1, \ldots, \ell_1+\ell_2+m\}$ with $j<i$ for which \texttt{CHANGE($j,i$)} is feasible, this could possibly yield an assignment of cost shares for which this right alternative for $e$ is not tight anymore and Player 1 could possibly increase her cost share on $e$. 

In the following, we assume that the cost shares are maximized for Player 2. The existence of such cost shares follows from the next lemma.

\begin{lemma}\label{lemma:properties}
Let $F$ be an optimal Steiner forest which is not enforceable. Then there is an optimal solution for \lp{} which is maximized for Player 2.
\end{lemma}

To show this lemma, we will use the following algorithm that yields cost shares with property~\ref{nc}.

\begin{algorithm}[H]
\KwData{Optimal Steiner forest $F$ and an optimal solution $(\xi_{i,f})_{i \in \{1,2\}, f \in P_i}$ for \lp{}.}
\KwResult{Transformed optimal solution $(\xi_{i,f})_{i \in \{1,2\}, f \in P_i}$ for \lp{}.}
\For{$i = \ell_1+\ell_2+m$ \upshape down to $\ell_1+\ell_2+1$}{
	\For{$j = i-1$ \upshape down to $\ell_1+\ell_2+1$}{
		\If{\upshape \texttt{CHANGE($j,i$)} is feasible}{
		\texttt{CHANGE($j,i$)}
		}
	}
}
\caption{\textsc{Change}} 
\label{Schiebe2rechts}
\end{algorithm}

\begin{proof}[Proof of \autoref{lemma:properties}]
As $F$ is not enforceable there is an optimal \ptl-solution for \lp{} with $\ell_1+\ell_2+m+1 \leq \ord(e)$ (where $e$ is the first edge which is not completely paid). 
Among all optimal assignments of cost shares in which $e$ is the first edge that is not completely paid, choose an assignment in which the sum of the cost shares of Player 2 is maximal. 
Now we transform these cost shares by using Algorithm~\nameref{Schiebe2rechts} and analyze the resulting cost shares. Obviously the transformed solution $(\xi_{i,f})_{i \in \{1,2\}, f \in P_i}$ is feasible and still optimal. Note that the cost shares are only changed for the commonly used edges and these edges remain completely paid. Therefore, $e$ is still the first edge that is not completely paid. Furthermore the sum of cost shares of Player 2 remains the same. Therefore property~\ref{2m} always holds. 


We now show that property~\ref{nc} holds after the execution of the procedure.  
Assume that this is not true and consider a pair $(i,j)$ with $j<i$ where \texttt{CHANGE($j,i$)} is feasible. In particular, this means that $\xi_{2,e_j}>0$ and $\xi_{2,e_i}<c(e_i)$. Let $\xi_{2,e_j}'$ and $\xi_{2,e_i}'$ be the cost shares for these two edges directly after we executed \texttt{CHANGE($j,i$)} during Algorithm~\nameref{Schiebe2rechts} (or detected infeasibility of this change). It is clear that $\xi_{2,e_i}' < c(e_i)$ holds since the cost shares for $e_i$ were never decreased after \texttt{CHANGE($j,i$)}.
There are two cases: Either $\xi_{2,e_j}'>0$ or $\xi_{2,e_j}'=0$ holds. 
In the first case there either has to be a left alternative of Player 1 for $e_j$ that does not substitute $e_i$, or a right alternative of Player 2 for $e_i$ that does not substitute $e_j$. But this is a contradiction to our assumption that \texttt{CHANGE($j,i$)} is now feasible, since left alternatives of Player 1 and right alternatives of Player 2 stay tight during the algorithm. 
Therefore $\xi_{2,e_j}'=0$ has to hold. Since $\xi_{2,e_j}>0$, there has to be a suitable $k<j$ so that \texttt{CHANGE($k,j$)} was feasible during Algorithm~\nameref{Schiebe2rechts}. 
Let $\xi_{2,e_k}''$ and $\xi_{2,e_i}''$ be the cost shares directly after we executed \texttt{CHANGE($k,i$)} in Algorithm~\nameref{Schiebe2rechts} (or detected infeasibility of this change). Note that $\xi_{2,e_k}''>0$ because \texttt{CHANGE($k,j$)} was feasible and $\xi_{2,e_i}''<c(e_i)$ holds since we never decreased this cost share after \texttt{CHANGE($k,i$)}. Therefore there either has to be a tight left alternative of Player 1 for $e_k$ that does not substitute $e_i$ or a tight right alternative of Player 2 for $e_i$ that does not substitute $e_k$. 
In the first case, this alternative has to substitute $e_j$ because otherwise \texttt{CHANGE($k,j$)} would not have been feasible. But then \texttt{CHANGE($j,i$)} can not be feasible now. 
In the other case this alternative has to substitute $e_j$ because otherwise \texttt{CHANGE($j,i$)} can not be feasible now. But then \texttt{CHANGE($k,j$)} would not have been feasible.
\end{proof}

As already mentioned in \autoref{subsec:proof_sketch}, we have to consider subgraphs which are ``almost'' Bad Configurations. In the following we give an exact definition. 
\begin{definition}\label{def:PBC}
We call a subgraph of $(G, (s_1,t_1), (s_2, t_2))$ a \textit{Preliminary Bad Configuration (PBC)}, if there are vertices $u,v,w,x$ 
satisfying the following conditions (see Figure~\ref{PBC} for illustration): 
\begin{enumerate}[(1)]
	\item $u,v$ are the terminal nodes of one player and $w,x$ the terminal nodes of the other player;
	\item There is a $u$-$v$-path
	$P_u$ and a $w$-$x$-path $P_{\ell}$ with 
	\begin{itemize}[leftmargin=*,itemsep=-0em]
		\item $P_u \cup P_{\ell}$ contains no cycle;
		\item $P_u$ and $P_{\ell}$ are not edge-disjoint (let $M$ be the unique subpath of $P_u$ which is also contained in $P_{\ell}$ (``commonly used edges''));
		\item If we consider $P_u$ as directed from $u$ to $v$ and $P_{\ell}$ as directed from $w$ to $x$, this induces the same direction on $M$;
		\item The subpath $R_u$ of $P_u$ which connects $M$ with $v$ contains at least one edge;
		\item The subpath $L_{\ell}$ of $P_{\ell}$ which connects $w$ with $M$ contains at least one edge;
		\item The subpath $R_{\ell}$ of $P_{\ell}$ which connects $x$ with $M$ contains at least one edge.
	\end{itemize}
	\item There is a path $q_1$ which closes a unique cycle $C_1$ with $P_u$, where $C_1$ contains edges from $R_u$ and $M$ (edges from $L_u$ are also allowed, but not necessary; this results in the two cases displayed in Figure~\ref{PBC}).
	\item There is a path $q_2$ which closes a unique cycle $C_2$ with $P_{\ell}$, where $C_2$ contains edges from $L_{\ell}$ and $M$, and $M \setminus C_2$ contains at least one edge.  
	
	\item There is a path $q_3$ which closes a unique cycle $C_3$ with $P_{\ell}$, where $C_3$ contains edges from $R_{\ell}$ and $M$, and $M \setminus C_3$ contains at least one edge. 
	\item $C:=C_2 \cap C_3 \cap M$ contains at least one edge and $C \subseteq C_1$. 
	\item $(M \cap C_1 \cap C_2)\setminus C$ contains at least one edge.
	\item $(M \cap C_1 \cap C_3)\setminus C$ contains at least one edge .
\end{enumerate}

\vspace{-1em}
\begin{figure}[H] \centering \psset{unit=1cm}
\subfloat[$C_1$ does not contain edges of $L_u$, ``$q_1$ small'']{
 \begin{pspicture}(-0.5,-1)(7.5,1) 
\Knoten{0}{1}{v1}\nput{180}{v1}{$u$}
\Knoten{0}{-1}{v2}\nput{180}{v2}{$w$}
\Knoten{0.5}{-0.5}{n2}
\Knoten{1}{0}{v3}
\Knoten{2.125}{0}{n5}
\Knoten{3.25}{0}{v4}
\Knoten{4.375}{0}{v6}
\Knoten{5.5}{0}{v7}
\Knoten{6}{0.5}{n3}
\Knoten{6}{-0.5}{n4}
\Knoten{6.5}{1}{v10}\nput{0}{v10}{$v$}
\Knoten{6.5}{-1}{v11}\nput{0}{v11}{$x$}

\ncline[linestyle=dashed]{-}{v1}{v3}
\ncline[linestyle=dashed]{-}{n2}{v3}
\ncline{-}{v2}{n2}
\ncline[linestyle=dashed]{-}{v3}{n5}
\ncline{-}{n5}{v4}
\ncline{-}{v4}{v6}
\ncline{-}{v6}{v7}
\ncline{-}{v7}{n3}
\ncline[linestyle=dashed]{-}{n3}{v10}
\ncline{-}{v7}{n4}
\ncline[linestyle=dashed]{-}{n4}{v11}
\Bogen{n5}{n3}{q_1}{35}
\Bogen{n2}{v6}{q_2}{-35}
\Bogen{v4}{n4}{q_3}{-35}
\end{pspicture}\label{PBCa}
}
\subfloat[$C_1$ contains at least one edge of $L_u$, ``$q_1$ big'']{
 \begin{pspicture}(-1,-1)(7,1) 
\Knoten{1}{1}{v1}\nput{180}{v1}{$u$}
\Knoten{1.5}{0.5}{n1}
\Knoten{1}{-1}{v2}\nput{180}{v2}{$w$}
\Knoten{1.5}{-0.5}{n2}
\Knoten{2}{0}{v3}
\Knoten{3.25}{0}{v4}
\Knoten{4.375}{0}{v6}
\Knoten{5.5}{0}{v7}
\Knoten{6}{0.5}{n3}
\Knoten{6}{-0.5}{n4}
\Knoten{6.5}{1}{v10}\nput{0}{v10}{$v$}
\Knoten{6.5}{-1}{v11}\nput{0}{v11}{$x$}

\ncline[linestyle=dashed]{-}{v1}{n1}
\ncline{-}{n1}{v3}
\ncline{-}{n2}{v3}
\ncline[linestyle=dashed]{-}{v2}{n2}
\ncline{-}{v3}{v4}
\ncline{-}{v4}{v6}
\ncline{-}{v6}{v7}
\ncline{-}{v7}{n3}
\ncline[linestyle=dashed]{-}{n3}{v10}
\ncline{-}{v7}{n4}
\ncline[linestyle=dashed]{-}{n4}{v11}
\Bogen{n1}{n3}{q_1}{25}
\Bogen{n2}{v6}{q_2}{-35}
\Bogen{v4}{n4}{q_3}{-35}
\end{pspicture}\label{PBCb}
}
\caption{PBCs; dashed lines may contain only one node.}
\label{PBC}
\end{figure}
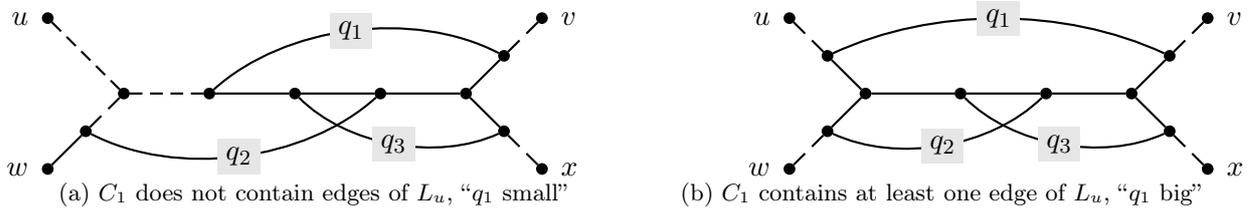
\end{definition}
Note that the uniqueness of the cycle $C_1$ implies that $q_1$ and $P_u$ are internal (i.e. except for the endnodes of $q_1$) node-disjoint. Analogously we get that $q_2$ ($q_3$) and $P_{\ell}$ are internal node-disjoint. 
However, the paths $q_1,q_2$ and $q_3$ do not have to be node-disjoint. 
Furthermore, $q_1$ can (internally; i.e. not an endnode of $q_1$) contain nodes of $L_{\ell}$ and/or $R_{\ell}$, as well as $q_2$ or $q_3$ can (internally) contain nodes of $L_u$ and/or $R_u$. In the following, we omit the term ``internally''. 
Depending on these properties, we call a PBC a Bad Configuration (BC) or a No Bad Configuration (NBC) (see \autoref{formaldef:BC} and \autoref{def:NBC}). 

We now start with the proof of Case $R$. Similar as for Case $M$, we get several subcases which mostly depend on properties of constructed tight alternatives. Figure~\ref{tree2} illustrates the most important subcases of Subcase $R.1$ (cf. \autoref{tree}); not displayed subcases easily yield cheaper Steiner forests. In the subcases $R.1.2.1$ and $R.1.2.2$ we analyze why \texttt{CHANGE($\rho,\sigma$)} (for certain edges $e_{\rho}$ and $e_{\sigma}$) is not feasible: either Player 1 has a tight alternative which substitutes $e_{\rho}$, but not $e_{\sigma}$ ($R.1.2.1$) or Player 2 has a tight alternative which substitutes $e_{\sigma}$, but not $e_{\rho}$ ($R.1.2.2$). 
The subcases labelled with PBC1 - PBC11 correspond to constructed subgraphs which are PBCs and these subcases are further analyzed in \autoref{subsubsec:pbcs}: We first derive properties for certain PBCs which are no BCs (so-called NBCs, see \autoref{def:NBC} and \autoref{lemmanbcanfang} - \autoref{lemmanbcende}) and show that the subgraphs PBC1-PBC11 have to be NBCs (see \autoref{lemma:pbc}). Finally we show that the properties of NBCs contradict the properties of PBC1-PBC11 (\autoref{lemmapbc1} - \autoref{lemmarest}).

For Subcase $R.2$, we get the two subcases \hyperlink{pbc12}{PBC12} and \hyperlink{pbc13}{PBC13} which are analyzed in the same manner as PBC1-PBC11 (see \autoref{subsubsec:pbcs}). In Subcase $R.3$ we have to consider \hyperlink{pbc14}{PBC14} (also analyzed in the same manner).
 
\begin{figure}[H]
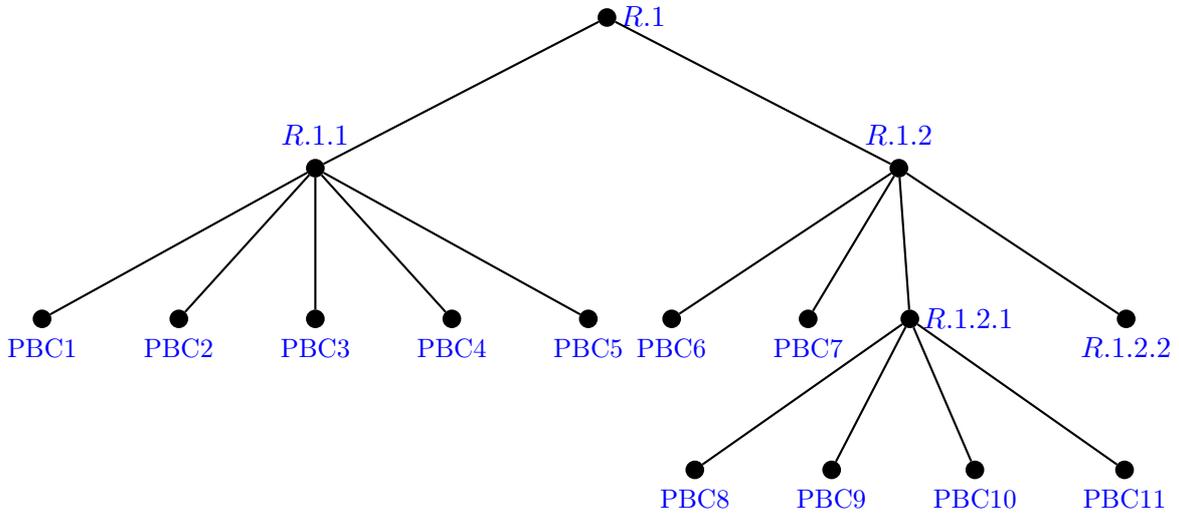
\centering\psset{treemode=D, dotsize=7pt}
		\pstree[thistreesep=5pt]{\Tdot~[tnpos=r]{\hyperlink{r1}{$R.1$}}}{
			\pstree[thistreesep=25pt]{\Tdot~[tnpos=a]{\hyperlink{r11}{$R.1.1$}}}{
				\Tdot~[tnpos=b]{\hyperlink{pbc1}{\small{PBC1}}}
				\Tdot~[tnpos=b]{\hyperlink{pbc2}{\small{PBC2}}}
				\Tdot~[tnpos=b]{\hyperlink{pbc3}{\small{PBC3}}}
				\Tdot~[tnpos=b]{\hyperlink{pbc4}{\small{PBC4}}}
				\Tdot~[tnpos=b]{\hyperlink{pbc4}{\small{PBC5}}}
			}		
			\pstree[thistreesep=25pt]{\Tdot~[tnpos=a]{\hyperlink{r12}{$R.1.2$}}}{
				\Tdot~[tnpos=b]{\hyperlink{pbc6}{\small{PBC6}}}
				\Tdot~[tnpos=b]{\hyperlink{pbc7}{\small{PBC7}}}
				\pstree[thistreesep=25pt]{\Tdot~[tnpos=r]{\hyperlink{r121}{$R.1.2.1$}}}{
					\Tdot~[tnpos=b]{\hyperlink{pbc8}{\small{PBC8}}}
					\Tdot~[tnpos=b]{\hyperlink{pbc9}{\small{PBC9}}}
					\Tdot~[tnpos=b]{\hyperlink{pbc10}{\small{PBC10}}}
					\Tdot~[tnpos=b]{\hyperlink{pbc11}{\small{PBC11}}}
				}		
				\Tdot~[tnpos=b]{\hyperlink{r122}{$R.1.2.2$}}
			}
		}
\caption{Tree for case distinction of Subcase R.1 in the proof of~\autoref{lemma3}}\label{tree2}
\end{figure}

\begin{proposition} \label{lemma3}
In Case \hypertarget{r}{\textbf{$R$}}, that means $\ell_1+\ell_2+m+1 \leq \ord(e)$, we get a contradiction.
\end{proposition}

\begin{proof} 

As already mentioned we can assume w.l.o.g. $e \in P_1$.
Since $e$ is not completely paid, Player 1 must have a tight right alternative which substitutes $e$. Let $q_1$, defined by $\beta$ and $\gamma$, be a smallest such alternative. It is clear that $\beta \leq \ell_1+\ell_2+m$ has to hold since otherwise we get a cheaper Steiner forest by using $q_1$.
If Player 2 pays all commonly used edges which are substituted by $q_1$, we get a cheaper Steiner forest by substituting $q_1$ for $\{e_{\ell_1+\ell_2+m+1}, \ldots, e, \ldots, e_{\gamma}\}$.
Therefore let $e_{\alpha}$ be the biggest commonly used edge which is not paid completely by Player 2. This implies $\xi_{2,e_i}=c(e_i)$ for all $i \in \{\alpha+1, \ldots, \ell_1+\ell_2+m\}$. 
Since our cost shares are maximized for Player 2, she has a tight alternative which substitutes $e_{\alpha}$. 
If there is such an alternative which is a right alternative, let $q_2$, defined by $\mu$ and $\nu$, be a smallest. 
Otherwise, let $q_2$ be any left alternative.
 Note that any tight alternative  which substitutes $e_{\alpha}$ has to be a right or left one, otherwise we get a contradiction to the optimality of $F$. Furthermore we can again exclude the case that there is a tight alternative for Player 1 which substitutes $e$ and a corresponding tight alternative for Player 2 which substitutes the same commonly used edges.
This leads to the following three different subcases which are illustrated in Figure~\ref{4case3c}:
\begin{enumerate}[itemsep=-0em,leftmargin=2.3cm]
	\item[Subcase $R.1$] $\nu \geq \ell_1+\ell_2+m+r_1+1$ and $\beta +1 \leq \mu \leq \alpha$;
	\item[Subcase $R.2$] $\nu \geq \ell_1+\ell_2+m+r_1+1$ and $\ell_1+1 \leq \mu \leq \beta-1$;
	\item[Subcase $R.3$] $\ell_1+1 \leq \mu \leq \ell_1+\ell_2$ and $\nu\leq\ell_1+\ell_2+m$.
\end{enumerate}

\vspace{-1em}
\begin{figure}[H] \centering \psset{unit=0.75cm}
\subfloat[Subcase $R.1$ (note that $\beta \leq \ell_1$ is also possible.)]{
\begin{pspicture}(-1.5,-1.8)(8,1.8) 
\Knoten{-1.1}{1.5}{v-1}\nput{180}{v-1}{$s_1$}
\Knoten{-1.1}{-1.5}{v0}\nput{180}{v0}{$s_2$}
\Knoten{0.4}{0}{v-3}
\Knoten{1.2}{0}{v-2}
\Knoten{2}{0}{v3}
\Knoten{2.8}{0}{v4}
\Knoten{3.6}{0}{v5}
\Knoten{4.4}{0}{v6}
\Knoten{5.2}{0}{v7}
\Knoten{6}{0}{v8}
\Knoten{7.5}{1.5}{v9}\nput{0}{v9}{$t_1$}
\Knoten{7.5}{-1.5}{v10}\nput{0}{v10}{$t_2$}
\Knoten{7.27}{1.27}{v1}
\Knoten{7.27}{-1.27}{v2}
\Knoten{6.87}{0.87}{v11}
\Knoten{6.87}{-0.87}{v12}
\Knoten{6.64}{0.64}{n1}
\Knoten{6.24}{0.24}{n2}

\ncline{-}{v-1}{v-3}
\ncline{-}{v0}{v-3}
\ncline{-}{v-3}{v4}
\Kante{v3}{v-2}{\beta}
\ncline{-}{v4}{v5}
\ncline{v4}{v3}
\ncline{-}{v5}{v6}
\Kante{v7}{v6}{\alpha}
\Kante{v4}{v5}{\mu}
\ncline{-}{v7}{v8}
\ncline{-}{v8}{v11}
\ncline{-}{v8}{v10}
\Kante{v11}{v1}{\gamma}
\Kante{v12}{v2}{\nu}
\ncline{-}{v2}{v10}
\ncline{-}{v1}{v9}
\Kante{n2}{n1}{e}
\Bogendashed{v1}{v-2}{q_1}{-35}
\Bogendashed{v4}{v2}{q_2}{-35}
\end{pspicture}
}
\hspace{0.5cm}
\subfloat[Subcase $R.2$ (note that $\mu \leq \ell_1+\ell_2$ is also possible.)]{
\begin{pspicture}(-1.5,-1.8)(8,1.8) 
\Knoten{-1.1}{1.5}{v-1}\nput{180}{v-1}{$s_1$}
\Knoten{-1.1}{-1.5}{v0}\nput{180}{v0}{$s_2$}
\Knoten{0.4}{0}{v-3}
\Knoten{1.2}{0}{v-2}
\Knoten{2}{0}{v3}
\Knoten{2.8}{0}{v4}
\Knoten{3.6}{0}{v5}
\Knoten{4.4}{0}{v6}
\Knoten{5.2}{0}{v7}
\Knoten{6}{0}{v8}
\Knoten{7.5}{1.5}{v9}\nput{0}{v9}{$t_1$}
\Knoten{7.5}{-1.5}{v10}\nput{0}{v10}{$t_2$}
\Knoten{7.27}{1.27}{v1}
\Knoten{7.27}{-1.27}{v2}
\Knoten{6.87}{0.87}{v11}
\Knoten{6.87}{-0.87}{v12}
\Knoten{6.64}{0.64}{n1}
\Knoten{6.24}{0.24}{n2}

\ncline{-}{v-1}{v-3}
\ncline{-}{v0}{v-3}
\ncline{-}{v-3}{v4}
\Kante{v3}{v-2}{\mu}
\ncline{-}{v4}{v5}
\ncline{v4}{v3}
\ncline{-}{v5}{v6}
\Kante{v7}{v6}{\alpha}
\Kante{v4}{v5}{\beta}
\ncline{-}{v7}{v8}
\ncline{-}{v8}{v11}
\ncline{-}{v8}{v10}
\Kante{v11}{v1}{\gamma}
\Kante{v12}{v2}{\nu}
\ncline{-}{v2}{v10}
\ncline{-}{v1}{v9}
\Kante{n2}{n1}{e}
\Bogendashed{v1}{v4}{q_1}{-35}
\Bogendashed{v-2}{v2}{q_2}{-35}
\end{pspicture}
}
\end{figure}
\begin{figure}[H] \centering \psset{unit=0.75cm}
\subfloat[Subcase $R.3$]{
\begin{pspicture}(-1.5,-1.8)(8,1.8) 
\Knoten{-1.1}{1.5}{v-1}\nput{180}{v-1}{$s_1$}
\Knoten{-1.1}{-1.5}{v0}\nput{180}{v0}{$s_2$}
\Knoten{0.4}{0}{v-3}
\Knoten{1.2}{0}{v-2}
\Knoten{2}{0}{v3}
\Knoten{2.8}{0}{v4}
\Knoten{3.6}{0}{v5}
\Knoten{4.4}{0}{v6}
\Knoten{5.2}{0}{v7}
\Knoten{6}{0}{v8}
\Knoten{7.5}{1.5}{v9}\nput{0}{v9}{$t_1$}
\Knoten{7.5}{-1.5}{v10}\nput{0}{v10}{$t_2$}
\Knoten{-0.6}{-1}{v13}
\Knoten{-0.1}{-0.5}{v14}
\Knoten{7.27}{1.27}{v1}
\Knoten{6.87}{0.87}{v11}
\Knoten{6.64}{0.64}{n1}
\Knoten{6.24}{0.24}{n2}

\ncline{-}{v-1}{v-3}
\ncline{-}{v0}{v-3}
\ncline{-}{v-3}{v4}
\Kante{v3}{v-2}{\beta}
\ncline{-}{v4}{v5}
\ncline{v4}{v3}
\ncline{-}{v5}{v6}
\Kante{v7}{v6}{\nu}
\Kante{v4}{v5}{\alpha}
\ncline{-}{v7}{v8}
\ncline{-}{v8}{v11}
\ncline{-}{v8}{v10}
\Kante{v11}{v1}{\gamma}
\Kante{v13}{v14}{\mu}
\ncline{-}{v8}{v10}
\ncline{-}{v1}{v9}
\Kante{n2}{n1}{e}
\Bogendashed{v1}{v-2}{q_1}{-35}
\Bogendashed{v13}{v7}{q_2}{-35}
\end{pspicture}
}
\caption{Subcases in Case $R$} \label{4case3c}
\end{figure}
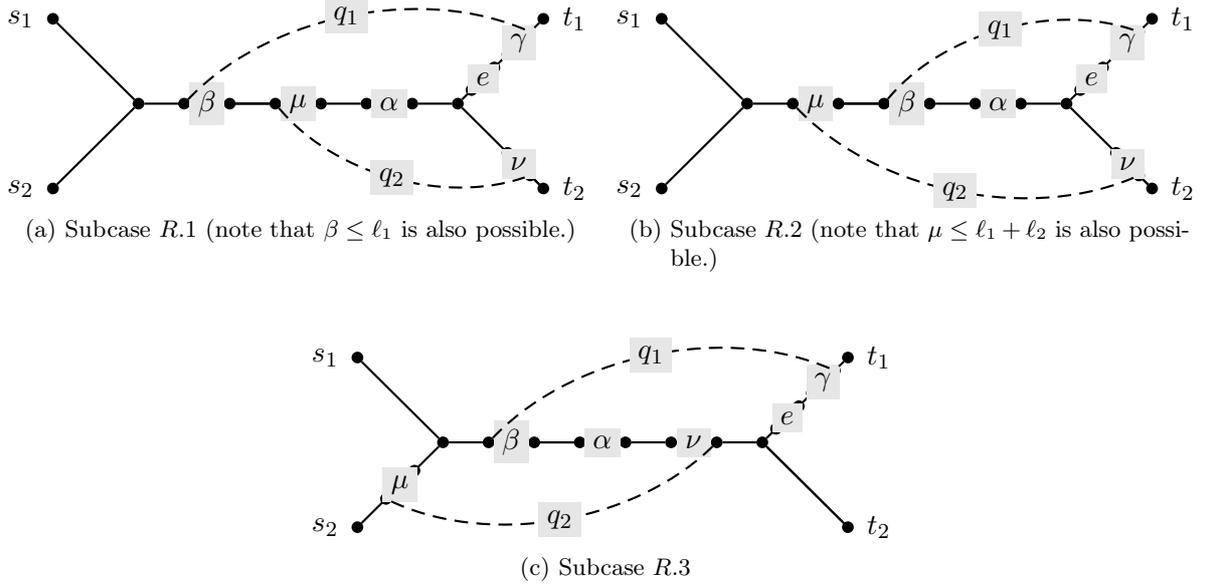


\vspace{1em}
\hypertarget{r1}{\textbf{Subcase R.1}}

If Player 2 pays all commonly used edges which are substituted by $q_1$, but not by $q_2$, we get a cheaper solution.
Therefore let $e_{\sigma}$ be the largest such edge which Player 2 does not pay completely. Since the cost shares are maximized by Player 2 there has to be a tight alternative for this edge. 

\vspace{1em}
\hypertarget{r11}{\textit{Subcase R.1.1}}

We first consider the subcase that there is a left alternative for $e_{\sigma}$; let $q_2'$, defined by $\mu'$ and $\nu'$, be a largest such one.
We now show that $\beta\geq \ell_1+\ell_2+1$ (i.e. $q_1$ small) and $\nu' < \mu$ has to hold: 

\vspace{0.5em}
First assume that $q_1$ is big, i.e. $\beta\leq \ell_1$ holds. If $\nu' \geq \alpha$ or $\nu' \leq \mu-1$ we can use $q_1, q_2'$ (and $q_2$ in the latter case) to get a cheaper solution. Therefore we can assume that $\nu'\in\{\mu, \ldots, \alpha-1\}$. It is clear that this has to hold for all left alternatives which substitute $e_{\sigma}$. In the following, we consider the smallest such alternative $q_2''$. Then $(P_1, P_2, q_1,q_2'',q_2)$ is a PBC for $(u,v,w,x)=(s_1,t_1,s_2,t_2)$.
We denote this subcase of $R.1.1$ as \hypertarget{pbc1}{PBC1} and analyze it in \autoref{lemmapbc1}, showing that PBC1 is not possible. Thus we have shown that $q_1$ cannot be big.

\vspace{0.5em}
Now assume that $\nu'\geq \mu$ holds (and $q_1$ is small). 
Then we distinguish between $\mu \leq \nu' < \ell_1+\ell_2+m$ and $\nu'\geq \ell_1+\ell_2+m$ (i.e. $\nu'=\ell_1+\ell_2+m$ or $\nu' \geq \ell_1+\ell_2+m+r_1+1$). 

If $\mu \leq \nu' < \ell_1+\ell_2+m$ holds, $(P_1,P_2,q_1,q_2',q_2)$ is a PBC for $(u,v,w,x)=(s_1,t_1,s_2,t_2)$ (Subcase \hypertarget{pbc2}{PBC2}). 
We treat this case in \autoref{lemmarest} and get $\nu' < \ell_1+\ell_2+m$ is not possible.
Now we consider the case $\nu'\geq \ell_1+\ell_2+m$. 
If Player 1 pays the edges $e_{\ell_1+\ell_2+1}, \ldots, e_{\beta-1}$ completely or $q_1$ substitutes all commonly used edges, we get a cheaper solution by using $q_1$ and $q_2'$. Therefore let $e_{\tau}$ be the largest such edge which Player 1 does not pay completely. 
Furthermore, let $e_{\rho}$ be the smallest edge in $\{e_{\beta}, \ldots, e_{\mu-1}\}$ which is not paid completely by Player 2. We now analyze why \texttt{CHANGE($\tau, \rho$)} is not feasible. 
Either there is a right alternative for Player 2 which substitutes $e_{\rho}$, but not $e_{\tau}$ (we say that the changes are \textit{not feasible for Player 2}), or there is a left alternative for Player 1 which substitutes $e_{\tau}$, but not $e_{\rho}$ (the changes are \textit{not feasible for Player 1}). 
If the changes are not feasible for Player 2, there has to be a right alternative which substitutes $e_{\rho}$, but not $e_{\tau}$. But then we get a cheaper solution by using this alternative and $q_1$ (note that Player 1 pays the edges $e_{\tau+1}, \ldots, e_{\beta-1}$ and Player 2 pays  $e_{\beta}, \ldots, e_{\rho-1}$). Therefore the changes for Player 1 can not be feasible. That means that there is a left alternative for Player 1 which substitutes $e_{\tau}$, but not $e_{\rho}$. Let $q_1'$, defined by $\beta'$ and $\gamma'$, be a smallest.
If $\tau \leq \gamma' \leq \beta-1$, we can use $q_1',q_1$ and $q_2'$ to construct a cheaper Steiner forest. 
Therefore, the situation is as illustrated in Figure~\ref{fig:R.1}. 
 
\begin{figure}[H] \centering \psset{unit=1cm}
\begin{pspicture}(1,-1.7)(14,1.8) 
\Knoten{1.2}{1.9}{v1}\nput{180}{v1}{$s_1$}
\Knoten{1.2}{-1.9}{v4}\nput{180}{v4}{$s_2$}
\Knoten{1.5}{-1.6}{v5}
\Knoten{2.0}{-1.1}{v6}
\Knoten{3.1}{0}{v7}
\Knoten{1.5}{1.6}{v8}
\Knoten{2.0}{1.1}{v9}
\Knoten{3.5}{0}{v10}
\Knoten{4.3}{0}{v11}
\Knoten{4.7}{0}{v12}
\Knoten{5.5}{0}{v13}
\Knoten{5.9}{0}{v14}
\Knoten{6.7}{0}{v15}
\Knoten{7.1}{0}{v16}
\Knoten{7.9}{0}{v17}
\Knoten{8.3}{0}{v18}
\Knoten{9.1}{0}{v19}
\Knoten{9.5}{0}{v20}
\Knoten{10.3}{0}{v21}
\Knoten{10.7}{0}{v22}
\Knoten{11.5}{0}{v23}
\Knoten{11.9}{0}{v24}
\Knoten{12.2}{0.3}{v25}
\Knoten{12.7}{0.8}{v26}
\Knoten{13.0}{1.1}{v27}
\Knoten{13.5}{1.6}{v28}
\Knoten{13.8}{1.9}{v29}\nput{0}{v29}{$t_1$}
\Knoten{12.2}{-0.3}{v30}
\Knoten{12.7}{-0.8}{v31}
\Knoten{13.0}{-1.1}{v32}
\Knoten{13.5}{-1.6}{v33}
\Knoten{13.8}{-1.9}{v34}\nput{0}{v34}{$t_2$}

\ncline{v1}{v7}
\ncline{v4}{v7}
\ncline{v7}{v24}
\ncline{v24}{v29}
\ncline{v24}{v34}

\Kante{v5}{v6}{\mu'}
\Kante{v8}{v9}{\beta'}
\Kante{v10}{v11}{\tau}
\Kante{v12}{v13}{\beta}
\Kante{v14}{v15}{\gamma'}
\Kante{v16}{v17}{\rho}
\Kante{v18}{v19}{\sigma}
\Kante{v20}{v21}{\mu}
\Kante{v22}{v23}{\alpha}
\Kante{v25}{v26}{e}
\Kante{v27}{v28}{\gamma}
\Kante{v30}{v31}{\nu}
\Kante{v32}{v33}{\nu'}

\Bogendashed{v8}{v15}{q_1'}{25}
\Bogendashed{v20}{v31}{q_2}{-35}
\Bogendashed{v5}{v33}{q_2'}{0}
\Bogendashed{v12}{v28}{q_1}{25}
\end{pspicture}
\caption{Subcase PBC3 (note that the order of $\nu$ and  $\nu'$ is not relevant)} \label{fig:R.1}
\end{figure}
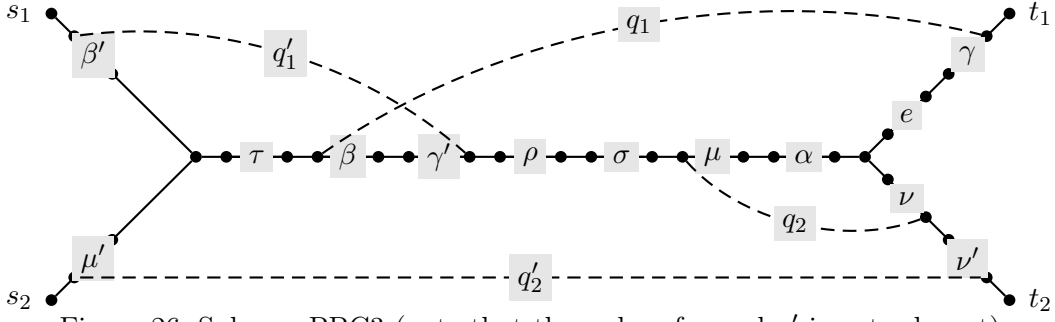

Now $(P_2, P_1, q_2',q_1,q_1')$ is a PBC for $(u,v,w,x)=(t_2,s_2,t_1,s_1)$ (Subcase \hypertarget{pbc3}{PBC3}).
\autoref{lemmarest} shows that $\nu' \geq \ell_1+\ell_2+m$ is also not possible.

\vspace{1em}
Overall we showed that $\nu' < \mu$ and $\beta \geq \ell_1+\ell_2+1$ has to hold, see Figure~\ref{case1i}.
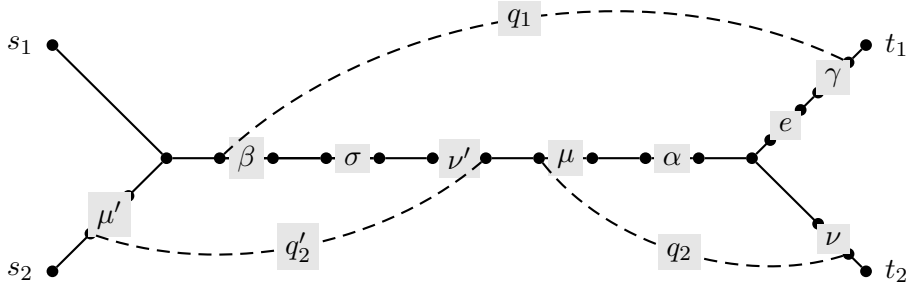
\begin{figure}[H] \centering \psset{unit=1cm}
\begin{pspicture}(-1.5,-1.8)(10,1.8) 
\Knoten{-1.1}{1.5}{v-1}\nput{180}{v-1}{$s_1$}
\Knoten{-1.1}{-1.5}{v0}\nput{180}{v0}{$s_2$}
\Knoten{0.4}{0}{v-3}
\Knoten{1.1}{0}{v-2}
\Knoten{1.8}{0}{v3}
\Knoten{2.5}{0}{v4}
\Knoten{3.2}{0}{v5}
\Knoten{3.9}{0}{v6}
\Knoten{4.6}{0}{v7}
\Knoten{5.3}{0}{v8}
\Knoten{6.0}{0}{n1}
\Knoten{6.7}{0}{n2}
\Knoten{7.4}{0}{n3}
\Knoten{8.1}{0}{n4}
\Knoten{9.6}{1.5}{v9}\nput{0}{v9}{$t_1$}
\Knoten{9.6}{-1.5}{v10}\nput{0}{v10}{$t_2$}
\Knoten{9.37}{1.27}{v1}
\Knoten{9.37}{-1.27}{v2}
\Knoten{8.97}{0.87}{v11}
\Knoten{8.97}{-0.87}{v12}
\Knoten{-0.6}{-1}{v13}
\Knoten{-0.1}{-0.5}{v14}
\Knoten{8.74}{0.64}{n5}
\Knoten{8.34}{0.24}{n6}

\ncline{-}{v-1}{v-3}
\ncline{-}{v0}{v-3}
\ncline{-}{v-3}{v4}
\Kante{v3}{v-2}{\beta}
\ncline{-}{v4}{v5}
\ncline{v4}{v3}
\ncline{-}{v5}{v6}
\Kante{v7}{v6}{\nu '}
\Kante{v4}{v5}{\sigma}
\ncline{-}{v7}{v8}
\Kante{v8}{n1}{\mu}
\ncline{n1}{n2}
\ncline{-}{n4}{v11}
\ncline{-}{n4}{v10}
\Kante{v11}{v1}{\gamma}
\Kante{v13}{v14}{\mu'}
\ncline{-}{v1}{v9}
\Kante{v12}{v2}{\nu}
\Kante{n2}{n3}{\alpha}
\ncline{n3}{n4}
\Kante{n5}{n6}{e}
\Bogendashed{v1}{v-2}{q_1}{-35}
\Bogendashed{v13}{v7}{q_2'}{-30}
\Bogendashed{v8}{v2}{q_2}{-35}
\end{pspicture}
\caption{Subcase of $R.1.1$}
\label{case1i}
\end{figure}

If Player 1 pays the edges $e_{\ell_1+\ell_2+1},\ldots,e_{\beta-1}$ or $q_1$ substitutes all commonly used edges, we get a cheaper solution by using $q_1, q_2'$ and $q_2$ because Player 2 pays the edges $e_{\nu'+1},\ldots,e_{\mu-1}$. Otherwise let $e_{\rho}$ be the largest such edge which Player 1 does not pay completely. Since the cost shares are maximized for Player 2, \texttt{CHANGE($\rho,\sigma$)} can not be feasible and we have to analyze why. 

If the changes for Player 1 are not feasible, there has to be a left alternative $q_1'$, defined by $\beta'$ and $\gamma'$, that substitutes $e_{\rho}$ but not $e_{\sigma}$. 
If $\gamma'\leq \beta-1$, we can use $q_1,q_1',q_2$ and $q_2'$ to construct a cheaper Steiner forest. Therefore $\gamma' \in \{\beta, \ldots, \sigma-1\}$ has to hold and $(P_2,P_1,q_2',q_1,q_1')$ is a PBC for $(u,v,w,x)=(t_2,s_2,t_1,s_1)$ (Subcase \hypertarget{pbc4}{PBC4}).
As \autoref{lemmarest} shows, Subcase PBC4 can not occur, i.e. the changes for Player~1 are feasible, but not the changes for Player~2.

Thus there has to be a right alternative $\bar{q_2}$ for Player 2, defined by $\bar{\mu}$ and $\bar{\nu}$, which substitutes $e_{\sigma}$ but not $e_{\rho}$. 
If $\bar{\mu}\leq \beta$ holds, we can use $q_1$ and $\bar{q_2}$ to get a cheaper solution (note that Player~1 pays the edges $e_{\rho+1},\ldots,e_{\beta-1}$ by the choice of $\rho$).
Therefore we get $\bar{\mu} \in \{\beta+1, \ldots, \sigma\}$ and $(P_1,P_2,q_1,q_2',\bar{q_2})$ is a PBC for $(u,v,w,x)=(s_1,t_1,s_2,t_2)$  (Subcase \hypertarget{pbc5}{PBC5}). 
As this subcase can not occur (see \autoref{lemmarest}), this completes Subcase R.1.1 (Player 2 has a left alternative for $e_{\sigma}$).

\vspace{1em}
\hypertarget{r12}{\textit{Subcase R.1.2}}

We can now assume that there are only right alternatives for $e_{\sigma}$ and consider a largest  such alternative $q_2'$ (defined by $\mu'$ and $\nu'$). 

\vspace{0.5em}
Let us first assume that $q_1$ substitutes more commonly used edges than $q_2'$ (if both alternatives substitute the same commonly used edges, we obviously get a cheaper Steiner forest). If Player~2 pays the edges of the commonly used path which are substituted by $q_1$, but not by $q_2'$, we can use $q_1$ and $q_2'$ to get a cheaper Steiner forest. Therefore let $e_{\rho}$ be the largest such edge that is not paid completely by Player 2. Since the cost shares are maximized for Player 2 and $q_2'$ is the largest right alternative, there has to be a left alternative for $e_{\rho}$. Let $\bar{q_2}$, defined by $\bar{\mu}$ and $\bar{\nu}$, be a smallest such alternative.
We now show that $\bar{\nu}\leq\mu'-1$ has to hold, since we get a contradiction if this is not true:

If $\bar{\nu}\geq \mu'$ holds, we get that $(P_1,P_2,q_1,\bar{q_2},q_2')$ is a PBC for $(u,v,w,x)=(s_1,t_1,s_2,t_2)$ ((Subcase \hypertarget{pbc6}{PBC6}); note that $\bar{\nu}<\sigma$ since there are no left alternatives for $e_{\sigma}$).
We have to distinguish between the two cases $q_1$ big and $q_1$ small.
In both cases we get that Subcase \hyperlink{pbc6}{PBC6} can not occur (if $q_1$ is big, see \autoref{lemmapbc7}, else see \autoref{lemmarest}). Therefore $\bar{\nu}\leq \mu'-1$ has to hold, cf. Figure~\ref{case1iia}. 

\begin{figure}[H] \centering \psset{unit=1cm}
\begin{pspicture}(0,-2)(13.8,2) 
\Knoten{0}{1.9}{v1}\nput{180}{v1}{$s_1$}
\Knoten{0}{-1.9}{v4}\nput{180}{v4}{$s_2$}
\Knoten{0.3}{-1.6}{v5}
\Knoten{0.8}{-1.1}{v6}
\Knoten{1.9}{0}{v7}
\Knoten{2.3}{0}{v8}
\Knoten{3.1}{0}{v9}
\Knoten{3.5}{0}{v10}
\Knoten{4.3}{0}{v11}
\Knoten{4.7}{0}{v12}
\Knoten{5.5}{0}{v13}
\Knoten{5.9}{0}{v14}
\Knoten{6.7}{0}{v15}
\Knoten{7.1}{0}{v16}
\Knoten{7.9}{0}{v17}
\Knoten{8.3}{0}{v18}
\Knoten{9.1}{0}{v19}
\Knoten{9.5}{0}{v20}
\Knoten{10.3}{0}{v21}
\Knoten{10.7}{0}{v22}
\Knoten{11.5}{0}{v23}
\Knoten{11.9}{0}{v24}
\Knoten{12.2}{0.3}{v25}
\Knoten{12.7}{0.8}{v26}
\Knoten{13.0}{1.1}{v27}
\Knoten{13.5}{1.6}{v28}
\Knoten{13.8}{1.9}{v29}\nput{0}{v29}{$t_1$}
\Knoten{12.2}{-0.3}{v30}
\Knoten{12.7}{-0.8}{v31}
\Knoten{13.0}{-1.1}{v32}
\Knoten{13.5}{-1.6}{v33}
\Knoten{13.8}{-1.9}{v34}\nput{0}{v34}{$t_2$}

\ncline{v1}{v7}
\ncline{v4}{v7}
\ncline{v7}{v24}
\ncline{v24}{v29}
\ncline{v24}{v34}

\Kante{v5}{v6}{\bar{\mu}}
\Kante{v8}{v9}{\omega}
\Kante{v10}{v11}{\beta}
\Kante{v12}{v13}{\rho}
\Kante{v14}{v15}{\bar{\nu}}
\Kante{v16}{v17}{\mu '}
\Kante{v18}{v19}{\sigma}
\Kante{v20}{v21}{\mu}
\Kante{v22}{v23}{\alpha}
\Kante{v25}{v26}{e}
\Kante{v27}{v28}{\gamma}
\Kante{v30}{v31}{\nu}
\Kante{v32}{v33}{\nu'}

\Bogendashed{v5}{v15}{\bar{q_2}}{-35}
\Bogendashed{v20}{v31}{q_2}{-35}
\Bogendashed{v16}{v33}{q_2'}{-35}
\Bogendashed{v10}{v28}{q_1}{25}
\end{pspicture}
\caption{Subcase of $R.1.2$ ($q_1$ substitutes more commonly used edges than $q_2'$)} \label{case1iia}
\end{figure}
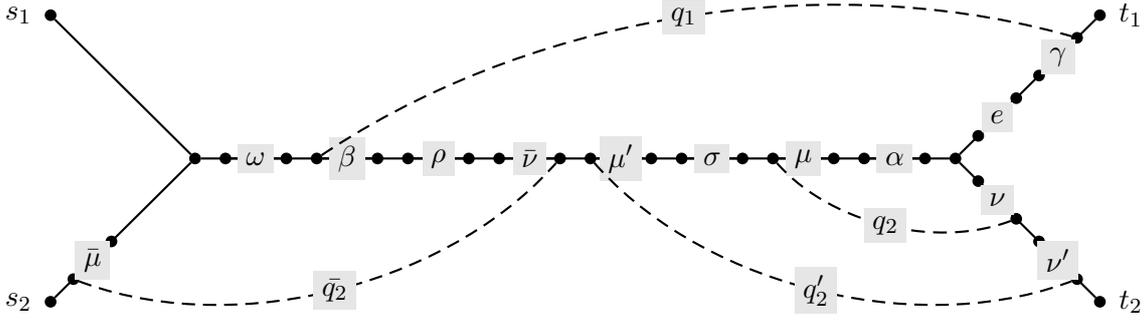

If $q_1$ substitutes all commonly used edges or Player 1 pays the edges $e_{\ell_1+\ell_2+1},\ldots,e_{\beta-1}$ completely, we get a cheaper solution. Otherwise consider the largest such edge $e_{\omega}$ that is not paid completely by Player 1. Since the cost shares are maximized for Player 2, \texttt{CHANGE($\omega,\rho$)} can not be feasible and we now analyze why. The changes are feasible for Player 2 because of the choice of $q_2'$. Therefore they can not be feasible for Player 1 and there is a left alternative for Player 1 that substitutes $e_{\omega}$, but not $e_{\rho}$. Let $q_1'$, defined by $\beta'$ and $\gamma'$, be such an alternative. If $\gamma'\leq \beta-1$ holds, we get a cheaper solution by using $q_1,q_1',\bar{q_2}$ and $q_2'$. Therefore $\beta \leq \gamma' \leq \rho-1$ holds and $(P_2,P_1,\bar{q_2}, q_1,q_1')$ is a PBC for $(u,v,w,x)=(t_2,s_2,t_1,s_1)$ (Subcase \hypertarget{pbc7}{PBC7}). 
\autoref{lemmarest} shows that this subcase is not possible, completing the case that $q_1$ substitutes more commonly used edges than $q_2'$.

\vspace{0.5em}
Therefore $q_1$ substitutes less commonly used edges than $q_2'$. If Player 1 pays the commonly used edges which are substituted by $q_2'$, but not by $q_1$, we can use $q_1$ and $q_2'$ to get a cheaper solution. Therefore let $e_{\rho}$ be the largest such edge that is not paid completely by Player 1 (see Figure~\ref{fig:r121}). Now \texttt{CHANGE($\rho,\sigma$)} must not be feasible. 

\begin{figure}[H] \centering \psset{unit=1cm}
\begin{pspicture}(3.6,-2)(15,2) 
\Knoten{3.6}{1.9}{v1}\nput{180}{v1}{$s_1$}
\Knoten{3.6}{-1.9}{v4}\nput{180}{v4}{$s_2$}
\Knoten{5.5}{0}{v7}
\Knoten{5.9}{0}{v10}
\Knoten{6.7}{0}{v11}
\Knoten{7.1}{0}{v12}
\Knoten{7.9}{0}{v13}
\Knoten{8.3}{0}{v16}
\Knoten{9.1}{0}{v17}
\Knoten{9.5}{0}{v20}
\Knoten{10.3}{0}{v21}
\Knoten{10.7}{0}{v22}
\Knoten{11.5}{0}{v23}
\Knoten{11.9}{0}{v24}
\Knoten{12.7}{0}{n1}
\Knoten{13.1}{0}{n2}
\Knoten{13.4}{0.3}{v25}
\Knoten{13.9}{0.8}{v26}
\Knoten{14.2}{1.1}{v27}
\Knoten{14.7}{1.6}{v28}
\Knoten{15}{1.9}{v29}\nput{0}{v29}{$t_1$}
\Knoten{13.4}{-0.3}{v30}
\Knoten{13.9}{-0.8}{v31}
\Knoten{14.2}{-1.1}{v32}
\Knoten{14.7}{-1.6}{v33}
\Knoten{15}{-1.9}{v34}\nput{0}{v34}{$t_2$}
\ncline{v1}{v7}
\ncline{v4}{v7}
\ncline{v7}{n2}
\ncline{n2}{v29}
\ncline{n2}{v34}
\ncline{v2}{v3}
\Kante{v10}{v11}{\mu'}
\Kante{v12}{v13}{\rho}
\Kante{v16}{v17}{\beta}
\Kante{v20}{v21}{\sigma}
\Kante{v22}{v23}{\mu}
\Kante{v25}{v26}{e}
\Kante{v27}{v28}{\gamma}
\Kante{v30}{v31}{\nu}
\Kante{v32}{v33}{\nu'}
\Kante{n1}{v24}{\alpha}
%
\Bogendashed{v22}{v31}{q_2}{-35}
\Bogendashed{v10}{v33}{q_2'}{-20}
\Bogendashed{v16}{v28}{q_1}{25}
\end{pspicture}
\caption{Subcase of $R.1.2$ ($q_1$ substitutes less commonly used edges than $q_2'$)}\label{fig:r121}
\end{figure}
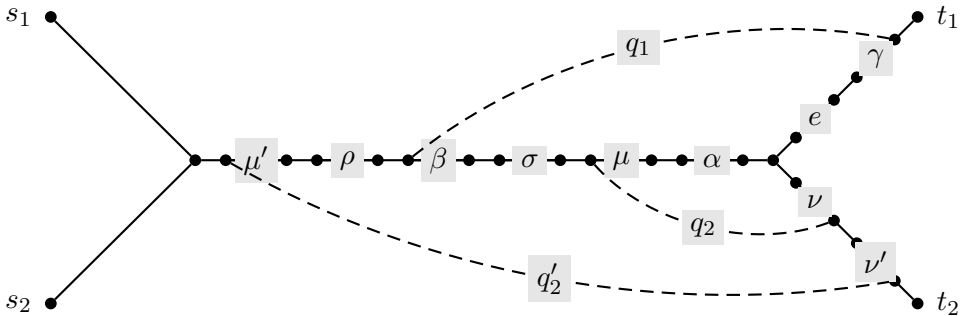

\vspace{0.5em}
\hypertarget{r121}{\textit{Subcase R.1.2.1}}

Let us first assume that the changes are not feasible for Player 1. Then there has to be a left alternative for Player 1 which substitutes $e_{\rho}$ but not $e_{\sigma}$. Let $q_1'$, defined by $\beta'$ and $\gamma'$, be a smallest such alternative. 
For $\gamma' \in \{\beta, \ldots, \sigma-1\}$ we get that $(P_2,P_1,q_2',q_1',q_1)$ is a PBC for $(u,v,w,x)=(s_2,t_2,s_1,t_1)$ ((Subcase \hypertarget{pbc8}{PBC8});
note that $\mu'\geq\ell_1+\ell_2+1$ holds since there are no left alternatives for $e_{\sigma}$).
This is only possible if the costs of the edges $e_{\beta}, \ldots, e_{\gamma'}$ are all zero (see \autoref{lemmapbc9}).
Therefore either $\gamma' \leq \beta-1$ holds, or, if this is not true, the costs of the edges $e_{\beta}, \ldots, e_{\gamma'}$ all have to be zero.

Now if $q_2'$ substitutes all commonly used edges or Player 2 pays the edges $e_{\ell_1+\ell_2+1},\ldots,e_{\mu'-1}$ we get a cheaper Steiner forest by using $q_1',q_1$ and $q_2'$. Thus let $e_{\omega}$ be the largest such edge that is not paid completely by Player 2. Since $q_2'$ is the largest right alternative, there are only left alternatives for $e_{\omega}$. We consider a smallest such alternative $\bar{q_2}$, defined by $\bar{\mu}$ and $\bar{\nu}$. 
If $\bar{\nu} \leq \mu'-1$, we can construct a cheaper Steiner forest by using $q_1',q_1,\bar{q_2}$ and $q_2'$. Furthermore, $\bar{\nu}\leq \sigma-1$ has to hold since there are no left alternatives for $e_{\sigma}$. Therefore either $\bar{\nu}\in\{\mu',\ldots,\gamma'-1\}$ or $\bar{\nu}\in\{\gamma',\ldots,\sigma-1\}$ holds.
In the first case $(P_1,P_2,q_1',q_2',\bar{q_2})$ is a PBC for $(u,v,w,x)=(t_1,s_1,t_2,s_2)$ (Subcase \hypertarget{pbc9}{PBC9}). This is only possible if Player 1 pays the edges $e_{\bar{\nu}+1}, \ldots, e_{\gamma'}$ completely, see \autoref{lemmarest}. Since this is not true for $e_{\rho}$, $\bar{\nu}\geq \rho$ has to hold.
Figure~\ref{case1iib} illustrates all remaining possibilities. As we will see, we can treat all these cases almost analogously.

\begin{figure}[H] \centering \psset{unit=1cm}
\subfloat[First remaining possibility (note: costs of $e_{\beta}, \ldots, e_{\gamma'}$ zero; $\bar{\nu}\in \{\rho,\ldots,\gamma'-1\}$)]{
\begin{pspicture}(0,-2.2)(13.9,2.2) 
\Knoten{0}{1.9}{v1}\nput{180}{v1}{$s_1$}
\Knoten{0.3}{1.6}{v2}
\Knoten{0.8}{1.1}{v3}
\Knoten{0}{-1.9}{v4}\nput{180}{v4}{$s_2$}
\Knoten{0.3}{-1.6}{v5}
\Knoten{0.8}{-1.1}{v6}
\Knoten{1.9}{0}{v7}
\Knoten{2.3}{0}{v8}
\Knoten{3.1}{0}{v9}
\Knoten{3.5}{0}{v10}
\Knoten{4.3}{0}{v11}
\Knoten{4.7}{0}{v12}
\Knoten{5.5}{0}{v13}
\Knoten{5.9}{0}{v14}
\Knoten{6.7}{0}{v15}
\Knoten{7.1}{0}{v16}
\Knoten{7.9}{0}{v17}
\Knoten{8.3}{0}{v18}
\Knoten{9.1}{0}{v19}
\Knoten{9.5}{0}{v20}
\Knoten{10.3}{0}{v21}
\Knoten{10.7}{0}{v22}
\Knoten{11.5}{0}{v23}
\Knoten{11.9}{0}{v24}
\Knoten{12.7}{0}{n1}
\Knoten{13.1}{0}{n2}
\Knoten{13.4}{0.3}{v25}
\Knoten{13.9}{0.8}{v26}
\Knoten{14.2}{1.1}{v27}
\Knoten{14.7}{1.6}{v28}
\Knoten{15}{1.9}{v29}\nput{0}{v29}{$t_1$}
\Knoten{13.4}{-0.3}{v30}
\Knoten{13.9}{-0.8}{v31}
\Knoten{14.2}{-1.1}{v32}
\Knoten{14.7}{-1.6}{v33}
\Knoten{15}{-1.9}{v34}\nput{0}{v34}{$t_2$}
\ncline{v1}{v7}
\ncline{v4}{v7}
\ncline{v7}{n2}
\ncline{n2}{v29}
\ncline{n2}{v34}
\Kante{v2}{v3}{\beta'}
\Kante{v5}{v6}{\bar{\mu}}
\Kante{v8}{v9}{\omega}
\Kante{v10}{v11}{\mu'}
\Kante{v12}{v13}{\rho}
\Kante{v14}{v15}{\bar{\nu}}
\Kante{v16}{v17}{\beta}
\Kante{v18}{v19}{\gamma'}
\Kante{v20}{v21}{\sigma}
\Kante{v22}{v23}{\mu}
\Kante{v25}{v26}{e}
\Kante{v27}{v28}{\gamma}
\Kante{v30}{v31}{\nu}
\Kante{v32}{v33}{\nu'}
\Kante{n1}{v24}{\alpha}
\Bogendashed{v5}{v15}{\bar{q_2}}{-25}
\Bogendashed{v22}{v31}{q_2}{-35}
\Bogendashed{v10}{v33}{q_2'}{-25}
\Bogendashed{v16}{v28}{q_1}{25}
\Bogendashed{v2}{v19}{q_1'}{25}
\end{pspicture}
\label{case1iib1}
}
\end{figure}

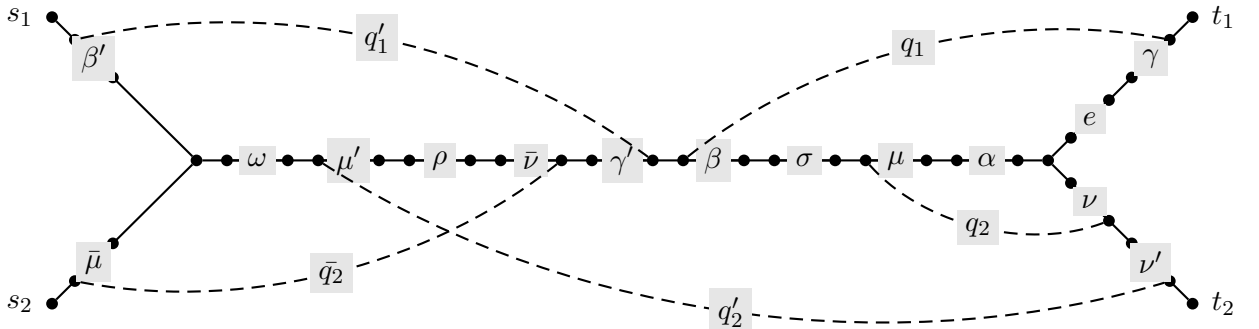
\begin{figure}[H] \centering \psset{unit=1cm}\captionsetup[subfigure]{labelformat=empty}
\subfloat[Second remaining possibility (note: $\bar{\nu}\in \{\rho,\ldots,\gamma'-1\}$)]{
\begin{pspicture}(0,-2.2)(13.9,2.2) 
\Knoten{0}{1.9}{v1}\nput{180}{v1}{$s_1$}
\Knoten{0.3}{1.6}{v2}
\Knoten{0.8}{1.1}{v3}
\Knoten{0}{-1.9}{v4}\nput{180}{v4}{$s_2$}
\Knoten{0.3}{-1.6}{v5}
\Knoten{0.8}{-1.1}{v6}
\Knoten{1.9}{0}{v7}
\Knoten{2.3}{0}{v8}
\Knoten{3.1}{0}{v9}
\Knoten{3.5}{0}{v10}
\Knoten{4.3}{0}{v11}
\Knoten{4.7}{0}{v12}
\Knoten{5.5}{0}{v13}
\Knoten{5.9}{0}{v14}
\Knoten{6.7}{0}{v15}
\Knoten{7.1}{0}{v16}
\Knoten{7.9}{0}{v17}
\Knoten{8.3}{0}{v18}
\Knoten{9.1}{0}{v19}
\Knoten{9.5}{0}{v20}
\Knoten{10.3}{0}{v21}
\Knoten{10.7}{0}{v22}
\Knoten{11.5}{0}{v23}
\Knoten{11.9}{0}{v24}
\Knoten{12.7}{0}{n1}
\Knoten{13.1}{0}{n2}
\Knoten{13.4}{0.3}{v25}
\Knoten{13.9}{0.8}{v26}
\Knoten{14.2}{1.1}{v27}
\Knoten{14.7}{1.6}{v28}
\Knoten{15}{1.9}{v29}\nput{0}{v29}{$t_1$}
\Knoten{13.4}{-0.3}{v30}
\Knoten{13.9}{-0.8}{v31}
\Knoten{14.2}{-1.1}{v32}
\Knoten{14.7}{-1.6}{v33}
\Knoten{15}{-1.9}{v34}\nput{0}{v34}{$t_2$}
\ncline{v1}{v7}
\ncline{v4}{v7}
\ncline{v7}{n2}
\ncline{n2}{v29}
\ncline{n2}{v34}
\Kante{v2}{v3}{\beta'}
\Kante{v5}{v6}{\bar{\mu}}
\Kante{v8}{v9}{\omega}
\Kante{v10}{v11}{\mu'}
\Kante{v12}{v13}{\rho}
\Kante{v14}{v15}{\bar{\nu}}
\Kante{v16}{v17}{\gamma'}
\Kante{v18}{v19}{\beta}
\Kante{v20}{v21}{\sigma}
\Kante{v22}{v23}{\mu}
\Kante{v25}{v26}{e}
\Kante{v27}{v28}{\gamma}
\Kante{v30}{v31}{\nu}
\Kante{v32}{v33}{\nu'}
\Kante{n1}{v24}{\alpha}
\Bogendashed{v5}{v15}{\bar{q_2}}{-25}
\Bogendashed{v22}{v31}{q_2}{-35}
\Bogendashed{v10}{v33}{q_2'}{-25}
\Bogendashed{v18}{v28}{q_1}{25}
\Bogendashed{v2}{v17}{q_1'}{25}
\end{pspicture}
\label{case1iib2}
}
\end{figure}

\begin{figure}[H] \centering \psset{unit=1cm}\captionsetup[subfigure]{labelformat=empty}
\subfloat[Third remaining possibility (note: costs of $e_{\beta}, \ldots, e_{\gamma'}$ zero; $\bar{\nu}\in\{\gamma',\ldots,\sigma-1\}$)]{
\begin{pspicture}(0,-2.2)(13.9,2.2) 
\Knoten{0}{1.9}{v1}\nput{180}{v1}{$s_1$}
\Knoten{0.3}{1.6}{v2}
\Knoten{0.8}{1.1}{v3}
\Knoten{0}{-1.9}{v4}\nput{180}{v4}{$s_2$}
\Knoten{0.3}{-1.6}{v5}
\Knoten{0.8}{-1.1}{v6}
\Knoten{1.9}{0}{v7}
\Knoten{2.3}{0}{v8}
\Knoten{3.1}{0}{v9}
\Knoten{3.5}{0}{v10}
\Knoten{4.3}{0}{v11}
\Knoten{4.7}{0}{v12}
\Knoten{5.5}{0}{v13}
\Knoten{5.9}{0}{v14}
\Knoten{6.7}{0}{v15}
\Knoten{7.1}{0}{v16}
\Knoten{7.9}{0}{v17}
\Knoten{8.3}{0}{v18}
\Knoten{9.1}{0}{v19}
\Knoten{9.5}{0}{v20}
\Knoten{10.3}{0}{v21}
\Knoten{10.7}{0}{v22}
\Knoten{11.5}{0}{v23}
\Knoten{11.9}{0}{v24}
\Knoten{12.7}{0}{n1}
\Knoten{13.1}{0}{n2}
\Knoten{13.4}{0.3}{v25}
\Knoten{13.9}{0.8}{v26}
\Knoten{14.2}{1.1}{v27}
\Knoten{14.7}{1.6}{v28}
\Knoten{15}{1.9}{v29}\nput{0}{v29}{$t_1$}
\Knoten{13.4}{-0.3}{v30}
\Knoten{13.9}{-0.8}{v31}
\Knoten{14.2}{-1.1}{v32}
\Knoten{14.7}{-1.6}{v33}
\Knoten{15}{-1.9}{v34}\nput{0}{v34}{$t_2$}
\ncline{v1}{v7}
\ncline{v4}{v7}
\ncline{v7}{n2}
\ncline{n2}{v29}
\ncline{n2}{v34}
\Kante{v2}{v3}{\beta'}
\Kante{v5}{v6}{\bar{\mu}}
\Kante{v8}{v9}{\omega}
\Kante{v10}{v11}{\mu'}
\Kante{v12}{v13}{\rho}
\Kante{v14}{v15}{\beta}
\Kante{v16}{v17}{\gamma'}
\Kante{v18}{v19}{\bar{\nu}}
\Kante{v20}{v21}{\sigma}
\Kante{v22}{v23}{\mu}
\Kante{v25}{v26}{e}
\Kante{v27}{v28}{\gamma}
\Kante{v30}{v31}{\nu}
\Kante{v32}{v33}{\nu'}
\Kante{n1}{v24}{\alpha}
\Bogendashed{v5}{v19}{\bar{q_2}}{-25}
\Bogendashed{v22}{v31}{q_2}{-35}
\Bogendashed{v10}{v33}{q_2'}{-25}
\Bogendashed{v14}{v28}{q_1}{25}
\Bogendashed{v2}{v17}{q_1'}{25}
\end{pspicture}
\label{case1iib3}
}
\end{figure}
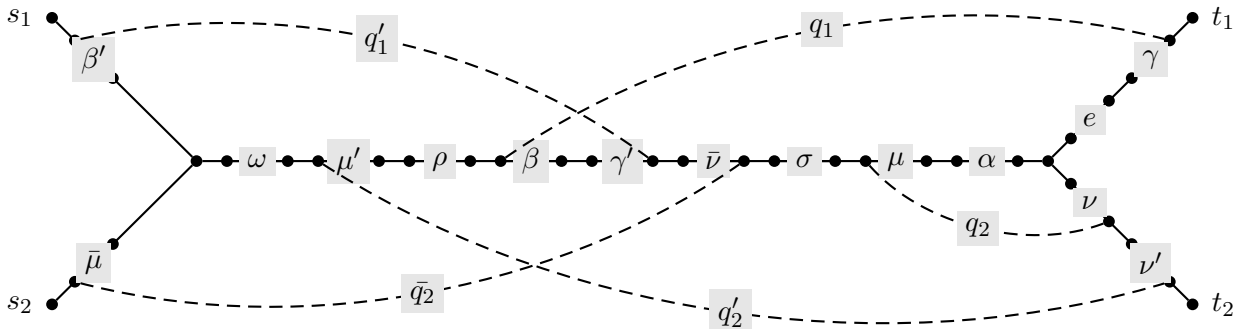

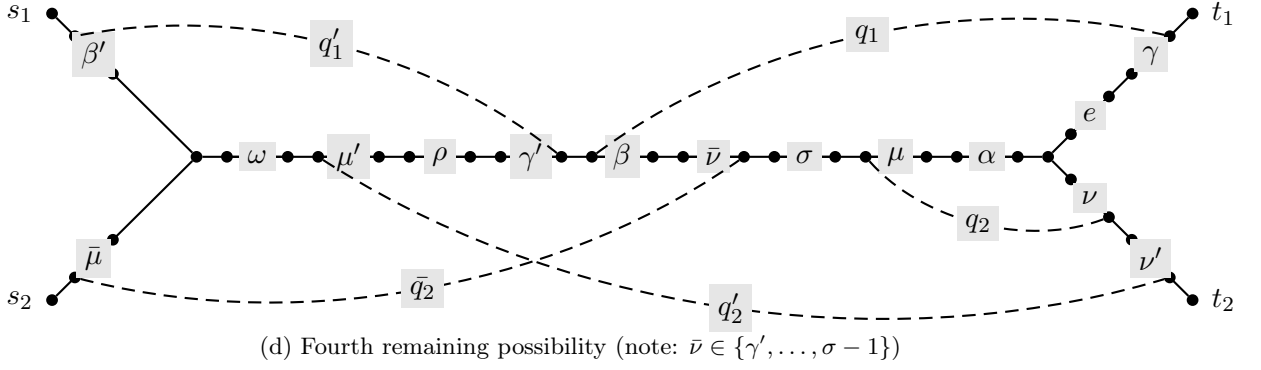
\begin{figure}[H] \centering \psset{unit=1cm}
\subfloat[Fourth remaining possibility (note: $\bar{\nu}\in\{\gamma',\ldots,\sigma-1\}$)]{
\begin{pspicture}(0,-2.2)(13.9,2.2) 
\Knoten{0}{1.9}{v1}\nput{180}{v1}{$s_1$}
\Knoten{0.3}{1.6}{v2}
\Knoten{0.8}{1.1}{v3}
\Knoten{0}{-1.9}{v4}\nput{180}{v4}{$s_2$}
\Knoten{0.3}{-1.6}{v5}
\Knoten{0.8}{-1.1}{v6}
\Knoten{1.9}{0}{v7}
\Knoten{2.3}{0}{v8}
\Knoten{3.1}{0}{v9}
\Knoten{3.5}{0}{v10}
\Knoten{4.3}{0}{v11}
\Knoten{4.7}{0}{v12}
\Knoten{5.5}{0}{v13}
\Knoten{5.9}{0}{v14}
\Knoten{6.7}{0}{v15}
\Knoten{7.1}{0}{v16}
\Knoten{7.9}{0}{v17}
\Knoten{8.3}{0}{v18}
\Knoten{9.1}{0}{v19}
\Knoten{9.5}{0}{v20}
\Knoten{10.3}{0}{v21}
\Knoten{10.7}{0}{v22}
\Knoten{11.5}{0}{v23}
\Knoten{11.9}{0}{v24}
\Knoten{12.7}{0}{n1}
\Knoten{13.1}{0}{n2}
\Knoten{13.4}{0.3}{v25}
\Knoten{13.9}{0.8}{v26}
\Knoten{14.2}{1.1}{v27}
\Knoten{14.7}{1.6}{v28}
\Knoten{15}{1.9}{v29}\nput{0}{v29}{$t_1$}
\Knoten{13.4}{-0.3}{v30}
\Knoten{13.9}{-0.8}{v31}
\Knoten{14.2}{-1.1}{v32}
\Knoten{14.7}{-1.6}{v33}
\Knoten{15}{-1.9}{v34}\nput{0}{v34}{$t_2$}
\ncline{v1}{v7}
\ncline{v4}{v7}
\ncline{v7}{n2}
\ncline{n2}{v29}
\ncline{n2}{v34}
\Kante{v2}{v3}{\beta'}
\Kante{v5}{v6}{\bar{\mu}}
\Kante{v8}{v9}{\omega}
\Kante{v10}{v11}{\mu'}
\Kante{v12}{v13}{\rho}
\Kante{v14}{v15}{\gamma'}
\Kante{v16}{v17}{\beta}
\Kante{v18}{v19}{\bar{\nu}}
\Kante{v20}{v21}{\sigma}
\Kante{v22}{v23}{\mu}
\Kante{v25}{v26}{e}
\Kante{v27}{v28}{\gamma}
\Kante{v30}{v31}{\nu}
\Kante{v32}{v33}{\nu'}
\Kante{n1}{v24}{\alpha}
\Bogendashed{v5}{v19}{\bar{q_2}}{-25}
\Bogendashed{v22}{v31}{q_2}{-35}
\Bogendashed{v10}{v33}{q_2'}{-25}
\Bogendashed{v16}{v28}{q_1}{25}
\Bogendashed{v2}{v15}{q_1'}{25}
\end{pspicture}
\label{case1iib4}
}
\caption{Remaining subcases in Subcase R.1.2.1} \label{case1iib}
\end{figure}

We now show (for all remaining cases) that the given assignment of cost shares cannot be
maximized for Player 2. To see this, we show that the following changes of cost shares (by a suitable small amount) preserve the feasibility for \lp{} (we just say that the changes are feasible):
\[
\begin{array}{l}
\text{Player 1 increases on $e_{\rho}$ and decreases on $e_{\omega}$ and $e_{\sigma}$;} \\
\text{Player 2 decreases on $e_{\rho}$ and increases on $e_{\omega}$ and $e_{\sigma}$.} 
\end{array}
\]

Since the sum of cost shares of Player 2 is then higher than before, we get a (final) contradiction to our assumption that the changes of \texttt{CHANGE($\rho,\sigma$)} are not feasible for Player 1.


It is clear that Player 1 can increase on $e_{\rho}$ and decrease on $e_{\omega}$ and $e_{\sigma}$ since every (left or right) alternative that substitutes $e_{\rho}$ also substitutes $e_{\omega}$ or $e_{\sigma}$.

\vspace{0.5em}
The changes of Player 2 are more complicated. First, decreasing on $e_{\rho}$ and increasing on $e_{\omega}$ is feasible because of our choice of $\bar{q_2}$. 
	Additionally we want to increase on $e_{\sigma}$. Assume that this is not possible. Since $q_2'$ is a largest right alternative and there are no left alternatives for $e_{\sigma}$, there has to be a right alternative which substitutes $e_{\sigma}$, but not $e_{\rho}$. Let $\widehat{q_2}$, defined by $\widehat{\mu}$ and $\widehat{\nu}$, be a largest such alternative. If $\rho+1 \leq \widehat{\mu} \leq \beta$ we get a cheaper solution by using $q_1$ and $\widehat{q_2}$ since Player 1 pays the edges $e_{\rho+1},\ldots,e_{\beta-1}$. 
	We now have to distinguish between the different cases of Figure~\ref{case1iib}. 
	
	\vspace{0.5em}
	We start with Figures~\ref{case1iib1} and~\ref{case1iib2}. In fact we can restrict to the case of Figure~\ref{case1iib2}. It is quite clear that the case of Figure~\ref{case1iib1} can be treated almost analogously, just imagine to contract the edges $e_{\beta}, \ldots, e_{\gamma'}$ (all those  edges have cost zero), and this yields the case of Figure~\ref{case1iib2} (for $\gamma'=\beta-1$). Therefore consider the case of Figure~\ref{case1iib2}.
	
If Player 2 pays the edges $e_{\beta}, \ldots, e_{\widehat{\mu}-1}$, we get a cheaper solution by using $q_1$ and $\widehat{q_2}$. Therefore let $e_{\sigma'}$ be the largest edge of $e_{\beta}, \ldots, e_{\widehat{\mu}-1}$ that is not completely paid by Player 2. 
If the described changes of cost shares are feasible for $\sigma'$ instead of $\sigma$, we redefine $\sigma$ by $\sigma'$ and get that the changes are feasible. 
 In the other case there either has to be a right alternative which substitutes $e_{\sigma'}$, but not $e_{\rho}$, a right alternative which substitutes $e_{\omega}$ or a left alternative which substitutes $e_{\sigma'}$. Since the first two cases yield contradictions to the choice of $\widehat{q_2}$ and $q_2'$, the third case has to be true. Let $\tilde{q_2}$, defined by $\tilde{\mu}$ and $\tilde{\nu}$, be a left alternative which substitutes $e_{\sigma'}$. Note that $\tilde{q_2}$ does not substitute $e_{\sigma}$ since there are no left alternatives for this edge. If $\tilde{\nu}\leq \widehat{\mu}-1$, we can construct a cheaper Steiner Forest by using $q_1',q_1,\tilde{q_2}$ and $\widehat{q_2}$. Thus $(P_1,P_2,q_1,\tilde{q_2},\widehat{q_2})$ has to be a PBC for $(u,v,w,x)=(s_1,t_1,s_2,t_2)$ (Subcase \hypertarget{pbc10}{PBC10}). \autoref{lemmarest} shows that the Subcase PBC10 can not occur. Therefore the changes of Player 2 are also feasible (perhaps for a slightly changed $\sigma$) for the situations of Figure~\ref{case1iib1} and Figure~\ref{case1iib2}. 

\vspace{0.5em}
Now consider the situations of Figures~\ref{case1iib3} and~\ref{case1iib4}. As above it is quite clear that Figure~\ref{case1iib3} can be treated very similar to Figure~\ref{case1iib4}, therefore we only consider the latter case.
	
For the situation displayed in Figure~\ref{case1iib4} we first consider the case $\widehat{\mu} \in \{\beta+1, \ldots,\bar{\nu}\}$. We get that $(P_1,P_2,q_1,\bar{q_2},\widehat{q_2})$ is a PBC for $(u,v,w,x)=(s_1,t_1,s_2,t_2)$ (Subcase \hypertarget{pbc11}{PBC11}). By \autoref{lemmarest} we get that $\widehat{\mu}\geq \bar{\nu}+1$ has to hold, because Subcase PBC11 is not possible. 
	 If Player 2 pays the edges $e_{\bar{\nu}+1},\ldots,e_{\widehat{\mu}-1}$ we use $q_1', q_1, \bar{q_2}$ and $\widehat{q_2}$ to construct a cheaper solution. 
	Therefore let $e_{\sigma'}$ be the largest edge of $e_{\bar{\nu}+1}, \ldots, e_{\widehat{\mu}-1}$ that is not completely paid by Player 2. If the changes are feasible for $\sigma'$ instead of $\sigma$, we redefine $\sigma$ by $\sigma'$ and get that the changes are feasible. 
	In the other case we get a contradiction. We omit the proof since it is analogous to the proof for the case of Figure~\ref{case1iib3}. This shows that the changes of Player 2 are also feasible (perhaps for a slightly changed $\sigma$) for the situations of Figure~\ref{case1iib3} and Figure~\ref{case1iib4}. 

\vspace{1em}
\hypertarget{r122}{\textit{Subcase R.1.2.2}}

Overall we showed now that the changes of \texttt{CHANGE($\rho,\sigma$)} are feasible for Player 1. Therefore they can not be feasible for Player 2. Then there has to be a right alternative that substitutes $e_{\sigma}$, but not $e_{\rho}$; let us consider a largest such alternative $\bar{q_2}$, defined by $\bar{\mu}$ and $\bar{\nu}$. Obviously, $\rho+1 \leq \bar{\mu} \leq \sigma$. The case $\beta < \bar{\mu}$ is illustrated in Figure~\ref{case1iibp2}. In this case we can use $q_1$ and $\bar{q_2}$ to get a cheaper solution if Player 2 pays the edges $e_{\beta},\ldots,e_{\bar{\mu}-1}$. Otherwise consider such an edge $e_{\omega}$ which is not paid completely by Player 2. But then we get a contradiction since \texttt{CHANGE($\rho,\omega$)} is feasible (feasible for Player 2 because of our choice of $\bar{q_2}$; for Player 1 because of the feasibility of \texttt{CHANGE($\rho,\sigma$)}).
Therefore $\bar{\mu} \leq \beta$ has to hold. Now use $\bar{q_2}$ and $q_1$ to construct a cheaper Steiner forest (note that Player 1 pays the edges $e_{\rho+1},\ldots,e_{\beta-1}$).

This completes our analysis of Subcase R.1 since we showed that this case can not occur.
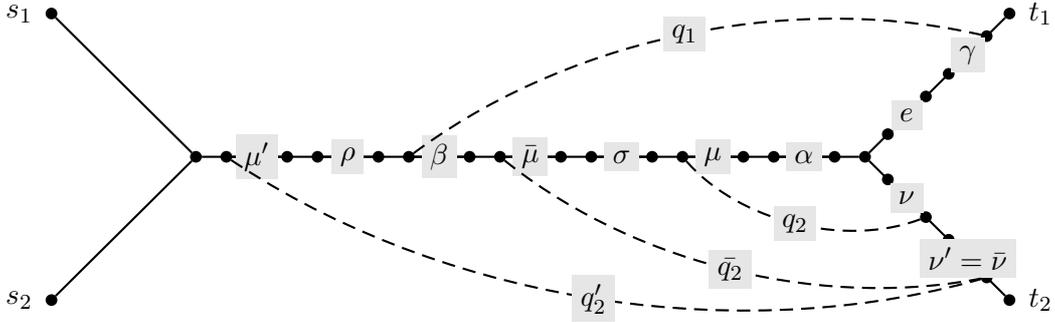
\begin{figure}[H] \centering \psset{unit=1cm}
\begin{pspicture}(1,-2)(14,2) 
\Knoten{1.2}{1.9}{v1}\nput{180}{v1}{$s_1$}
\Knoten{1.2}{-1.9}{v4}\nput{180}{v4}{$s_2$}
\Knoten{3.1}{0}{v7}
\Knoten{3.5}{0}{v10}
\Knoten{4.3}{0}{v11}
\Knoten{4.7}{0}{v12}
\Knoten{5.5}{0}{v13}
\Knoten{5.9}{0}{v14}
\Knoten{6.7}{0}{v15}
\Knoten{7.1}{0}{v16}
\Knoten{7.9}{0}{v17}
\Knoten{8.3}{0}{v18}
\Knoten{9.1}{0}{v19}
\Knoten{9.5}{0}{v20}
\Knoten{10.3}{0}{v21}
\Knoten{10.7}{0}{v22}
\Knoten{11.5}{0}{v23}
\Knoten{11.9}{0}{v24}
\Knoten{12.2}{0.3}{v25}
\Knoten{12.7}{0.8}{v26}
\Knoten{13.0}{1.1}{v27}
\Knoten{13.5}{1.6}{v28}
\Knoten{13.8}{1.9}{v29}\nput{0}{v29}{$t_1$}
\Knoten{12.2}{-0.3}{v30}
\Knoten{12.7}{-0.8}{v31}
\Knoten{13.0}{-1.1}{v32}
\Knoten{13.5}{-1.6}{v33}
\Knoten{13.8}{-1.9}{v34}\nput{0}{v34}{$t_2$}

\ncline{v1}{v7}
\ncline{v4}{v7}
\ncline{v7}{v24}
\ncline{v24}{v29}
\ncline{v24}{v34}

\Kante{v10}{v11}{\mu'}
\Kante{v12}{v13}{\rho}
\Kante{v14}{v15}{\beta}
\Kante{v16}{v17}{\bar{\mu}}
\Kante{v18}{v19}{\sigma}
\Kante{v20}{v21}{\mu}
\Kante{v22}{v23}{\alpha}
\Kante{v25}{v26}{e}
\Kante{v27}{v28}{\gamma}
\Kante{v30}{v31}{\nu}
\Kante{v32}{v33}{\nu'=\bar{\nu}}

\Bogendashed{v16}{v33}{\bar{q_2}}{-25}
\Bogendashed{v20}{v31}{q_2}{-35}
\Bogendashed{v10}{v33}{q_2'}{-25}
\Bogendashed{v14}{v28}{q_1}{25}
\end{pspicture}
\caption{Subcase R.1.2.2 (note that the order of $\nu, \nu'$ and $\bar{\nu}$ is not relevant, therefore we have chosen $\nu'=\bar{\nu}$ for simplification)} \label{case1iibp2}
\end{figure}

\vspace{1em}
\hypertarget{r2}{\textbf{Subcase R.2}}

Now we analyze Subcase R.2 ($q_2$ smallest right alternative for $e_{\alpha}$ and $\mu \leq \beta-1$).
If $q_1$ substitutes all commonly used edges or Player 1 pays all edges of the commonly used path which are substituted by $q_2$, but not by $q_1$, using $q_1$ and $q_2$ yields a cheaper solution. Thus let $e_{\sigma}$ be the largest such edge which Player 1 does not pay completely. Figure~\ref{r2} illustrates the situation for $q_2$ small.

\begin{figure}[H] \centering \psset{unit=0.95cm}
\begin{pspicture}(-0.3,-2)(11,1.7) 
\Knoten{0.2}{1.5}{s1}\nput{180}{s1}{$s_1$}
\Knoten{0.2}{-1.5}{s2}\nput{180}{s2}{$s_2$}
%
\Knoten{1.7}{0}{v5}
\Knoten{2.5}{0}{v6}
\Knoten{3.3}{0}{v7}
\Knoten{4.1}{0}{v8}
\Knoten{4.9}{0}{v9}
\Knoten{5.7}{0}{v10}
\Knoten{6.5}{0}{v11}
\Knoten{7.7}{0}{v14}
\Knoten{8.5}{0}{v15}
\Knoten{9.3}{0}{v16}
\Knoten{10.16}{0.86}{v17}
\Knoten{10.56}{1.26}{v18}
\Knoten{10.8}{1.5}{t1}\nput{0}{t1}{$t_1$}
\Knoten{9.8}{-0.5}{v19}
\Knoten{10.3}{-1}{v20}
\Knoten{10.8}{-1.5}{t2}\nput{0}{t2}{$t_2$}
\Knoten{9.53}{0.23}{n1}
\Knoten{9.93}{0.63}{n2}
\ncline{s1}{v5}
\ncline{s2}{v5}
\ncline{v5}{v6}
\Kante{v6}{v7}{\mu}
\ncline{v7}{v8}
\Kante{v8}{v9}{\sigma}
\ncline{v9}{v10}
\Kante{v10}{v11}{\beta}
\ncline{v11}{v14}
\Kante{v14}{v15}{\alpha}
\ncline{v15}{v16}
\ncline{v16}{n1}
\ncline{n2}{v17}
\Kante{v17}{v18}{\gamma}
\ncline{v18}{t1}
\ncline{v16}{v19}
\Kante{v19}{v20}{\nu}
\ncline{v20}{t2}
\Kante{n1}{n2}{e}

\Bogendashed{v10}{v18}{q_1}{35}
\Bogendashed{v6}{v20}{q_2}{-35}
\end{pspicture}
\caption{Subcase R.2 for $q_2$ small} \label{r2}
\end{figure}
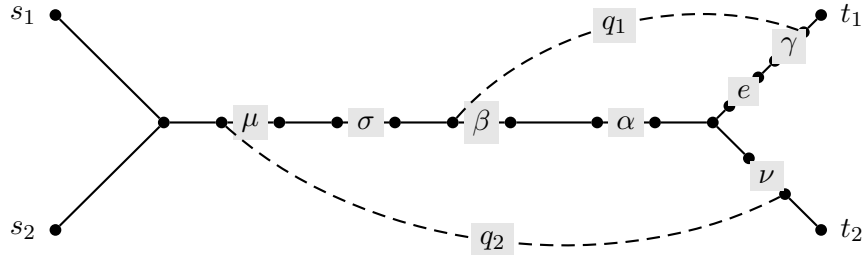

We have to analyze why \texttt{CHANGE($\sigma, \alpha$)} is not feasible. It it clear that the changes are feasible for Player~2 because of our choice of $q_2$. Therefore the changes can not be feasible for Player~1, thus there has to be a left alternative for Player~1 which substitutes $e_{\sigma}$, but not $e_{\alpha}$. We consider a smallest such alternative $q_1'$ (defined by $\beta'$ and $\gamma'$). 
If $\gamma'\geq\beta$ holds, $(P_2,P_1,q_2,q_1',q_1)$ is a PBC for $(u,v,w,x)=(s_2,t_2,s_1,t_1)$ (Subcase \hypertarget{pbc12}{PBC12}). 
We have to distinguish if $q_2$ is small or big. In both cases we get that Subcase PBC12 is not possible (the case where $q_2$ is small is discussed in \autoref{lemmapbc14}, the case $q_2$ big in \autoref{lemmarest}). Therefore $\sigma \leq \gamma' \leq \beta-1$ has to hold, see Figure~\ref{case3b2} (for $q_2$ small).

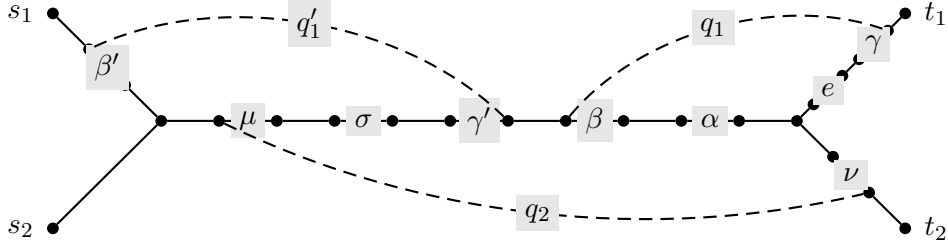
\begin{figure}[H] \centering \psset{unit=0.95cm}
\begin{pspicture}(-1.5,-1.5)(11,1.7) 
\Knoten{-1}{1.5}{s1}\nput{180}{s1}{$s_1$}
\Knoten{-0.5}{1}{v1}
\Knoten{0}{0.5}{v2}
\Knoten{-1}{-1.5}{s2}\nput{180}{s2}{$s_2$}
%
\Knoten{0.5}{0}{v5}
\Knoten{1.3}{0}{v6}
\Knoten{2.1}{0}{v7}
\Knoten{2.9}{0}{v8}
\Knoten{3.7}{0}{v9}
\Knoten{4.5}{0}{v10}
\Knoten{5.3}{0}{v11}
\Knoten{6.1}{0}{v12}
\Knoten{6.9}{0}{v13}
\Knoten{7.7}{0}{v14}
\Knoten{8.5}{0}{v15}
\Knoten{9.3}{0}{v16}
\Knoten{10.16}{0.86}{v17}
\Knoten{10.56}{1.26}{v18}
\Knoten{10.8}{1.5}{t1}\nput{0}{t1}{$t_1$}
\Knoten{9.8}{-0.5}{v19}
\Knoten{10.3}{-1}{v20}
\Knoten{10.8}{-1.5}{t2}\nput{0}{t2}{$t_2$}
\Knoten{9.53}{0.23}{n1}
\Knoten{9.93}{0.63}{n2}
\ncline{s1}{v1}
\Kante{v1}{v2}{\beta '}
\ncline{v2}{v5}
\ncline{s2}{v5}
\ncline{v5}{v6}
\Kante{v6}{v7}{\mu}
\ncline{v7}{v8}
\Kante{v8}{v9}{\sigma}
\ncline{v9}{v10}
\Kante{v10}{v11}{\gamma '}
\ncline{v11}{v12}
\Kante{v12}{v13}{\beta}
\ncline{v13}{v14}
\Kante{v14}{v15}{\alpha}
\ncline{v15}{v16}
\ncline{v16}{n1}
\ncline{n2}{v17}
\Kante{v17}{v18}{\gamma}
\ncline{v18}{t1}
\ncline{v16}{v19}
\Kante{v19}{v20}{\nu}
\ncline{v20}{t2}
\Kante{n1}{n2}{e}

\Bogendashed{v1}{v11}{q_1'}{35}
\Bogendashed{v12}{v18}{q_1}{35}
\Bogendashed{v6}{v20}{q_2}{-20}
\end{pspicture}
\caption{Subcase of R.2 for $q_2$ small} \label{case3b2}
\end{figure}

Now let $q_2'$ be a largest right alternative for Player 2. If $q_2'$ substitutes all commonly used edges or Player 2 pays all commonly used edges which are not substituted by $q_2'$, we get a cheaper Steiner forest by using $q_1,q_1'$ and $q_2'$ (note that Player 1 pays the edges $e_{\gamma'+1}, \ldots, e_{\beta-1}$).
Therefore let $e_{\tau}$ be the largest edge in $\{\ell_1+\ell_2+1, \ldots, \mu'-1\}$ which Player 2 does not pay completely.
Since the cost shares are maximized for Player 2 there has to be a tight alternative for Player 2 which substitutes $e_{\tau}$. 
Since $q_2'$ is the largest right alternative, this has to be a left alternative. 
Let $\bar{q_2}$, defined by $\bar{\mu}$ and $\bar{\nu}$, be a smallest such alternative. 
If $\bar{\nu}\leq\mu'-1$ holds, we can construct a cheaper solution by using $q_1',q_1,q_2'$ and $\bar{q_2}$ since Player 2 pays the edges $e_{\tau+1}, \ldots, e_{\mu'-1}$ by the choice of $\tau$.
For $\bar{\nu}\geq\alpha$ we also get a cheaper solution since Player 2 pays the edges $e_{\alpha+1}, \ldots, e_{\ell_1+\ell_2+m}$. Therefore $\bar{\nu} \in \{\mu',\ldots, \alpha-1\}$ holds.
For $\mu' \leq \bar{\nu} \leq \sigma-1$ we get that $(P_1,P_2,q_1',q_2',\bar{q_2})$ is a PBC for $(u,v,w,x)=(t_1,s_1,t_2,s_2)$ (Subcase \hypertarget{pbc13}{PBC13}) We get that this subcase is not possible (see \autoref{lemmarest}), therefore $\sigma \leq \bar{\nu}\leq \alpha-1$ is the only remaining possibility, see Figure~\ref{abb:R.2.R2} where we have illustrated the situation for $\bar{\nu}\leq \gamma'$. 
%
Now we simultaneously change the cost shares (by a suitably small amount) as described below to get a contradiction to the assumption that the cost shares are maximized for Player 2. 
\[
\begin{array}{l}
\text{Player 1 increases on $e_{\sigma}$ and decreases on $e_{\alpha}$ and $e_{\tau}$:} \\
\text{Player 2 decreases on $e_{\sigma}$ and increases on $e_{\alpha}$ and $e_{\tau}$.} 
\end{array}
\]
It is clear that Player 1 can increase on $e_{\sigma}$ and decrease on $e_{\alpha}$ and $e_{\tau}$ since every alternative that substitutes $e_{\sigma}$ also substitutes $e_{\alpha}$ or $e_{\tau}$.
It is also clear that Player 2 can decrease on $e_{\sigma}$ and increase on $e_{\tau}$ since $\bar{q_2}$ is a smallest left alternative for $e_{\tau}$. Additionally increasing on $e_{\alpha}$ is also possible: There can not be a right alternative which substitutes $e_{\alpha}$, but not $e_{\sigma}$ (note that $q_2$ also substitutes $e_{\sigma}$); there can not be a right alternative which substitutes $e_{\tau}$ because $q_2'$ is the largest right alternative and a left alternative which substitutes $e_{\alpha}$ would lead to a cheaper Steiner forest (together with $q_1'$ and $q_1$).

\begin{figure}[H] \centering \psset{unit=1cm}
\begin{pspicture}(-5,-2.3)(10,1.7) 
\Knoten{-5}{1.5}{s1}\nput{180}{s1}{$s_1$}
\Knoten{-4.5}{1}{v1}
\Knoten{-4}{0.5}{v2}
\Knoten{-5}{-1.5}{s2}\nput{180}{s2}{$s_2$}
\Knoten{-4.5}{-1}{v3}
\Knoten{-4}{-0.5}{v4}
\Knoten{-3.5}{0}{v5}
\Knoten{-2.7}{0}{n1}
\Knoten{-1.9}{0}{n2}
\Knoten{-1.1}{0}{v6}
\Knoten{-0.3}{0}{v7}
\Knoten{0.5}{0}{v8}
\Knoten{1.3}{0}{v9}
\Knoten{2.1}{0}{v10}
\Knoten{2.9}{0}{v11}
\Knoten{3.7}{0}{n3}
\Knoten{4.5}{0}{n4}
\Knoten{5.3}{0}{v12}
\Knoten{6.1}{0}{v13}
\Knoten{6.9}{0}{v14}
\Knoten{7.7}{0}{v15}
\Knoten{8.5}{0}{v16}
\Knoten{8.73}{0.23}{n5}
\Knoten{9.13}{0.63}{n6}
\Knoten{9.36}{0.86}{v17}
\Knoten{9.76}{1.26}{v18}
\Knoten{10}{1.5}{t1}\nput{0}{t1}{$t_1$}
\Knoten{9}{-0.5}{v19}
\Knoten{9.5}{-1}{v20}
\Knoten{10}{-1.5}{t2}\nput{0}{t2}{$t_2$}
\ncline{s1}{v1}
\Kante{v1}{v2}{\beta '}
\ncline{v2}{v5}
\ncline{s2}{v3}
\ncline{v4}{v5}
\ncline{v5}{n1}
\Kante{n1}{n2}{\tau}
\ncline{n2}{v6}
\Kante{v6}{v7}{\mu'}
\Kante{v3}{v4}{\bar{\mu}}
\ncline{v7}{v8}
\Kante{v8}{v9}{\sigma}
\ncline{v9}{v10}
\Kante{v10}{v11}{\bar{\nu}}
\ncline{v11}{n3}
\Kante{n3}{n4}{\gamma '}
\ncline{n4}{v12}
\Kante{v12}{v13}{\beta}
\ncline{v13}{v14}
\Kante{v14}{v15}{\alpha}
\ncline{v15}{v16}
\ncline{v16}{n5}
\Kante{n5}{n6}{e}
\ncline{n6}{v17}
\Kante{v17}{v18}{\gamma}
\ncline{v18}{t1}
\ncline{v16}{v19}
\Kante{v19}{v20}{\nu'}
\ncline{v20}{t2}

\Bogendashed{v1}{n4}{q_1'}{25}
\Bogendashed{v12}{v18}{q_1}{30}
\Bogendashed{v6}{v20}{q_2'}{-25}
\Bogendashed{v3}{v11}{\bar{q_2}}{-25}
\end{pspicture}
\caption{Remaining possibility for Subcase R.2 ($\gamma'<\bar{\nu}$ also possible)} \label{abb:R.2.R2}
\end{figure}
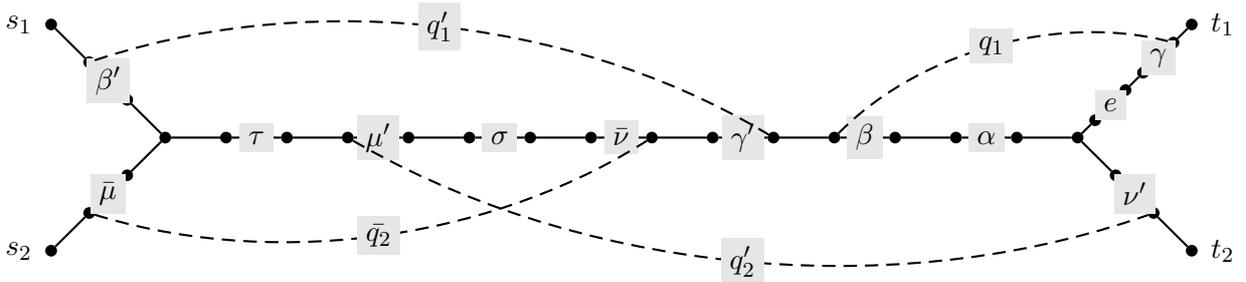

This completes the proof of Subcase R.2 since we showed that this case can not occur.  

\vspace{1em}
\hypertarget{r3}{\textbf{Subcase R.3}}

In Subcase R.3, we assume that $\ell_1+1 \leq \mu \leq \ell_1+\ell_2$ and $\nu\leq\ell_1+\ell_2+m$ (remember that $q_2$ is a tight left alternative which substitutes $e_{\alpha}$ and there are no tight right alternatives for $e_{\alpha}$), see Figure~\ref{case3bild}. 

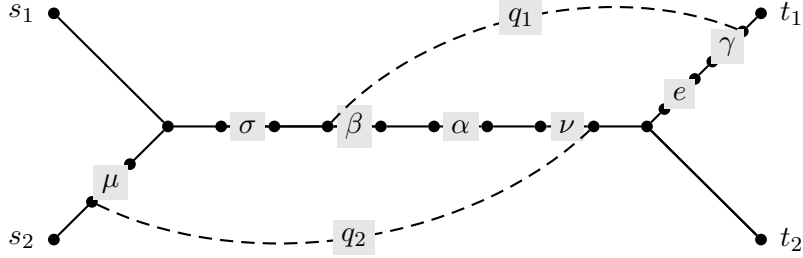
\begin{figure}[H] \centering \psset{unit=1cm}
\begin{pspicture}(-1.5,-1.7)(8.6,1.7) 
\Knoten{-1.1}{1.5}{v-1}\nput{180}{v-1}{$s_1$}
\Knoten{-1.1}{-1.5}{v0}\nput{180}{v0}{$s_2$}
\Knoten{0.4}{0}{v-3}
\Knoten{1.1}{0}{v-2}
\Knoten{1.8}{0}{v3}
\Knoten{2.5}{0}{v4}
\Knoten{3.2}{0}{v5}
\Knoten{3.9}{0}{v6}
\Knoten{4.6}{0}{v7}
\Knoten{5.3}{0}{v8}
\Knoten{6.0}{0}{n1}
\Knoten{6.7}{0}{n2}

\Knoten{6.93}{0.23}{n3}
\Knoten{7.33}{0.63}{n4}
\Knoten{8.2}{1.5}{v9}\nput{0}{v9}{$t_1$}
\Knoten{8.2}{-1.5}{v10}\nput{0}{v10}{$t_2$}
\Knoten{7.96}{1.26}{v1}
\Knoten{7.56}{0.86}{v11}
\Knoten{-0.6}{-1}{v13}
\Knoten{-0.1}{-0.5}{v14}

\ncline{-}{v-1}{v-3}
\ncline{-}{v0}{v-3}
\ncline{-}{v-3}{v4}
\Kante{v3}{v-2}{\sigma}
\ncline{-}{v4}{v5}
\ncline{v4}{v3}
\ncline{-}{v5}{v6}
\Kante{v7}{v6}{\alpha}
\Kante{v4}{v5}{\beta}
\ncline{-}{v7}{v8}
\Kante{v8}{n1}{\nu}
\ncline{n1}{n2}
\ncline{n2}{n3}
\Kante{n3}{n4}{e}
\ncline{n4}{v11}
\ncline{-}{n2}{v10}
\Kante{v11}{v1}{\gamma}
\Kante{v13}{v14}{\mu}
\ncline{-}{n2}{v10}
\ncline{-}{v1}{v9}
\Bogendashed{v1}{v4}{q_1}{-35}
\Bogendashed{v13}{n1}{q_2}{-35}
\end{pspicture}
\caption{Subcase R.3}
\label{case3bild}
\end{figure}
As Player 2 pays the edges $e_{\alpha+1},\ldots,e_{\ell_1+\ell_2+m}$ completely, we get a cheaper solution if $q_1$ substitutes all commonly used edges or Player 1 pays the edges $e_{\ell_1+\ell_2+1},\ldots,e_{\beta-1}$ completely. 
Thus let $e_{\sigma}$ be the largest of these edges that Player 1 does not pay completely. We now analyze why \texttt{CHANGE($\sigma,\alpha$)} is not feasible. 
The changes for Player 2 are feasible since there is no right alternative for $e_{\alpha}$.
Therefore, the changes of Player 1 must not be feasible. Thus there has to be a tight left alternative for Player 1 which substitutes $e_{\sigma}$, but not $e_{\alpha}$. Let $q_1'$, defined by $\beta'$ and $\gamma'$, be a smallest such one. If $\gamma'\leq\beta-1$, we get a cheaper Steiner forest by using $q_1',q_1$ and $q_2$ because Player 1 pays the edges $e_{\sigma+1},\ldots,e_{\beta-1}$ completely. 
Therefore $\beta \leq \gamma'\leq \alpha-1$ has to hold and $(P_2,P_1,q_2,q_1,q_1')$ is a PBC for $(u,v,w,x)=(t_2,s_2,t_1,s_1)$ (Subcase \hypertarget{pbc14}{PBC14}). \autoref{lemmarest} shows that this is also not possible, completing the whole proof of \autoref{lemma3}.
\end{proof}
\subsubsection{Analysis of constructed PBCs} \label{subsubsec:pbcs}
The rest of the paper analyzes the PBCs which we constructed in the proof of \autoref{lemma3}. 
Since we excluded the existence of BCs, these subgraphs cannot be BCs. We first define twelve types of PBCs which are no BCs (called NBCs, see~\autoref{def:NBC}) and show that each PBC has to be an NBC if we exclude the existence of BCs (see \autoref{lemma:pbc}). 
Then we derive properties for each type of NBC (see \autoref{lemmanbcanfang} - \autoref{lemmanbcende}) which are finally used to get contradictions to the properties of the PBCs constructed in \autoref{lemma3} (see \autoref{lemmapbc1} - \autoref{lemmarest}).

\vspace{1em}
Using the definition of a PBC, we first give an exact definition of Bad Configurations.
To this aim it is useful to consider the paths $q_1, q_2$ and $q_3$ as directed paths. For $q_1$, we choose the direction as follows: If we consider $P_u$ as directed from $u$ to $v$, the first node of this directed path which is contained in $q_1$ is the start node of $q_1$. 
Considering Figure \ref{PBC} this means that the start node of $q_1$ is the left node of the two endnodes of $q_1$. 
Therefore we refer to this direction as ``directed from left to right'' and write $\vec{q_1}$ for this directed version of $q_1$. The directed versions $\vec{q_2}$ and $\vec{q_3}$ of $q_2$ and $q_3$ are chosen analogously by considering $P_{\ell}$ as directed from $w$ to $x$. 

 Furthermore we need to subdivide the paths $q_1, q_2$ and $q_3$ in subpaths. In terms of notation we will always use $\alpha_i$ for subpaths of $q_1$, $\beta_i$ for subpaths of $q_2$ and $\gamma_i$ for $q_3$. The (directed) subpaths of $\vec{q_1}$, $\vec{q_2}$ and $\vec{q_3}$ are written as $\vec{\alpha_i}$, $\vec{\beta_i}$ and $\vec{\gamma_i}$.

\begin{definition}\label{formaldef:BC}
We call a PBC $(P_u, P_{\ell}, q_1,q_2,q_3)$ a \emph{BC$i$} for an $i \in \{1,2,3,4\}$, if the properties of BC$i$ (described below) are fulfilled. 
Note that some of the BC$i$s have different subtypes which are described and illustrated in the corresponding figures.

\newpage
	\textbf{BC1:}
	\begin{itemize}[itemsep=-0em,leftmargin=*]
		\item $q_1,q_2$ and $q_3$ are pairwise node-disjoint;
		\item $q_1,q_2$ and $q_3$ are internal node-disjoint with $P_u \cup P_{\ell}$.
	\end{itemize}
	\begin{figure}[H] \centering \psset{unit=1cm}
	\captionsetup[subfigure]{labelformat=empty}
\subfloat[BC1a ($q_1$ small)]{
 \begin{pspicture}(-0.5,-1.2)(7,1.2) 
\Knoten{0}{1}{v1}\nput{180}{v1}{$u$}
\Knoten{0}{-1}{v2}\nput{180}{v2}{$w$}
\Knoten{0.5}{-0.5}{n2}
\Knoten{1}{0}{v3}
\Knoten{2.125}{0}{n5}
\Knoten{3.25}{0}{v4}
\Knoten{4.375}{0}{v6}
\Knoten{5.5}{0}{v7}
\Knoten{6}{0.5}{n3}
\Knoten{6}{-0.5}{n4}
\Knoten{6.5}{1}{v10}\nput{0}{v10}{$v$}
\Knoten{6.5}{-1}{v11}\nput{0}{v11}{$x$}

\ncline{-}{v1}{v3}
\ncline{-}{n2}{v3}
\ncline{-}{v2}{n2}
\ncline{-}{v3}{v4}
\ncline{-}{v4}{n5}
\ncline{-}{v4}{v6}
\ncline{-}{v6}{v7}
\ncline{-}{v7}{n3}
\ncline{-}{n3}{v10}
\ncline{-}{v7}{n4}
\ncline{-}{n4}{v11}
\Bogen{n5}{n3}{q_1}{35}
\Bogen{n2}{v6}{q_2}{-35}
\Bogen{v4}{n4}{q_3}{-35}
\end{pspicture}
}
\hspace{1cm}
\subfloat[BC1b ($q_1$ big)]{
 \begin{pspicture}(0.5,-1.2)(7,1.2) 
\Knoten{1}{1}{v1}\nput{180}{v1}{$u$}
\Knoten{1.5}{0.5}{n1}
\Knoten{1}{-1}{v2}\nput{180}{v2}{$w$}
\Knoten{1.5}{-0.5}{n2}
\Knoten{2}{0}{v3}
\Knoten{3.25}{0}{v4}
\Knoten{4.375}{0}{v6}
\Knoten{5.5}{0}{v7}
\Knoten{6}{0.5}{n3}
\Knoten{6}{-0.5}{n4}
\Knoten{6.5}{1}{v10}\nput{0}{v10}{$v$}
\Knoten{6.5}{-1}{v11}\nput{0}{v11}{$x$}

\ncline{-}{v1}{v3}
\ncline{-}{n2}{v3}
\ncline{-}{v2}{n2}
\ncline{-}{v3}{v4}
\ncline{-}{v4}{v3}
\ncline{-}{v4}{v6}
\ncline{-}{v6}{v7}
\ncline{-}{v7}{n3}
\ncline{-}{n3}{v10}
\ncline{-}{v7}{n4}
\ncline{-}{n4}{v11}
\Bogen{n1}{n3}{q_1}{25}
\Bogen{n2}{v6}{q_2}{-35}
\Bogen{v4}{n4}{q_3}{-35}
\end{pspicture}
}
\caption{BC1a and b}
\end{figure}
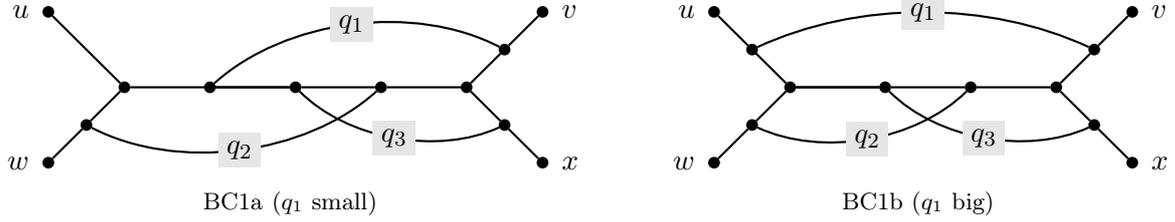

	\textbf{BC2:}
	\begin{itemize}[itemsep=-0em,leftmargin=*]
		\item $q_3$ is node-disjoint with $q_1$ and with $q_2$;
		\item $q_1$ and $q_2$ are not node-disjoint and $\alpha_2=\beta_2$, where $\alpha_2$ $(\beta_2)$ is the subpath of $\vec{q_1}$ $(\vec{q_2})$ from the first until the last node which is contained in $q_2$ $(q_1)$;
		\item $q_1,q_2$ and $q_3$ are internal node-disjoint with $P_u \cup P_{\ell}$.
	\end{itemize}
\begin{figure}[H] \centering \psset{unit=1cm}
\captionsetup[subfigure]{labelformat=empty}
\subfloat[BC2a ($q_1$ small, $\vec{\alpha_2}=\vec{\beta_2}$)]{
\begin{pspicture}(-0.5,-1)(7,1.3) 
\Knoten{0}{1}{v1}\nput{180}{v1}{$u$}
\Knoten{0}{-1}{v2}\nput{180}{v2}{$w$}
\Knoten{0.5}{-0.5}{n2}
\Knoten{1}{0}{v3}
\Knoten{2.125}{0}{n5}
\Knoten{3.25}{0}{v4}
\Knoten{4.375}{0}{v6}
\Knoten{5.5}{0}{v7}
\Knoten{6}{0.5}{n3}
\Knoten{6}{-0.5}{n4}
\Knoten{6.5}{1}{v10}\nput{0}{v10}{$v$}
\Knoten{6.5}{-1}{v11}\nput{0}{v11}{$x$}

\Knoten{3.25}{0.75}{n6}
\Knoten{4.375}{0.75}{n7}

\Kante{n6}{n7}{\alpha_2}
\ncline{n7}{n3}\ncput{\colorbox{almostwhite}{$\alpha_3$}}

\ncline{n5}{n6}\ncput{\colorbox{almostwhite}{$\alpha_1$}}
\ncline{n7}{v6}\ncput{\colorbox{almostwhite}{$\beta_3$}}

\ncline{-}{v1}{v3}
\ncline{-}{n2}{v3}
\ncline{-}{v2}{n2}
\ncline{-}{v3}{v4}
\ncline{-}{v4}{n5}
\ncline{-}{v4}{v6}
\ncline{-}{v6}{v7}
\ncline{-}{v7}{n3}
\ncline{-}{n3}{v10}
\ncline{-}{v7}{n4}
\ncline{-}{n4}{v11}

\pscurve(0.5,-0.5)(-0.5,1.1)(0.5,1.3)(3.25,0.75)\rput{0}(1,1.25){\colorbox{almostwhite}{$\beta_1$}}
\Bogen{v4}{n4}{q_3}{-35}
\end{pspicture}
}
\hspace{0.2cm}
\subfloat[BC2b ($q_1$ small, $\vec{\alpha_2}\neq\vec{\beta_2}$)]{
\begin{pspicture}(-0.5,-1)(7,1.3) 
\Knoten{0}{1}{v1}\nput{180}{v1}{$u$}
\Knoten{0}{-1}{v2}\nput{180}{v2}{$w$}
\Knoten{0.5}{-0.5}{n2}
\Knoten{1}{0}{v3}
\Knoten{2.125}{0}{n5}
\Knoten{3.25}{0}{v4}
\Knoten{4.375}{0}{v6}
\Knoten{5.5}{0}{v7}
\Knoten{6}{0.5}{n3}
\Knoten{6}{-0.5}{n4}
\Knoten{6.5}{1}{v10}\nput{0}{v10}{$v$}
\Knoten{6.5}{-1}{v11}\nput{0}{v11}{$x$}

\Knoten{3.25}{0.85}{n6}
\Knoten{4.375}{0.85}{n7}

\ncline{n5}{n6}\ncput{\colorbox{almostwhite}{$\alpha_1$}}
\Kante{n6}{n7}{\alpha_2}
\ncline{n7}{n3}\ncput{\colorbox{almostwhite}{$\alpha_3$}}
\ncline{n6}{v6}\ncput{\colorbox{almostwhite}{$\beta_3$}}

\ncline{-}{v1}{v3}
\ncline{-}{n2}{v3}
\ncline{-}{v2}{n2}
\ncline{-}{v3}{v4}
\ncline{-}{v4}{n5}
\ncline{-}{v4}{v6}
\ncline{-}{v6}{v7}
\ncline{-}{v7}{n3}
\ncline{-}{n3}{v10}
\ncline{-}{v7}{n4}
\ncline{-}{n4}{v11}

\pscurve(0.5,-0.5)(-0.5,1.1)(0.5,1.5)(4.375,0.85)\rput{0}(1,1.4){\colorbox{almostwhite}{$\beta_1$}}
\Bogen{v4}{n4}{q_3}{-35}
\end{pspicture} 
}
\end{figure}

\begin{figure}[H] \centering \psset{unit=1cm}
\captionsetup[subfigure]{labelformat=empty}
\subfloat[BC2c ($q_1$ big, $\vec{\alpha_2}=\vec{\beta_2}$)]{
\begin{pspicture}(0.5,-1)(7,1) 
\Knoten{1}{1}{v1}\nput{180}{v1}{$u$}
\Knoten{1.5}{0.5}{n1}
\Knoten{1}{-1}{v2}\nput{180}{v2}{$w$}
\Knoten{1.5}{-0.5}{n2}
\Knoten{2}{0}{v3}
\Knoten{3.25}{0}{v4}
\Knoten{4.375}{0}{v6}
\Knoten{5.5}{0}{v7}
\Knoten{6}{0.5}{n3}
\Knoten{6}{-0.5}{n4}
\Knoten{6.5}{1}{v10}\nput{0}{v10}{$v$}
\Knoten{6.5}{-1}{v11}\nput{0}{v11}{$x$}

\Knoten{3.25}{0.75}{n6}
\Knoten{4.375}{0.75}{n7}

\Kante{n6}{n7}{\alpha_2}
\ncline{n7}{n3}\ncput{\colorbox{almostwhite}{$\alpha_3$}}

\ncline{n1}{n6}\ncput{\colorbox{almostwhite}{$\alpha_1$}}
\ncline{n7}{v6}\ncput{\colorbox{almostwhite}{$\beta_3$}}

\ncline{-}{v1}{v3}
\ncline{-}{n2}{v3}
\ncline{-}{v2}{n2}
\ncline{-}{v3}{v4}
\ncline{-}{v4}{v6}
\ncline{-}{v6}{v7}
\ncline{-}{v7}{n3}
\ncline{-}{n3}{v10}
\ncline{-}{v7}{n4}
\ncline{-}{n4}{v11}

\pscurve(1.5,-0.5)(0.5,1.1)(3.25,0.75)\rput{0}(1.5,1.3){\colorbox{almostwhite}{$\beta_1$}}
\Bogen{v4}{n4}{q_3}{-35}
\end{pspicture}
}
\hspace{0.5cm}
\subfloat[BC2d ($q_1$ big, $\vec{\alpha_2}\neq\vec{\beta_2}$)]{
\begin{pspicture}(0.5,-1)(7,1) 
\Knoten{1}{1}{v1}\nput{180}{v1}{$u$}
\Knoten{1.5}{0.5}{n1}
\Knoten{1}{-1}{v2}\nput{180}{v2}{$w$}
\Knoten{1.5}{-0.5}{n2}
\Knoten{2}{0}{v3}
\Knoten{3.25}{0}{v4}
\Knoten{4.375}{0}{v6}
\Knoten{5.5}{0}{v7}
\Knoten{6}{0.5}{n3}
\Knoten{6}{-0.5}{n4}
\Knoten{6.5}{1}{v10}\nput{0}{v10}{$v$}
\Knoten{6.5}{-1}{v11}\nput{0}{v11}{$x$}

\Knoten{3.25}{0.85}{n6}
\Knoten{4.375}{0.85}{n7}

\ncline{n1}{n6}\ncput{\colorbox{almostwhite}{$\alpha_1$}}
\Kante{n6}{n7}{\alpha_2}
\ncline{n7}{n3}\ncput{\colorbox{almostwhite}{$\alpha_3$}}
\ncline{n6}{v6}\ncput{\colorbox{almostwhite}{$\beta_3$}}

\ncline{-}{v1}{v3}
\ncline{-}{n2}{v3}
\ncline{-}{v2}{n2}
\ncline{-}{v3}{v4}
\ncline{-}{v4}{v6}
\ncline{-}{v6}{v7}
\ncline{-}{v7}{n3}
\ncline{-}{n3}{v10}
\ncline{-}{v7}{n4}
\ncline{-}{n4}{v11}

\pscurve(1.5,-0.5)(0.5,1.1)(4.375,0.9)\rput{0}(1.5,1.5){\colorbox{almostwhite}{$\beta_1$}}
\Bogen{v4}{n4}{q_3}{-35}
\end{pspicture} 
}
\caption{BC2a, b, c and d}
\end{figure}
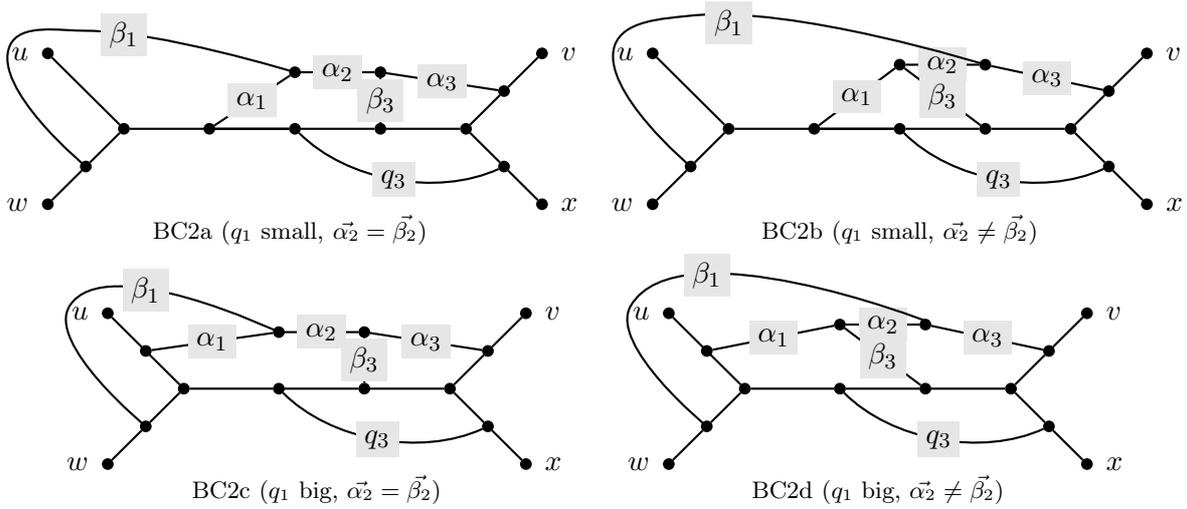

\textbf{BC3:}
\begin{itemize}[itemsep=-0em,leftmargin=*]
	\item $q_1$ small and substitutes all commonly used edges;
	\item $q_1,q_2$ and $q_3$ are pairwise node-disjoint;
	\item $q_2$ and $q_3$ do not contain nodes of $R_u$ or $L_u$;
	\item $q_1$ does not contain nodes of $R_{\ell}$;
	\item $q_1$ contains nodes of $L_{\ell}$, but not the start node of $\vec{q_2}$, and $\alpha_1 \subseteq C_2$, where $\alpha_1$ is the subpath of $\vec{q_1}$ beginning with the start node of $\vec{q_1}$ and ending with the last node which is contained in $L_{\ell}$. 
	\end{itemize}
	
	\begin{figure}[H] \centering \psset{unit=1.2cm}
	\begin{pspicture}(0.5,-1)(6.5,1) 
\Knoten{1.2}{0.8}{v1}\nput{180}{v1}{$u$}
\Knoten{1}{-1}{v2}\nput{180}{v2}{$w$}
\Knoten{1.3}{-0.7}{n2}
\Knoten{1.6}{-0.4}{n7}
\Knoten{2}{0}{v3}
\Knoten{3.25}{0}{v4}
\Knoten{4.375}{0}{v6}
\Knoten{5.5}{0}{v7}
\Knoten{6}{0.5}{n3}
\Knoten{6}{-0.5}{n4}
\Knoten{6.5}{1}{v10}\nput{0}{v10}{$v$}
\Knoten{6.5}{-1}{v11}\nput{0}{v11}{$x$}

\ncline{-}{v1}{v3}
\ncline{-}{n2}{v3}
\ncline{-}{v2}{n2}
\ncline{-}{v3}{v4}
\ncline{-}{v4}{v6}
\ncline{-}{v6}{v7}
\ncline{-}{v7}{n3}
\ncline{-}{n3}{v10}
\ncline{-}{v7}{n4}
\ncline{-}{n4}{v11}
\Kante{v3}{n7}{\alpha_1}
\pscurve(1.6,-0.4)(0.4,0.7)(6,0.5)\rput(2.5,1.3){\colorbox{almostwhite}{$\alpha_2$}}
\Bogen{n2}{v6}{q_2}{-35}
\Bogen{v4}{n4}{q_3}{-35}
\end{pspicture}
\caption{BC3} 
\end{figure}

%
%

\textbf{BC4:}
\begin{itemize}[itemsep=-0em,leftmargin=*]
  \item $q_1$ small and substitutes all commonly used edges;
	\item $q_3$ is node-disjoint with $q_1$ and with $q_2$;
	\item $q_1$ and $q_2$ are not node-disjoint and $\alpha_3=\beta_2$, where $\alpha_3$ $(\beta_2)$ is the subpath of $\vec{q_1}$ $(\vec{q_2})$ from the first until the last node which is contained in $q_2$ $(q_1)$;
	\item $q_2$ and $q_3$ do not contain nodes of $R_u$ or $L_u$;
	\item $q_1$ does not contain nodes of $R_{\ell}$;
	\item $q_1$ contains nodes of $L_{\ell}$, but not the start node of $\vec{q_2}$, and $\alpha_1 \subseteq C_2$, where $\alpha_1$ is the subpath of $\vec{q_1}$ beginning with the start node of $\vec{q_1}$ and ending with the last node which is contained in $L_{\ell}$.
\end{itemize}

\begin{figure}[H] \centering \psset{unit=1.1cm}\captionsetup[subfigure]{labelformat=empty}
\subfloat[BC4a ($\vec{\alpha_3}=\vec{\beta_2}$)]{
\begin{pspicture}(0.5,-1)(6.75,1.3) 
\Knoten{1.2}{0.8}{v1}\nput{180}{v1}{$u$}
\Knoten{1}{-1}{v2}\nput{180}{v2}{$w$}
\Knoten{1.25}{-0.75}{n9}
\Knoten{1.5}{-0.5}{n2}
\Knoten{2}{0}{v3}
\Knoten{3.25}{0}{v4}
\Knoten{4.375}{0}{v6}
\Knoten{5.5}{0}{v7}
\Knoten{6}{0.5}{n3}
\Knoten{6}{-0.5}{n4}
\Knoten{6.5}{1}{v10}\nput{0}{v10}{$v$}
\Knoten{6.5}{-1}{v11}\nput{0}{v11}{$x$}

\Knoten{3.25}{0.65}{n6}
\Knoten{4.375}{0.65}{n7}

\Kante{n6}{n7}{\alpha_3}
\ncline{n7}{n3}\ncput{\colorbox{almostwhite}{$\alpha_4$}}

\ncline{n7}{v6}\ncput{\colorbox{almostwhite}{$\beta_3$}}

\Kante{n2}{v3}{\alpha_1}
\ncline{-}{v1}{v3}
\ncline{-}{v2}{n9}
\ncline{-}{n9}{n2}
\ncline{-}{v3}{v4}
\ncline{-}{v4}{v6}
\ncline{-}{v6}{v7}
\ncline{-}{v7}{n3}
\ncline{-}{n3}{v10}
\ncline{-}{v7}{n4}
\ncline{-}{n4}{v11}

\pscurve(1.5,-0.5)(0.5,0.9)(3.25,0.65)\rput{0}(1,0.1){\colorbox{almostwhite}{$\alpha_2$}}
\pscurve(1.25,-0.75)(0,1.3)(3.25,0.65)\rput{0}(1,1.5){\colorbox{almostwhite}{$\beta_1$}}

\Bogen{v4}{n4}{q_3}{-35}
\end{pspicture}
}
\hspace{1cm}
\subfloat[BC4b ($\vec{\alpha_3}\neq\vec{\beta_2})$]{
\begin{pspicture}(0.25,-1)(6.5,1.3) 
\Knoten{1.2}{0.8}{v1}\nput{180}{v1}{$u$}
\Knoten{1}{-1}{v2}\nput{180}{v2}{$w$}
\Knoten{1.25}{-0.75}{n9}
\Knoten{1.5}{-0.5}{n2}
\Knoten{2}{0}{v3}
\Knoten{3.25}{0}{v4}
\Knoten{4.375}{0}{v6}
\Knoten{5.5}{0}{v7}
\Knoten{6}{0.5}{n3}
\Knoten{6}{-0.5}{n4}
\Knoten{6.5}{1}{v10}\nput{0}{v10}{$v$}
\Knoten{6.5}{-1}{v11}\nput{0}{v11}{$x$}

\Knoten{3.25}{0.65}{n6}
\Knoten{4.375}{0.65}{n7}

\Kante{n6}{n7}{\alpha_3}
\ncline{n7}{n3}\ncput{\colorbox{almostwhite}{$\alpha_4$}}

\ncline{n6}{v6}\ncput{\colorbox{almostwhite}{$\beta_3$}}

\Kante{n2}{v3}{\alpha_1}
\ncline{-}{v1}{v3}
\ncline{-}{v2}{n9}
\ncline{-}{n9}{n2}
\ncline{-}{v3}{v4}
\ncline{-}{v4}{v6}
\ncline{-}{v6}{v7}
\ncline{-}{v7}{n3}
\ncline{-}{n3}{v10}
\ncline{-}{v7}{n4}
\ncline{-}{n4}{v11}

\pscurve(1.5,-0.5)(0.5,0.9)(3.25,0.65)\rput{0}(1,0.1){\colorbox{almostwhite}{$\alpha_2$}}
\pscurve(1.25,-0.75)(0,1.3)(4.375,0.65)\rput{0}(1,1.6){\colorbox{almostwhite}{$\beta_1$}}

\Bogen{v4}{n4}{q_3}{-35}
\end{pspicture}
}
\caption{BC4a and b}
\end{figure}
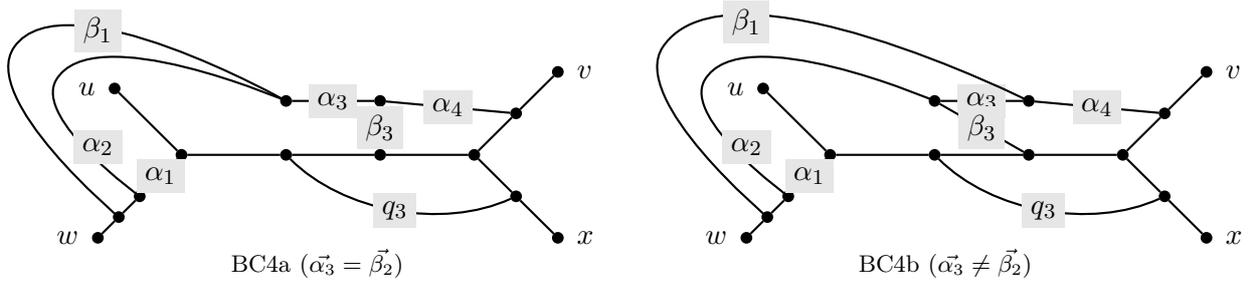
\end{definition}

In the following definition, we group PBCs which are no BCs into twelve ``types''.  

\begin{definition}\label{def:NBC}
We call a PBC $(P_u, P_{\ell}, q_1,q_2,q_3)$ a \emph{NBC$i$} for an $i \in \{1,2,\ldots,12\}$, if the properties of NBC$i$ (described below) are fulfilled. 
Note that some of the NBC$i$s have different subtypes, which are described and illustrated in the corresponding figures.
	
	\vspace{1em}
	\textbf{NBC1:} $q_1$ is small and $q_1$ and $q_3$ are not node-disjoint 

\vspace{0.5em}	
	\textbf{NBC2:} $q_2$ and $q_3$ are not node-disjoint
	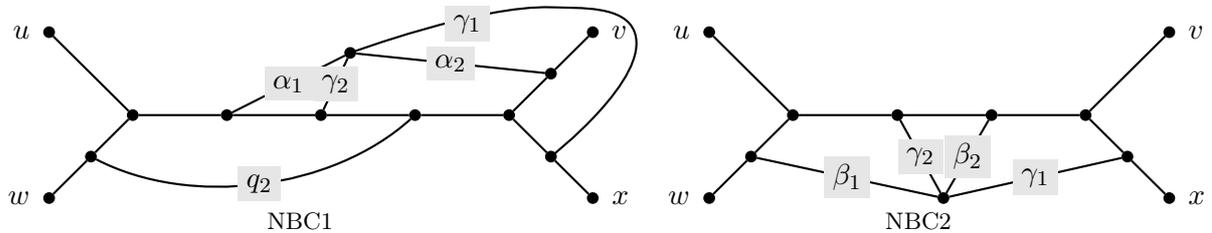
\begin{figure}[H] \centering \psset{unit=1.1cm}\captionsetup[subfigure]{labelformat=empty}
	\subfloat[NBC1]{
 \begin{pspicture}(-0.5,-1)(6.5,1.2) 
\Knoten{0}{1}{s1}\nput{180}{s1}{$u$}
\Knoten{0}{-1}{s2}\nput{180}{s2}{$w$}
\Knoten{0.5}{-0.5}{v1}
\Knoten{1}{0}{v2}
\Knoten{2.125}{0}{v3}
\Knoten{3.25}{0}{v4}
\Knoten{4.375}{0}{v5}
\Knoten{5.5}{0}{v6}
\Knoten{3.6}{0.75}{v8}
\Knoten{6}{0.5}{v9}
\Knoten{6.5}{1}{t1}\nput{0}{t1}{$v$}
\Knoten{6}{-0.5}{v10}
\Knoten{6.5}{-1}{t2}\nput{0}{t2}{$x$}
\ncline{s1}{v2}
\ncline{s2}{v1}
\ncline{v1}{v2}
\ncline{v2}{v3}
\ncline{v3}{v4}
\ncline{v4}{v5}
\ncline{v5}{v6}
\ncline{v6}{v9}
\ncline{v6}{v10}
\ncline{v9}{t1}
\ncline{v10}{t2}
\Bogen{v1}{v5}{q_2}{-35}
\pscurve(3.6,0.75)(6,1.3)(7,1)(6,-0.5)\rput(5,1.1){\colorbox{almostwhite}{$\gamma_1$}}
\Kante{v8}{v9}{\alpha_2}
\Kante{v3}{v8}{\alpha_1}
\Kante{v4}{v8}{\gamma_2}
\end{pspicture}
}
\hspace{0.6cm}
\subfloat[NBC2]{
\begin{pspicture}(0.5,-1)(6.5,1.2) 
\Knoten{1}{1}{s1}\nput{180}{s1}{$u$}
\Knoten{1}{-1}{s2}\nput{180}{s2}{$w$}
\Knoten{1.5}{-0.5}{v1}
\Knoten{2}{0}{v2}
\Knoten{3.25}{0}{v4}
\Knoten{4.375}{0}{v5}
\Knoten{5.5}{0}{v6}
\Knoten{3.8}{-1}{v7}
\Knoten{6.5}{1}{t1}\nput{0}{t1}{$v$}
\Knoten{6}{-0.5}{v10}
\Knoten{6.5}{-1}{t2}\nput{0}{t2}{$x$}
\ncline{s1}{v2}
\ncline{s2}{v1}
\ncline{v1}{v2}
\ncline{v2}{v4}
\ncline{v4}{v5}
\ncline{v5}{v6}
\ncline{v6}{v10}
\ncline{v6}{t1}

\ncline{v10}{t2}
\Kante{v1}{v7}{\beta_1}
\Kante{v10}{v7}{\gamma_1}
\Kante{v7}{v4}{\gamma_2}
\Kante{v7}{v5}{\beta_2}
%
\end{pspicture}
}
\caption{NBC1 and 2} \label{NBC1_2}
\end{figure}

\newpage
\textbf{NBC3:} $q_2$ contains a node of $L_u$ or $R_u$
\begin{figure}[H] \centering \psset{unit=1cm}\captionsetup[subfigure]{labelformat=empty}
\subfloat[NBC3a ($q_1$ small, $q_2$ contains a node of $L_u$)]{
 \begin{pspicture}(-0.5,-1.3)(7.5,1) 
\Knoten{0}{1}{v1}\nput{180}{v1}{$u$}
\Knoten{0.5}{0.5}{n1}
\Knoten{0}{-1}{v2}\nput{180}{v2}{$w$}
\Knoten{0.5}{-0.5}{n2}
\Knoten{1}{0}{v3}
\Knoten{2.125}{0}{n5}
\Knoten{3.25}{0}{v4}
\Knoten{4.375}{0}{v6}
\Knoten{5.5}{0}{v7}
\Knoten{6}{0.5}{n3}
\Knoten{6}{-0.5}{n4}
\Knoten{6.5}{1}{v10}\nput{0}{v10}{$v$}
\Knoten{6.5}{-1}{v11}\nput{0}{v11}{$x$}

\ncline{-}{v1}{v3}
\ncline{-}{n2}{v3}
\ncline{-}{v2}{n2}
\ncline{-}{v3}{v4}
\ncline{-}{v4}{n5}
\ncline{-}{v4}{v6}
\ncline{-}{v6}{v7}
\ncline{-}{v7}{n3}
\ncline{-}{n3}{v10}
\ncline{-}{v7}{n4}
\ncline{-}{n4}{v11}
\Bogen{n5}{n3}{q_1}{35}
\Bogen{n1}{v6}{\beta_2}{35}
\Kante{n2}{n1}{\beta_1}
\Bogen{v4}{n4}{q_3}{-35}
\end{pspicture}\label{NBC3a}
}
\subfloat[NBC3b ($q_1$ small, $q_2$ contains a node of $R_u$ which is not substituted by $q_1$)]{
 \begin{pspicture}(-0.5,-1.3)(7,1) 
\Knoten{0}{1}{v1}\nput{180}{v1}{$u$}
\Knoten{0}{-1}{v2}\nput{180}{v2}{$w$}
\Knoten{0.5}{-0.5}{n2}
\Knoten{1}{0}{v3}
\Knoten{2.125}{0}{n5}
\Knoten{3.25}{0}{v4}
\Knoten{4.375}{0}{v6}
\Knoten{5.5}{0}{v7}
\Knoten{6}{0.5}{n3}
\Knoten{6.25}{0.75}{n6}
\Knoten{6}{-0.5}{n4}
\Knoten{6.5}{1}{v10}\nput{0}{v10}{$v$}
\Knoten{6.3}{-0.8}{v11}\nput{0}{v11}{$x$}

\ncline{-}{v1}{v3}
\ncline{-}{n2}{v3}
\ncline{-}{v2}{n2}
\ncline{-}{v3}{v4}
\ncline{-}{v4}{n5}
\ncline{-}{v4}{v6}
\ncline{-}{v6}{v7}
\ncline{-}{v7}{n3}
\ncline{-}{n3}{v10}
\ncline{-}{v7}{n4}
\ncline{-}{n4}{v11}
\Bogen{n5}{n3}{q_1}{35}
\Bogen{v6}{n6}{\beta_2}{25}
\Bogen{v4}{n4}{q_3}{-35}
\pscurve(0.5,-0.5)(7.2,-0.5)(6.25,0.75)\rput(1.4,-0.9){\colorbox{almostwhite}{$\beta_1$}}
\end{pspicture}\label{NBC3b}
}

\subfloat[NBC3c ($q_1$ small, $q_2$ contains a node of $R_u$ which is substituted by $q_1$)]{
 \begin{pspicture}(-1.5,-1.5)(8.5,1) 
\Knoten{0}{1}{v1}\nput{180}{v1}{$u$}
\Knoten{0}{-1}{v2}\nput{180}{v2}{$w$}
\Knoten{0.5}{-0.5}{n2}
\Knoten{1}{0}{v3}
\Knoten{2.125}{0}{n5}
\Knoten{3.25}{0}{v4}
\Knoten{4.375}{0}{v6}
\Knoten{5.5}{0}{v7}
\Knoten{6}{0.5}{n3}
\Knoten{6.25}{0.75}{n6}
\Knoten{6}{-0.5}{n4}
\Knoten{6.5}{1}{v10}\nput{0}{v10}{$v$}
\Knoten{6.3}{-0.8}{v11}\nput{0}{v11}{$x$}

\ncline{-}{v1}{v3}
\ncline{-}{n2}{v3}
\ncline{-}{v2}{n2}
\ncline{-}{v3}{v4}
\ncline{-}{v4}{n5}
\ncline{-}{v4}{v6}
\ncline{-}{v6}{v7}
\ncline{-}{v7}{n3}
\ncline{-}{n3}{v10}
\ncline{-}{v7}{n4}
\ncline{-}{n4}{v11}
\Bogen{n5}{n6}{q_1}{35}
\Bogen{v6}{n3}{\beta_2}{25}
\Bogen{v4}{n4}{q_3}{-35}
\pscurve(0.5,-0.5)(7.2,-0.5)(6,0.5)\rput(1.4,-0.9){\colorbox{almostwhite}{$\beta_1$}}
\end{pspicture}\label{NBC3c}
}
\end{figure}

\begin{figure}[H] \centering \psset{unit=1cm}\captionsetup[subfigure]{labelformat=empty}
\subfloat[NBC3d ($q_1$ big, $q_2$ contains a node of $L_u$ which is not substituted by $q_1$)]{
 \begin{pspicture}(0.5,-1.2)(7,1.5) 
\Knoten{1}{1}{v1}\nput{180}{v1}{$u$}
\Knoten{1.5}{0.5}{n1}
\Knoten{1.25}{0.75}{n6}
\Knoten{1}{-1}{v2}\nput{180}{v2}{$w$}
\Knoten{1.5}{-0.5}{n2}
\Knoten{2}{0}{v3}
\Knoten{3.25}{0}{v4}
\Knoten{4.375}{0}{v6}
\Knoten{5.5}{0}{v7}
\Knoten{6}{0.5}{n3}
\Knoten{6}{-0.5}{n4}
\Knoten{6.5}{1}{v10}\nput{0}{v10}{$v$}
\Knoten{6.5}{-1}{v11}\nput{0}{v11}{$x$}

\ncline{-}{v1}{v3}
\ncline{-}{n2}{v3}
\ncline{-}{v2}{n2}
\ncline{-}{v3}{v4}
\ncline{-}{v4}{v3}
\ncline{-}{v4}{v6}
\ncline{-}{v6}{v7}
\ncline{-}{v7}{n3}
\ncline{-}{n3}{v10}
\ncline{-}{v7}{n4}
\ncline{-}{n4}{v11}
\Bogen{n1}{n3}{q_1}{25}
\Bogen{n6}{v6}{\beta_2}{25}
\Bogen{v4}{n4}{q_3}{-35}
\Kante{n2}{n6}{\beta_1}
\end{pspicture}\label{NBC3d}
}
\hspace{1cm}
\subfloat[NBC3e ($q_1$ big, $q_2$ contains a node of $L_u$ which is substituted by $q_1$)]{
 \begin{pspicture}(0.5,-1.2)(7,1.5) 
\Knoten{1}{1}{v1}\nput{180}{v1}{$u$}
\Knoten{1.5}{0.5}{n1}
\Knoten{1.25}{0.75}{n6}
\Knoten{1}{-1}{v2}\nput{180}{v2}{$w$}
\Knoten{1.5}{-0.5}{n2}
\Knoten{2}{0}{v3}
\Knoten{3.25}{0}{v4}
\Knoten{4.375}{0}{v6}
\Knoten{5.5}{0}{v7}
\Knoten{6}{0.5}{n3}
\Knoten{6}{-0.5}{n4}
\Knoten{6.5}{1}{v10}\nput{0}{v10}{$v$}
\Knoten{6.5}{-1}{v11}\nput{0}{v11}{$x$}

\ncline{-}{v1}{v3}
\ncline{-}{n2}{v3}
\ncline{-}{v2}{n2}
\ncline{-}{v3}{v4}
\ncline{-}{v4}{v3}
\ncline{-}{v4}{v6}
\ncline{-}{v6}{v7}
\ncline{-}{v7}{n3}
\ncline{-}{n3}{v10}
\ncline{-}{v7}{n4}
\ncline{-}{n4}{v11}
\Bogen{n6}{n3}{q_1}{25}
\Bogen{n1}{v6}{\beta_2}{25}
\Bogen{v4}{n4}{q_3}{-35}
\Kante{n2}{n1}{\beta_1}
\end{pspicture}\label{NBC3e}
}
\end{figure}

\begin{figure}[H] \centering \psset{unit=1cm}
\captionsetup[subfigure]{labelformat=empty}
\subfloat[NBC3f ($q_1$ big, $q_2$ contains a node of $R_u$ which is not substituted by $q_1$)]{
 \begin{pspicture}(0.5,-1.7)(7,1.7) 
\Knoten{1}{1}{v1}\nput{180}{v1}{$u$}
\Knoten{1.5}{0.5}{n1}
\Knoten{1}{-1}{v2}\nput{180}{v2}{$w$}
\Knoten{1.5}{-0.5}{n2}
\Knoten{2}{0}{v3}
\Knoten{3.25}{0}{v4}
\Knoten{4.375}{0}{v6}
\Knoten{5.5}{0}{v7}
\Knoten{6}{0.5}{n3}
\Knoten{6}{-0.5}{n4}
\Knoten{6.25}{0.75}{n6}
\Knoten{6.5}{1}{v10}\nput{0}{v10}{$v$}
\Knoten{6.3}{-0.8}{v11}\nput{0}{v11}{$x$}

\ncline{-}{v1}{v3}
\ncline{-}{n2}{v3}
\ncline{-}{v2}{n2}
\ncline{-}{v3}{v4}
\ncline{-}{v4}{v3}
\ncline{-}{v4}{v6}
\ncline{-}{v6}{v7}
\ncline{-}{v7}{n3}
\ncline{-}{n3}{v10}
\ncline{-}{v7}{n4}
\ncline{-}{n4}{v11}
\Bogen{n1}{n3}{q_1}{25}
\Bogen{v4}{n4}{q_3}{-35}
\pscurve(1.5,-0.5)(7.2,-0.5)(6.25,0.75)\rput(2.4,-0.8){\colorbox{almostwhite}{$\beta_1$}}
\Bogen{v6}{n6}{\beta_2}{25}
\end{pspicture}\label{NBC3f}
}
\hspace{1cm}
\subfloat[NBC3g ($q_1$ big, $q_2$ contains a node of $R_u$ which is substituted by $q_1$)]{
 \begin{pspicture}(0.5,-1.7)(7,1.7) 
\Knoten{1}{1}{v1}\nput{180}{v1}{$u$}
\Knoten{1.5}{0.5}{n1}
\Knoten{1}{-1}{v2}\nput{180}{v2}{$w$}
\Knoten{1.5}{-0.5}{n2}
\Knoten{2}{0}{v3}
\Knoten{3.25}{0}{v4}
\Knoten{4.375}{0}{v6}
\Knoten{5.5}{0}{v7}
\Knoten{6}{0.5}{n3}
\Knoten{6}{-0.5}{n4}
\Knoten{6.25}{0.75}{n6}
\Knoten{6.5}{1}{v10}\nput{0}{v10}{$v$}
\Knoten{6.3}{-0.8}{v11}\nput{0}{v11}{$x$}

\ncline{-}{v1}{v3}
\ncline{-}{n2}{v3}
\ncline{-}{v2}{n2}
\ncline{-}{v3}{v4}
\ncline{-}{v4}{v3}
\ncline{-}{v4}{v6}
\ncline{-}{v6}{v7}
\ncline{-}{v7}{n3}
\ncline{-}{n3}{v10}
\ncline{-}{v7}{n4}
\ncline{-}{n4}{v11}
\Bogen{n1}{n6}{q_1}{25}
\Bogen{v4}{n4}{q_3}{-35}
\pscurve(1.5,-0.5)(7.2,-0.5)(6,0.5)\rput(2.4,-0.8){\colorbox{almostwhite}{$\beta_1$}}
\Bogen{v6}{n3}{\beta_2}{25}
\end{pspicture}\label{NBC3g}
}
\caption{NBC3a, b, c, d, e, f and g}
\label{NBC3}
\end{figure}

\textbf{NBC4:} $q_3$ contains a node of $L_u$ or $R_u$ 
\begin{figure}[H] \centering \psset{unit=1cm}\captionsetup[subfigure]{labelformat=empty}
\subfloat[NBC4a ($q_1$ small, $q_3$ contains a node of $L_u$)]{
 \begin{pspicture}(-0.5,-1.4)(7,1) 
\Knoten{0}{1}{v1}\nput{180}{v1}{$u$}
\Knoten{0.5}{0.5}{n1}
\Knoten{0.2}{-0.8}{v2}\nput{180}{v2}{$w$}
\Knoten{0.5}{-0.5}{n2}
\Knoten{1}{0}{v3}
\Knoten{2.125}{0}{n5}
\Knoten{3.25}{0}{v4}
\Knoten{4.375}{0}{v6}
\Knoten{5.5}{0}{v7}
\Knoten{6}{0.5}{n3}
\Knoten{6}{-0.5}{n4}
\Knoten{6.5}{1}{v10}\nput{0}{v10}{$v$}
\Knoten{6.5}{-1}{v11}\nput{0}{v11}{$x$}

\ncline{-}{v1}{v3}
\ncline{-}{n2}{v3}
\ncline{-}{v2}{n2}
\ncline{-}{v3}{v4}
\ncline{-}{v4}{n5}
\ncline{-}{v4}{v6}
\ncline{-}{v6}{v7}
\ncline{-}{v7}{n3}
\ncline{-}{n3}{v10}
\ncline{-}{v7}{n4}
\ncline{-}{n4}{v11}
\Bogen{n5}{n3}{q_1}{35}
\Bogen{n2}{v6}{q_2}{-35}
\Bogen{n1}{v4}{\gamma_2}{25}
\pscurve(6,-0.5)(-0.8,-0.6)(0.5,0.5)\rput(4.375,-1){\colorbox{almostwhite}{$\gamma_1$}}
\end{pspicture}\label{NBC4a}
}
\hspace{0.2cm}
\subfloat[NBC4b ($q_1$ small, $q_3$ contains a node of $R_u$ which is not substituted by $q_1$)]{
 \begin{pspicture}(-0.5,-1.4)(7,1) 
\Knoten{0}{1}{v1}\nput{180}{v1}{$u$}
\Knoten{0}{-1}{v2}\nput{180}{v2}{$w$}
\Knoten{0.5}{-0.5}{n2}
\Knoten{1}{0}{v3}
\Knoten{2.125}{0}{n5}
\Knoten{3.25}{0}{v4}
\Knoten{4.375}{0}{v6}
\Knoten{5.5}{0}{v7}
\Knoten{6}{0.5}{n3}
\Knoten{6.25}{0.75}{n6}
\Knoten{6}{-0.5}{n4}
\Knoten{6.5}{1}{v10}\nput{0}{v10}{$v$}
\Knoten{6.5}{-1}{v11}\nput{0}{v11}{$x$}

\ncline{-}{v1}{v3}
\ncline{-}{n2}{v3}
\ncline{-}{v2}{n2}
\ncline{-}{v3}{v4}
\ncline{-}{v4}{n5}
\ncline{-}{v4}{v6}
\ncline{-}{v6}{v7}
\ncline{-}{v7}{n3}
\ncline{-}{n3}{v10}
\ncline{-}{v7}{n4}
\ncline{-}{n4}{v11}
\Bogen{n2}{v6}{q_2}{-35}
\Bogen{n5}{n3}{q_1}{35}
\Bogen{v4}{n6}{\gamma_2}{20}
\Kante{n4}{n6}{\gamma_1}
\end{pspicture}\label{NBC4b}
}
\end{figure}

\begin{figure}[H] \centering \psset{unit=1cm}\captionsetup[subfigure]{labelformat=empty}
\subfloat[NBC4c ($q_1$ small, $q_3$ contains a node of $R_u$ which is substituted by $q_1$)]{
 \begin{pspicture}(-1.5,-1)(8.2,1) 
\Knoten{0}{1}{v1}\nput{180}{v1}{$u$}
\Knoten{0}{-1}{v2}\nput{180}{v2}{$w$}
\Knoten{0.5}{-0.5}{n2}
\Knoten{1}{0}{v3}
\Knoten{2.125}{0}{n5}
\Knoten{3.25}{0}{v4}
\Knoten{4.375}{0}{v6}
\Knoten{5.5}{0}{v7}
\Knoten{6}{0.5}{n3}
\Knoten{6.25}{0.75}{n6}
\Knoten{6}{-0.5}{n4}
\Knoten{6.5}{1}{v10}\nput{0}{v10}{$v$}
\Knoten{6.5}{-1}{v11}\nput{0}{v11}{$x$}

\ncline{-}{v1}{v3}
\ncline{-}{n2}{v3}
\ncline{-}{v2}{n2}
\ncline{-}{v3}{v4}
\ncline{-}{v4}{n5}
\ncline{-}{v4}{v6}
\ncline{-}{v6}{v7}
\ncline{-}{v7}{n3}
\ncline{-}{n3}{v10}
\ncline{-}{v7}{n4}
\ncline{-}{n4}{v11}
\Bogen{n2}{v6}{q_2}{-35}
\Bogen{n5}{n6}{q_1}{35}
\Bogen{v4}{n3}{\gamma_2}{20}
\Kante{n4}{n3}{\gamma_1}
\end{pspicture}\label{NBC4c}
}
\end{figure}

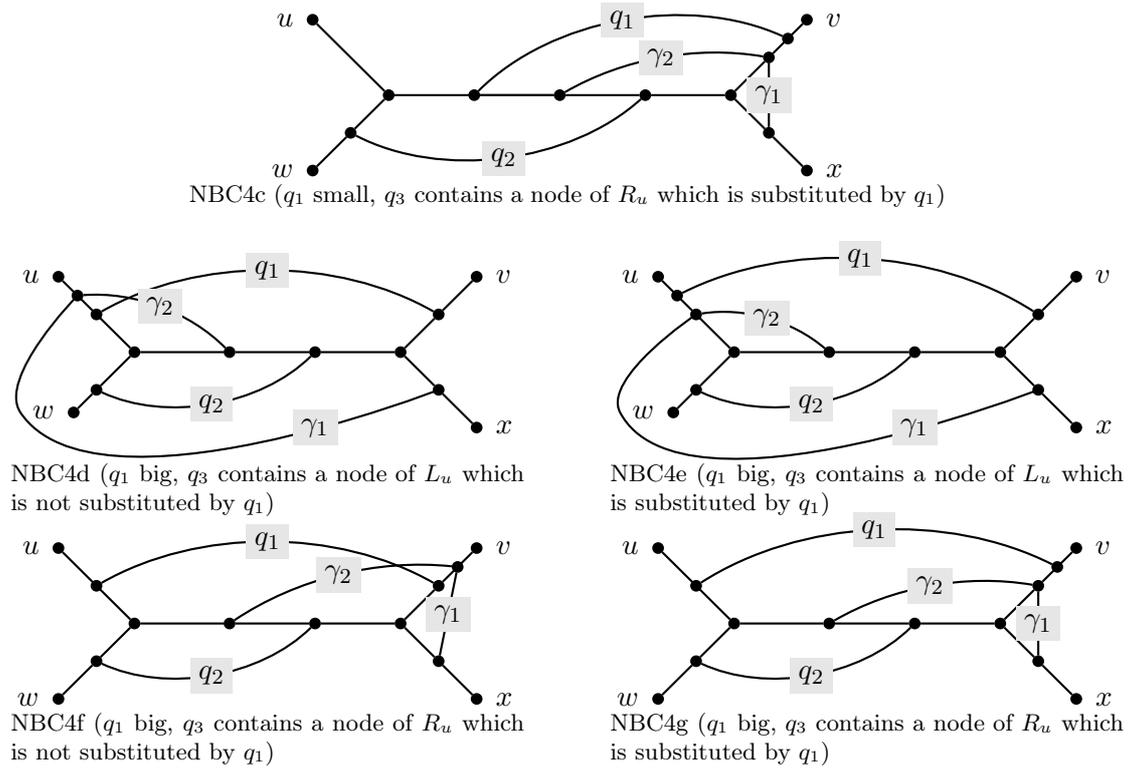
\begin{figure}[H] \centering \psset{unit=1cm}\captionsetup[subfigure]{labelformat=empty}
\subfloat[NBC4d ($q_1$ big, $q_3$ contains a node of $L_u$  which is not substituted by $q_1$)]{
 \begin{pspicture}(0.5,-1.3)(7,1) 
\Knoten{1}{1}{v1}\nput{180}{v1}{$u$}
\Knoten{1.5}{0.5}{n1}
\Knoten{1.25}{0.75}{n6}
\Knoten{1.2}{-0.8}{v2}\nput{180}{v2}{$w$}
\Knoten{1.5}{-0.5}{n2}
\Knoten{2}{0}{v3}
\Knoten{3.25}{0}{v4}
\Knoten{4.375}{0}{v6}
\Knoten{5.5}{0}{v7}
\Knoten{6}{0.5}{n3}
\Knoten{6}{-0.5}{n4}
\Knoten{6.5}{1}{v10}\nput{0}{v10}{$v$}
\Knoten{6.5}{-1}{v11}\nput{0}{v11}{$x$}

\ncline{-}{v1}{v3}
\ncline{-}{n2}{v3}
\ncline{-}{v2}{n2}
\ncline{-}{v3}{v4}
\ncline{-}{v4}{v6}
\ncline{-}{v6}{v7}
\ncline{-}{v7}{n3}
\ncline{-}{n3}{v10}
\ncline{-}{v7}{n4}
\ncline{-}{n4}{v11}
\Bogen{n1}{n3}{q_1}{30}
\Bogen{n2}{v6}{q_2}{-35}
\Bogen{n6}{v4}{\gamma_2}{25}
\pscurve(6,-0.5)(0.5,-0.8)(1.25,0.75)\rput(4.375,-1){\colorbox{almostwhite}{$\gamma_1$}}
\end{pspicture}\label{NBC4d}
}
\hspace{1cm}
\subfloat[NBC4e ($q_1$ big, $q_3$ contains a node of $L_u$ which is substituted by $q_1$)]{
 \begin{pspicture}(0.5,-1.3)(7,1) 
\Knoten{1}{1}{v1}\nput{180}{v1}{$u$}
\Knoten{1.5}{0.5}{n1}
\Knoten{1.25}{0.75}{n6}
\Knoten{1.2}{-0.8}{v2}\nput{180}{v2}{$w$}
\Knoten{1.5}{-0.5}{n2}
\Knoten{2}{0}{v3}
\Knoten{3.25}{0}{v4}
\Knoten{4.375}{0}{v6}
\Knoten{5.5}{0}{v7}
\Knoten{6}{0.5}{n3}
\Knoten{6}{-0.5}{n4}
\Knoten{6.5}{1}{v10}\nput{0}{v10}{$v$}
\Knoten{6.5}{-1}{v11}\nput{0}{v11}{$x$}

\ncline{-}{v1}{v3}
\ncline{-}{n2}{v3}
\ncline{-}{v2}{n2}
\ncline{-}{v3}{v4}
\ncline{-}{v4}{v6}
\ncline{-}{v6}{v7}
\ncline{-}{v7}{n3}
\ncline{-}{n3}{v10}
\ncline{-}{v7}{n4}
\ncline{-}{n4}{v11}
\Bogen{n6}{n3}{q_1}{30}
\Bogen{n2}{v6}{q_2}{-35}
\Bogen{n1}{v4}{\gamma_2}{25}
\pscurve(6,-0.5)(0.5,-0.8)(1.5,0.5)\rput(4.375,-1){\colorbox{almostwhite}{$\gamma_1$}}
\end{pspicture}\label{NBC4e}
}

\subfloat[NBC4f ($q_1$ big, $q_3$ contains a node of $R_u$ which is not substituted by $q_1$)]{
 \begin{pspicture}(0.5,-1)(7,1) 
\Knoten{1}{1}{v1}\nput{180}{v1}{$u$}
\Knoten{1.5}{0.5}{n1}
\Knoten{1}{-1}{v2}\nput{180}{v2}{$w$}
\Knoten{1.5}{-0.5}{n2}
\Knoten{2}{0}{v3}
\Knoten{3.25}{0}{v4}
\Knoten{4.375}{0}{v6}
\Knoten{5.5}{0}{v7}
\Knoten{6}{0.5}{n3}
\Knoten{6.25}{0.75}{n6}
\Knoten{6}{-0.5}{n4}
\Knoten{6.5}{1}{v10}\nput{0}{v10}{$v$}
\Knoten{6.5}{-1}{v11}\nput{0}{v11}{$x$}

\ncline{-}{v1}{v3}
\ncline{-}{n2}{v3}
\ncline{-}{v2}{n2}
\ncline{-}{v3}{v4}
\ncline{-}{v4}{v6}
\ncline{-}{v6}{v7}
\ncline{-}{v7}{n3}
\ncline{-}{n3}{v10}
\ncline{-}{v7}{n4}
\ncline{-}{n4}{v11}
\Bogen{n2}{v6}{q_2}{-35}
\Bogen{n1}{n3}{q_1}{30}
\Bogen{v4}{n6}{\gamma_2}{20}
\Kante{n4}{n6}{\gamma_1}
\end{pspicture}\label{NBC4f}
}
\hspace{1cm}
\subfloat[NBC4g ($q_1$ big, $q_3$ contains a node of $R_u$ which is substituted by $q_1$)]{
 \begin{pspicture}(0.5,-1)(7,1) 
\Knoten{1}{1}{v1}\nput{180}{v1}{$u$}
\Knoten{1.5}{0.5}{n1}
\Knoten{1}{-1}{v2}\nput{180}{v2}{$w$}
\Knoten{1.5}{-0.5}{n2}
\Knoten{2}{0}{v3}
\Knoten{3.25}{0}{v4}
\Knoten{4.375}{0}{v6}
\Knoten{5.5}{0}{v7}
\Knoten{6}{0.5}{n3}
\Knoten{6.25}{0.75}{n6}
\Knoten{6}{-0.5}{n4}
\Knoten{6.5}{1}{v10}\nput{0}{v10}{$v$}
\Knoten{6.5}{-1}{v11}\nput{0}{v11}{$x$}

\ncline{-}{v1}{v3}
\ncline{-}{n2}{v3}
\ncline{-}{v2}{n2}
\ncline{-}{v3}{v4}
\ncline{-}{v4}{v6}
\ncline{-}{v6}{v7}
\ncline{-}{v7}{n3}
\ncline{-}{n3}{v10}
\ncline{-}{v7}{n4}
\ncline{-}{n4}{v11}
\Bogen{n2}{v6}{q_2}{-35}
\Bogen{n1}{n6}{q_1}{30}
\Bogen{v4}{n3}{\gamma_2}{20}
\Kante{n4}{n3}{\gamma_1}
\end{pspicture}\label{NBC4g}
}
\caption{NBC4a, b, c, d, e, f and g}
\label{NBC4}
\end{figure}
\textbf{NBC5:} $q_1$ contains a node of $L_{\ell}$ and a node of $R_{\ell}$
\begin{figure}[H] \centering \psset{unit=1.1cm}\captionsetup[subfigure]{labelformat=empty}
\subfloat[NBC5a ($q_1$ small, $\vec{q_1}$ contains first a node of $L_{\ell}$)]{
\begin{pspicture}(1.5,-1)(7.5,1) 
\Knoten{2}{1}{s1}\nput{180}{s1}{$u$}
\Knoten{2}{-1}{s2}\nput{180}{s2}{$w$}
\Knoten{2.5}{-0.5}{v1}
\Knoten{3}{0}{v2}
\Knoten{4.25}{0}{v3}

\Knoten{5.5}{0}{v6}
\Knoten{6}{0.5}{v9}
\Knoten{6.5}{1}{t1}\nput{0}{t1}{$v$}
\Knoten{6}{-0.5}{v10}
\Knoten{6.5}{-1}{t2}\nput{0}{t2}{$x$}
\ncline{s1}{v2}
\ncline{s2}{v1}
\ncline{v1}{v2}
\ncline{v2}{v3}
\ncline{v3}{v6}
\ncline{v6}{v9}
\ncline{v6}{v10}
\ncline{v9}{t1}
\ncline{v10}{t2}
\Kante{v1}{v3}{\alpha_1}
\Kante{v10}{v9}{\alpha_3}
\Bogen{v1}{v10}{\alpha_2}{-15}
\end{pspicture}
	\label{NBC5a}
	}
	\hspace{1cm}
	\subfloat[NBC5b ($q_1$ small, $\vec{q_1}$ contains first a node of $R_{\ell}$)]{
	 \begin{pspicture}(-1,-1)(5,1.2) 
\Knoten{0.2}{0.8}{s1}\nput{180}{s1}{$u$}
\Knoten{0}{-1}{s2}\nput{180}{s2}{$w$}
\Knoten{0.5}{-0.5}{v1}
\Knoten{1}{0}{v2}
\Knoten{2.25}{0}{v3}
\Knoten{3.5}{0}{v6}

\Knoten{4}{0.5}{v9}
\Knoten{4.5}{1}{t1}\nput{0}{t1}{$v$}
\Knoten{4}{-0.5}{v10}
\Knoten{4.5}{-1}{t2}\nput{0}{t2}{$x$}
\ncline{s1}{v2}
\ncline{s2}{v1}
\ncline{v1}{v2}
\ncline{v2}{v3}
\ncline{v3}{v6}
\ncline{v6}{v9}
\ncline{v6}{v10}
\ncline{v9}{t1}
\ncline{v10}{t2}
\Kante{v3}{v10}{\alpha_1}
\pscurve(0.5,-0.5)(-0.8,0.7)(4,0.5)\rput(1.4,1.2){\colorbox{almostwhite}{$\alpha_3$}}
\Bogen{v1}{v10}{\alpha_2}{-15}
\end{pspicture}
	\label{NBC5b}
	}
	
	\subfloat[NBC5c ($q_1$ big, $\vec{q_1}$ contains first a node of $L_{\ell}$)]{
	\begin{pspicture}(1.5,-1)(7.5,1.2) 
	\Knoten{2}{1}{s1}\nput{180}{s1}{$u$}
\Knoten{2}{-1}{s2}\nput{180}{s2}{$w$}
\Knoten{2.5}{-0.5}{v1}
\Knoten{2.5}{0.5}{n1}
\Knoten{3}{0}{v2}
\Knoten{5.5}{0}{v6}
\Knoten{6}{0.5}{v9}
\Knoten{6.5}{1}{t1}\nput{0}{t1}{$v$}
\Knoten{6}{-0.5}{v10}
\Knoten{6.5}{-1}{t2}\nput{0}{t2}{$x$}
\ncline{s1}{v2}
\ncline{s2}{v1}
\ncline{v1}{v2}
\ncline{v2}{v6}
\ncline{v6}{v9}
\ncline{v6}{v10}
\ncline{v9}{t1}
\ncline{v10}{t2}
\Kante{v1}{n1}{\alpha_1}
\Kante{v10}{v9}{\alpha_3}
\Bogen{v1}{v10}{\alpha_2}{-15}
\end{pspicture}
	\label{NBC5c}
	}
	\hspace{1cm}
	\subfloat[NBC5d ($q_1$ big, $\vec{q_1}$ contains first a node of $R_{\ell}$)]{
	\begin{pspicture}(1,-1)(7.5,1.2) 
\Knoten{2.2}{0.8}{s1}\nput{180}{s1}{$u$}
\Knoten{2}{-1}{s2}\nput{180}{s2}{$w$}
\Knoten{2.5}{-0.5}{v1}
\Knoten{2.5}{0.5}{n1}
\Knoten{3}{0}{v2}
\Knoten{5.5}{0}{v6}
\Knoten{6}{0.5}{v9}
\Knoten{6.3}{0.8}{t1}\nput{0}{t1}{$v$}
\Knoten{6}{-0.5}{v10}
\Knoten{6.5}{-1}{t2}\nput{0}{t2}{$x$}
\ncline{s1}{v2}
\ncline{s2}{v1}
\ncline{v1}{v2}
\ncline{v2}{v6}
\ncline{v6}{v9}
\ncline{v6}{v10}
\ncline{v9}{t1}
\ncline{v10}{t2}
\pscurve(2.5,-0.5)(1.5,0.7)(6,0.5)\rput(1.7,0){\colorbox{almostwhite}{$\alpha_3$}}
\pscurve(6,-0.5)(7.1,0.7)(2.5,0.5)\rput(6,1.2){\colorbox{almostwhite}{$\alpha_1$}}
\Bogen{v1}{v10}{\alpha_2}{-15}
\end{pspicture}
	\label{NBC5d}
	}
	\caption{NBC5a, b, c and d}\label{NBC5}
	\end{figure}
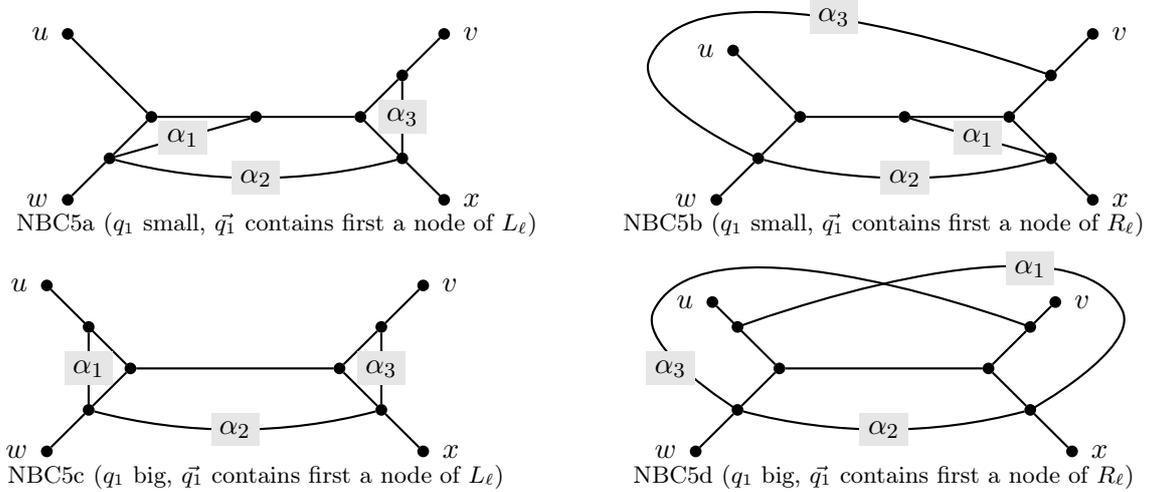
	
	\textbf{NBC6:} $q_1$ is small and contains a node of $R_{\ell}$
	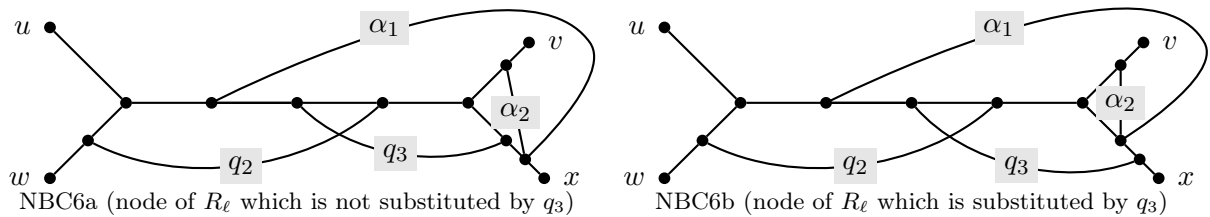
\begin{figure}[H] \centering \psset{unit=1cm}\captionsetup[subfigure]{labelformat=empty}
	\subfloat[NBC6a (node of $R_{\ell}$ which is not substituted by $q_3$)]{
	\begin{pspicture}(-0.5,-1)(7,1) 
\Knoten{0}{1}{v1}\nput{180}{v1}{$u$}
\Knoten{0}{-1}{v2}\nput{180}{v2}{$w$}
\Knoten{0.5}{-0.5}{n2}
\Knoten{1}{0}{v3}
\Knoten{2.125}{0}{n5}
\Knoten{3.25}{0}{v4}
\Knoten{4.375}{0}{v6}
\Knoten{5.5}{0}{v7}
\Knoten{6}{0.5}{n3}
\Knoten{6}{-0.5}{n4}
\Knoten{6.25}{-0.75}{n6}
\Knoten{6.3}{0.8}{v10}\nput{0}{v10}{$v$}
\Knoten{6.5}{-1}{v11}\nput{0}{v11}{$x$}

\ncline{-}{v1}{v3}
\ncline{-}{n2}{v3}
\ncline{-}{v2}{n2}
\ncline{-}{v3}{v4}
\ncline{-}{v4}{n5}
\ncline{-}{v4}{v6}
\ncline{-}{v6}{v7}
\ncline{-}{v7}{n3}
\ncline{-}{n3}{v10}
\ncline{-}{v7}{n4}
\ncline{-}{n4}{v11}
\Kante{n6}{n3}{\alpha_2}
\pscurve(6.25,-0.75)(7.1,0.8)(2.125,0)\rput(4.4,1){\colorbox{almostwhite}{$\alpha_1$}}
\Bogen{n2}{v6}{q_2}{-35}
\Bogen{v4}{n4}{q_3}{-35}
\end{pspicture}
	\label{NBC6a}
	}
	\hspace{0.2cm}
	\subfloat[NBC6b (node of $R_{\ell}$ which is substituted by $q_3$)]{
 	\begin{pspicture}(-0.5,-1)(7,1) 
\Knoten{0}{1}{v1}\nput{180}{v1}{$u$}
\Knoten{0}{-1}{v2}\nput{180}{v2}{$w$}
\Knoten{0.5}{-0.5}{n2}
\Knoten{1}{0}{v3}
\Knoten{2.125}{0}{n5}
\Knoten{3.25}{0}{v4}
\Knoten{4.375}{0}{v6}
\Knoten{5.5}{0}{v7}
\Knoten{6}{0.5}{n3}
\Knoten{6}{-0.5}{n4}
\Knoten{6.25}{-0.75}{n6}
\Knoten{6.3}{0.8}{v10}\nput{0}{v10}{$v$}
\Knoten{6.5}{-1}{v11}\nput{0}{v11}{$x$}

\ncline{-}{v1}{v3}
\ncline{-}{n2}{v3}
\ncline{-}{v2}{n2}
\ncline{-}{v3}{v4}
\ncline{-}{v4}{n5}
\ncline{-}{v4}{v6}
\ncline{-}{v6}{v7}
\ncline{-}{v7}{n3}
\ncline{-}{n3}{v10}
\ncline{-}{v7}{n4}
\ncline{-}{n4}{v11}
\Kante{n4}{n3}{\alpha_2}
\pscurve(6,-0.5)(7.1,0.8)(2.125,0)\rput(4.4,1){\colorbox{almostwhite}{$\alpha_1$}}
\Bogen{n2}{v6}{q_2}{-35}
\Bogen{v4}{n6}{q_3}{-35}
\end{pspicture}
	\label{NBC6b}
	}
	\caption{NBC6a and b}\label{NBC6}
	\end{figure}
	
	\newpage
	\textbf{NBC7:} $q_1$ is big and contains a node of $R_{\ell}$ which is not substituted by $q_3$
	\begin{figure}[H] \centering \psset{unit=1cm}
	\begin{pspicture}(-0.5,-1)(7.5,1.2) 
\Knoten{1}{1}{v1}\nput{180}{v1}{$u$}
\Knoten{1.5}{0.5}{n1}
\Knoten{1}{-1}{v2}\nput{180}{v2}{$w$}
\Knoten{1.5}{-0.5}{n2}
\Knoten{2}{0}{v3}
\Knoten{3.25}{0}{v4}
\Knoten{4.375}{0}{v6}
\Knoten{5.5}{0}{v7}
\Knoten{6}{0.5}{n3}
\Knoten{6}{-0.5}{n4}
\Knoten{6.25}{-0.75}{n5}
\Knoten{6.3}{0.8}{v10}\nput{0}{v10}{$v$}
\Knoten{6.5}{-1}{v11}\nput{0}{v11}{$x$}

\ncline{-}{v1}{v3}
\ncline{-}{n2}{v3}
\ncline{-}{v2}{n2}
\ncline{-}{v3}{v4}
\ncline{-}{v4}{v6}
\ncline{-}{v6}{v7}
\ncline{-}{v7}{n3}
\ncline{-}{n3}{v10}
\ncline{-}{v7}{n4}
\ncline{-}{n4}{v11}
\Kante{n5}{n3}{\alpha_2}
\pscurve(6.25,-0.75)(7.1,0.6)(1.5,0.5)\rput(4.4,1.2){\colorbox{almostwhite}{$\alpha_1$}}
\Bogen{n2}{v6}{q_2}{-35}
\Bogen{v4}{n4}{q_3}{-35}
\end{pspicture}
 \label{NBC7} \caption{NBC7}
\end{figure}

\textbf{NBC8:} $q_1$ is big and either
\begin{itemize}
	\item $q_1$ is not node-disjoint with $q_2$ and contains a node of $R_{\ell}$ which is substituted by $q_3$, or
	\item $q_1$ is not node-disjoint with $q_3$ and contains a node of $L_{\ell}$ which is substituted by $q_2$
\end{itemize}

\begin{figure}[H] \centering \psset{unit=1cm}\captionsetup[subfigure]{labelformat=empty}
\subfloat[NBC8a ($q_2$, $R_{\ell}$, $\vec{q_1}$ contains first a node of $q_2$)]{
	\begin{pspicture}(0.5,-1)(7.5,1.6) 
\Knoten{1.2}{0.8}{v1}\nput{180}{v1}{$u$}
\Knoten{1.5}{0.5}{n1}
\Knoten{1.2}{-0.8}{v2}\nput{180}{v2}{$w$}
\Knoten{1.5}{-0.5}{n2}
\Knoten{2}{0}{v3}
\Knoten{3.25}{0}{v4}
\Knoten{4.375}{0}{v6}
\Knoten{4.375}{1}{n7}
\Knoten{5.5}{0}{v7}
\Knoten{6}{0.5}{n3}
\Knoten{6}{-0.5}{n4}
\Knoten{6.25}{-0.75}{n6}
\Knoten{6.3}{0.8}{v10}\nput{0}{v10}{$v$}
\Knoten{6.5}{-1}{v11}\nput{0}{v11}{$x$}

\ncline{-}{v1}{v3}
\ncline{-}{n2}{v3}
\ncline{-}{v2}{n2}
\ncline{-}{v3}{v4}
\ncline{-}{v4}{v6}
\ncline{-}{v6}{v7}
\ncline{-}{v7}{n3}
\ncline{-}{n3}{v10}
\ncline{-}{v7}{n4}
\ncline{-}{n4}{v11}
\Kante{n7}{v6}{\beta_2}
\Kante{n4}{n3}{\alpha_3}
\Kante{n7}{n1}{\alpha_1}
\pscurve(6,-0.5)(7.1,1)(4.375,1)\rput(5.2,1.2){\colorbox{almostwhite}{$\alpha_2$}}
\pscurve(1.5,-0.5)(0.5,0.8)(4.375,1)\rput(2.2,1.3){\colorbox{almostwhite}{$\beta_1$}}
\Bogen{v4}{n6}{q_3}{-30}
\end{pspicture}\label{NBC8a}
} 
\hspace{1cm}
\subfloat[NBC8b ($q_2$, $R_{\ell}$, $\vec{q_1}$ contains first a node of $R_{\ell}$)]{
	\begin{pspicture}(0.5,-1)(7.5,1.6) 
\Knoten{1.2}{0.8}{v1}\nput{180}{v1}{$u$}
\Knoten{1.5}{0.5}{n1}
\Knoten{1.2}{-0.8}{v2}\nput{180}{v2}{$w$}
\Knoten{1.5}{-0.5}{n2}
\Knoten{2}{0}{v3}
\Knoten{3.25}{0}{v4}
\Knoten{4.375}{0}{v6}
\Knoten{4.375}{1}{n7}
\Knoten{5.5}{0}{v7}
\Knoten{6}{0.5}{n3}
\Knoten{6}{-0.5}{n4}
\Knoten{6.25}{-0.75}{n6}
\Knoten{6.3}{0.8}{v10}\nput{0}{v10}{$v$}
\Knoten{6.5}{-1}{v11}\nput{0}{v11}{$x$}

\ncline{-}{v1}{v3}
\ncline{-}{n2}{v3}
\ncline{-}{v2}{n2}
\ncline{-}{v3}{v4}
\ncline{-}{v4}{v6}
\ncline{-}{v6}{v7}
\ncline{-}{v7}{n3}
\ncline{-}{n3}{v10}
\ncline{-}{v7}{n4}
\ncline{-}{n4}{v11}
\Kante{n7}{v6}{\beta_2}
\Kante{n7}{n3}{\alpha_3}
\pscurve(6,-0.5)(7.1,1)(4.375,1)\rput(6.25,1.3){\colorbox{almostwhite}{$\alpha_2$}}
\pscurve(1.5,-0.5)(0.5,0.8)(4.375,1)\rput(2.2,1.3){\colorbox{almostwhite}{$\beta_1$}}
\pscurve(6,-0.5)(7.4,1.3)(1.5,0.5)\rput(4.2,1.6){\colorbox{almostwhite}{$\alpha_1$}}
\Bogen{v4}{n6}{q_3}{-30}
\end{pspicture}\label{NBC8b}
}
\end{figure}

\begin{figure}[H] \centering \psset{unit=1cm}\captionsetup[subfigure]{labelformat=empty}
\subfloat[NBC8c ($q_3$, $L_{\ell}$, $\vec{q_1}$ contains first a node of $q_3$)]{
	\begin{pspicture}(0.5,-1)(7.5,1.3) 
\Knoten{1.2}{0.8}{v1}\nput{180}{v1}{$u$}
\Knoten{1.5}{0.5}{n1}
\Knoten{1}{-1}{v2}\nput{180}{v2}{$w$}
\Knoten{1.5}{-0.5}{n2}
\Knoten{2}{0}{v3}
\Knoten{3.25}{0}{v4}
\Knoten{4.375}{0}{v6}
\Knoten{3.25}{0.9}{n7}
\Knoten{5.5}{0}{v7}
\Knoten{6}{0.5}{n3}
\Knoten{6}{-0.5}{n4}
\Knoten{1.25}{-0.75}{n6}
\Knoten{6.3}{0.8}{v10}\nput{0}{v10}{$v$}
\Knoten{6.5}{-1}{v11}\nput{0}{v11}{$x$}

\ncline{-}{v1}{v3}
\ncline{-}{n2}{v3}
\ncline{-}{v2}{n2}
\ncline{-}{v3}{v4}
\ncline{-}{v4}{v6}
\ncline{-}{v6}{v7}
\ncline{-}{v7}{n3}
\ncline{-}{n3}{v10}
\ncline{-}{v7}{n4}
\ncline{-}{n4}{v11}
\Kante{n7}{v4}{\gamma_2}
\Kante{n1}{n7}{\alpha_1}
\pscurve(6,-0.5)(7.1,0.8)(3.25,0.9)\rput(6.25,1.3){\colorbox{almostwhite}{$\gamma_1$}}
\pscurve(1.5,-0.5)(0.5,0.8)(3.25,0.9)\rput(2,1.1){\colorbox{almostwhite}{$\alpha_2$}}
\pscurve(6,0.5)(0.3,0.8)(1.5,-0.5)\rput(5.3,0.8){\colorbox{almostwhite}{$\alpha_3$}}
\Bogen{n6}{v6}{q_2}{-30}
\end{pspicture}\label{NBC8c}
}
\hspace{1cm}
\subfloat[NBC8d ($q_3$, $L_{\ell}$, $\vec{q_1}$ contains first a node of $L_{\ell}$)]{
	\begin{pspicture}(0.5,-1)(7.5,1.3) 
\Knoten{1.2}{0.8}{v1}\nput{180}{v1}{$u$}
\Knoten{1.5}{0.5}{n1}
\Knoten{1}{-1}{v2}\nput{180}{v2}{$w$}
\Knoten{1.5}{-0.5}{n2}
\Knoten{2}{0}{v3}
\Knoten{3.25}{0}{v4}
\Knoten{4.375}{0}{v6}
\Knoten{3.25}{1}{n7}
\Knoten{5.5}{0}{v7}
\Knoten{6}{0.5}{n3}
\Knoten{6}{-0.5}{n4}
\Knoten{1.25}{-0.75}{n6}
\Knoten{6.3}{0.8}{v10}\nput{0}{v10}{$v$}
\Knoten{6.5}{-1}{v11}\nput{0}{v11}{$x$}

\ncline{-}{v1}{v3}
\ncline{-}{n2}{v3}
\ncline{-}{v2}{n2}
\ncline{-}{v3}{v4}
\ncline{-}{v4}{v6}
\ncline{-}{v6}{v7}
\ncline{-}{v7}{n3}
\ncline{-}{n3}{v10}
\ncline{-}{v7}{n4}
\ncline{-}{n4}{v11}
\Kante{n7}{v4}{\gamma_2}
\Kante{n7}{n3}{\alpha_3}
\Kante{n1}{n2}{\alpha_1}
\pscurve(6,-0.5)(7,0.8)(3.25,1)\rput(5.2,1.3){\colorbox{almostwhite}{$\gamma_1$}}
\pscurve(1.5,-0.5)(0.2,0.8)(3.25,1)\rput(2.3,1.2){\colorbox{almostwhite}{$\alpha_2$}}
\Bogen{n6}{v6}{q_2}{-30}
\end{pspicture}\label{NBC8d}
} 
\label{NBC8}\caption{NBC8a, b, c and d}
\end{figure}

\textbf{NBC9:} $q_1$ contains a node of $L_{\ell}$ which is not substituted by $q_2$
\begin{figure}[H] \centering \psset{unit=1.1cm}\captionsetup[subfigure]{labelformat=empty}
\subfloat[NBC9a ($q_1$ small)]{
	\begin{pspicture}(-0.5,-1)(6.5,1.2) 
\Knoten{0.2}{0.8}{v1}\nput{180}{v1}{$u$}
\Knoten{0}{-1}{v2}\nput{180}{v2}{$w$}
\Knoten{0.5}{-0.5}{n2}
\Knoten{1}{0}{v3}
\Knoten{2.125}{0}{n5}
\Knoten{3.25}{0}{v4}
\Knoten{4.375}{0}{v6}
\Knoten{5.5}{0}{v7}
\Knoten{6}{0.5}{n3}
\Knoten{6}{-0.5}{n4}
\Knoten{0.25}{-0.75}{n6}
\Knoten{6.5}{1}{v10}\nput{0}{v10}{$v$}
\Knoten{6.5}{-1}{v11}\nput{0}{v11}{$x$}

\ncline{-}{v1}{v3}
\ncline{-}{n2}{v3}
\ncline{-}{v2}{n2}
\ncline{-}{v3}{v4}
\ncline{-}{v4}{n5}
\ncline{-}{v4}{v6}
\ncline{-}{v6}{v7}
\ncline{-}{v7}{n3}
\ncline{-}{n3}{v10}
\ncline{-}{v7}{n4}
\ncline{-}{n4}{v11}
\pscurve(6,0.5)(-0.3,0.7)(0.25,-0.75)\rput(4.4,1){\colorbox{almostwhite}{$\alpha_2$}}
\Bogen{n2}{v6}{q_2}{-30}
\Bogen{v4}{n4}{q_3}{-35}
\Kante{n6}{n5}{\alpha_1}
\end{pspicture} 
}
\hspace{0.5cm}
\subfloat[NBC9b ($q_1$ big)
]{
	\begin{pspicture}(0.5,-1)(6.5,1.2) 
\Knoten{1.2}{0.8}{v1}\nput{180}{v1}{$u$}
\Knoten{1.5}{0.5}{n1}
\Knoten{1}{-1}{v2}\nput{180}{v2}{$w$}
\Knoten{1.5}{-0.5}{n2}
\Knoten{2}{0}{v3}
\Knoten{3.25}{0}{v4}
\Knoten{4.375}{0}{v6}
\Knoten{5.5}{0}{v7}
\Knoten{6}{0.5}{n3}
\Knoten{6}{-0.5}{n4}
\Knoten{1.25}{-0.75}{n6}
\Knoten{6.5}{1}{v10}\nput{0}{v10}{$v$}
\Knoten{6.5}{-1}{v11}\nput{0}{v11}{$x$}

\ncline{-}{v1}{v3}
\ncline{-}{n2}{v3}
\ncline{-}{v2}{n2}
\ncline{-}{v3}{v4}
\ncline{-}{v4}{v6}
\ncline{-}{v6}{v7}
\ncline{-}{v7}{n3}
\ncline{-}{n3}{v10}
\ncline{-}{v7}{n4}
\ncline{-}{n4}{v11}
\pscurve(6,0.5)(0.5,0.7)(1.25,-0.75)\rput(4.4,1){\colorbox{almostwhite}{$\alpha_2$}}
\Bogen{n2}{v6}{q_2}{-35}
\Bogen{v4}{n4}{q_3}{-35}
\Kante{n6}{n1}{\alpha_1}
\end{pspicture} 
}
\label{NBC9} \caption{NBC9a and b}
\end{figure}
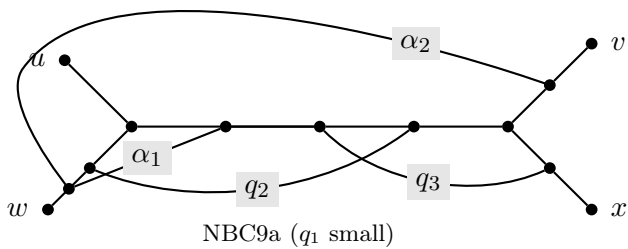
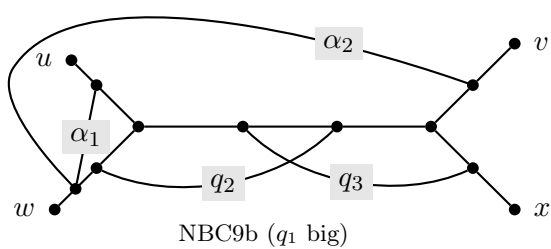




\textbf{NBC10:} $q_1$ is small and $\vec{q_1}$ contains a node of $L_{\ell}$ which is substituted by $q_2$ after a node of $q_2$
\begin{figure}[H] \centering \psset{unit=1cm}
	\begin{pspicture}(-0.5,-1.3)(7,1.5) 
\Knoten{0.2}{0.8}{v1}\nput{180}{v1}{$u$}
\Knoten{0}{-1}{v2}\nput{180}{v2}{$w$}
\Knoten{0.5}{-0.5}{n2}
\Knoten{1}{0}{v3}
\Knoten{2.125}{0}{n5}
\Knoten{3.25}{0}{v4}
\Knoten{4.375}{0}{v6}
\Knoten{5.5}{0}{v7}
\Knoten{6}{0.5}{n3}
\Knoten{6}{-0.5}{n4}
\Knoten{0.25}{-0.75}{n6}
\Knoten{6.5}{1}{v10}\nput{0}{v10}{$v$}
\Knoten{6.3}{-0.8}{v11}\nput{0}{v11}{$x$}
\Knoten{3.25}{0.9}{n7}

\ncline{-}{v1}{v3}
\ncline{-}{n2}{v3}
\ncline{-}{v2}{n2}
\ncline{-}{v3}{v4}
\ncline{-}{v4}{n5}
\ncline{-}{v4}{v6}
\ncline{-}{v6}{v7}
\ncline{-}{v7}{n3}
\ncline{-}{n3}{v10}
\ncline{-}{v7}{n4}
\ncline{-}{n4}{v11}
\pscurve(0.25,-0.75)(-1,1.3)(3.25,0.9)\rput(-1,1.3){\colorbox{almostwhite}{$\beta_1$}}
\pscurve(0.5,-0.5)(-0.5,1)(3.25,0.9)\rput(0,1.3){\colorbox{almostwhite}{$\alpha_2$}}
\pscurve(0.5,-0.5)(6.9,-0.8)(6,0.5)\rput(1.7,-1){\colorbox{almostwhite}{$\alpha_3$}}
\Bogen{v4}{n4}{q_3}{-35}
\Kante{n5}{n7}{\alpha_1}
\Kante{n7}{v6}{\beta_2}
\end{pspicture} 
\label{NBC10} \caption{NBC10}
\end{figure}

\textbf{NBC11:} $q_1$ is big and either 
\begin{itemize}
\item $\vec{q_1}$ contains a node of $L_{\ell}$ which is substituted by $q_2$ after a node of $q_2$ or
	%
		\item $\vec{q_1}$ contains a node of $q_3$ after a node of $R_{\ell}$ which is substituted by $q_3$ 
	%
\end{itemize}

\begin{figure}[H] \centering \psset{unit=1.1cm}\captionsetup[subfigure]{labelformat=empty}
\subfloat[NBC11a ($q_2$, $L_{\ell}$)]{
	\begin{pspicture}(0.5,-1.4)(7,1.2) 
\Knoten{1.2}{0.8}{v1}\nput{180}{v1}{$u$}
\Knoten{1.5}{0.5}{n1}
\Knoten{1}{-1}{v2}\nput{180}{v2}{$w$}
\Knoten{1.5}{-0.5}{n2}
\Knoten{2}{0}{v3}
\Knoten{3.25}{0}{v4}
\Knoten{4.375}{0}{v6}
\Knoten{5.5}{0}{v7}
\Knoten{6}{0.5}{n3}
\Knoten{6}{-0.5}{n4}
\Knoten{1.25}{-0.75}{n6}
\Knoten{6.5}{1}{v10}\nput{0}{v10}{$v$}
\Knoten{6.3}{-0.8}{v11}\nput{0}{v11}{$x$}
\Knoten{0.4}{0.2}{n7}

\ncline{-}{v1}{v3}
\ncline{-}{n2}{v3}
\ncline{-}{v2}{n2}
\ncline{-}{v3}{v4}
\ncline{-}{v4}{v6}
\ncline{-}{v6}{v7}
\ncline{-}{v7}{n3}
\ncline{-}{n3}{v10}
\ncline{-}{v7}{n4}
\ncline{-}{n4}{v11}
\pscurve(0.4,0.1)(0.8,1)(4.375,0)\rput(2.5,0.9){\colorbox{almostwhite}{$\beta_2$}}
\pscurve(1.5,-0.5)(6.9,-0.8)(6,0.5)\rput(2.5,-0.9){\colorbox{almostwhite}{$\alpha_3$}}

\Bogen{v4}{n4}{q_3}{-35}
\Kante{n1}{n7}{\alpha_1}
\Kante{n7}{n2}{\alpha_2}
\Bogen{n6}{n7}{\beta_1}{70}
\end{pspicture} 
}
\hspace{0.5cm}
\subfloat[NBC11b ($R_{\ell}$, $q_3$)]{
	\begin{pspicture}(0.5,-1.4)(7,1.2) 
\Knoten{1}{1}{v1}\nput{180}{v1}{$u$}
\Knoten{1.5}{0.5}{n1}
\Knoten{1.2}{-0.8}{v2}\nput{180}{v2}{$w$}
\Knoten{1.5}{-0.5}{n2}
\Knoten{2}{0}{v3}
\Knoten{3.25}{0}{v4}
\Knoten{4.375}{0}{v6}
\Knoten{5.5}{0}{v7}
\Knoten{6}{0.5}{n3}
\Knoten{6}{-0.5}{n4}
\Knoten{6.25}{-0.75}{n6}
\Knoten{6.3}{0.8}{v10}\nput{0}{v10}{$v$}
\Knoten{6.5}{-1}{v11}\nput{0}{v11}{$x$}
\Knoten{7}{0.2}{n7}

\ncline{-}{v1}{v3}
\ncline{-}{n2}{v3}
\ncline{-}{v2}{n2}
\ncline{-}{v3}{v4}
\ncline{-}{v4}{v6}
\ncline{-}{v6}{v7}
\ncline{-}{v7}{n3}
\ncline{-}{n3}{v10}
\ncline{-}{v7}{n4}
\ncline{-}{n4}{v11}
\pscurve(3.25,0)(7,1)(7,0.2)\rput(4.5,0.6){\colorbox{almostwhite}{$\gamma_2$}}
\pscurve(1.5,0.5)(0.5,-0.8)(6,-0.5)\rput(4.8,-1){\colorbox{almostwhite}{$\alpha_1$}}

\Bogen{n2}{v6}{q_2}{-35}
\Kante{n3}{n7}{\alpha_3}
\Kante{n4}{n7}{\alpha_2}
\Bogen{n6}{n7}{\gamma_1}{-70}
\end{pspicture} 
}
\label{NBC11} \caption{NBC11a and b}
\end{figure}
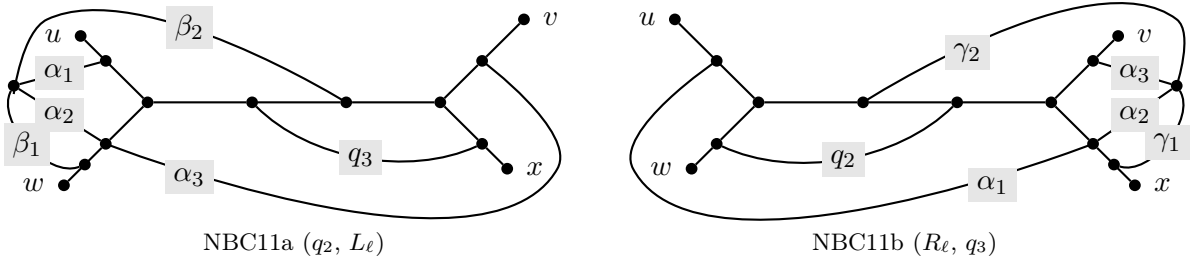
\textbf{NBC12:} $q_1$ is big, not node-disjoint with $q_2$ and not node-disjoint with $q_3$

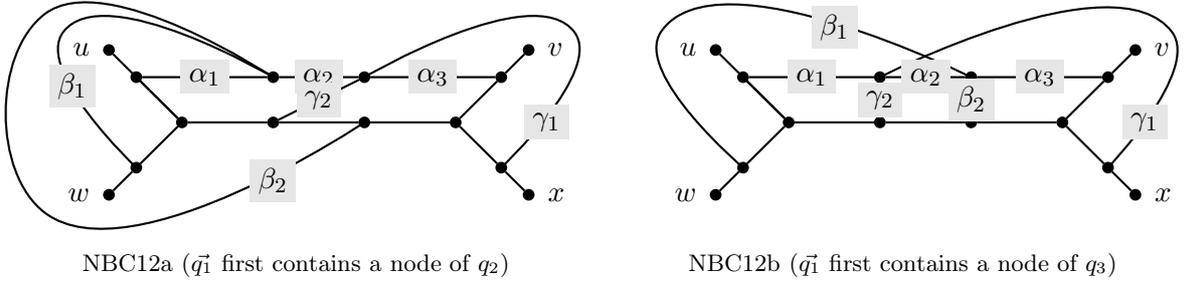
\begin{figure}[H] \centering \psset{unit=1.2cm}\captionsetup[subfigure]{labelformat=empty}
\subfloat[NBC12a ($\vec{q_1}$ first contains a node of $q_2$)]{
 \begin{pspicture}(-0.5,-1.3)(5,1.3) 
\Knoten{0.2}{0.8}{s1}\nput{180}{s1}{$u$}
\Knoten{0.5}{0.5}{n1}
\Knoten{0.2}{-0.8}{s2}\nput{180}{s2}{$w$}
\Knoten{0.5}{-0.5}{v1}
\Knoten{1}{0}{v2}
\Knoten{2}{0}{v4}
\Knoten{3}{0}{v5}
\Knoten{4}{0}{v6}
\Knoten{2}{0.5}{n3}
\Knoten{3}{0.5}{v7}
\Knoten{4.5}{0.5}{v9}
\Knoten{4.8}{0.8}{t1}\nput{0}{t1}{$v$}
\Knoten{4.5}{-0.5}{v10}
\Knoten{4.8}{-0.8}{t2}\nput{0}{t2}{$x$}
\ncline{s1}{v2}
\ncline{s2}{v1}
\ncline{n1}{v2}
\ncline{v1}{v2}
\ncline{v2}{v4}
\ncline{v4}{v5}
\ncline{v5}{v6}
\ncline{v6}{v9}
\ncline{v6}{v10}
\ncline{v9}{t1}
\ncline{v10}{t2}

\pscurve(3,0.5)(5.3,1)(4.5,-0.5)\rput(5,0){\colorbox{almostwhite}{$\gamma_1$}}
\pscurve(0.5,-0.5)(-0.3,1)(2,0.5)\rput(-0.2,0.4){\colorbox{almostwhite}{$\beta_1$}}
\pscurve(2,0.5)(-0.5,1.2)(-0.5,-1)(3,0)\rput(2,-0.65){\colorbox{almostwhite}{$\beta_2$}}
\Kante{v7}{v9}{\alpha_3}
\Kante{n1}{n3}{\alpha_1}
\Kante{n3}{v7}{\alpha_2}
\Kante{v4}{v7}{\gamma_2}
\end{pspicture}
}
\hspace{1cm}
\subfloat[NBC12b ($\vec{q_1}$ first contains a node of $q_3$)]{
\begin{pspicture}(-0.5,-1.3)(5,1.3) 
\Knoten{0.2}{0.8}{s1}\nput{180}{s1}{$u$}
\Knoten{0.5}{0.5}{n1}
\Knoten{0.2}{-0.8}{s2}\nput{180}{s2}{$w$}
\Knoten{0.5}{-0.5}{v1}
\Knoten{1}{0}{v2}
\Knoten{2}{0}{v4}
\Knoten{3}{0}{v5}
\Knoten{4}{0}{v6}
\Knoten{2}{0.5}{v7}
\Knoten{3}{0.5}{n2}
\Knoten{4.5}{0.5}{v9}
\Knoten{4.8}{0.8}{t1}\nput{0}{t1}{$v$}
\Knoten{4.5}{-0.5}{v10}
\Knoten{4.8}{-0.8}{t2}\nput{0}{t2}{$x$}
\ncline{s1}{v2}
\ncline{s2}{v1}
\ncline{n1}{v2}
\ncline{v1}{v2}
\ncline{v2}{v4}
\ncline{v4}{v5}
\ncline{v5}{v6}
\ncline{v6}{v9}
\ncline{v6}{v10}
\ncline{v9}{t1}
\ncline{v10}{t2}
\pscurve(2,0.5)(5.2,1)(4.5,-0.5)\rput(4.9,0){\colorbox{almostwhite}{$\gamma_1$}}
\pscurve(0.5,-0.5)(-0.4,1)(3,0.5)\rput(1.5,1.05){\colorbox{almostwhite}{$\beta_1$}}
\Kante{n2}{v5}{\beta_2}
\Kante{v7}{n2}{\alpha_2}
\Kante{n2}{v9}{\alpha_3}
\Kante{n1}{v7}{\alpha_1}
\Kante{v4}{v7}{\gamma_2}
\end{pspicture}
}
\caption{NBC12a and b}
\label{NBC12}
\end{figure}

\end{definition} 

The following lemma shows that the PBCs constructed in the proof of \autoref{lemma3} have to be NBCs.
\begin{lemma}\label{lemma:pbc}
	If $(G,(s_1,t_1),(s_2,t_2))$ does not contain a BC, then each PBC has to be a NBC.
\end{lemma}

\begin{proof}
Let us assume (by contradiction) that $(G,(s_1,t_1),(s_2,t_2))$ does not contain a BC and $(P_u,P_{\ell},q_1,q_2,q_3)$ is a PBC, but no NBC. 

The following properties are closely related to the properties defining the different types of BCs and NBCs. By investigating which of these properties are fulfilled, we construct a contradiction.

\begin{enumerate}[itemsep=-0em,leftmargin=*,label=(P\arabic*)]
	\item\label{p1} $q_1$ small;
	\item\label{p2} $q_1$ big;
	\item\label{p3} $q_1$ node-disjoint with $q_2$;
	\item\label{p4} $q_1$ node-disjoint with $q_3$;
	\item\label{p5} $q_2$ node-disjoint with $q_3$;
	\item\label{p6} $q_1$ contains nodes of $L_{\ell}$ or $R_{\ell}$;
	\item\label{p7} $q_2$ contains nodes of $L_u$ or $R_u$;
	\item\label{p8} $q_3$ contains nodes of $L_u$ or $R_u$.
\end{enumerate}

Since $(P_u,P_{\ell},q_1,q_2,q_3)$ is no NBC, the property \ref{p5} has to hold (no NBC2), whereas \ref{p7} and \ref{p8} are not fulfilled (no NBC3 and 4). According to \ref{p6} we get that $q_1$ does not contain a node of $L_{\ell}$ \textbf{and} a node of $R_{\ell}$ (no NBC5). 

We first assume that \ref{p1} is fulfilled (thus \ref{p2} not) and show that we get a contradiction in this case. 
Since $(P_u,P_{\ell},q_1,q_2,q_3)$ is no NBC1, \ref{p4} holds. Regarding \ref{p6} we get that $q_1$ does not contain a node of $R_{\ell}$ (no NBC6) and no node of $L_{\ell}$ which is not substituted by $q_2$ (no NBC9). 
What remains according to \ref{p6} is if the following property holds: 
\begin{enumerate}[itemsep=-0em,leftmargin=*,label=(P\arabic*')]
\setcounter{enumi}{5}
\item\label{p6'} $q_1$ does not contain a node of $L_{\ell}$ which is substituted by $q_2$.
\end{enumerate}
Furthermore we have not investigated if \ref{p3} is fulfilled. 
Since $(P_u,P_{\ell},q_1,q_2,q_3)$ cannot be a BC1, at least one of the properties \ref{p3} and \ref{p6'} does not hold.
Assume that \ref{p3} does not hold, but \ref{p6'} is true. As we will see, this cannot happen since we excluded the existence of a BC2. To this aim let us subdivide $q_1$ from left to right into $\alpha_1,\alpha_2,\alpha_3$, where $\alpha_2$ is the subpath from the first node $v_a$ to the last node $v_b$ of $q_1$ which is contained in $q_2$. 
Furthermore, let $\beta$ be the subpath of $q_2$ from $v_a$ to $v_b$. But now either $\alpha_2=\beta$ holds, thus $(P_u,P_{\ell},q_1,q_2,q_3)$ is a BC2, or $(P_u,P_{\ell},q_1',q_2,q_3)$ is a BC2, where $q_1'$ is generated by using $\beta$ instead of $\alpha_2$ in $q_1$.
The argumentation is very similar for the case that \ref{p3} holds, but not \ref{p6'}: Let $v_a$ be the last node of $q_1$ which is contained in $L_{\ell}$ and $v_b$ the first node of the commonly used path (considering it from left to right). Furthermore, $\alpha_1$ is the subpath of $q_1$ beginning with the start node of $q_1$ and ending with $v_a$, and $\ell$ is the subpath of $L_{\ell}$ from $v_b$ to $v_a$.
Now either $(P_u,P_{\ell},q_1,q_2,q_3)$ is a BC3, or we get a BC3 by using $\ell$ instead of $\alpha_1$ in $q_1$. 
%
%
The only remaining case is that both properties are violated. Since $(P_u,P_{\ell},q_1,q_2,q_3)$ is no NBC10, we can subdivide $q_1$ from left to right into $\alpha_1,\alpha_2,\alpha_3,\alpha_4,\alpha_5$, where $\alpha_2$ is the subpath between the first node $v_a$ and the last node $v_b$ of $q_1$ which is contained in $L_{\ell}$ and $\alpha_4$ is the subpath between the first node $v_c$ and the last node $v_d$ of $q_1$ which is contained in $q_2$. Furthermore let $v_e$ be the first node of the commonly used path.
But now we again get a contradiction since we can show that there is a BC4:
Let $\ell$ be the subpath of $L_{\ell}$ from $v_e$ to $v_b$ and $\beta$ the subpath of $q_2$ from $v_c$ to $v_d$. 
Now either $(P_u,P_{\ell},q_1,q_2,q_3)$ is a BC4 or we get a BC4 by using $\ell$ instead of $\alpha_1$ and $\alpha_2$, and $\beta$ instead of $\alpha_4$ in $q_1$.

This shows that \ref{p1} cannot hold and therefore \ref{p2} has to be fulfilled. 
Since $(P_u,P_{\ell},q_1,q_2,q_3)$ is no NBC7 and no NBC9, $q_1$ does not contain a node of $R_{\ell}$ which is not substituted by $q_3$, and not a node of $L_{\ell}$ which is not substituted by $q_2$. 
It remains to investigate if \ref{p3}, \ref{p4}, \ref{p6'} and the following property \ref{p6''} hold:
\begin{enumerate}[itemsep=-0em,leftmargin=*,label=(P\arabic*'')]
\setcounter{enumi}{5}
\item\label{p6''} $q_1$ does not contain a node of $R_{\ell}$ which is substituted by $q_3$.
\end{enumerate}

Since $(P_u,P_{\ell},q_1,q_2,q_3)$ cannot be a BC1, at least one of these properties does not hold. Also note that at most one of the properties \ref{p3} and \ref{p4} can be violated (no NBC12), and also at most one of \ref{p6'} and \ref{p6''} (no NBC5). Therefore at most two properties can be violated.
Let us first consider the case that exactly one of these properties is violated: 
If this is \ref{p3} or \ref{p4}, we get a contradiction since there is no BC2 (either $(P_u,P_{\ell},q_1,q_2,q_3)$ or $(P_u,P_{\ell},q_1,q_3,q_2)$ is a BC2, or we can change $q_1$ to $q_1'$ (as in the case of $q_1$ small) and get that $(P_u,P_{\ell},q_1',q_2,q_3)$ or $(P_u,P_{\ell},q_1',q_3,q_2)$ is a BC2).
If \ref{p6'} or \ref{p6''} does not hold, we also get contradictions since there is no BC3 (by a suitable change of $q_1$ to $q_1'$, we get that either $(P_u,P_{\ell},q_1',q_2,q_3)$ or $(P_u,P_{\ell},q_1',q_3,q_2)$ is a BC3).
Now we consider the case that two properties are violated: First assume that $q_1$ is not node-disjoint with $q_2$ (i.e. \ref{p4} is true, \ref{p3} not). Since $(P_u,P_{\ell},q_1,q_2,q_3)$ is no NBC8, \ref{p6''} has to be true, whereas \ref{p6'} is violated. Furthermore $(P_u,P_{\ell},q_1,q_2,q_3)$ is no NBC11, thus we can subdivide $q_1$ from left to right into $\alpha_1,\alpha_2,\alpha_3,\alpha_4,\alpha_5$, where $\alpha_2$ is the subpath between the first and the last node of $q_1$ which is contained in $L_{\ell}$ and $\alpha_4$ is the subpath between the first and the last node of $q_1$ which is contained in $q_2$.  
But now we get a contradiction, since we get a BC4 by a suitable change of $q_1$ (as for $q_1$ small). 
Therefore \ref{p3} has to hold, whereas \ref{p4} is violated. 
As $(P_u,P_{\ell},q_1,q_2,q_3)$ is no NBC8, \ref{p6'} has to be true and \ref{p6''} is violated. 
Furthermore (no NBC11) we can subdivide $q_1$ from left to right into $\alpha_1,\alpha_2,\alpha_3,\alpha_4,\alpha_5$, where $\alpha_2$ is the subpath between the first and the last node of $q_1$ which is contained in $q_3$ and $\alpha_4$ is the subpath between the first and the last node of $q_1$ which is contained in $R_{\ell}$.
But now we get a contradiction, since $(P_u,P_{\ell},q_1',q_3,q_2)$ is a BC4 for a suitable change of $q_1$ to $q_1'$ (as for $q_1$ small). 

This shows that our initial assumption, that $(P_u,P_{\ell},q_1,q_2,q_3)$ is no NBC, cannot be true.
\end{proof}

Next we derive some properties of the NBCs which we use in \autoref{lemmapbc1} - \autoref{lemmarest} to analyze the PBCs constructed in the proof of \autoref{lemma3}. In view of the proof of \autoref{lemma3} it is reasonable to assume that $P_u \cup P_{\ell}=F$ and $q_1,q_2$ and $q_3$ are tight alternatives. In particular, all commonly used edges are completely paid, and $e$ is either in $L_u \cup L_{\ell}$ or in $R_u \cup R_{\ell}$. 
Furthermore we denote the player who corresponds to the source-sink-pair $(u,v)$ as the \textit{upper Player} and the other one as the \textit{lower Player}, and their cost shares by $\xi_{u,.}$ and $\xi_{\ell,.}$, respectively. 
To simplify notation, we use the displayed path labels for the paths itself, but also for the costs of the corresponding paths. Additionally we just write $\xi_{i,p}$ for the sum of cost shares that Player $i$ pays on the edges of a path $p$. 

	\begin{figure}[H] \centering \psset{unit=1.3cm}
 \begin{pspicture}(-0.5,-1)(7,1) 
\Knoten{0.2}{0.8}{s1}\nput{180}{s1}{$u$}
\Knoten{0.2}{-0.8}{s2}\nput{180}{s2}{$w$}
\Knoten{0.5}{-0.5}{v1}
\Knoten{1}{0}{v2}
\Knoten{2.125}{0}{v3}
\Knoten{3.25}{0}{v4}
\Knoten{4.375}{0}{v5}
\Knoten{5.5}{0}{v6}
\Knoten{3.6}{0.75}{v8}
\Knoten{6}{0.5}{v9}
\Knoten{6.3}{0.8}{t1}\nput{0}{t1}{$v$}
\Knoten{6}{-0.5}{v10}
\Knoten{6.3}{-0.8}{t2}\nput{0}{t2}{$x$}
\ncline{s1}{v2}
\ncline{s2}{v1}
\ncline{v1}{v2}
\ncline{v2}{v3}
\Kante{v3}{v4}{a}
\Kante{v4}{v5}{b}
\Kante{v5}{v6}{c}
\Kante{v6}{v9}{d}
\Kante{v6}{v10}{f}
\ncline{v9}{t1}
\ncline{v10}{t2}
\Bogendashed{v1}{v5}{q_2}{-35}
\pscurve[linestyle=dashed](3.6,0.75)(6,1.3)(7,0.8)(6,-0.5)\rput(5,1.1){\colorbox{almostwhite}{$\gamma_1$}}
\Kantedashed{v8}{v9}{\alpha_2}
\Kantedashed{v3}{v8}{\alpha_1}
\Kantedashed{v4}{v8}{\gamma_2}
\end{pspicture}
\caption{NBC1}\label{propNBC1}
\end{figure}

\begin{lemma} \label{lemmanbcanfang} For NBC1 (cf. Figure~\ref{propNBC1}), we get the following properties: 
\begin{enumerate}[itemsep=-0em,leftmargin=*,label=\roman*)]
	\item All edges which are substituted by $q_1$ or by $q_3$ are completely paid.
	\item For $(u,v,w,x) \in \{(s_1,t_1,s_2,t_2), (t_1,s_1,t_2,s_2)\}$, Player 1 pays all edges of the commonly used path which are substituted by $q_1$, but not by $q_3$ $(a)$.
	\item For $(u,v,w,x) \in \{(s_2,t_2,s_1,t_1), (t_2,s_2,t_1,s_1)\}$, Player 2 pays all edges of the commonly used path which are substituted by $q_1$, but not by $q_3$ $(a)$.
\end{enumerate}
\end{lemma}

\begin{proof}
Since $q_1$ is a tight alternative for the upper player and $q_3$ is a tight alternative for the lower player, we get
\[
\alpha_1+\alpha_2=\xi_{u,a}+\xi_{u,b}+\xi_{u,c}+\xi_{u,d} \text{ and }  \gamma_1+\gamma_2=\xi_{\ell,b}+\xi_{\ell,c}+\xi_{\ell,f}.
\]
Adding these equalities yields
\[
\alpha_1+\alpha_2+\gamma_1+\gamma_2=\xi_{u,a}+b+c+\xi_{u,d}+\xi_{\ell,f}.
\]
Since $P_u \cup P_{\ell}$ is an optimal Steiner forest and adding $\alpha_1, \alpha_2$ and $\gamma_1$, while deleting $a,b,c,d$ and $f$, yields another Steiner forest, $\alpha_1+\alpha_2+ \gamma_1 \geq a+b+c+d+f$ has to hold. Altogether we get that $\xi_{u,a}+\xi_{u,d}+\xi_{\ell,f} \geq a+d+f$ and therefore $\xi_{u,a}=a, \xi_{u,d}=d$ and $\xi_{\ell,f}=f$ holds.
\end{proof}

	\begin{figure}[H] \centering \psset{unit=1.3cm}
 \begin{pspicture}(0.5,-1)(7,0.8) 
\Knoten{1.2}{0.8}{s1}\nput{180}{s1}{$u$}
\Knoten{1.2}{-0.8}{s2}\nput{180}{s2}{$w$}
\Knoten{1.5}{-0.5}{v1}
\Knoten{2}{0}{v2}
\Knoten{3.25}{0}{v4}
\Knoten{4.375}{0}{v5}
\Knoten{5.5}{0}{v6}
\Knoten{3.8}{-1}{v7}
%
\Knoten{6.3}{0.8}{t1}\nput{0}{t1}{$v$}
\Knoten{6}{-0.5}{v10}
\Knoten{6.3}{-0.8}{t2}\nput{0}{t2}{$x$}
\ncline{s1}{v2}
\ncline{s2}{v1}
\Kante{v1}{v2}{a}
\Kante{v2}{v4}{b}
\Kante{v4}{v5}{c}
\Kante{v5}{v6}{d}
\ncline{v6}{t1}
\Kante{v6}{v10}{f}
\ncline{v10}{t2}
\Kantedashed{v1}{v7}{\beta_1}
\Kantedashed{v10}{v7}{\gamma_1}
\Kantedashed{v7}{v4}{\gamma_2}
\Kantedashed{v7}{v5}{\beta_2}
%
\end{pspicture}
\caption{NBC2}\label{propNBC2}
\end{figure}

\begin{lemma} For NBC2 (cf. Figure~\ref{propNBC2}), we get the following properties:
	\begin{enumerate}[itemsep=-0em,leftmargin=*,label=\roman*)]
	\item For $(u,v,w,x) \in \{(s_1,t_1,s_2,t_2), (t_1,s_1,t_2,s_2)\}$, there is a tight alternative for Player 2 which substitutes all edges of $P_{2}$ which are substituted by $q_2$ or $q_3$ (in particular all commonly used edges).
	Furthermore, Player 2 pays all edges of the commonly used path which are substituted by $q_2$ \textbf{and} $q_3$ $(c)$. 
	\item For $(u,v,w,x) \in \{(s_2,t_2,s_1,t_1), (t_2,s_2,t_1,s_1)\}$, there is a tight alternative for Player 1 which substitutes all edges of $P_{1}$ which are substituted by $q_2$ or $q_3$ (in particular all commonly used edges). 
	Furthermore, Player 1 pays all edges of the commonly used path which are substituted by $q_2$ \textbf{and} $q_3$ $(c)$. 
\end{enumerate} 
\end{lemma}

\begin{proof}
Since $q_2$ and $q_3$ are tight alternatives for the lower player we get
\[
\beta_1+\beta_2=\xi_{\ell,a}+\xi_{\ell,b}+\xi_{\ell,c} \text{ and } \gamma_1+\gamma_2=\xi_{\ell,c}+\xi_{\ell,d}+\xi_{\ell,f}.
\]
Adding these two equalities yields 
\[
\beta_1+\beta_2+\gamma_1+\gamma_2=\xi_{\ell,a}+\xi_{\ell,b}+2\xi_{\ell,c}+\xi_{\ell,d}+\xi_{\ell,f}.
\]
Furthermore there is an alternative $q$ for the lower Player which substitutes $a,b,c,d$ and $f$ (contained in the union of the paths $\beta_1$ and $\gamma_1$) and therefore 
\[
\beta_1+\gamma_1 \geq c(q) \geq \xi_{\ell,a}+\xi_{\ell,b}+\xi_{\ell,c}+\xi_{\ell,d}+\xi_{\ell,f}
\]
holds, where $c(q)$ denotes the cost of the alternative $q$. Using this we get
\[
\xi_{\ell,a}+\xi_{\ell,b}+2\xi_{\ell,c}+\xi_{\ell,d}+\xi_{\ell,f}=\beta_1+\beta_2+\gamma_1+\gamma_2 \geq  \xi_{\ell,a}+\xi_{\ell,b}+\xi_{\ell,c}+\xi_{\ell,d}+\xi_{\ell,f} + \beta_2+\gamma_2
\]
and thus $\xi_{\ell,c} \geq \beta_2+\gamma_2$ holds. On the other hand, $\xi_{\ell,c}\leq c \leq \gamma_2+\beta_2$, since adding $\gamma_2$ and $\beta_2$ while deleting $c$ also yields a Steiner forest. This implies $\xi_{\ell,c}= c = \gamma_2+\beta_2$. Furthermore we get that $q$ is a tight alternative since $\xi_{\ell,a}+\xi_{\ell,b}+2\xi_{\ell,c}+\xi_{\ell,d}+\xi_{\ell,f}=\beta_1+\gamma_1+\xi_{\ell,c}$ and therefore $\xi_{\ell,a}+\xi_{\ell,b}+\xi_{\ell,c}+\xi_{\ell,d}+\xi_{\ell,f}=\beta_1+\gamma_1=c(q)$ holds.
\end{proof}

	
	\begin{figure}[H] \centering \psset{unit=1cm}\captionsetup[subfigure]{labelformat=empty}
\subfloat[NBC3a]{
 \begin{pspicture}(-0.5,-1.5)(7,0.7) 
\Knoten{0}{1}{v1}\nput{180}{v1}{$u$}
\Knoten{0.5}{0.5}{n1}
\Knoten{0}{-1}{v2}\nput{180}{v2}{$w$}
\Knoten{0.5}{-0.5}{n2}
\Knoten{1}{0}{v3}
\Knoten{2.125}{0}{n5}
\Knoten{3.25}{0}{v4}
\Knoten{4.375}{0}{v6}
\Knoten{5.5}{0}{v7}
\Knoten{6}{0.5}{n3}
\Knoten{6}{-0.5}{n4}
\Knoten{6.5}{1}{v10}\nput{0}{v10}{$v$}
\Knoten{6.5}{-1}{v11}\nput{0}{v11}{$x$}

\ncline{-}{v1}{n1}
\Kante{n1}{v3}{a}
\Kante{n2}{v3}{b}
\ncline{-}{v2}{n2}
\Kante{v3}{n5}{c}
\Kante{v4}{n5}{d}
\Kante{v4}{v6}{f}
\Kante{v6}{v7}{g}
\Kante{v7}{n3}{h}
\ncline{-}{n3}{v10}
\Kante{v7}{n4}{i}
\ncline{-}{n4}{v11}
\Bogendashed{n5}{n3}{q_1}{35}
\Bogendashed{n1}{v6}{\beta_2}{35}
\Bogendashed{n2}{n1}{\beta_1}{30}
\Bogendashed{v4}{n4}{q_3}{-35}
\end{pspicture}
}
\hspace{0.2cm}
\subfloat[NBC3b]{
 \begin{pspicture}(-0.5,-1.5)(7.2,0.7) 
\Knoten{0}{1}{v1}\nput{180}{v1}{$u$}
\Knoten{0}{-1}{v2}\nput{180}{v2}{$w$}
\Knoten{0.5}{-0.5}{n2}
\Knoten{1}{0}{v3}
\Knoten{2.125}{0}{n5}
\Knoten{3.25}{0}{v4}
\Knoten{4.375}{0}{v6}
\Knoten{5.5}{0}{v7}
\Knoten{6}{0.5}{n3}
\Knoten{6.45}{0.95}{n6}
\Knoten{6}{-0.5}{n4}
\Knoten{6.7}{1.2}{v10}\nput{0}{v10}{$v$}
\Knoten{6.3}{-0.8}{v11}\nput{0}{v11}{$x$}

\ncline{-}{v1}{v3}
\Kante{n2}{v3}{a}
\ncline{-}{v2}{n2}
\Kante{v3}{n5}{b}
\Kante{v4}{n5}{c}
\Kante{v4}{v6}{d}
\Kante{v6}{v7}{f}
\Kante{v7}{n3}{g}
\Kante{n3}{n6}{h}
\ncline{-}{n6}{v10}
\ncline{-}{v7}{n4}
\ncline{-}{n4}{v11}
\Bogendashed{n5}{n3}{q_1}{45}
\Bogendashed{v6}{n6}{\beta_2}{20}
\Bogendashed{v4}{n4}{q_3}{-35}
\pscurve[linestyle=dashed](0.5,-0.5)(6.9,-0.9)(6.45,0.95)\rput(1.4,-0.9){\colorbox{almostwhite}{$\beta_1$}}
\end{pspicture}
}

\subfloat[NBC3c]{
 \begin{pspicture}(-0.5,-1.5)(7,1) 
\Knoten{0.2}{0.8}{v1}\nput{180}{v1}{$u$}
\Knoten{0.2}{-0.8}{v2}\nput{180}{v2}{$w$}
\Knoten{0.5}{-0.5}{n2}
\Knoten{1}{0}{v3}
\Knoten{2.125}{0}{n5}
\Knoten{3.25}{0}{v4}
\Knoten{4.375}{0}{v6}
\Knoten{5.5}{0}{v7}
\Knoten{6}{0.5}{n3}
\Knoten{6.45}{0.95}{n6}
\Knoten{6}{-0.5}{n4}
\Knoten{6.7}{1.2}{v10}\nput{0}{v10}{$v$}
\Knoten{6.3}{-0.8}{v11}\nput{0}{v11}{$x$}

\ncline{-}{v1}{v3}
\Kante{n2}{v3}{a}
\ncline{-}{v2}{n2}
\Kante{v3}{n5}{b}
\Kante{v4}{n5}{c}
\Kante{v4}{v6}{d}
\Kante{v6}{v7}{f}
\Kante{v7}{n3}{g}
\Kante{n3}{n6}{h}
\ncline{-}{n6}{v10}
\ncline{-}{v7}{n4}
\ncline{-}{n4}{v11}
\Bogendashed{n5}{n6}{q_1}{35}
\Bogendashed{v6}{n3}{\beta_2}{25}
\Bogendashed{v4}{n4}{q_3}{-35}
\pscurve[linestyle=dashed](0.5,-0.5)(6.9,-0.8)(6,0.5)\rput(1.4,-0.9){\colorbox{almostwhite}{$\beta_1$}}
\end{pspicture}
}

\subfloat[NBC3d]{
 \begin{pspicture}(0.5,-1)(7,1.3) 
\Knoten{0.8}{1.2}{v1}\nput{180}{v1}{$u$}
\Knoten{1.5}{0.5}{n1}
\Knoten{1.05}{0.95}{n6}
\Knoten{1.2}{-0.8}{v2}\nput{180}{v2}{$w$}
\Knoten{1.5}{-0.5}{n2}
\Knoten{2}{0}{v3}
\Knoten{3.25}{0}{v4}
\Knoten{4.375}{0}{v6}
\Knoten{5.5}{0}{v7}
\Knoten{6}{0.5}{n3}
\Knoten{6}{-0.5}{n4}
\Knoten{6.3}{0.8}{v10}\nput{0}{v10}{$v$}
\Knoten{6.3}{-0.8}{v11}\nput{0}{v11}{$x$}

\ncline{-}{v1}{n6}
\Kante{n6}{n1}{a}
\Kante{n1}{v3}{b}
\Kante{n2}{v3}{c}
\ncline{-}{v2}{n2}
\Kante{v3}{v4}{d}
\Kante{v4}{v6}{f}
\Kante{v6}{v7}{g}
\Kante{v7}{n3}{h}
\ncline{-}{n3}{v10}
\Kante{v7}{n4}{i}
\ncline{-}{n4}{v11}
\Bogendashed{n1}{n3}{q_1}{35}
\Bogendashed{n6}{v6}{\beta_2}{18}
\Bogendashed{v4}{n4}{q_3}{-35}
\Bogendashed{n2}{n6}{\beta_1}{20}
\end{pspicture}
}
\hspace{0.5cm}
\subfloat[NBC3e]{
 \begin{pspicture}(0.5,-1)(7,1.3) 
\Knoten{0.8}{1.2}{v1}\nput{180}{v1}{$u$}
\Knoten{1.5}{0.5}{n1}
\Knoten{1.05}{0.95}{n6}
\Knoten{1.2}{-0.8}{v2}\nput{180}{v2}{$w$}
\Knoten{1.5}{-0.5}{n2}
\Knoten{2}{0}{v3}
\Knoten{3.25}{0}{v4}
\Knoten{4.375}{0}{v6}
\Knoten{5.5}{0}{v7}
\Knoten{6}{0.5}{n3}
\Knoten{6}{-0.5}{n4}
\Knoten{6.3}{0.8}{v10}\nput{0}{v10}{$v$}
\Knoten{6.3}{-0.8}{v11}\nput{0}{v11}{$x$}

\ncline{-}{v1}{n6}
\Kante{n6}{n1}{a}
\Kante{n1}{v3}{b}
\Kante{n2}{v3}{c}
\ncline{-}{v2}{n2}
\Kante{v3}{v4}{d}
\Kante{v4}{v6}{f}
\Kante{v6}{v7}{g}
\Kante{v7}{n3}{h}
\ncline{-}{n3}{v10}
\Kante{v7}{n4}{i}
\ncline{-}{n4}{v11}
\Bogendashed{n6}{n3}{q_1}{25}
\Bogendashed{n1}{v6}{\beta_2}{25}
\Bogendashed{v4}{n4}{q_3}{-35}
\Bogendashed{n2}{n1}{\beta_1}{30}
\end{pspicture}
}
\end{figure}

\begin{figure}[H] \centering \psset{unit=1cm}\captionsetup[subfigure]{labelformat=empty}
\subfloat[NBC3f]{
 \begin{pspicture}(0.5,-1.7)(7,1.1) 
\Knoten{1.2}{0.8}{v1}\nput{180}{v1}{$u$}
\Knoten{1.5}{0.5}{n1}
\Knoten{1.2}{-0.8}{v2}\nput{180}{v2}{$w$}
\Knoten{1.5}{-0.5}{n2}
\Knoten{2}{0}{v3}
\Knoten{3.25}{0}{v4}
\Knoten{4.375}{0}{v6}
\Knoten{5.5}{0}{v7}
\Knoten{6}{0.5}{n3}
\Knoten{6}{-0.5}{n4}
\Knoten{6.45}{0.95}{n6}
\Knoten{6.7}{1.2}{v10}\nput{0}{v10}{$v$}
\Knoten{6.3}{-0.8}{v11}\nput{0}{v11}{$x$}

\ncline{-}{v1}{n1}
\Kante{n1}{v3}{a}
\Kante{n2}{v3}{b}
\ncline{-}{v2}{n2}
\Kante{v3}{v4}{c}
\Kante{v4}{v6}{d}
\Kante{v6}{v7}{f}
\Kante{v7}{n3}{g}
\Kante{n3}{n6}{h}
\ncline{-}{n6}{v10}
\Kante{v7}{n4}{i}
\ncline{-}{n4}{v11}
\Bogendashed{n1}{n3}{q_1}{35}
\Bogendashed{v4}{n4}{q_3}{-35}
\pscurve[linestyle=dashed](1.5,-0.5)(7,-0.8)(6.45,0.95)\rput(2.4,-0.8){\colorbox{almostwhite}{$\beta_1$}}
\Bogendashed{v6}{n6}{\beta_2}{20}
\end{pspicture}
}
\hspace{0.5cm}
\subfloat[NBC3g]{
 \begin{pspicture}(0.5,-1.7)(7,1.1) 
\Knoten{1.2}{0.8}{v1}\nput{180}{v1}{$u$}
\Knoten{1.5}{0.5}{n1}
\Knoten{1.2}{-0.8}{v2}\nput{180}{v2}{$w$}
\Knoten{1.5}{-0.5}{n2}
\Knoten{2}{0}{v3}
\Knoten{3.25}{0}{v4}
\Knoten{4.375}{0}{v6}
\Knoten{5.5}{0}{v7}
\Knoten{6}{0.5}{n3}
\Knoten{6}{-0.5}{n4}
\Knoten{6.45}{0.95}{n6}
\Knoten{6.7}{1.2}{v10}\nput{0}{v10}{$v$}
\Knoten{6.3}{-0.8}{v11}\nput{0}{v11}{$x$}

\ncline{-}{v1}{n1}
\Kante{n1}{v3}{a}
\Kante{n2}{v3}{b}
\ncline{-}{v2}{n2}
\Kante{v3}{v4}{c}
\Kante{v4}{v6}{d}
\Kante{v6}{v7}{f}
\Kante{v7}{n3}{g}
\Kante{n3}{n6}{h}
\ncline{-}{n6}{v10}
\Kante{v7}{n4}{i}
\ncline{-}{n4}{v11}
\Bogendashed{n1}{n6}{q_1}{20}
\Bogendashed{v4}{n4}{q_3}{-35}
\pscurve[linestyle=dashed](1.5,-0.5)(7,-0.8)(6,0.5)\rput(2.4,-0.8){\colorbox{almostwhite}{$\beta_1$}}
\Bogendashed{v6}{n3}{\beta_2}{25}
\end{pspicture}
}
\caption{NBC3}
\label{propNBC3}
\end{figure}
\begin{lemma} For NBC3 (cf. Figure~\ref{propNBC3}) we get the following properties:
\begin{enumerate}[itemsep=-0em,leftmargin=*,label=\roman*)]
\item For NBC3a, NBC3d and NBC3e: All edges which are substituted by $q_1, q_2$ or by $q_3$ are completely paid. 
\item 
For NBC3a and $(u,v,w,x)=(t_1,s_1,t_2,s_2)$: Player 2 pays all commonly used edges which are not substituted by $q_1$ $(c)$. 

\item For NBC3b, NBC3c, NBC3f and NBC3g: 
\begin{itemize}[itemsep=-0em,leftmargin=*]
		\item For $(u,v,w,x)=(s_1,t_1,s_2,t_2)$, Player 2 has a tight left alternative which substitutes all commonly used edges.
		\item For $(u,v,w,x)=(t_1,s_1,t_2,s_2)$, Player 2 has a tight right alternative which substitutes all commonly used edges. 
		\item For $(u,v,w,x)=(s_2,t_2,s_1,t_1)$, Player 2 pays all edges of the commonly used path which are not substituted by $q_2$ $(f)$. 
		\item For $(u,v,w,x)=(t_2,s_2,t_1,s_1)$, all edges of the commonly used path which are not substituted by $q_2$ $(f)$ have cost 0. 
	\end{itemize}
	\end{enumerate}
	\end{lemma}

\begin{proof}
We start with NBC3a. Since $q_1,q_2$ and $q_3$ are tight alternatives we get 
\[
	q_1=\xi_{u,d}+\xi_{u,f}+\xi_{u,g}+\xi_{u,h}, \ 
	\beta_1+\beta_2=\xi_{\ell,b}+\xi_{\ell,c}+\xi_{\ell,d}+\xi_{\ell,f} \text{ and }
	q_3=\xi_{\ell,f}+\xi_{\ell,g}+\xi_{\ell,i} 
	\]
	and (by adding)
	\[
q_1+\beta_1+\beta_2+q_3=\xi_{\ell,b}+\xi_{\ell,c}+d+f+\xi_{\ell,f}+g+\xi_{u,h}+\xi_{\ell,i}.
	\]
	Since $P_u \cup P_{\ell}$ is optimal, $\beta_1+\beta_2 \geq a+b+c+d+f$ and $q_1+q_3 \geq f+g+h+i$ has to hold. This implies $q_1+\beta_1+\beta_2+q_3 \geq a+b+c+d+2f+g+h+i$ and therefore 
$\xi_{\ell,b}=b, \xi_{\ell,c}=c, \xi_{u,h}=h$ and $\xi_{\ell,i}=i$ holds.

\vspace{1em}
Now we consider NBC3d. Adding the tight alternatives $q_1,q_2$ and $q_3$ yields
\[
q_1+\beta_1+\beta_2+q_3=\xi_{u,b}+\xi_{\ell,c}+d+f+\xi_{\ell,f}+g+\xi_{u,h}+\xi_{\ell,i}.
\]
Since $P_u \cup P_{\ell}$ is optimal, $\beta_1+\beta_2 \geq a+b+c+d+f$ and $q_1+q_3 \geq b+f+g+h+i$ and therefore $q_1+\beta_1+\beta_2+q_3 \geq a+2b+c+d+2f+g+h+i$ has to hold. This implies $\xi_{u,b}=b=0, \xi_{\ell,c}=c, \xi_{u,h}=h$ and $\xi_{\ell,i}=i$.

\vspace{1em}
We omit the proof for NBC3e since it it almost analogous to the proof for NBC3d.
	
	\vspace{1em}
	For NBC3b we get
	\[
	\beta_1+\beta_2=\xi_{\ell,a}+\xi_{\ell,b}+\xi_{\ell,c}+\xi_{\ell,d} 
	\]
	because $q_2$ is a tight alternative of the lower Player. Since $\beta_1$ followed by $h$ and $g$ is an alternative $q$ for the lower Player which substitutes $a,b,c,d$ and $f$, 
	\[
	\beta_1+g+h \geq \xi_{\ell,a}+\xi_{\ell,b}+\xi_{\ell,c}+\xi_{\ell,d}+\xi_{\ell,f}
	\]
	holds. Furthermore we get $\beta_2 \geq g+h$ because $P_u \cup P_{\ell}$ is optimal. Together this implies
	\[
	\xi_{\ell,a}+\xi_{\ell,b}+\xi_{\ell,c}+\xi_{\ell,d} =\beta_1+\beta_2 \geq \xi_{\ell,a}+\xi_{\ell,b}+\xi_{\ell,c}+\xi_{\ell,d}+\xi_{\ell,f}
	\]
	and therefore $\xi_{\ell,f}=0$, $q$ is tight and $\beta_2=g+h$.
	
	We now use $\beta_2 \geq \xi_{u,f}+\xi_{u,g}+\xi_{u,h}$. It is clear that this holds if $\beta_2$ is an alternative for the upper Player; if not, the statement is still true (we will prove this below). 
Therefore we get 
	\[
	g+h=\beta_2 \geq \xi_{u,f}+\xi_{u,g}+\xi_{u,h}=f+\xi_{u,g}+\xi_{u,h}
	\]
	and, if $\xi_{u,g}=g$ and $\xi_{u,h}=h$ holds, $f=0$. 

\vspace{1em}
We now show that $\beta_2 \geq \xi_{u,f}+\xi_{u,g}+\xi_{u,h}$ always holds.
Consider the path $\beta_2$ and assume that $\beta_2$ is not an alternative for the upper player. That means that $\beta_2$ has to contain at least one node of $L_u$ or $R_u$ (additionally to the endnode which is already in $R_u$). 
 
Assume that $\beta_2$ contains a node of $L_u$. Considering $\beta_2$ as directed from left to right, let $q$ be the subpath of $\beta_2$ from the last node which is contained in $L_u$ until the next node which is in $R_u$. This is an alternative for the upper Player which substitutes a subpath $\ell$ of $L_u$, all commonly used edges and a subpath $r$ of $R_u$. Subdivide $g$ and $h$ into $g_1,g_2$ and $h_1,h_2$, where $g_1$ and $h_1$ are the subpaths of $g$ and $h$ which are contained in $r$ (note that $g_1=g$, $h_1=h$ or $h_2=h$ is possible). 
We then get 
\[
q\geq\xi_{u,\ell}+\xi_{u,b}+\xi_{u,c}+\xi_{u,d}+\xi_{u,f}+\xi_{u,r}\geq \xi_{u,f}+\xi_{u,g_1}+\xi_{u,h_1}.
\] 
Furthermore, $\beta_2 \geq q+g_2+h_2$ has to hold, since the remaining subpath of $\beta_2$ after $q$ has cost at least $g_2+h_2$ ($P_u \cup P_{\ell}$ is optimal).  
Together with $g_2+h_2 \geq \xi_{u,g_2}+\xi_{u,h_2}$ this yields the desired inequality. 

Now assume that $\beta_2$ does not contain a node of $L_u$. 
Again considering $\beta_2$ as directed from left to right, let $q$ be the subpath from the beginning of $\beta_2$ until the first node which is contained in $R_u$. This is an alternative for the upper Player which substitutes $f$, a part $g_1$ of $g$ and a part $h_1$ of $h$ (again by subdividing $g$ in $g_1,g_2$ and $h$ in $h_1,h_2$). Using $q \geq \xi_{u,f}+\xi_{u,g_1}+\xi_{u,h_1}$ and $\beta_2 \geq q+g_2+h_2$ yields the desired result in this case as well.

	\vspace{1em}
	We omit the proofs for NBC3c, NBC3f and NBC3g since they are almost analogous to the proof for NBC3b.
\end{proof}


\begin{figure}[H] \centering \psset{unit=1cm}\captionsetup[subfigure]{labelformat=empty}
\subfloat[NBC4a]{
 \begin{pspicture}(0,-1.3)(7,0.7) 
\Knoten{0.2}{0.8}{v1}\nput{180}{v1}{$u$}
\Knoten{0.5}{0.5}{n1}
\Knoten{0.2}{-0.8}{v2}\nput{180}{v2}{$w$}
\Knoten{0.5}{-0.5}{n2}
\Knoten{1}{0}{v3}
\Knoten{2.125}{0}{n5}
\Knoten{3.25}{0}{v4}
\Knoten{4.375}{0}{v6}
\Knoten{5.5}{0}{v7}
\Knoten{6}{0.5}{n3}
\Knoten{6}{-0.5}{n4}
\Knoten{6.3}{0.8}{v10}\nput{0}{v10}{$v$}
\Knoten{6.3}{-0.8}{v11}\nput{0}{v11}{$x$}

\ncline{-}{v1}{n1}
\Kante{n1}{v3}{a}
\Kante{n2}{v3}{b}
\ncline{-}{v2}{n2}
\Kante{v3}{n5}{c}
\Kante{v4}{n5}{d}
\Kante{v4}{v6}{f}
\Kante{v6}{v7}{g}
\Kante{v7}{n3}{h}
\ncline{-}{n3}{v10}
\Kante{v7}{n4}{i}
\ncline{-}{n4}{v11}
\Bogendashed{n5}{n3}{q_1}{35}
\Bogendashed{n2}{v6}{q_2}{-35}
\Bogendashed{n1}{v4}{\gamma_2}{25}
\pscurve[linestyle=dashed](6,-0.5)(-0.5,-0.8)(0.5,0.5)\rput(4.375,-1){\colorbox{almostwhite}{$\gamma_1$}}
\end{pspicture}
}
\hspace{0.5cm}
\subfloat[NBC4b]{
 \begin{pspicture}(0,-1.3)(7.2,0.7) 
\Knoten{0.2}{0.8}{v1}\nput{180}{v1}{$u$}
\Knoten{0.2}{-0.8}{v2}\nput{180}{v2}{$w$}
\Knoten{0.5}{-0.5}{n2}
\Knoten{1}{0}{v3}
\Knoten{2.125}{0}{n5}
\Knoten{3.25}{0}{v4}
\Knoten{4.375}{0}{v6}
\Knoten{5.5}{0}{v7}
\Knoten{6}{0.5}{n3}
\Knoten{6.4}{0.9}{n6}
\Knoten{6}{-0.5}{n4}
\Knoten{6.7}{1.2}{v10}\nput{0}{v10}{$v$}
\Knoten{6.3}{-0.8}{v11}\nput{0}{v11}{$x$}

\ncline{-}{v1}{v3}
\ncline{-}{n2}{v3}
\ncline{-}{v2}{n2}
\ncline{-}{v3}{n5}
\Kante{v4}{n5}{a}
\Kante{v4}{v6}{b}
\Kante{v6}{v7}{c}
\Kante{v7}{n3}{d}
\Kante{n3}{n6}{f}
\ncline{-}{n6}{v10}
\Kante{v7}{n4}{g}
\ncline{-}{n4}{v11}
\Bogendashed{n2}{v6}{q_2}{-35}
\Bogendashed{n5}{n3}{q_1}{40}
\Bogendashed{v4}{n6}{\gamma_2}{15}
\Bogendashed{n4}{n6}{\gamma_1}{-30}
\end{pspicture}
}

\subfloat[NBC4c]{
 \begin{pspicture}(0,-1)(7,1) 
\Knoten{0.2}{0.8}{v1}\nput{180}{v1}{$u$}
\Knoten{0.2}{-0.8}{v2}\nput{180}{v2}{$w$}
\Knoten{0.5}{-0.5}{n2}
\Knoten{1}{0}{v3}
\Knoten{2.125}{0}{n5}
\Knoten{3.25}{0}{v4}
\Knoten{4.375}{0}{v6}
\Knoten{5.5}{0}{v7}
\Knoten{6}{0.5}{n3}
\Knoten{6.25}{0.75}{n6}
\Knoten{6}{-0.5}{n4}
\Knoten{6.5}{1}{v10}\nput{0}{v10}{$v$}
\Knoten{6.3}{-0.8}{v11}\nput{0}{v11}{$x$}

\ncline{-}{v1}{v3}
\ncline{-}{n2}{v3}
\ncline{-}{v2}{n2}
\ncline{-}{v3}{n5}
\ncline{-}{v4}{n5}
\Kante{v4}{v6}{a}
\Kante{v6}{v7}{b}
\Kante{v7}{n3}{c}
\ncline{-}{n3}{v10}
\Kante{v7}{n4}{d}
\ncline{-}{n4}{v11}
\Bogendashed{n2}{v6}{q_2}{-35}
\Bogendashed{n5}{n6}{q_1}{35}
\Bogendashed{v4}{n3}{\gamma_2}{20}
\Bogendashed{n4}{n3}{\gamma_1}{-30}
\end{pspicture}
}

\subfloat[NBC4d]{
 \begin{pspicture}(0.5,-1.4)(7.2,1) 
\Knoten{0.8}{1.2}{v1}\nput{180}{v1}{$u$}
\Knoten{1.5}{0.5}{n1}
\Knoten{1.05}{0.95}{n6}
\Knoten{1.2}{-0.8}{v2}\nput{180}{v2}{$w$}
\Knoten{1.5}{-0.5}{n2}
\Knoten{2}{0}{v3}
\Knoten{3.25}{0}{v4}
\Knoten{4.375}{0}{v6}
\Knoten{5.5}{0}{v7}
\Knoten{6}{0.5}{n3}
\Knoten{6}{-0.5}{n4}
\Knoten{6.3}{0.8}{v10}\nput{0}{v10}{$v$}
\Knoten{6.3}{-0.8}{v11}\nput{0}{v11}{$x$}

\ncline{-}{v1}{n6}
\Kante{n6}{n1}{a}
\Kante{n1}{v3}{b}
\Kante{n2}{v3}{c}
\ncline{-}{v2}{n2}
\Kante{v3}{v4}{d}
\Kante{v4}{v6}{f}
\Kante{v6}{v7}{g}
\Kante{v7}{n3}{h}
\ncline{-}{n3}{v10}
\Kante{v7}{n4}{i}
\ncline{-}{n4}{v11}
\Bogendashed{n1}{n3}{q_1}{35}
\Bogendashed{n2}{v6}{q_2}{-35}
\Bogendashed{n6}{v4}{\gamma_2}{25}
\pscurve[linestyle=dashed](6,-0.5)(0.5,-0.8)(1.05,0.95)\rput(4.375,-1){\colorbox{almostwhite}{$\gamma_1$}}
\end{pspicture}
}
\hspace{1cm}
\subfloat[NBC4e]{
 \begin{pspicture}(0.5,-1.4)(7,1) 
\Knoten{0.8}{1.2}{v1}\nput{180}{v1}{$u$}
\Knoten{1.5}{0.5}{n1}
\Knoten{1.05}{0.95}{n6}
\Knoten{1.2}{-0.8}{v2}\nput{180}{v2}{$w$}
\Knoten{1.5}{-0.5}{n2}
\Knoten{2}{0}{v3}
\Knoten{3.25}{0}{v4}
\Knoten{4.375}{0}{v6}
\Knoten{5.5}{0}{v7}
\Knoten{6}{0.5}{n3}
\Knoten{6}{-0.5}{n4}
\Knoten{6.3}{0.8}{v10}\nput{0}{v10}{$v$}
\Knoten{6.3}{-0.8}{v11}\nput{0}{v11}{$x$}

\ncline{-}{v1}{n6}
\Kante{n6}{n1}{a}
\Kante{n1}{v3}{b}
\Kante{n2}{v3}{c}
\ncline{-}{v2}{n2}
\Kante{v3}{v4}{d}
\Kante{v4}{v6}{f}
\Kante{v6}{v7}{g}
\Kante{v7}{n3}{h}
\ncline{-}{n3}{v10}
\Kante{v7}{n4}{i}
\ncline{-}{n4}{v11}
\Bogendashed{n6}{n3}{q_1}{25}
\Bogendashed{n2}{v6}{q_2}{-35}
\Bogendashed{n1}{v4}{\gamma_2}{25}
\pscurve[linestyle=dashed](6,-0.5)(0.5,-0.8)(1.5,0.5)\rput(4.375,-1){\colorbox{almostwhite}{$\gamma_1$}}
\end{pspicture}
}
\end{figure}

\begin{figure}[H] \centering \psset{unit=1cm}\captionsetup[subfigure]{labelformat=empty}
\subfloat[NBC4f]{
 \begin{pspicture}(0.5,-1)(7,1) 
\Knoten{1.2}{0.8}{v1}\nput{180}{v1}{$u$}
\Knoten{1.5}{0.5}{n1}
\Knoten{1.2}{-0.8}{v2}\nput{180}{v2}{$w$}
\Knoten{1.5}{-0.5}{n2}
\Knoten{2}{0}{v3}
\Knoten{3.25}{0}{v4}
\Knoten{4.375}{0}{v6}
\Knoten{5.5}{0}{v7}
\Knoten{6}{0.5}{n3}
\Knoten{6.45}{0.95}{n6}
\Knoten{6}{-0.5}{n4}
\Knoten{6.7}{1.2}{v10}\nput{0}{v10}{$v$}
\Knoten{6.3}{-0.8}{v11}\nput{0}{v11}{$x$}

\ncline{-}{v1}{n1}
\Kante{n1}{v3}{a}
\Kante{n2}{v3}{b}
\ncline{-}{v2}{n2}
\Kante{v3}{v4}{c}
\Kante{v4}{v6}{d}
\Kante{v6}{v7}{f}
\Kante{v7}{n3}{g}
\Kante{n3}{n6}{h}
\ncline{-}{n6}{v10}
\Kante{v7}{n4}{i}
\ncline{-}{n4}{v11}
\Bogendashed{n2}{v6}{q_2}{-35}
\Bogendashed{n1}{n3}{q_1}{35}
\Bogendashed{v4}{n6}{\gamma_2}{20}
\Bogendashed{n4}{n6}{\gamma_1}{-30}
\end{pspicture}
}
\hspace{1cm}
\subfloat[NBC4g]{
 \begin{pspicture}(0.5,-1)(7,1) 
\Knoten{1.2}{0.8}{v1}\nput{180}{v1}{$u$}
\Knoten{1.5}{0.5}{n1}
\Knoten{1.2}{-0.8}{v2}\nput{180}{v2}{$w$}
\Knoten{1.5}{-0.5}{n2}
\Knoten{2}{0}{v3}
\Knoten{3.25}{0}{v4}
\Knoten{4.375}{0}{v6}
\Knoten{5.5}{0}{v7}
\Knoten{6}{0.5}{n3}
\Knoten{6.45}{0.95}{n6}
\Knoten{6}{-0.5}{n4}
\Knoten{6.7}{1.2}{v10}\nput{0}{v10}{$v$}
\Knoten{6.3}{-0.8}{v11}\nput{0}{v11}{$x$}

\ncline{-}{v1}{n1}
\Kante{n1}{v3}{a}
\Kante{n2}{v3}{b}
\ncline{-}{v2}{n2}
\Kante{v3}{v4}{c}
\Kante{v4}{v6}{d}
\Kante{v6}{v7}{f}
\Kante{v7}{n3}{g}
\Kante{n3}{n6}{h}
\ncline{-}{n6}{v10}
\Kante{v7}{n4}{i}
\ncline{-}{n4}{v11}

\Bogendashed{n2}{v6}{q_2}{-35}
\Bogendashed{n1}{n6}{q_1}{25}
\Bogendashed{v4}{n3}{\gamma_2}{20}
\Bogendashed{n4}{n3}{\gamma_1}{-30}
\end{pspicture}
}
\caption{NBC4}
\label{propNBC4}
\end{figure}
\begin{lemma} For NBC4 (cf. Figure~\ref{propNBC4}) we get the following properties:
\begin{enumerate}[itemsep=-0em,leftmargin=*,label=\roman*)]
\item For NBC4a:
\begin{itemize}[itemsep=-0em,leftmargin=*]
	\item All edges which are substituted by $q_1$ or by $q_3$ are completely paid.
	\item For $(u,v,w,x)=(t_1,s_1,t_2,s_2)$, Player 2 has a tight left alternative which substitutes all commonly used edges. 
	\item For $(u,v,w,x)=(t_2,s_2,t_1,s_1)$, Player 2 pays all commonly used edges which are not substituted by $q_3$ (c,d). 
	\end{itemize}

\item 
For NBC4b:
\begin{itemize}[itemsep=-0em,leftmargin=*]
	\item All edges which are substituted by $q_1$ or by $q_3$ are completely paid.
	\item All commonly used edges which are substituted by $q_3$ $(b,c)$ have cost 0. 
\end{itemize}

\item For NBC4c:
\begin{itemize}[itemsep=-0em,leftmargin=*]
	\item All edges which are substituted by $q_3$ are completely paid.
	\item For $(u,v,w,x) \in \{(s_1,t_1,s_2,t_2), (t_1,s_1,t_2,s_2)\}$, Player 2 pays all commonly used edges which are substituted by $q_3$ $(a,b)$.
	\item For $(u,v,w,x) \in \{(s_2,t_2,s_1,t_1), (t_2,s_2,t_1,s_1)\}$, Player 1 pays all commonly used edges which are substituted by $q_3$ $(a,b)$.
\end{itemize}

\item For NBC4d and NBC4e: 
\begin{itemize}[itemsep=-0em,leftmargin=*]
		\item For $(u,v,w,x)=(s_1,t_1,s_2,t_2)$, Player 2 has a tight right alternative which substitutes all commonly used edges. 
		\item For $(u,v,w,x)=(s_2,t_2,s_1,t_1)$, Player 1 has a tight right alternative which substitutes all commonly used edges.
		\item For $(u,v,w,x)=(t_2,s_2,t_1,s_1)$, Player 2 pays all edges of the commonly used path which are not substituted by $q_3$ $(d)$. 
	\end{itemize}
	
	
\item For NBC4f and NBC4g: All edges which are substituted by $q_1, q_2$ or by $q_3$ are completely paid.
\end{enumerate}
\end{lemma}
\begin{proof}
We start with NBC4a. Since $q_3$ is a tight alternative of the lower player, 
\[
\gamma_1+\gamma_2=\xi_{\ell,f}+\xi_{\ell,g}+\xi_{\ell,i}
\]
holds. 
Furthermore, $\gamma_1$ followed by $a$ is an alternative $q$ for the lower player which substitutes $c,d,f,g$ and $i$. This implies
\[
\gamma_1+a \geq \xi_{\ell,c}+\xi_{\ell,d}+\xi_{\ell,f}+\xi_{\ell,g}+\xi_{\ell,i}.
\]
Using $\gamma_2 \geq a$  ($P_u \cup P_{\ell}$ is optimal) yields that $q$ is tight and $\xi_{\ell,c}=\xi_{\ell,d}=0$ holds. \\
Since $q_1$ is also tight we get
\[
q_1+\gamma_1+\gamma_2=d+f+g+\xi_{u,h}+\xi_{\ell,i}
\]
and the optimality of $P_u \cup P_{\ell}$ implies $q_1 +\gamma_1 \geq d+f+g+h+i$. This results in $\xi_{u,h}=h$ and $\xi_{\ell,i}=i$. 

\vspace{1em}
For NBC4b we get
\[
\xi_{\ell,b}+\xi_{\ell,c}+\xi_{\ell,g}=\gamma_1+\gamma_2 \geq b+c+d+f+g 
\]
since $q_3$ is tight and $P_u \cup P_{\ell}$ is optimal and therefore $\xi_{\ell,b}=b,\xi_{\ell,c}=c,\xi_{u,d}=d=0$ and $\xi_{\ell,g}=g$ holds. Furthermore 
\[
\xi_{u,a}=\xi_{u,a}+\xi_{u,b}+\xi_{u,c}+\xi_{u,d}=q_1\geq a+b+c
\]
and thus $b=c=0$ holds.

\vspace{1em}
For NBC4c we use
\[
\xi_{\ell,a}+\xi_{\ell,b}+\xi_{\ell,d} =\gamma_1+\gamma_2 \geq a+b+c+d 
\]
to get $\xi_{\ell,a}=a, \xi_{\ell,b}=b$ and $\xi_{\ell,d}=d$.

\vspace{1em}
The properties for NBC4d, NBC4e, NBC4f and NBC4g follow from the properties of the corresponding cases of NBC3 (by symmetry).
\end{proof}
	
	\begin{figure}[H] \centering \psset{unit=1.2cm}\captionsetup[subfigure]{labelformat=empty}
	\subfloat[NBC5a]{
	\begin{pspicture}(1.5,-1)(7.5,1) 
\Knoten{2}{1}{s1}\nput{180}{s1}{$u$}
\Knoten{2}{-1}{s2}\nput{180}{s2}{$w$}
\Knoten{2.5}{-0.5}{v1}
\Knoten{3}{0}{v2}
\Knoten{4.25}{0}{v3}

\Knoten{5.5}{0}{v6}
\Knoten{6}{0.5}{v9}
\Knoten{6.5}{1}{t1}\nput{0}{t1}{$v$}
\Knoten{6}{-0.5}{v10}
\Knoten{6.5}{-1}{t2}\nput{0}{t2}{$x$}
\ncline{s1}{v2}
\ncline{s2}{v1}
\Kante{v1}{v2}{a}
\Kante{v2}{v3}{b}
\Kante{v3}{v6}{c}
\Kante{v6}{v9}{d}
\Kante{v6}{v10}{f}
\ncline{v9}{t1}
\ncline{v10}{t2}
\Bogendashed{v1}{v3}{\alpha_1}{-10}
\Bogendashed{v10}{v9}{\alpha_3}{-15}
\Bogendashed{v1}{v10}{\alpha_2}{-15}
\end{pspicture}
	}
	\subfloat[NBC5b]{
	
	\begin{pspicture}(-1,-1)(5,1.2) 
\Knoten{0.2}{0.8}{s1}\nput{180}{s1}{$u$}
\Knoten{0}{-1}{s2}\nput{180}{s2}{$w$}
\Knoten{0.5}{-0.5}{v1}
\Knoten{1}{0}{v2}
\Knoten{2.25}{0}{v3}
\Knoten{3.5}{0}{v6}

\Knoten{4}{0.5}{v9}
\Knoten{4.5}{1}{t1}\nput{0}{t1}{$v$}
\Knoten{4}{-0.5}{v10}
\Knoten{4.5}{-1}{t2}\nput{0}{t2}{$x$}
\ncline{s1}{v2}
\ncline{s2}{v1}
\Kante{v1}{v2}{a}
\Kante{v2}{v3}{b}
\Kante{v3}{v6}{c}
\Kante{v6}{v9}{d}
\Kante{v6}{v10}{f}
\ncline{v9}{t1}
\ncline{v10}{t2}
\Bogendashed{v3}{v10}{\alpha_1}{-10}
\pscurve[linestyle=dashed](0.5,-0.5)(-0.8,0.7)(4,0.5)\rput(1.4,1.2){\colorbox{almostwhite}{$\alpha_3$}}
\Bogendashed{v1}{v10}{\alpha_2}{-15}
\end{pspicture}
	}
	
	\subfloat[NBC5c]{
	\begin{pspicture}(1.5,-1)(7.5,1.2) 
	\Knoten{2}{1}{s1}\nput{180}{s1}{$u$}
\Knoten{2}{-1}{s2}\nput{180}{s2}{$w$}
\Knoten{2.5}{-0.5}{v1}
\Knoten{2.5}{0.5}{n1}
\Knoten{3}{0}{v2}
\Knoten{5.5}{0}{v6}
\Knoten{6}{0.5}{v9}
\Knoten{6.5}{1}{t1}\nput{0}{t1}{$v$}
\Knoten{6}{-0.5}{v10}
\Knoten{6.5}{-1}{t2}\nput{0}{t2}{$x$}
\ncline{s1}{n1}
\Kante{n1}{v2}{a}
\ncline{s2}{v1}
\Kante{v1}{v2}{b}
\Kante{v2}{v6}{c}
\Kante{v6}{v9}{d}
\Kante{v6}{v10}{f}
\ncline{v9}{t1}
\ncline{v10}{t2}
\Bogendashed{v1}{v10}{\alpha_2}{-15}
\Bogendashed{v1}{n1}{\alpha_1}{15}
\Bogendashed{v10}{v9}{\alpha_3}{-15}
\end{pspicture}
	}
	\subfloat[NBC5d]{
	
		\begin{pspicture}(1,-1)(7.5,1.2) 
\Knoten{2.2}{0.8}{s1}\nput{180}{s1}{$u$}
\Knoten{2}{-1}{s2}\nput{180}{s2}{$w$}
\Knoten{2.5}{-0.5}{v1}
\Knoten{2.5}{0.5}{n1}
\Knoten{3}{0}{v2}
\Knoten{5.5}{0}{v6}
\Knoten{6}{0.5}{v9}
\Knoten{6.3}{0.8}{t1}\nput{0}{t1}{$v$}
\Knoten{6}{-0.5}{v10}
\Knoten{6.5}{-1}{t2}\nput{0}{t2}{$x$}
\ncline{s1}{n1}
\Kante{n1}{v2}{a}
\ncline{s2}{v1}
\Kante{v1}{v2}{b}
\Kante{v2}{v6}{c}
\Kante{v6}{v9}{d}
\Kante{v6}{v10}{f}
\ncline{v9}{t1}
\ncline{v10}{t2}
\pscurve[linestyle=dashed](2.5,-0.5)(1.5,0.7)(6,0.5)\rput(1.7,0){\colorbox{almostwhite}{$\alpha_3$}}
\pscurve[linestyle=dashed](6,-0.5)(7.1,0.7)(2.5,0.5)\rput(6,1.2){\colorbox{almostwhite}{$\alpha_1$}}
\Bogendashed{v1}{v10}{\alpha_2}{-15}
\end{pspicture}
	}
	\caption{NBC5} \label{propNBC5}
	\end{figure}
	\begin{lemma}For NBC5 (cf. Figure~\ref{propNBC5}) we get the following properties:
	\begin{enumerate}[itemsep=-0em,leftmargin=*,label=\roman*)]
	\item All edges which are substituted by $q_1$ are completely paid.
		\item For $(u,v,w,x) \in \{(s_1,t_1,s_2,t_2), (t_1,s_1,t_2,s_2)\}$, all commonly used edges are completely paid by Player 1. 
	\item For $(u,v,w,x) \in \{(s_2,t_2,s_1,t_1), (t_2,s_2,t_1,s_1)\}$,  all commonly used edges are completely paid by Player 2. 
	\end{enumerate}
	\end{lemma}
	\begin{proof}

	%
	
	For NBC5a and NBC5b we get $\alpha_1+\alpha_2+\alpha_3=\xi_{u,c}+\xi_{u,d}$ since $q_1$ is tight. The optimality of $P_u \cup P_{\ell}$ implies $\alpha_2 \geq b+c$	and $\alpha_3 \geq d$ and therefore $\alpha_2+\alpha_3 \geq b+c+d$. Altogether we get $\xi_{u,c}=c, \xi_{u,d}=d$ and $\xi_{u,b}=b=0$.
	
	\vspace{1em}
	For NBC5c and NBC5d we use $\alpha_1+\alpha_2+\alpha_3=\xi_{u,a}+\xi_{u,c}+\xi_{u,d}$, $\alpha_1 \geq a$, $\alpha_2 \geq c$ and	$\alpha_3 \geq d$ to get $\xi_{u,a}=a, \xi_{u,c}=c$ and $\xi_{u,d}=d$.
	\end{proof}

	\begin{figure}[H] \centering \psset{unit=1.06cm}\captionsetup[subfigure]{labelformat=empty}
	\subfloat[NBC6a]{
	\begin{pspicture}(-0.5,-1)(7,1) 
\Knoten{0.2}{0.8}{v1}\nput{180}{v1}{$u$}
\Knoten{0.2}{-0.8}{v2}\nput{180}{v2}{$w$}
\Knoten{0.5}{-0.5}{n2}
\Knoten{1}{0}{v3}
\Knoten{2.125}{0}{n5}
\Knoten{3.25}{0}{v4}
\Knoten{4.375}{0}{v6}
\Knoten{5.5}{0}{v7}
\Knoten{6}{0.5}{n3}
\Knoten{6}{-0.5}{n4}
\Knoten{6.45}{-0.95}{n6}
\Knoten{6.3}{0.8}{v10}\nput{0}{v10}{$v$}
\Knoten{6.7}{-1.2}{v11}\nput{0}{v11}{$x$}

\ncline{-}{v1}{v3}
\ncline{n2}{v3}
\ncline{-}{v2}{n2}
\ncline{v3}{n5}
\Kante{v4}{n5}{a}
\Kante{v4}{v6}{b}
\Kante{v6}{v7}{c}
\Kante{v7}{n3}{d}
\ncline{-}{n3}{v10}
\Kante{v7}{n4}{f}
\Kante{n4}{n6}{g}
\ncline{-}{n6}{v11}
\Bogendashed{n6}{n3}{\alpha_2}{-20}
\pscurve[linestyle=dashed](6.45,-0.95)(7.1,1)(2.125,0)\rput(4.4,1){\colorbox{almostwhite}{$\alpha_1$}}
\Bogendashed{n2}{v6}{q_2}{-35}
\Bogendashed{v4}{n4}{q_3}{-35}
\end{pspicture}
	}
	\subfloat[NBC6b]{
 	\begin{pspicture}(-0.5,-1)(7,1) 
\Knoten{0.2}{0.8}{v1}\nput{180}{v1}{$u$}
\Knoten{0.2}{-0.8}{v2}\nput{180}{v2}{$w$}
\Knoten{0.5}{-0.5}{n2}
\Knoten{1}{0}{v3}
\Knoten{2.125}{0}{n5}
\Knoten{3.25}{0}{v4}
\Knoten{4.375}{0}{v6}
\Knoten{5.5}{0}{v7}
\Knoten{6}{0.5}{n3}
\Knoten{6}{-0.5}{n4}
\Knoten{6.45}{-0.95}{n6}
\Knoten{6.3}{0.8}{v10}\nput{0}{v10}{$v$}
\Knoten{6.7}{-1.2}{v11}\nput{0}{v11}{$x$}

\ncline{-}{v1}{v3}
\ncline{n2}{v3}
\ncline{-}{v2}{n2}
\ncline{v3}{n5}
\Kante{v4}{n5}{a}
\Kante{v4}{v6}{b}
\Kante{v6}{v7}{c}
\Kante{v7}{n3}{d}
\ncline{-}{n3}{v10}
\Kante{v7}{n4}{f}
\Kante{n4}{n6}{g}
\ncline{-}{n6}{v11}
\Bogendashed{n4}{n3}{\alpha_2}{-20}
\pscurve[linestyle=dashed](6,-0.5)(7.1,1)(2.125,0)\rput(4.4,1){\colorbox{almostwhite}{$\alpha_1$}}
\Bogendashed{n2}{v6}{q_2}{-35}
\Bogendashed{v4}{n6}{q_3}{-35}
\end{pspicture}
	}
	\caption{NBC6}\label{propNBC6}
	\end{figure}
	\begin{lemma}For NBC6 (cf. Figure~\ref{propNBC6}) we get the following properties:
	\begin{enumerate}[itemsep=-0em,leftmargin=*,label=\roman*)]
		\item All edges which are substituted by $q_1$ or by $q_3$ are completely paid.
		\item For $(u,v,w,x) \in \{(s_1,t_1,s_2,t_2), (t_1,s_1,t_2,s_2)\}$, Player 1 pays all commonly used edges which are substituted by $q_1$ $(a,b,c)$. 
	\item For $(u,v,w,x) \in \{(s_2,t_2,s_1,t_1), (t_2,s_2,t_1,s_1)\}$,  Player 2 pays all commonly used edges which are substituted by $q_1$ $(a,b,c)$. 
	\end{enumerate}
	\end{lemma}
	
	\begin{proof}
	We start with NBC6a and get
	\[
	\xi_{u,a}+\xi_{u,b}+\xi_{u,c}+\xi_{u,d}= \alpha_1+\alpha_2 \geq a+b+c+d+f+g 
	\]
	since $q_1$ is tight and $P_u \cup P_{\ell}$ is optimal. Therefore $\xi_{u,a}=a, \xi_{u,b}=b, \xi_{u,c}=c, \xi_{u,d}=d$ and $\xi_{\ell,f}=f=0$ holds.
	
	\vspace{1em}
	For NBC6b we analogously get that $\xi_{u,a}=a, \xi_{u,b}=b, \xi_{u,c}=c, \xi_{u,d}=d$ and $\xi_{\ell,f}=f=0$ holds. Together with 
	\[
	\xi_{\ell,g}=\xi_{\ell,b}+\xi_{\ell,c}+\xi_{\ell,f}+\xi_{\ell,g}=q_3 \geq f+g=g
	\]
	this shows $\xi_{\ell,g}=g$.
	\end{proof}
	
	\begin{figure}[H] \centering \psset{unit=1.3cm}
	\begin{pspicture}(0.5,-1)(7.5,1) 
\Knoten{1.2}{0.8}{v1}\nput{180}{v1}{$u$}
\Knoten{1.5}{0.5}{n1}
\Knoten{1.2}{-0.8}{v2}\nput{180}{v2}{$w$}
\Knoten{1.5}{-0.5}{n2}
\Knoten{2}{0}{v3}
\Knoten{3.25}{0}{v4}
\Knoten{4.375}{0}{v6}
\Knoten{5.5}{0}{v7}
\Knoten{6}{0.5}{n3}
\Knoten{5.9}{-0.4}{n4}
\Knoten{6.25}{-0.75}{n5}
\Knoten{6.3}{0.8}{v10}\nput{0}{v10}{$v$}
\Knoten{6.5}{-1}{v11}\nput{0}{v11}{$x$}
\ncline{v1}{n1}
\Kante{n1}{v3}{a}
\Kante{n2}{v3}{b}
\ncline{-}{v2}{n2}
\Kante{v3}{v4}{c}
\Kante{v4}{v6}{d}
\Kante{v6}{v7}{f}
\Kante{v7}{n3}{g}
\ncline{-}{n3}{v10}
\Kante{v7}{n4}{h}
\Kante{n4}{n5}{i}
\ncline{-}{n5}{v11}
\Bogendashed{n5}{n3}{\alpha_2}{-20}
\pscurve[linestyle=dashed](6.25,-0.75)(7.1,0.6)(1.5,0.5)\rput(4.4,1.2){\colorbox{almostwhite}{$\alpha_1$}}
\Bogendashed{n2}{v6}{q_2}{-35}
\Bogendashed{v4}{n4}{q_3}{-35}
\end{pspicture}
\caption{NBC7}\label{propNBC7}
\end{figure}
\begin{lemma}For NBC7 (cf. Figure~\ref{propNBC7}) we get that all edges which are substituted by $q_1, q_2$ or by $q_3$ are completely paid.\end{lemma}
\begin{proof}
Using that $q_1, q_2$ and $q_3$ are tight yields
\[
\alpha_1+\alpha_2+q_2+q_3=\xi_{u,a}+\xi_{\ell,b}+c+d+\xi_{\ell,d}+f+\xi_{u,g}+\xi_{\ell,h}.
\]
The optimality of $P_u \cup P_{\ell}$ implies $\alpha_1+q_2 \geq a+b+c+d$ and 
	$\alpha_2+q_3 \geq d+f+g+h$ and therefore 
	\[
	\alpha_1+\alpha_2+q_2+q_3 \geq a+b+c+2d+f+g+h.
	\]
Altogether we get $\xi_{u,a}=a, \xi_{\ell,b}=b, \xi_{u,g}=g$ and $\xi_{\ell,h}=h$.
	\end{proof}
	
	
\begin{figure}[H] \centering \psset{unit=1cm}\captionsetup[subfigure]{labelformat=empty}
\subfloat[NBC8a]{
\begin{pspicture}(0.5,-1)(7.5,1.2) 
\Knoten{1.2}{0.8}{v1}\nput{180}{v1}{$u$}
\Knoten{1.5}{0.5}{n1}
\Knoten{1.2}{-0.8}{v2}\nput{180}{v2}{$w$}
\Knoten{1.5}{-0.5}{n2}
\Knoten{2}{0}{v3}
\Knoten{3.25}{0}{v4}
\Knoten{4.375}{0}{v6}
\Knoten{4.375}{1}{n7}
\Knoten{5.5}{0}{v7}
\Knoten{6}{0.5}{n3}
\Knoten{6}{-0.5}{n4}
\Knoten{6.4}{-0.9}{n6}
\Knoten{6.3}{0.8}{v10}\nput{0}{v10}{$v$}
\Knoten{6.7}{-1.2}{v11}\nput{0}{v11}{$x$}

\ncline{-}{v1}{n1}
\Kante{n1}{v3}{a}
\Kante{n2}{v3}{b}
\ncline{-}{v2}{n2}
\Kante{v3}{v4}{c}
\Kante{v4}{v6}{d}
\Kante{v6}{v7}{f}
\Kante{v7}{n3}{g}
\ncline{-}{n3}{v10}
\Kante{v7}{n4}{h}
\Kante{n4}{n6}{i}
\ncline{-}{n6}{v11}
\Kantedashed{n7}{v6}{\beta_2}
\Bogendashed{n4}{n3}{\alpha_3}{-20}
\Kantedashed{n7}{n1}{\alpha_1}
\pscurve[linestyle=dashed](6,-0.5)(7.1,1)(4.375,1)\rput(5.2,1.2){\colorbox{almostwhite}{$\alpha_2$}}
\pscurve[linestyle=dashed](1.5,-0.5)(0.5,0.8)(4.375,1)\rput(2.2,1.3){\colorbox{almostwhite}{$\beta_1$}}
\Bogendashed{v4}{n6}{q_3}{-25}
\end{pspicture}
} 
\hspace{1cm}
\subfloat[NBC8b]{
	\begin{pspicture}(0.5,-1)(7.5,1.2) 
\Knoten{1.2}{0.8}{v1}\nput{180}{v1}{$u$}
\Knoten{1.5}{0.5}{n1}
\Knoten{1.2}{-0.8}{v2}\nput{180}{v2}{$w$}
\Knoten{1.5}{-0.5}{n2}
\Knoten{2}{0}{v3}
\Knoten{3.25}{0}{v4}
\Knoten{4.375}{0}{v6}
\Knoten{4.375}{1}{n7}
\Knoten{5.5}{0}{v7}
\Knoten{6}{0.5}{n3}
\Knoten{6}{-0.5}{n4}
\Knoten{6.4}{-0.9}{n6}
\Knoten{6.3}{0.8}{v10}\nput{0}{v10}{$v$}
\Knoten{6.7}{-1.2}{v11}\nput{0}{v11}{$x$}

\ncline{-}{v1}{n1}
\Kante{n1}{v3}{a}
\Kante{n2}{v3}{b}
\ncline{-}{v2}{n2}
\Kante{v3}{v4}{c}
\Kante{v4}{v6}{d}
\Kante{v6}{v7}{f}
\Kante{v7}{n3}{g}
\ncline{-}{n3}{v10}
\Kante{v7}{n4}{h}
\Kante{n4}{n6}{i}
\ncline{-}{n6}{v11}
\Kantedashed{n7}{v6}{\beta_2}
\Kantedashed{n7}{n3}{\alpha_3}

\pscurve[linestyle=dashed](6,-0.5)(7.4,1.3)(1.5,0.5)\rput(4.2,1.6){\colorbox{almostwhite}{$\alpha_1$}}
\pscurve[linestyle=dashed](6,-0.5)(7.1,1)(4.375,1)\rput(6.25,1.3){\colorbox{almostwhite}{$\alpha_2$}}
\pscurve[linestyle=dashed](1.5,-0.5)(0.5,0.8)(4.375,1)\rput(2.2,1.3){\colorbox{almostwhite}{$\beta_1$}}
\Bogendashed{v4}{n6}{q_3}{-25}
\end{pspicture}
}

\subfloat[NBC8c]{
	\begin{pspicture}(0.5,-1)(7.5,1.3) 
\Knoten{1.2}{0.8}{v1}\nput{180}{v1}{$u$}
\Knoten{1.5}{0.5}{n1}
\Knoten{0.8}{-1.2}{v2}\nput{180}{v2}{$w$}
\Knoten{1.1}{-0.9}{n6}
\Knoten{1.5}{-0.5}{n2}
\Knoten{2}{0}{v3}
\Knoten{3.25}{0}{v4}
\Knoten{4.375}{0}{v6}
\Knoten{3.25}{0.9}{n7}
\Knoten{5.5}{0}{v7}
\Knoten{6}{0.5}{n3}
\Knoten{6}{-0.5}{n4}
\Knoten{6.3}{0.8}{v10}\nput{0}{v10}{$v$}
\Knoten{6.5}{-1}{v11}\nput{0}{v11}{$x$}

\ncline{-}{v1}{n1}
\Kante{n1}{v3}{a}
\Kante{n2}{v3}{c}
\ncline{-}{v2}{n6}
\Kante{n6}{n2}{b}
\Kante{v3}{v4}{d}
\Kante{v4}{v6}{f}
\Kante{v6}{v7}{g}
\Kante{v7}{n3}{h}
\ncline{-}{n3}{v10}
\Kante{v7}{n4}{i}
\ncline{-}{n4}{v11}
\Kantedashed{n7}{v4}{\gamma_2}
\Kantedashed{n1}{n7}{\alpha_1}

\pscurve[linestyle=dashed](1.5,-0.5)(0.5,0.8)(3.25,0.9)\rput(2,1.1){\colorbox{almostwhite}{$\alpha_2$}}
\pscurve[linestyle=dashed](6,0.5)(0.3,0.8)(1.5,-0.5)\rput(5.3,0.8){\colorbox{almostwhite}{$\alpha_3$}}
\pscurve[linestyle=dashed](6,-0.5)(7.1,1)(3.25,0.9)\rput(6.25,1.4){\colorbox{almostwhite}{$\gamma_1$}}
\Bogendashed{n6}{v6}{q_2}{-25}
\end{pspicture}
}
\hspace{1cm}
\subfloat[NBC8d]{
	\begin{pspicture}(0.5,-1)(7.5,1.3) 
\Knoten{1.2}{0.8}{v1}\nput{180}{v1}{$u$}
\Knoten{1.5}{0.5}{n1}
\Knoten{0.8}{-1.2}{v2}\nput{180}{v2}{$w$}
\Knoten{1.1}{-0.9}{n6}
\Knoten{1.5}{-0.5}{n2}
\Knoten{2}{0}{v3}
\Knoten{3.25}{0}{v4}
\Knoten{4.375}{0}{v6}
\Knoten{3.25}{1}{n7}
\Knoten{5.5}{0}{v7}
\Knoten{6}{0.5}{n3}
\Knoten{6}{-0.5}{n4}
\Knoten{6.3}{0.8}{v10}\nput{0}{v10}{$v$}
\Knoten{6.5}{-1}{v11}\nput{0}{v11}{$x$}

\ncline{-}{v1}{n1}
\Kante{n1}{v3}{a}
\Kante{n2}{v3}{c}
\ncline{-}{v2}{n6}
\Kante{n6}{n2}{b}
\Kante{v3}{v4}{d}
\Kante{v4}{v6}{f}
\Kante{v6}{v7}{g}
\Kante{v7}{n3}{h}
\ncline{-}{n3}{v10}
\Kante{v7}{n4}{i}
\ncline{-}{n4}{v11}
\Kantedashed{n7}{v4}{\gamma_2}
\Kantedashed{n7}{n3}{\alpha_3}
\Bogendashed{n1}{n2}{\alpha_1}{-30}
\pscurve[linestyle=dashed](1.5,-0.5)(0.2,0.8)(3.25,1)\rput(2.3,1.2){\colorbox{almostwhite}{$\alpha_2$}}
\pscurve[linestyle=dashed](6,-0.5)(7.1,1)(3.25,1)\rput(5.2,1.4){\colorbox{almostwhite}{$\gamma_1$}}
\Bogendashed{n6}{v6}{q_2}{-25}
\end{pspicture}
} 
\caption{NBC8}\label{propNBC8}
\end{figure}
\begin{lemma}For NBC8 (cf. Figure~\ref{propNBC8}) we get the following properties:
\begin{enumerate}[itemsep=-0em,leftmargin=*,label=\roman*)]
\item For NBC8a and NBC8b:
\begin{itemize}[itemsep=-0em,leftmargin=*]
	\item All edges which are substituted by $q_1$ or by $q_2$ are completely paid.
	\item For $(u,v,w,x)=(s_2,t_2,s_1,t_1)$, Player 2 pays all edges of the commonly used path which are not substituted by $q_2$ $(f)$. 
\end{itemize}
\item
For NBC8c and NBC8d:
\begin{itemize}[itemsep=-0em,leftmargin=*]
	\item All edges which are substituted by $q_1$ or by $q_3$ are completely paid. 
	\item For $(u,v,w,x)=(t_2,s_2,t_1,s_1)$, Player 2 pays all edges of the commonly used path which are not substituted by $q_3$ $(d)$.
\end{itemize}\end{enumerate}
\end{lemma}
 
\begin{proof}
For NBC8a and NBC8b we get
\[
\alpha_1+\alpha_2+\alpha_3+\beta_1+\beta_2=\xi_{u,a}+\xi_{\ell,b}+c+d+\xi_{u,f}+\xi_{u,g}
\] 
since $q_1$ and $q_2$ are tight alternatives. On the other hand the optimality of $L_u \cup L_{\ell}$ implies
\[
\alpha_1+\alpha_2+\alpha_3+\beta_1 \geq a+b+c+d+f+g+h
\]
and therefore $ \xi_{u,a}=a, \xi_{\ell,b}=b, \xi_{u,f}=f$ and $\xi_{u,g}=g$ holds. 

\vspace{1em}
The properties for NBC8c and NBC8d follow by symmetry reasons. 
\end{proof}


\begin{figure}[H] \centering \psset{unit=1.1cm}\captionsetup[subfigure]{labelformat=empty}
\subfloat[NBC9a]{
	\begin{pspicture}(-0.5,-1)(6.5,1.2) 
\Knoten{0.2}{0.8}{v1}\nput{180}{v1}{$u$}
\Knoten{-0.2}{-1.2}{v2}\nput{180}{v2}{$w$}
\Knoten{0.05}{-0.95}{n6}
\Knoten{0.5}{-0.5}{n2}
\Knoten{1}{0}{v3}
\Knoten{2.125}{0}{n5}
\Knoten{3.25}{0}{v4}
\Knoten{4.375}{0}{v6}
\Knoten{5.5}{0}{v7}
\Knoten{6}{0.5}{n3}
\Knoten{6}{-0.5}{n4}
\Knoten{6.3}{0.8}{v10}\nput{0}{v10}{$v$}
\Knoten{6.3}{-0.8}{v11}\nput{0}{v11}{$x$}

\ncline{-}{v1}{v3}
\ncline{-}{n2}{v3}
\ncline{-}{v2}{n6}
\Kante{n6}{n2}{a}
\Kante{n2}{v3}{b}
\Kante{v3}{n5}{c}
\Kante{v4}{n5}{d}
\Kante{v4}{v6}{f}
\Kante{v6}{v7}{g}
\Kante{v7}{n3}{h}
\ncline{-}{n3}{v10}
\Kante{v7}{n4}{i}
\ncline{-}{n4}{v11}
\pscurve[linestyle=dashed](6,0.5)(-0.5,0.8)(0.05,-0.95)\rput(4.4,1.1){\colorbox{almostwhite}{$\alpha_2$}}
\Bogendashed{n2}{v6}{q_2}{-40}
\Bogendashed{v4}{n4}{q_3}{-35}
\Bogendashed{n6}{n5}{\alpha_1}{-15}
\end{pspicture} 
}
\hspace{0.5cm}
\subfloat[NBC9b]{
\begin{pspicture}(0.5,-1)(6.5,1.2) 
\Knoten{1.2}{0.8}{v1}\nput{180}{v1}{$u$}
\Knoten{1.5}{0.5}{n1}
\Knoten{0.8}{-1.2}{v2}\nput{180}{v2}{$w$}
\Knoten{1.05}{-0.95}{n6}
\Knoten{1.5}{-0.5}{n2}
\Knoten{2}{0}{v3}
\Knoten{3.25}{0}{v4}
\Knoten{4.375}{0}{v6}
\Knoten{5.5}{0}{v7}
\Knoten{6}{0.5}{n3}
\Knoten{6}{-0.5}{n4}
\Knoten{6.3}{0.8}{v10}\nput{0}{v10}{$v$}
\Knoten{6.3}{-0.8}{v11}\nput{0}{v11}{$x$}

\ncline{-}{v1}{n1}
\Kante{n1}{v3}{a}
\Kante{n2}{v3}{c}
\ncline{v2}{n6}
\Kante{n6}{n2}{b}
\Kante{v3}{v4}{d}
\Kante{v4}{v6}{f}
\Kante{v6}{v7}{g}
\Kante{v7}{n3}{h}
\ncline{-}{n3}{v10}
\Kante{v7}{n4}{i}
\ncline{-}{n4}{v11}
\pscurve[linestyle=dashed](6,0.5)(0.5,0.7)(1.05,-0.95)\rput(4.4,1){\colorbox{almostwhite}{$\alpha_2$}}
\Bogendashed{n2}{v6}{q_2}{-35}
\Bogendashed{v4}{n4}{q_3}{-35}
\Bogendashed{n6}{n1}{\alpha_1}{20}
\end{pspicture} 
}
\caption{NBC9}\label{propNBC9}
\end{figure}
\begin{lemma}For NBC9 (cf. Figure~\ref{propNBC9}) we get the following properties:
\begin{enumerate}[itemsep=-0em,leftmargin=*,label=\roman*)]
	\item All edges which are substituted by $q_1, q_2$ or by $q_3$ are completely paid. 
	\item If $q_1$ is small and $(u,v,w,x)=(t_1,s_1,t_2,s_2)$, we additionally get that Player 2 pays all edges of the commonly used path which are not substituted by $q_1$ $(c)$.
\end{enumerate} 
\end{lemma}
	%

\begin{proof}
For NBC9a we get
\[
\alpha_1+\alpha_2+q_2+q_3=\xi_{\ell,b}+\xi_{\ell,c}+d+f+\xi_{\ell,f}+g+\xi_{u,h}+\xi_{\ell,i}
\]
since $q_1,q_2$ and $q_3$ are tight. On the other hand the optimality of $L_u \cup L_{\ell}$ implies $\alpha_2+q_3 \geq f+g+h+i$, $q_2 \geq c+d+f$ and $\alpha_1 \geq a+b$ and therefore 
\[
\alpha_1+\alpha_2+q_2+q_3 \geq a+b+c+d+2f+g+h+i.
\]
Altogether we get that $\xi_{\ell,b}=b, \xi_{\ell,c}=c, \xi_{u,h}=h$ and $\xi_{\ell,i}=i$ holds. 

\vspace{1em}
For NBC9b, the desired properties follow from the properties of NBC7 (by symmetry). 
\end{proof}


\begin{figure}[H] \centering \psset{unit=1.1cm}
	\begin{pspicture}(-0.5,-1.3)(6.5,1.2) 
\Knoten{0.2}{0.8}{v1}\nput{180}{v1}{$u$}
\Knoten{-0.2}{-1.2}{v2}\nput{180}{v2}{$w$}
\Knoten{0.05}{-0.95}{n6}
\Knoten{0.5}{-0.5}{n2}
\Knoten{1}{0}{v3}
\Knoten{2.125}{0}{n5}
\Knoten{3.25}{0}{v4}
\Knoten{4.375}{0}{v6}
\Knoten{5.5}{0}{v7}
\Knoten{6}{0.5}{n3}
\Knoten{6}{-0.5}{n4}
\Knoten{6.3}{0.8}{v10}\nput{0}{v10}{$v$}
\Knoten{6.3}{-0.8}{v11}\nput{0}{v11}{$x$}
\Knoten{3.25}{0.9}{n7}

\ncline{-}{v1}{v3}
\ncline{-}{n6}{v2}
\Kante{n6}{n2}{a}
\Kante{v3}{n2}{b}
\Kante{v3}{n5}{c}
\Kante{v4}{n5}{d}
\Kante{v4}{v6}{f}
\Kante{v6}{v7}{g}
\Kante{v7}{n3}{h}
\ncline{-}{n3}{v10}
\Kante{v7}{n4}{i}
\ncline{-}{n4}{v11}

\pscurve[linestyle=dashed](0.05,-0.95)(-1,1.3)(3.25,0.9)\rput(-1,1.3){\colorbox{almostwhite}{$\beta_1$}}
\pscurve[linestyle=dashed](0.5,-0.5)(-0.5,1)(3.25,0.9)\rput(0,1.3){\colorbox{almostwhite}{$\alpha_2$}}
\pscurve[linestyle=dashed](0.5,-0.5)(6.9,-0.8)(6,0.5)\rput(1.7,-1){\colorbox{almostwhite}{$\alpha_3$}}
\Bogendashed{v4}{n4}{q_3}{-35}
\Kantedashed{n5}{n7}{\alpha_1}
\Kantedashed{n7}{v6}{\beta_2}
\end{pspicture} 
\caption{NBC10}\label{propNBC10}
\end{figure}
\begin{lemma}For NBC10 (cf. Figure~\ref{propNBC10}) we get the following properties:
\begin{enumerate}[itemsep=-0em,leftmargin=*,label=\roman*)]
	\item All edges which are substituted by $q_1, q_2$ or by $q_3$ are completely paid. 
	\item For $(u,v,w,x)=(t_1,s_1,t_2,s_2)$, Player 2 pays all edges of the commonly used path which are not substituted by $q_1$ $(c)$. 
\end{enumerate} 
\end{lemma}

\begin{proof}
Using that $q_1, q_2$ and $q_3$ are tight we get
\[
\alpha_1+\alpha_2+\alpha_3+\beta_1+\beta_2+q_3=\xi_{\ell,a}+\xi_{\ell,b}+\xi_{\ell,c}+d+f+\xi_{\ell,f}+g+\xi_{u,h}+\xi_{\ell,i}.
\]
The optimality of $L_u \cup L_{\ell}$ yields $\alpha_1+\beta_1 \geq a+b$, $\alpha_2+\beta_2 \geq c+d+f$ and $\alpha_3+q_3 \geq f+g+h+i$ and therefore
\[
\alpha_1+\alpha_2+\alpha_3+\beta_1+\beta_2+q_3 \geq a+b+c+d+2f+g+h+i.
\]
This results in $\xi_{\ell,a}=a, \xi_{\ell,b}=b, \xi_{\ell,c}=c, \xi_{u,h}=h$ and $\xi_{\ell,i}=i$.
\end{proof}


\begin{figure}[H] \centering \psset{unit=1.1cm}\captionsetup[subfigure]{labelformat=empty}
\subfloat[NBC11a]{
		\begin{pspicture}(0.5,-1.4)(7,1.2) 
\Knoten{1.2}{0.8}{v1}\nput{180}{v1}{$u$}
\Knoten{1.5}{0.5}{n1}
\Knoten{0.8}{-1.2}{v2}\nput{180}{v2}{$w$}
\Knoten{1.05}{-0.95}{n6}
\Knoten{1.5}{-0.5}{n2}
\Knoten{2}{0}{v3}
\Knoten{3.25}{0}{v4}
\Knoten{4.375}{0}{v6}
\Knoten{5.5}{0}{v7}
\Knoten{6}{0.5}{n3}
\Knoten{6}{-0.5}{n4}
\Knoten{6.3}{0.8}{v10}\nput{0}{v10}{$v$}
\Knoten{6.3}{-0.8}{v11}\nput{0}{v11}{$x$}
\Knoten{0.4}{0.2}{n7}

\ncline{-}{v1}{n1}
\Kante{n1}{v3}{a}
\Kante{n6}{n2}{b}
\Kante{n2}{v3}{c}
\ncline{-}{v2}{n6}
\Kante{v3}{v4}{d}
\Kante{v4}{v6}{f}
\Kante{v6}{v7}{g}
\Kante{v7}{n3}{h}
\ncline{-}{n3}{v10}
\Kante{v7}{n4}{i}
\ncline{-}{n4}{v11}

\pscurve[linestyle=dashed](0.4,0.1)(0.8,1)(4.375,0)\rput(2.5,0.9){\colorbox{almostwhite}{$\beta_2$}}
\pscurve[linestyle=dashed](1.5,-0.5)(6.9,-0.8)(6,0.5)\rput(2.5,-0.9){\colorbox{almostwhite}{$\alpha_3$}}

\Bogendashed{v4}{n4}{q_3}{-35}
\Kantedashed{n1}{n7}{\alpha_1}
\Kantedashed{n7}{n2}{\alpha_2}
\Bogendashed{n6}{n7}{\beta_1}{50}
\end{pspicture} 
}
\hspace{0.5cm}
\subfloat[NBC11b]{
	\begin{pspicture}(0.5,-1.4)(7,1.2) 
\Knoten{1.2}{0.8}{v1}\nput{180}{v1}{$u$}
\Knoten{1.5}{0.5}{n1}
\Knoten{1.2}{-0.8}{v2}\nput{180}{v2}{$w$}
\Knoten{1.5}{-0.5}{n2}
\Knoten{2}{0}{v3}
\Knoten{3.25}{0}{v4}
\Knoten{4.375}{0}{v6}
\Knoten{5.5}{0}{v7}
\Knoten{6}{0.5}{n3}
\Knoten{6}{-0.5}{n4}
\Knoten{6.45}{-0.95}{n6}
\Knoten{6.3}{0.8}{v10}\nput{0}{v10}{$v$}
\Knoten{6.7}{-1.2}{v11}\nput{0}{v11}{$x$}
\Knoten{7}{0.2}{n7}

\ncline{-}{v1}{n1}
\Kante{n1}{v3}{a}
\Kante{n2}{v3}{b}
\ncline{-}{v2}{n2}
\Kante{v3}{v4}{c}
\Kante{v4}{v6}{d}
\Kante{v6}{v7}{f}
\Kante{v7}{n3}{g}
\ncline{-}{n3}{v10}
\Kante{v7}{n4}{h}
\Kante{n4}{n6}{i}
\ncline{-}{n6}{v11}

\pscurve[linestyle=dashed](1.5,0.5)(0.5,-0.8)(6,-0.5)\rput(4.8,-1){\colorbox{almostwhite}{$\alpha_1$}}
\pscurve[linestyle=dashed](3.25,0)(7,1.3)(7,0.2)\rput(4.8,1){\colorbox{almostwhite}{$\gamma_2$}}

\Bogendashed{n2}{v6}{q_2}{-35}
\Kantedashed{n3}{n7}{\alpha_3}
\Kantedashed{n4}{n7}{\alpha_2}
\Bogendashed{n6}{n7}{\gamma_1}{-50}
\end{pspicture} 
}
\caption{NBC11}\label{propNBC11}
\end{figure}
\begin{lemma}For NBC11 (cf. Figure~\ref{propNBC11}) we get that all edges which are substituted by $q_1,q_2$ or by $q_3$ are completely paid. \end{lemma}

\begin{proof}
For NBC11a, we get 
\[
\alpha_1+\alpha_2+\alpha_3+\beta_1+\beta_2+q_3=\xi_{u,a}+\xi_{\ell,b}+\xi_{\ell,c}+d+f+\xi_{\ell,f}+g+\xi_{u,h}+\xi_{\ell,i}
\]
by using that $q_1, q_2$ and $q_3$ are tight. On the other hand the optimality of $L_u \cup L_{\ell}$ implies $\alpha_3+q_3 \geq f+g+h+i$ and $\alpha_1+\beta_1+\beta_2 \geq a+b+c+d+f $ and therefore 
\[
\alpha_1+\alpha_2+\alpha_3+\beta_1+\beta_2+q_3 \geq a+b+c+d+2f+g+h+i.
\]
This yields $\xi_{u,a}=a, \xi_{\ell,b}=b, \xi_{\ell,c}=c, \xi_{u,h}=h$ and $\xi_{\ell,i}=i$. 

\vspace{1em}
The properties for NBC11b follow analogously (by symmetry).
\end{proof}

\begin{figure}[H] \centering \psset{unit=1.3cm}\captionsetup[subfigure]{labelformat=empty}
\subfloat[NBC12a]{
 \begin{pspicture}(-0.5,-1)(5,1.2) 
\Knoten{0.2}{0.8}{s1}\nput{180}{s1}{$u$}
\Knoten{0.5}{0.5}{n1}
\Knoten{0.2}{-0.8}{s2}\nput{180}{s2}{$w$}
\Knoten{0.5}{-0.5}{v1}
\Knoten{1}{0}{v2}
\Knoten{2}{0}{v4}
\Knoten{3}{0}{v5}
\Knoten{4}{0}{v6}
\Knoten{2}{0.5}{n3}
\Knoten{3}{0.5}{v7}
\Knoten{4.5}{0.5}{v9}
\Knoten{4.8}{0.8}{t1}\nput{0}{t1}{$v$}
\Knoten{4.5}{-0.5}{v10}
\Knoten{4.8}{-0.8}{t2}\nput{0}{t2}{$x$}
\ncline{s1}{v2}
\ncline{s2}{v1}
\Kante{n1}{v2}{a}
\Kante{v1}{v2}{b}
\Kante{v2}{v4}{c}
\Kante{v4}{v5}{d}
\Kante{v5}{v6}{f}
\Kante{v6}{v9}{g}
\Kante{v6}{v10}{h}
\ncline{v9}{t1}
\ncline{v10}{t2}

\pscurve[linestyle=dashed](3,0.5)(5.3,1)(4.5,-0.5)\rput(5,0){\colorbox{almostwhite}{$\gamma_1$}}
\pscurve[linestyle=dashed](0.5,-0.5)(-0.3,1)(2,0.5)\rput(-0.2,0.4){\colorbox{almostwhite}{$\beta_1$}}
\pscurve[linestyle=dashed](2,0.5)(-0.5,1.2)(-0.5,-1)(3,0)\rput(2,-0.65){\colorbox{almostwhite}{$\beta_2$}}
\Kantedashed{v7}{v9}{\alpha_3}
\Kantedashed{n1}{n3}{\alpha_1}
\Kantedashed{n3}{v7}{\alpha_2}
\Kantedashed{v4}{v7}{\gamma_2}
\end{pspicture}
}
\hspace{0.5cm}
\subfloat[NBC12b]{
\begin{pspicture}(-0.5,-1)(5,1.2) 
\Knoten{0.2}{0.8}{s1}\nput{180}{s1}{$u$}
\Knoten{0.5}{0.5}{n1}
\Knoten{0.2}{-0.8}{s2}\nput{180}{s2}{$w$}
\Knoten{0.5}{-0.5}{v1}
\Knoten{1}{0}{v2}
\Knoten{2}{0}{v4}
\Knoten{3}{0}{v5}
\Knoten{4}{0}{v6}
\Knoten{2}{0.5}{v7}
\Knoten{3}{0.5}{n2}
\Knoten{4.5}{0.5}{v9}
\Knoten{4.8}{0.8}{t1}\nput{0}{t1}{$v$}
\Knoten{4.5}{-0.5}{v10}
\Knoten{4.8}{-0.8}{t2}\nput{0}{t2}{$x$}
\ncline{s1}{v2}
\ncline{s2}{v1}
\Kante{n1}{v2}{a}
\Kante{v1}{v2}{b}
\Kante{v2}{v4}{c}
\Kante{v4}{v5}{d}
\Kante{v5}{v6}{f}
\Kante{v6}{v9}{g}
\Kante{v6}{v10}{h}
\ncline{v9}{t1}
\ncline{v10}{t2}
\pscurve[linestyle=dashed](2,0.5)(5.2,1)(4.5,-0.5)\rput(4.9,0){\colorbox{almostwhite}{$\gamma_1$}}
\pscurve[linestyle=dashed](0.5,-0.5)(-0.4,1)(3,0.5)\rput(1.5,1.05){\colorbox{almostwhite}{$\beta_1$}}
\Kantedashed{n2}{v5}{\beta_2}
\Kantedashed{v7}{n2}{\alpha_2}
\Kantedashed{n2}{v9}{\alpha_3}
\Kantedashed{n1}{v7}{\alpha_1}
\Kantedashed{v4}{v7}{\gamma_2}
\end{pspicture}
}
\caption{NBC12}
\label{propNBC12}
\end{figure}
\begin{lemma} \label{lemmanbcende} For NBC12 (cf. Figure~\ref{propNBC12}) we get the following properties:
\begin{enumerate}[itemsep=-0em,leftmargin=*,label=\roman*)]
	\item For $(u,v,w,x) \in \{(s_1,t_1,s_2,t_2), (t_1,s_1,t_2,s_2)\}$ and $q_1, q_2$ or $q_3$ substitutes an edge which is not completely paid, Player 2 has a positive cost share for an edge of the commonly used path which is substituted by $q_2$ \textbf{and} by $q_3$ $(d)$.
	\item For $(u,v,w,x) \in \{(s_2,t_2,s_1,t_1), (t_2,s_2,t_1,s_1)\}$ and $q_1, q_2$ or $q_3$ substitutes an edge which is not completely paid, Player 1 has a positive cost share for an edge of the commonly used path which is substituted by $q_2$ \textbf{and} by $q_3$ $(d)$. 
\end{enumerate}
\end{lemma}
\begin{proof}
We first show that NBC12a is not possible if $\xi_{u,a}+\xi_{\ell,b}+\xi_{u,g}+\xi_{\ell,h}<a+b+g+h$ holds:
Since $q_1,q_2$ and $q_3$ are tight, we get
\[
\alpha_1+\alpha_2+\alpha_3+\beta_1+\beta_2+\gamma_1+\gamma_2=\xi_{u,a}+\xi_{\ell,b}+c+d+\xi_{\ell,d}+f+\xi_{u,g}+\xi_{\ell,h}.
\]
Since $P_u \cup P_{\ell}$ is an optimal Steiner forest, we get $\alpha_1+\beta_1+\beta_2\geq a+b+c+d$ and $\alpha_3+\gamma_1+\gamma_2\geq d+f+g+h$  and together
\[
\alpha_1+\alpha_2+\alpha_3+\beta_1+\beta_2+\gamma_1+\gamma_2 \geq a+b+c+2d+f+g+h.
\]
This shows that $\xi_{u,a}=a, \xi_{\ell,b}=b, \xi_{u,g}=g$ and $\xi_{\ell,h}=h$ has to hold, a contradiction to our assumption. 

\vspace{1em}
Now consider NBC12b. Using that the alternatives are tight yields as above
\[
\alpha_1+\alpha_2+\alpha_3+\beta_1+\beta_2+\gamma_1+\gamma_2=\xi_{u,a}+\xi_{\ell,b}+c+d+\xi_{\ell,d}+f+\xi_{u,g}+\xi_{\ell,h} 
\] 
and the optimality of $P_u \cup P_{\ell}$ implies
\[
\alpha_1+\alpha_2+\alpha_3+\beta_1+\gamma_1 \geq a+b+c+d+f+g+h.
\]
Therefore $\xi_{\ell,d} \geq a+b+g+h-(\xi_{u,a}+\xi_{\ell,b}+\xi_{u,g}+\xi_{\ell,h}) >0$ has to hold. 
\end{proof}
Next we have to show that our statements according to PBCs in the proof of \autoref{lemma3} are correct. We will discuss the most involved occurrences of PBCs in detail. We can always assume that we are analyzing NBCs, since we excluded the existence of BCs in the proof of \autoref{lemma3}.

\begin{lemma}[\hyperlink{pbc1}{PBC1}] \label{lemmapbc1}
Assume that $(P_1,P_2,q_1,q_2'',q_2)$ is a NBC for $(u,v,w,x)=(s_1,t_1, s_2,t_2)$ with the properties of \hyperlink{pbc1}{PBC1}, i.e. 
\begin{itemize}[itemsep=-0em,leftmargin=*]
	\item $q_1$ (defined by $\beta$, $\gamma$) smallest right alternative of Player 1 for $e$;
	\item $q_1$ big;
		\item $e_{\alpha}$ largest edge in $\{\ell_1+\ell_2+1, \ldots, \ell_1+\ell_2+m\}$ which Player 2 does not  pay completely;
		\item $q_2$ (defined by $\mu$, $\nu$) smallest right alternative of Player 2 for $e_{\alpha}$; 
		\item $e_{\sigma}$ largest edge in $\{\ell_1+\ell_2+1, \ldots, \mu-1\}$ which Player 2 does not pay completely;
		\item  $q_2'$ (defined by $\mu'$, $\nu'$) largest left alternative of Player 2 for $e_{\sigma}$; 
		\item $\mu \leq \nu'\leq \alpha-1$;
		\item $q_2''$, defined by $\mu''$ and $\nu''$, smallest left alternative for Player 2 which substitutes $e_{\sigma}$;
		\item $\mu \leq \nu''\leq \alpha-1$;
\end{itemize}
see Figure~\ref{S3} for illustration. This leads to contradictions for all possible types of NBCs. 
\end{lemma}
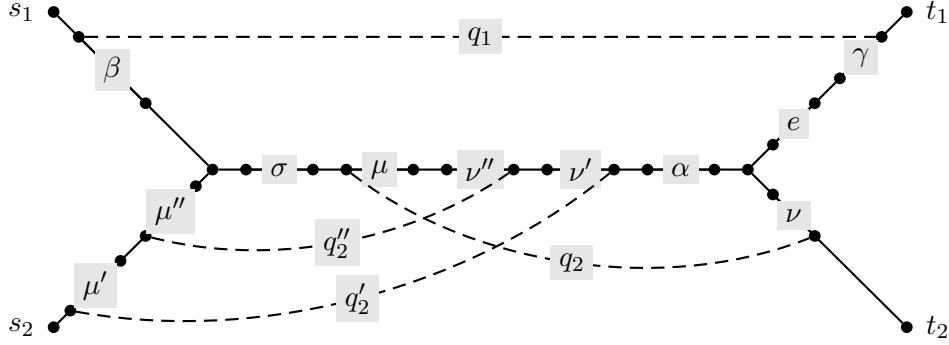
\begin{figure}[H] \centering \psset{unit=1.1cm}
\begin{pspicture}(3.2,-2)(13.8,2) 
\Knoten{3.6}{1.9}{v1}\nput{180}{v1}{$s_1$}
\Knoten{3.9}{1.6}{v2}
\Knoten{4.7}{0.8}{v3}
\Knoten{3.6}{-1.9}{v4}\nput{180}{v4}{$s_2$}
\Knoten{3.8}{-1.7}{v5}
\Knoten{4.4}{-1.1}{v6}
\Knoten{4.7}{-0.8}{v61}
\Knoten{5.3}{-0.2}{v62}
\Knoten{5.5}{0}{v7}
\Knoten{5.9}{0}{v14}
\Knoten{6.7}{0}{v15}
\Knoten{7.1}{0}{v16}
\Knoten{7.9}{0}{v17}
\Knoten{8.3}{0}{v18}
\Knoten{9.1}{0}{v19}
\Knoten{9.5}{0}{v20}
\Knoten{10.3}{0}{v21}
\Knoten{10.7}{0}{v22}
\Knoten{11.5}{0}{v23}
\Knoten{11.9}{0}{v24}
\Knoten{12.2}{0.3}{v25}
\Knoten{12.7}{0.8}{v26}
\Knoten{13.0}{1.1}{v27}
\Knoten{13.5}{1.6}{v28}
\Knoten{13.8}{1.9}{v29}\nput{0}{v29}{$t_1$}
\Knoten{12.2}{-0.3}{v30}
\Knoten{12.7}{-0.8}{v31}
\Knoten{13.8}{-1.9}{v34}\nput{0}{v34}{$t_2$}
\ncline{v1}{v7}
\ncline{v4}{v7}
\ncline{v7}{v24}
\ncline{v24}{v29}
\ncline{v24}{v34}

\Kante{v2}{v3}{\beta}
\Kante{v5}{v6}{\mu'}
\Kante{v61}{v62}{\mu''}
\Kante{v14}{v15}{\sigma}
\Kante{v16}{v17}{\mu}
\Kante{v18}{v19}{\nu''}
\Kante{v20}{v21}{\nu'}
\Kante{v22}{v23}{\alpha}
\Kante{v25}{v26}{e}
\Kante{v27}{v28}{\gamma}
\Kante{v30}{v31}{\nu}

\Bogendashed{v5}{v21}{q_2'}{-25}
\Bogendashed{v61}{v19}{q_2''}{-25}
\Bogendashed{v16}{v31}{q_2}{-30}
\Bogendashed{v2}{v28}{q_1}{0}
\end{pspicture}
\caption{PBC1 ($\mu'' \leq \mu'$ also possible)} \label{S3}
\end{figure}

\begin{proof} We get the following contradictions:
\vspace{-0.5em}
\begin{table}[H] 
\begin{center}
\begin{tabular}{l c@{\hspace{0.5cm}} l}
\toprule
Type of NBC & & contradiction \\ \midrule
1, 3a-c, 4a-c, 6, 10 & & not possible since $q_1$ is big \\ \addlinespace[0.2em]
2, 3f, 3g, 4d, 4e & & cheaper Steiner forest ($q_1$ and tight alternative for Player 2)  \\ \addlinespace[0.2em]
3d, 3e, 4f, 4g, 5, 7, 8, 9, 11 & & $e$ is completely paid \\ 
\bottomrule
\end{tabular}
\end{center} \end{table}
\vspace{-1.5em}

For NBC12, we do not get a contradiction directly. We only get that there has to be an edge $e_{\tau}$ in $\{e_{\mu},\ldots,e_{\nu''}\}$ which is not completely paid by Player 1.

Since the given cost shares are an an optimal solution for \lp{}, the following changes of the cost shares (by a suitably small amount) can not yield a feasible solution for \lp{} (the sum of all collected cost shares would be higher than before):
\[
\left.
\begin{array}{l}
	 \text{Player 2 increases on $e_{\sigma}$ and $e_{\alpha}$ and decreases on $e_{\tau},$} \\
	 \text{Player 1 increases on $e_{\tau}$ and $e$ and decreases on $e_{\sigma}$ and $e_{\alpha}.$}
\end{array}
\right\} \ (*)
\]
Therefore we now analyze which of these changes preserve the feasibility for \lp{} (we simply say that the changes are feasible), and which not.

The changes for Player 1 are obviously feasible since $q_1$ is the smallest right alternative for $e$. For Player 2, increasing on $e_{\alpha}$ while decreasing on $e_{\tau}$ is also feasible ($q_2$ smallest right alternative for $e_{\alpha}$), therefore additionally increasing on $e_{\sigma}$ can not be feasible. Either there is a left alternative which substitutes $e_{\sigma}$ but not $e_{\tau}$ (which is a contradiction to our choice of $q_2''$), there is a left alternative which substitutes $e_{\sigma},e_{\tau}$ and $e_{\alpha}$ (which is a contradiction to the optimality of the Steiner forest) or there is a right alternative which substitutes $e_{\sigma},e_{\tau}$ and $e_{\alpha}$. 
In the following, we change the cost shares iteratively to achieve that there is no such alternative, i.e. no right alternative for $e_{\sigma}$, meanwhile no alternatives get tight during the process, $q_1$ ($q_2$) remains the smallest right alternative for $e$ ($e_{\alpha}$) and the cost shares of all edges with order larger or equal than $\sigma$ stay unchanged. Furthermore, all left alternatives of the second player remain tight. Altogether we achieve that changing the cost shares as described in $(*)$ is then feasible and this is a contradiction to the optimality of the cost shares.

Now we describe how to change the cost shares iteratively. Let $\bar{q_2}$ be the largest right alternative for Player 2. If $\bar{q_2}$ substitutes all commonly used edges or Player 2 pays the edges $e_{\ell_1+\ell_2+1},\ldots,e_{\bar{\mu}-1}$ we get a cheaper solution. Therefore let $e_{\rho}$ be the largest such edge which Player 2 does not pay completely. Since the cost shares are maximized for Player 2 there is an alternative for this edge and this has to be a left one because of our choice of $\bar{q_2}$. We consider the smallest such alternative $\widehat{q_2}$, defined by $\widehat{\mu}$ and $\widehat{\nu}$. Note that $\widehat{\nu}\leq \nu''$ holds since $q_2''$ is a left alternative for $e_{\rho}$. 
If $\widehat{\nu}\leq \bar{\mu}-1$ holds, we can construct a cheaper Steiner forest. For $\sigma \leq \widehat{\nu}$ (i.e. $\widehat{\nu}=\nu''$ because of the choice of $q_2''$), the changes described in $(*)$ would be feasible for $\rho$ instead of $\sigma$.
Therefore $\bar{\mu}\leq \widehat{\nu} \leq \sigma-1$ has to hold, see Figure~\ref{caser1mitlinks}. 
\begin{figure}[H] \centering \psset{unit=1cm}
\begin{pspicture}(-1.2,-2)(13.8,2) 
\Knoten{-1.2}{1.9}{v1}\nput{180}{v1}{$s_1$}
\Knoten{-0.9}{1.6}{v2}
\Knoten{-0.2}{0.9}{v3}
\Knoten{-1.2}{-1.9}{v4}\nput{180}{v4}{$s_2$}
\Knoten{-1.0}{-1.7}{v5}
\Knoten{-0.4}{-1.1}{v6}
\Knoten{-0.1}{-0.8}{v61}
\Knoten{0.5}{-0.2}{v62}
\Knoten{0.7}{0}{v7}
\Knoten{1.1}{0}{v71}
\Knoten{1.9}{0}{v72}
\Knoten{2.3}{0}{v8}
\Knoten{3.1}{0}{v9}
\Knoten{3.5}{0}{v10}
\Knoten{4.3}{0}{v11}
\Knoten{4.7}{0}{v12}
\Knoten{5.5}{0}{v13}
\Knoten{5.9}{0}{v14}
\Knoten{6.7}{0}{v15}
\Knoten{7.1}{0}{v16}
\Knoten{7.9}{0}{v17}
\Knoten{8.3}{0}{v18}
\Knoten{9.1}{0}{v19}
\Knoten{9.5}{0}{v20}
\Knoten{10.3}{0}{v21}
\Knoten{10.7}{0}{v22}
\Knoten{11.5}{0}{v23}
\Knoten{11.9}{0}{v24}
\Knoten{12.2}{0.3}{v25}
\Knoten{12.7}{0.8}{v26}
\Knoten{13.0}{1.1}{v27}
\Knoten{13.5}{1.6}{v28}
\Knoten{13.8}{1.9}{v29}\nput{0}{v29}{$t_1$}
\Knoten{12.2}{-0.3}{v30}
\Knoten{12.7}{-0.8}{v31}
\Knoten{13.0}{-1.1}{v32}
\Knoten{13.5}{-1.6}{v33}
\Knoten{13.8}{-1.9}{v34}\nput{0}{v34}{$t_2$}

\ncline{v1}{v7}
\ncline{v4}{v7}
\ncline{v7}{v24}
\ncline{v24}{v29}
\ncline{v24}{v34}

\Kante{v2}{v3}{\beta}
\Kante{v5}{v6}{\mu''}
\Kante{v61}{v62}{\widehat{\mu}}
\Kante{v71}{v72}{\rho}
\Kante{v8}{v9}{\bar{\mu}}
\Kante{v10}{v11}{\omega}
\Kante{v12}{v13}{\widehat{\nu}}
\Kante{v14}{v15}{\sigma}
\Kante{v16}{v17}{\mu}
\Kante{v18}{v19}{\tau}
\Kante{v20}{v21}{\nu''}
\Kante{v22}{v23}{\alpha}
\Kante{v25}{v26}{e}
\Kante{v27}{v28}{\gamma}
\Kante{v30}{v31}{\nu}
\Kante{v32}{v33}{\bar{\nu}}

\Bogendashed{v5}{v21}{q_2''}{-20}
\Bogendashed{v16}{v31}{q_2}{-30}
\Bogendashed{v8}{v33}{\bar{q_2}}{-20}
\Bogendashed{v61}{v13}{\widehat{q_2}}{-30}
\Bogendashed{v2}{v28}{q_1}{0}
\end{pspicture}
\caption{NBC12 in \autoref{lemmapbc1}} \label{caser1mitlinks}
\end{figure}

But now, $(P_1,P_2,q_1,\widehat{q_2},\bar{q_2})$ is again a NBC for $(u,v,w,x)=(s_1,t_1, s_2,t_2)$ and (analogous as before) there has to be an edge $e_{\omega}$ in $\{\bar{\mu}, \ldots, \widehat{\nu}\}$ which Player 1 does not pay completely. 
Now it is feasible that Player 1 decreases her cost share on $e_{\rho}$ and increases on $e_{\omega}$ (because of our choice of $q_1$) and that Player 2 decreases on $e_{\omega}$ and increases on $e_{\rho}$ (by the choice of $\widehat{q_2}$). Note that we can ensure that no alternatives get tight by changing the cost shares by a suitably small amount. Furthermore, $\bar{q_2}$ is not tight anymore (but $q_1,q_2$ and all tight left alternatives for the second player obviously stay tight). If there is no right alternative for $e_{\sigma}$ which is tight according to the changed cost shares, we can stop the procedure. If not, redefine $\bar{q_2}$ as a largest such alternative and repeat the procedure above.
This shows that NBC12 also leads to a contradiction.
\end{proof}

\begin{lemma}[\hyperlink{pbc6}{PBC6}, in the case $q_1$ big] \label{lemmapbc7}
Assume that $(P_1,P_2,q_1,\bar{q_2},q_2')$ is a NBC for $(u,v,w,x)=(s_1,t_1,s_2,t_2)$ with the properties of \hyperlink{pbc6}{PBC6} and $q_1$ big, i.e.  
\begin{itemize}[itemsep=-0em,leftmargin=*]
	\item $q_1$ (defined by $\beta$, $\gamma$) smallest right alternative of Player 1 for $e$;
	\item $q_1$ big;
		\item $e_{\alpha}$ largest edge in $\{\ell_1+\ell_2+1, \ldots, \ell_1+\ell_2+m\}$ which Player 2 does not pay completely;
		\item $q_2$ (defined by $\mu$, $\nu$) smallest right alternative of Player 2 for $e_{\alpha}$; 
		\item $e_{\sigma}$ largest edge in $\{\ell_1+\ell_2+1, \ldots, \mu-1\}$ which Player 2 does not pay completely;
		\item only right alternatives of Player 2 for $e_{\sigma}$; $q_2'$ (defined by $\mu'$, $\nu'$) largest right alternative of Player 2 for $e_{\sigma}$;
		\item $e_{\rho}$ largest edge in $\{\ell_1+\ell_2+1, \ldots, \mu'-1\}$ which Player 2 does not pay completely;
		\item $\bar{q_2}$, defined by $\bar{\mu}$ and $\bar{\nu}$, smallest left alternative of Player 2 for $e_{\rho}$;
		\item $\mu'\leq\bar{\nu}\leq\sigma-1$;
\end{itemize}
see Figure~\ref{S7} for illustration. This leads to contradictions for all possible types of NBCs. 
\end{lemma}

\begin{figure}[H] \centering \psset{unit=1.1cm}
\begin{pspicture}(2.2,-2)(14,2) 
\Knoten{2.4}{1.9}{v1}\nput{180}{v1}{$s_1$}
\Knoten{2.4}{-1.9}{v4}\nput{180}{v4}{$s_2$}
\Knoten{2.7}{-1.6}{v5}
\Knoten{3.4}{-0.9}{v6}
\Knoten{4.3}{0}{v7}
\Knoten{2.7}{1.6}{v8}
\Knoten{3.4}{0.9}{v9}
\Knoten{4.7}{0}{v10}
\Knoten{5.5}{0}{v11}
\Knoten{5.9}{0}{v12}
\Knoten{6.7}{0}{v13}
\Knoten{7.1}{0}{v16}
\Knoten{7.9}{0}{v17}
\Knoten{8.3}{0}{v18}
\Knoten{9.1}{0}{v19}
\Knoten{9.5}{0}{v20}
\Knoten{10.3}{0}{v21}
\Knoten{10.7}{0}{v22}
\Knoten{11.5}{0}{v23}
\Knoten{11.9}{0}{v24}
\Knoten{12.2}{0.3}{v25}
\Knoten{12.7}{0.8}{v26}
\Knoten{13.0}{1.1}{v27}
\Knoten{13.5}{1.6}{v28}
\Knoten{13.8}{1.9}{v29}\nput{0}{v29}{$t_1$}
\Knoten{12.2}{-0.3}{v30}
\Knoten{12.7}{-0.8}{v31}
\Knoten{13.0}{-1.1}{v32}
\Knoten{13.5}{-1.6}{v33}
\Knoten{13.8}{-1.9}{v34}\nput{0}{v34}{$t_2$}

\ncline{v1}{v7}
\ncline{v4}{v7}
\ncline{v7}{v24}
\ncline{v24}{v29}
\ncline{v24}{v34}

\Kante{v5}{v6}{\bar{\mu}}
\Kante{v8}{v9}{\beta}
\Kante{v10}{v11}{\rho}
\Kante{v12}{v13}{\mu'}
\Kante{v16}{v17}{\bar{\nu}}
\Kante{v18}{v19}{\sigma}
\Kante{v20}{v21}{\mu}
\Kante{v22}{v23}{\alpha}
\Kante{v25}{v26}{e}
\Kante{v27}{v28}{\gamma}
\Kante{v30}{v31}{\nu}
\Kante{v32}{v33}{\nu'}

\Bogendashed{v20}{v31}{q_2}{-35}
\Bogendashed{v5}{v17}{\bar{q_2}}{-25}
\Bogendashed{v12}{v33}{q_2'}{-25}
\Bogendashed{v8}{v28}{q_1}{0}
\end{pspicture}
\caption{PBC6, $q_1$ big} \label{S7}
\end{figure}
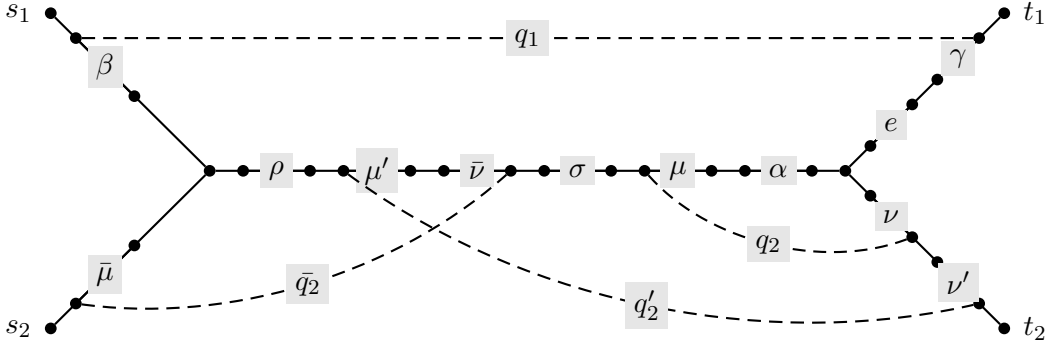

\begin{proof} We get the following contradictions:
\vspace{-0.5em}
\begin{table}[H] 
\begin{center}
\begin{tabular}{l c@{\hspace{0.5cm}} l}
\toprule
Type of NBC & & contradiction \\ \midrule
1, 3a-c, 4a-c, 6, 10 & & not possible since $q_1$ is big \\ \addlinespace[0.2em]
2, 3f, 3g, 4d, 4e & & cheaper Steiner forest ($q_1$ and tight alternative for Player 2)  \\ \addlinespace[0.2em]
3d, 3e, 4f, 4g, 5, 7, 8, 9, 11 & & $e$ is completely paid \\ 
\bottomrule
\end{tabular}
\end{center} \end{table}
\vspace{-1.5em}

For NBC12, we do not get a contradiction directly. We only get that there is an edge in $\{\mu', \ldots, \bar{\nu}\}$ that is not completely paid by Player 1. Let $e_{\tau}$ be the largest such edge. The situation is illustrated in Figure~\ref{fig:bc4mitlinks}.

\begin{figure}[H] \centering \psset{unit=1.1cm}
\begin{pspicture}(1,-2)(14,2) 
\Knoten{1.2}{1.9}{v1}\nput{180}{v1}{$s_1$}
\Knoten{1.2}{-1.9}{v4}\nput{180}{v4}{$s_2$}
\Knoten{1.5}{-1.6}{v5}
\Knoten{2.2}{-0.9}{v6}
\Knoten{3.1}{0}{v7}
\Knoten{1.5}{1.6}{v8}
\Knoten{2.2}{0.9}{v9}
\Knoten{3.5}{0}{v10}
\Knoten{4.3}{0}{v11}
\Knoten{4.7}{0}{v12}
\Knoten{5.5}{0}{v13}
\Knoten{5.9}{0}{v14}
\Knoten{6.7}{0}{v15}
\Knoten{7.1}{0}{v16}
\Knoten{7.9}{0}{v17}
\Knoten{8.3}{0}{v18}
\Knoten{9.1}{0}{v19}
\Knoten{9.5}{0}{v20}
\Knoten{10.3}{0}{v21}
\Knoten{10.7}{0}{v22}
\Knoten{11.5}{0}{v23}
\Knoten{11.9}{0}{v24}
\Knoten{12.2}{0.3}{v25}
\Knoten{12.7}{0.8}{v26}
\Knoten{13.0}{1.1}{v27}
\Knoten{13.5}{1.6}{v28}
\Knoten{13.8}{1.9}{v29}\nput{0}{v29}{$t_1$}
\Knoten{12.2}{-0.3}{v30}
\Knoten{12.7}{-0.8}{v31}
\Knoten{13.0}{-1.1}{v32}
\Knoten{13.5}{-1.6}{v33}
\Knoten{13.8}{-1.9}{v34}\nput{0}{v34}{$t_2$}

\ncline{v1}{v7}
\ncline{v4}{v7}
\ncline{v7}{v24}
\ncline{v24}{v29}
\ncline{v24}{v34}

\Kante{v5}{v6}{\bar{\mu}}
\Kante{v8}{v9}{\beta}
\Kante{v10}{v11}{\rho}
\Kante{v12}{v13}{\mu'}
\Kante{v14}{v15}{\tau}
\Kante{v16}{v17}{\bar{\nu}}
\Kante{v18}{v19}{\sigma}
\Kante{v20}{v21}{\mu}
\Kante{v22}{v23}{\alpha}
\Kante{v25}{v26}{e}
\Kante{v27}{v28}{\gamma}
\Kante{v30}{v31}{\nu}
\Kante{v32}{v33}{\nu'}

\Bogendashed{v20}{v31}{q_2}{-35}
\Bogendashed{v5}{v17}{\bar{q_2}}{-25}
\Bogendashed{v12}{v33}{q_2'}{-25}
\Bogendashed{v8}{v28}{q_1}{0}
\end{pspicture}
\caption{NBC12 in \autoref{lemmapbc7}} \label{fig:bc4mitlinks}
\end{figure}

We now change the cost shares for this case iteratively. First, Player 1/2 in/decreases her cost share on $e_{\tau}$ and de/increases on $e_{\rho}$ (by a suitably small amount to ensure that no alternatives get tight). This is obviously feasible for both Players because of the choice of $q_1$ and $\bar{q_2}$. But after that, $q_2'$ is not tight anymore (according to the changed cost shares). If there is no tight right alternative left which substitutes $e_{\sigma}$, Player 2 can increase her cost share on this edge, while Player 1 decreases it and then Player 1 can increase on $e$, which is a contradiction. Otherwise, redefine $q_2'$ as the largest right alternative according to the changed cost shares. It is clear that $\mu'\geq \tau+1$ has to hold, otherwise $q_2'$ could not be tight now. If $\mu'\leq\bar{\nu}$ holds, we get a contradiction: Our choice of $\tau$ implies that Player 1 pays the edges $e_{\mu'}, \ldots, e_{\bar{\nu}}$ completely. But we can obviously repeat our arguments above and get (by NBC12) that there has to be an edge in $e_{\mu'}, \ldots, e_{\bar{\nu}}$ that is not completely paid by Player 1. Therefore $\mu'\geq\bar{\nu}+1$ has to hold. If $\mu'=\bar{\nu}+1$ or Player 2 pays the edges $e_{\bar{\nu}+1}, \ldots, e_{\mu'-1}$, we can construct a cheaper Steiner forest by using $q_1,\bar{q_2}$ and $q_2'$. Thus redefine $e_{\rho}$ as the largest edge in $\{\bar{\nu}+1, \ldots, \mu'-1\}$ that is not completely paid by Player~2. If Player 2 has no tight alternative for this edge, she can increase her cost share on this edge, while Player 1 decreases it and additionally increases on $e$, which is a contradiction. Therefore there has to be a left alternative for $e_{\rho}$ (note that there can not be a right one) and we redefine $\bar{q_2}$ as the smallest such alternative. Analogously as above, there is an edge in $\{\mu', \ldots, \bar{\nu}\}$ that is not completely paid by Player 1 and we redefine $e_{\tau}$ as the largest such edge. Now we have reached the situation of Figure~\ref{fig:bc4mitlinks} again and can repeat the argumentation until we finally get a contradiction (this obviously terminates).
\end{proof}
\begin{lemma}[\hyperlink{pbc8}{PBC8}] \label{lemmapbc9}
Assume that $(P_2,P_1,q_2',q_1',q_1)$ is a NBC for $(u,v,w,x)=(s_2,t_2,s_1,t_1)$ with the properties of \hyperlink{pbc8}{PBC8}, i.e. 
\begin{itemize}[itemsep=-0em,leftmargin=*]
	\item $q_1$ (defined by $\beta$, $\gamma$) smallest right alternative of Player 1 for $e$;
		\item $e_{\alpha}$ largest edge in $\{\beta, \ldots, \ell_1+\ell_2+m\}$ which Player 2 does not pay completely;
		\item $q_2$ (defined by $\mu$, $\nu$) smallest right alternative of Player 2 for $e_{\alpha}$; 
		\item $e_{\sigma}$ largest edge in $\{\beta, \ldots, \mu-1\}$ which Player 2 does not pay completely;
		\item only right alternatives of Player 2 for $e_{\sigma}$; $q_2'$ (defined by $\mu'$, $\nu'$) largest right alternative of Player 2 for $e_{\sigma}$; 
		\item $q_2'$ small;
		\item $e_{\rho}$ largest edge in $\{\mu', \ldots, \beta-1\}$ which Player 1 does not pay completely;
		\item $q_1'$ (defined by $\beta'$, $\gamma'$) smallest left alternative of Player 1 which substitutes $e_{\rho}$, but not $e_{\sigma}$;
		\item $\beta \leq \gamma' \leq \sigma -1$;
\end{itemize}
see Figure~\ref{S9} for illustration. This is only possible if the edges $e_{\beta}, \ldots, e_{\gamma'}$ all have cost 0.
\end{lemma}

\begin{figure}[H] \centering \psset{unit=1cm}
\begin{pspicture}(2.4,-2)(15,2) 
\Knoten{2.4}{1.9}{v1}\nput{180}{v1}{$s_1$}
\Knoten{2.7}{1.6}{v2}
\Knoten{3.4}{0.9}{v3}
\Knoten{2.4}{-1.9}{v4}\nput{180}{v4}{$s_2$}
\Knoten{4.3}{0}{v7}
\Knoten{4.7}{0}{v10}
\Knoten{5.5}{0}{v11}
\Knoten{5.9}{0}{v12}
\Knoten{6.7}{0}{v13}
\Knoten{7.1}{0}{v16}
\Knoten{7.9}{0}{v17}
\Knoten{8.3}{0}{v18}
\Knoten{9.1}{0}{v19}
\Knoten{9.5}{0}{v20}
\Knoten{10.3}{0}{v21}
\Knoten{10.7}{0}{v22}
\Knoten{11.5}{0}{v23}
\Knoten{11.9}{0}{v24}
\Knoten{12.7}{0}{n1}
\Knoten{13.1}{0}{n2}
\Knoten{13.4}{0.3}{v25}
\Knoten{13.9}{0.8}{v26}
\Knoten{14.2}{1.1}{v27}
\Knoten{14.7}{1.6}{v28}
\Knoten{15}{1.9}{v29}\nput{0}{v29}{$t_1$}
\Knoten{13.4}{-0.3}{v30}
\Knoten{13.9}{-0.8}{v31}
\Knoten{14.15}{-1.05}{v32}
\Knoten{14.75}{-1.65}{v33}
\Knoten{15}{-1.9}{v34}\nput{0}{v34}{$t_2$}
\ncline{v1}{v7}
\ncline{v4}{v7}
\ncline{v7}{n2}
\ncline{n2}{v29}
\ncline{n2}{v34}
\Kante{v2}{v3}{\beta'}
\Kante{v10}{v11}{\mu'}
\Kante{v12}{v13}{\rho}
\Kante{v16}{v17}{\beta}
\Kante{v18}{v19}{\gamma'}
\Kante{v20}{v21}{\sigma}
\Kante{v22}{v23}{\mu}
\Kante{v25}{v26}{e}
\Kante{v27}{v28}{\gamma}
\Kante{v30}{v31}{\nu}
\Kante{v32}{v33}{\nu'}
\Kante{n1}{v24}{\alpha}
\Bogendashed{v22}{v31}{q_2}{-35}
\Bogendashed{v10}{v33}{q_2'}{-20}
\Bogendashed{v16}{v28}{q_1}{25}
\Bogendashed{v2}{v19}{q_1'}{25}
\end{pspicture}
\caption{PBC8}\label{S9}
\end{figure}

\begin{proof} The following types of NBCs yield contradictions:
\vspace{-0.5em}
\begin{table}[H] 
\begin{center}
\begin{tabular}{l c@{\hspace{0.5cm}} l}
\toprule
Type of NBC & & contradiction \\ \midrule
3d-g, 4d-g, 7, 8, 11, 12 & & not possible since $q_2'$ is small \\ \addlinespace[0.2em]
1, 3a, 4a-c, 6, 9, 10 & & $e$ is completely paid  \\ \addlinespace[0.2em]
3b, 3c, 5 & & Player 2 pays $e_{\alpha}$ \\
\bottomrule
\end{tabular}
\end{center} \end{table}
\vspace{-1.5em}
%
For NBC2, we do not get a contradiction directly. But we get that Player 1 pays the edges $e_{\beta}, \ldots, e_{\gamma'}$. If there is such an edge with cost larger than zero, let $e_{\tau}$ be the smallest such one. That implies that $\xi_{1,e_{\tau}}=c(e_{\tau})>0$. Now let us consider \texttt{CHANGE($\rho,\tau$)} and analyze why this is not feasible: The changes for Player 1 are feasible since $q_1'$ is the smallest alternative for Player~1 which substitutes $e_{\rho}$, but not $e_{\sigma}$. Therefore the changes can not be feasible for Player 2 and there has to be a right alternative for Player 2 which substitutes $e_{\tau}$, but not $e_{\rho}$. But then we can construct a cheaper Steiner forest by using this alternative together with $q_1$, since Player 1 pays the edges $e_{\rho+1}, \ldots, e_{\beta-1}$ and the edges $e_{\beta}, \ldots, e_{\tau-1}$ have costs zero. Therefore the costs of the edges $e_{\beta}, \ldots, e_{\gamma'}$ all have to be 0.
\end{proof}
\begin{lemma}[\hyperlink{pbc12}{PBC12}, in the case $q_2$ small] \label{lemmapbc14}
Assume that $(P_2,P_1,q_2,q_1',q_1)$ is a NBC for $(u,v,w,x)=(s_2,t_2,s_1,t_1)$ with the properties of \hyperlink{pbc12}{PBC12} and $q_2$ small, i.e. 
\begin{itemize}[itemsep=-0em,leftmargin=*]
	\item $q_1$ (defined by $\beta$, $\gamma$) smallest right alternative of Player 1 for $e$;
		\item $e_{\alpha}$ largest edge in $\{\beta, \ldots, \ell_1+\ell_2+m\}$ which Player 2 does not pay completely;
		\item $q_2$ (defined by $\mu$, $\nu$) smallest right alternative of Player 2 for $e_{\alpha}$; 
		\item $e_{\sigma}$ largest edge in $\{\mu, \ldots, \beta-1\}$ which Player 1 does not pay completely;
		\item $q_2$ small;
		\item $q_1'$ (defined by $\beta'$, $\gamma'$) smallest left alternative of Player 1 which substitutes $e_{\sigma}$, but not $e_{\alpha}$;
		\item $\beta \leq \gamma' \leq \alpha-1$;
\end{itemize}
see Figure~\ref{S14} for illustration. This leads to contradictions for all possible types of NBCs. 
\end{lemma}

\begin{figure}[H] \centering \psset{unit=1cm}
\begin{pspicture}(-1.5,-2)(11.3,2) 
\Knoten{-1.2}{1.7}{s1}\nput{180}{s1}{$s_1$}
\Knoten{-0.7}{1.2}{v1}
\Knoten{0}{0.5}{v2}
\Knoten{-1.2}{-1.7}{s2}\nput{180}{s2}{$s_2$}
%
\Knoten{0.5}{0}{v5}
\Knoten{1.3}{0}{v6}
\Knoten{2.1}{0}{v7}
\Knoten{2.9}{0}{v8}
\Knoten{3.7}{0}{v9}
\Knoten{4.5}{0}{v10}
\Knoten{5.3}{0}{v11}
\Knoten{6.1}{0}{v12}
\Knoten{6.9}{0}{v13}
\Knoten{7.7}{0}{v14}
\Knoten{8.5}{0}{v15}
\Knoten{9.3}{0}{v16}
\Knoten{9.53}{0.23}{n1}
\Knoten{10.03}{0.73}{n2}
\Knoten{10.26}{0.96}{v17}
\Knoten{10.76}{1.46}{v18}
\Knoten{11}{1.7}{t1}\nput{0}{t1}{$t_1$}
\Knoten{9.8}{-0.5}{v19}
\Knoten{10.5}{-1.2}{v20}
\Knoten{11}{-1.7}{t2}\nput{0}{t2}{$t_2$}

\ncline{s1}{v1}
\Kante{v1}{v2}{\beta '}
\ncline{v2}{v5}
\ncline{s2}{v5}
\ncline{v5}{v6}
\Kante{v6}{v7}{\mu}
\ncline{v7}{v8}
\Kante{v8}{v9}{\sigma}
\ncline{v9}{v10}
\Kante{v10}{v11}{\beta}
\ncline{v11}{v12}
\Kante{v12}{v13}{\gamma '}
\ncline{v13}{v14}
\Kante{v14}{v15}{\alpha}
\ncline{v15}{v16}
\ncline{v16}{n1}
\ncline{n2}{v17}
\Kante{v17}{v18}{\gamma}
\ncline{v18}{t1}
\ncline{v16}{v19}
\Kante{v19}{v20}{\nu}
\ncline{v20}{t2}
\Kante{n1}{n2}{e}

\Bogendashed{v1}{v13}{q_1'}{35}
\Bogendashed{v10}{v18}{q_1}{35}
\Bogendashed{v6}{v20}{q_2}{-35}
\end{pspicture}
\caption{PBC12, $q_2$ small} \label{S14}
\end{figure}
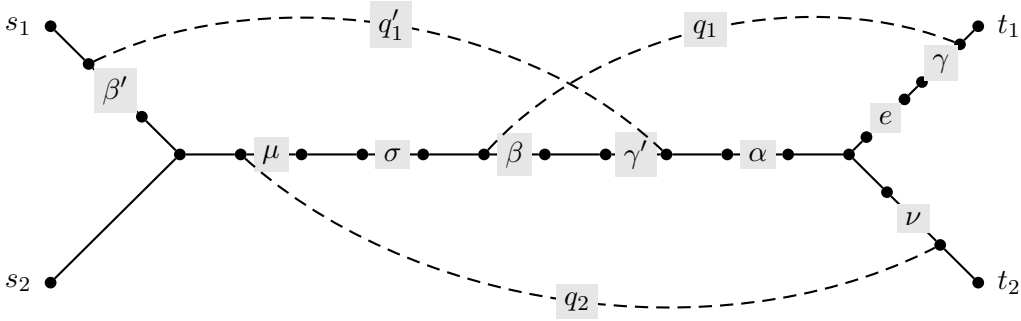

\begin{proof} The following types of NBCs yield contradictions:
\vspace{-0.5em}
\begin{table}[H] 
\begin{center}
\begin{tabular}{l c@{\hspace{0.5cm}} l}
\toprule
Type of NBC & & contradiction \\ \midrule
3d-g, 4d-g, 7, 8, 11, 12 & & not possible since $q_2$ is small \\ \addlinespace[0.2em]
1, 3a, 4a-c, 6, 9, 10 & & $e$ is completely paid  \\ \addlinespace[0.2em]
3b, 3c, 5 & & Player 2 pays $e_{\alpha}$ \\
\bottomrule
\end{tabular}
\end{center} \end{table}
\vspace{-1.5em}
For NBC2, we do not get a contradiction directly. But we get that Player 1 has a tight alternative $\bar{q_1}$ which substitutes $e$ and all commonly used edges and is also a left alternative.
Let $q_2'$ (defined by $\mu'$ and $\nu'$) be the largest right alternative for Player 2. 
Now if $q_2'$ substitutes all commonly used edges or Player 2 pays the edges $e_{\ell_1+\ell_2+1}, \ldots, e_{\mu'-1}$, we get a cheaper Steiner forest by using $\bar{q_1}$ and $q_2'$. Thus let $e_{\tau}$ be the largest such edge which Player 2 does not pay completely. 
Since the cost shares are maximized for Player 2, she has a tight alternative which substitutes $e_{\tau}$ and this has to be a left alternative because of the choice of $q_2'$.
We consider a smallest left alternative $\bar{q_2}$ according to $\tau$, defined by $\bar{\mu}$ and $\bar{\nu}$. If $\bar{\nu}\leq\mu'-1$ holds, we can construct a cheaper solution by using $\bar{q_1},q_2'$ and $\bar{q_2}$ since Player 2 pays the edges $e_{\tau+1}, \ldots, e_{\mu'-1}$ by the choice of $\tau$.
For $\bar{\nu}\geq\alpha$ we also get a cheaper solution since Player 2 pays the edges $e_{\alpha+1}, \ldots, e_{\ell_1+\ell_2+m}$. Therefore $\bar{\nu}\in\{\mu', \ldots, \alpha-1\}$ has to hold. We distinguish between $\bar{\nu}\leq \sigma-1$ and $\bar{\nu}\geq\sigma$. 

For the first case, $(P_1,P_2,q_1',q_2',\bar{q_2})$ is a NBC for $(u,v,w,x)=(t_1,s_1,t_2,s_2)$, see Figure~\ref{S15}.

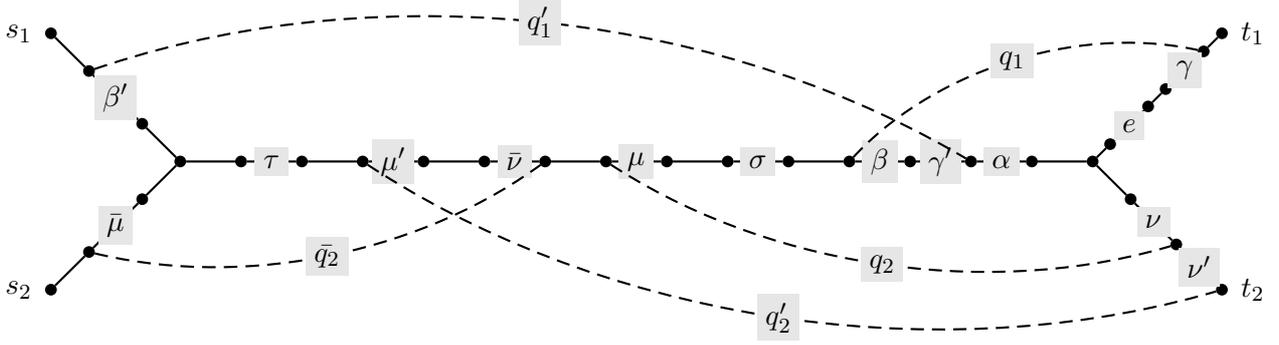
\begin{figure}[H] \centering \psset{unit=1cm}
\begin{pspicture}(-5,-2.3)(10,2) 
\Knoten{-5.2}{1.7}{s1}\nput{180}{s1}{$s_1$}
\Knoten{-4.7}{1.2}{v1}
\Knoten{-4}{0.5}{v2}
\Knoten{-5.2}{-1.7}{s2}\nput{180}{s2}{$s_2$}
\Knoten{-4.7}{-1.2}{v3}
\Knoten{-4}{-0.5}{v4}
\Knoten{-3.5}{0}{v5}
\Knoten{-2.7}{0}{n1}
\Knoten{-1.9}{0}{n2}
\Knoten{-1.1}{0}{v6}
\Knoten{-0.3}{0}{v7}
\Knoten{0.5}{0}{v8}
\Knoten{1.3}{0}{v9}
\Knoten{2.1}{0}{v10}
\Knoten{2.9}{0}{v11}
\Knoten{3.7}{0}{n3}
\Knoten{4.5}{0}{n4}
\Knoten{5.3}{0}{v12}
\Knoten{6.1}{0}{v13}
\Knoten{6.9}{0}{v14}
\Knoten{7.7}{0}{v15}
\Knoten{8.5}{0}{v16}
\Knoten{8.73}{0.23}{n5}
\Knoten{9.23}{0.73}{n6}
\Knoten{9.46}{0.96}{v17}
\Knoten{9.96}{1.46}{v18}
\Knoten{10.2}{1.7}{t1}\nput{0}{t1}{$t_1$}
\Knoten{9}{-0.5}{v19}
\Knoten{9.6}{-1.1}{v20}
\Knoten{10.2}{-1.7}{t2}\nput{0}{t2}{$t_2$}
\ncline{s1}{v1}
\Kante{v1}{v2}{\beta '}
\ncline{v2}{v5}
\ncline{s2}{v3}
\ncline{v4}{v5}
\ncline{v5}{n1}
\Kante{n1}{n2}{\tau}
\ncline{n2}{v6}
\Kante{v6}{v7}{\mu'}
\Kante{v3}{v4}{\bar{\mu}}
\ncline{v7}{v8}
\Kante{v8}{v9}{\bar{\nu}}
\ncline{v9}{v10}
\Kante{v10}{v11}{\mu}
\ncline{v11}{n3}
\Kante{n3}{n4}{\sigma}
\ncline{n4}{v12}
\Kante{v12}{v13}{\beta}
\Kante{v13}{v14}{\gamma'}
\Kante{v14}{v15}{\alpha}
\ncline{v15}{v16}
\ncline{v16}{n5}
\Kante{n5}{n6}{e}
\ncline{n6}{v17}
\Kante{v17}{v18}{\gamma}
\ncline{v18}{t1}
\ncline{v16}{v19}
\Kante{v19}{v20}{\nu}
\Kante{v20}{t2}{\nu'}

\Bogendashed{v1}{v14}{q_1'}{25}
\Bogendashed{v12}{v18}{q_1}{30}
\Bogendashed{v10}{v20}{q_2}{-25}
\Bogendashed{v6}{t2}{q_2'}{-25}
\Bogendashed{v3}{v9}{\bar{q_2}}{-25}
\end{pspicture}
\caption{First case of NBC2 in \autoref{lemmapbc14}  ($\mu \leq \bar{\nu}\leq \sigma-1$ also possible)}\label{S15}
\end{figure}
Analyzing the different types yields the following contradictions:
\vspace{-0.5em}
\begin{table}[H] 
\begin{center}
\begin{tabular}{l c@{\hspace{0.5cm}} l}
\toprule
Type of NBC & & contradiction \\ \midrule
3d-g, 4d-g, 7, 8, 11, 12 & & not possible since $q_1'$ is small \\ \addlinespace[0.2em]
1, 5, 6 & & Player $1$ pays $e_{\sigma}$  \\ \addlinespace[0.2em]
2, 3b, 3c, 4a & & cheaper Steiner forest ($\bar{q}_1$ and a tight alternative for Player 2) \\ \addlinespace[0.2em]
3a, 9, 10 & & Player $2$ pays $e_{\alpha}$ \\ \addlinespace[0.2em]
4b & & cost of $e_{\tau}$ is zero \\ \addlinespace[0.2em]
4c & & Player $2$ pays $e_{\tau}$ \\
\bottomrule
\end{tabular}
\end{center} \end{table}
\vspace{-1.5em}
%

Therefore $\bar{\nu}\geq\sigma$ has to hold. 
Since the cost shares are maximized for Player 2, the following changes of cost shares (by a suitably small amount) cannot yield a feasible assignment, since the sum of cost shares of Player 2 would be higher than before:
\[
\begin{array}{l}
	 \text{Player 2 increases on $e_{\alpha}$ and $e_{\tau}$ and decreases on $e_{\sigma},$} \\
	 \text{Player 1 increases on $e_{\sigma}$ and decreases on $e_{\alpha}$ and $e_{\tau}.$}
\end{array}
\] 
It is clear that Player 1 can increase on $e_{\sigma}$ and decrease on $e_{\alpha}$ and $e_{\tau}$ since every alternative that substitutes $e_{\sigma}$ also substitutes $e_{\alpha}$ or $e_{\tau}$.
But the changes are also feasible for Player 2, thus we get a contradiction:
It is clear that Player 2 can decrease on $e_{\sigma}$ and increase on $e_{\tau}$ since $\bar{q_2}$ is the smallest left alternative for $e_{\tau}$. Additionally increasing on $e_{\alpha}$ is also possible: There can not be a right alternative which substitutes $e_{\alpha}$, but not $e_{\sigma}$ (note that $q_2$ also substitutes $e_{\sigma}$); there can not be a right alternative which substitutes $e_{\tau}$ because $q_2'$ is the largest right alternative and a left alternative which substitutes $e_{\alpha}$ would lead to a cheaper Steiner forest (together with $\bar{q_1}$).
Overall this shows that a NBC2 in the original situation of this lemma also yields a contradiction.
\end{proof}

\begin{lemma} \label{lemmarest}
In the Subcases \hyperlink{pbc2}{PBC2}, \hyperlink{pbc3}{PBC3}, \hyperlink{pbc4}{PBC4}, \hyperlink{pbc5}{PBC5}, \hyperlink{pbc6}{PBC6} (in the case that $q_1$ is small), \hyperlink{pbc7}{PBC7}, \hyperlink{pbc9}{PBC9}, \hyperlink{pbc10}{PBC10}, \hyperlink{pbc11}{PBC11}, \hyperlink{pbc12}{PBC12} (in the case that $q_2$ is big), \hyperlink{pbc13}{PBC13} and \hyperlink{pbc14}{PBC14} we get contradictions for all possible types of NBCs.
\end{lemma}

The proofs of these cases are left to the reader, where one has to use the properties discussed in the proof of \autoref{lemma3} to get contradictions to the properties for all possible types of NBCs discussed in \autoref{lemmanbcanfang} to \autoref{lemmanbcende}. These proofs are similar to the proofs of \autoref{lemmapbc1} to \autoref{lemmapbc14}.

\subsubsection{Procedure for computing \ptl-cost shares}\label{subsubsec:push}
\begin{algorithm}[H]
\caption{\textsc{PushLeft}} \label{push}
\KwData{Steiner forest $F$ and an optimal solution $(\xi_{i,e})_{i \in \{1,2\}, e \in P_i}$ for (LP($F$)).}
\KwResult{Transformed optimal solution $(\xi_{i,e})_{i \in \{1,2\}, e \in P_i}$ for (LP($F$)).}
Choose one player, w.l.o.g. Player 1. 

Let $P_1= e_{1}, \ldots, e_{\ell_1}, e_{\ell_1+\ell_2+1}, \ldots, e_{\ell_1+\ell_2+m+r_1}$ the (directed) path of Player 1 in $F$ from her source $s_1$ to her sink $t_1$ and $\mathcal A_1$ the ordered set of indices of edges in $P_1$. 

		\If{$P_1 \cap P_2 \neq \emptyset$}{

			
	%
			
			\If{\upshape considering $P_2$ as directed from $s_2$ to $t_2$ induces the opposite direction on $P_1 \cap P_2$ (compared to the direction of $P_1$)}{
		
				swap $s_2$ and $t_2$
				}
			Let $P_2=e_{\ell_1+1}, \ldots, e_{\ell_1+\ell_2+m}, e_{\ell_1+\ell_2+m+r_1+1} \ldots, e_{\ell_1+\ell_2+m+r_1+r_2}$ the (directed) path of Player $2$ in $F$ from her source $s_2$ to her sink $t_2$ and $\mathcal A_2$ the ordered set of indices of edges in $P_2$.
			}

change $\leftarrow$ \texttt{true}

\While{\upshape change $=$ \texttt{true}}{ 

	change $\leftarrow$ \texttt{false} 

	\For{$i=1,2$}{

		\For{$\alpha \in \mathcal A_i$}{ 

			$E_{\alpha}\leftarrow \{e_{\delta} : \delta \in \mathcal A_i, \delta \leq \alpha\}$

			\textbf{solve}  
		 \vspace{-1em}
		\begin{align*}
		\max \ \ & \sigma_{i,e_{\alpha}} \\
		s.t. \ \ & \sum_{e \in E_{\alpha}}{\sigma_{i,e}}=\sum_{e \in E_{\alpha}}{\xi_{i,e}} \\
		& \sum_{j \in N\setminus i}{\xi_{j,e}}+\sigma_{i,e} \leq c(e) & \forall  e \in E_{\alpha} \\
		& \sum_{e \in P_i \setminus P_i' \setminus E_{\alpha}}{\xi_{i,e}} + \sum_{e \in (P_i \setminus P_i') \cap E_{\alpha}}{\sigma_{i,e}} \leq \sum_{e \in P_i' \setminus P_i }{c(e)} & \forall P_i' \in \mathcal P_i \\
		& \sigma_{i,e} \geq 0 & \forall e \in E_{\alpha}
		\end{align*}
		
		\If{$\sigma_{i, e_{\alpha}} \neq \xi_{i, e_{\alpha}}$}{
		
			change $\leftarrow$ \texttt{true}
		}
		
		\For{$e \in E_{\alpha}$}{
		
			$\xi_{i,e} \leftarrow \sigma_{i,e}$
		}
		}
	}
}
\end{algorithm} 

 \newpage
\bibliography{master-bib}   
\bibliographystyle{abbrv}
\end{document}